\documentclass[a4paper,twocolumn,11pt,accepted=2026-01-11]{quantumarticle}
\pdfoutput=1
\usepackage{hyperref}
\usepackage[utf8]{inputenc}
\usepackage[english]{babel}
\usepackage[T1]{fontenc}
\usepackage[numbers]{natbib}
\usepackage{graphicx}
\usepackage{caption}
\usepackage{subfig}

\usepackage{amsmath,amssymb}
\usepackage{physics} 
\usepackage{dsfont} 
\usepackage{bm}
\usepackage{mathrsfs} 
\usepackage{mathtools} 
\usepackage{amsthm} 

\newcommand{\Hilb}[1]{\mathcal{H}^{#1}}
\newcommand{\Alg}[2]{\mathscr{#1}_{#2}}
\newcommand{\Span}[1]{\text{Span}\left\{ #1 \right\}}

\newcommand{\LinOp}[1]{\mathcal{L}\left( #1 \right)}

\newcommand{\LinOpB}[2]{\mathcal{L}\left( #1 ,  #2 \right)}

\newcommand{\MapX}[2]{ \frac{\mathds{1}^{#1}}{d_{#1}} \otimes \Tr_{#1}\left[ #2 \right] }
\newcommand{\TrX}[2]{\mathrm{Tr}_{#1}\left[ #2 \right]}

\newcommand{\InProd}[2]{\left\langle #1 \,,\, #2 \right\rangle}

\newcommand{\GGB}[2]{\sigma^{#1}_{#2}}

\newcommand{\Proj}[2]{\mathcal{P}^{#1}_{#2}}
\newcommand{\ProjOn}[3]{\mathcal{P}^{#1}_{#2}\left\{#3\right\}}

\newcommand{\CompAlg}[1]{\overline{\mathscr{#1}}}
\newcommand{\CompProj}[2]{\overline{\mathcal{P}}^{#1}_{#2}}
\newcommand{\CompProjOn}[3]{\overline{\mathcal{P}}^{#1}_{#2}\left\{ #3 \right\}}

\newcommand{\TProj}[2]{\widetilde{\mathcal{P}}^{#1}_{#2}}

\newcommand{\MOn}[1]{\mathcal{M}\left( #1 \right)}
\newcommand{\MapOn}[3]{\mathcal{M}^{#1}_{#2}\left( #3 \right)}

\newcommand{\ChanScr}[2]{\mathscr{#1} \rightarrow \mathscr{#2}}

\newcommand{\parr}{%
\mathrel{\rotatebox[origin=c]{180}{$\&$}}}
\newtheorem{theo}{Theorem}
\newtheorem{defi}{Definition}
\newtheorem{lemm}{Lemma}
\newtheorem{prop}{Proposition}
\newtheorem{coro}{Corollary}

\begin{document}

\title{Projective characterization of higher-order quantum transformations}

\author{Timothée Hoffreumon}
 \email{t.hoffreumon@gmail.com}
 \thanks{currently at the Mathematical Institute of the Slovak Academy of Sciences, Bratislava, Slovakia.}
\author{Ognyan Oreshkov}
\affiliation{%
 Centre for Quantum Information and Communication (QuIC), {\'E}cole polytechnique de Bruxelles, CP 165/59,\\ Universit{\'e} libre de Bruxelles, Avenue F. D. Roosevelt 50, 1050 Brussels, Belgium.
}%

\begin{abstract}
Transformations of transformations, also called higher-order transformations, is a natural concept in information processing, which has recently attracted significant interest in the study of quantum causal relations. In this work, a framework for characterizing higher-order quantum transformations which relies on the use of superoperator projectors is presented.
More precisely, working with projectors in the Choi-Jamio{\l}kowski picture is shown to provide a handy way of defining the characterization constraints on any class of higher-order transformations. The algebraic properties of these projectors are furthermore shown to obey rules similar to \textit{multiplicative additive linear logic (MALL)}, providing an intuitive way of comparing any two classes through their projectors.
The main novelty of this work is the introduction to the algebra of the `prec' connector. It is used for the characterization of maps that are no signaling from input to output or the other way around. 
This allows to assess the possible signaling structure of any transformation characterized within the projective framework.
The properties of the prec are moreover shown to yield a normal form for projective expressions. This hints towards a general way to compare different classes of higher-order transformations.
\end{abstract}

\maketitle

\section{Introduction\label{sec:intro}}
The same way that a quantum channel describes the most general transformation mapping an input quantum state to an output quantum state \cite{Kraus1983}, a \textit{quantum supermap} describes the most general transformation mapping an input quantum channel to an output quantum channel \cite{Chiribella2008}.
Interpreting the quantum channel as a transformation between states, the supermap is then a transformation of transformations. For that reason, it is called \textit{higher-order} transformation. Since nothing forbids \textit{a priori} to nest transformations of transformations, one can consider successive nestings to recursively build whole hierarchies of higher-order transformations \cite{Chiribella2009,Perinotti2016,Bisio2018}. 

Fragments of a quantum circuit are a concrete instance of the use of a higher-order hierarchy. A fragment of quantum circuit that `goes around' a channel is a supermap: it takes a channel as input and outputs a channel; it is a `second-order object' with respect to seeing the channels as the `first-order objects'. This supermap, albeit composed of several interconnected channels, can in turn be seen as a single fragment around which a larger circuit can go. Doing so defines a `third-order object' that can be called a super-supermap: it takes a supermap as input and outputs a channel. And so on. This ensuing hierarchy has been defined under the name \textit{quantum comb formalism} \cite{Chiribella2009}, which has proven to be a valuable tool in the field of quantum information theory.

With a different goal than modeling circuit fragments, supermaps with multiple inputs were subsequently studied. First, the \textit{quantum switch} was proposed as a supermap that takes two channels and outputs them in an order that depends on a control qubit \cite{Chiribella2013}. Soon after, \textit{Process Matrices} (PM) were proposed as a general framework of supermaps that take a fixed number of quantum instruments \cite{Davies1970} and map them to a joint probability for their outcomes \cite{OCB2012}. (In the case when the inputs are channels, that probability is 1.)
Both concepts led to the identification of Indefinite Causal Order (ICO) as a feature of supermaps.

One can then wonder what differentiates the switch from a comb, or the PM from a comb. In particular, why do certain maps and hierarchies of maps feature ICO while others do not?
Motivated by these considerations, the goal of this work is to present a framework that formalizes and characterizes higher-order quantum transformations. The two main questions answered are `Given an operator on a set of input and output Hilbert spaces, does it represent (the Choi-Jamio{\l}kowski operator of) a higher order object?' and `What is (are) the underlying causal structure(s) of such an object?'. This work extends two previous characterizations, one done using type theory \cite{Perinotti2016,Bisio2018}, and the other using category theory \cite{Kissinger_2019}.
This extension relies on the use of superoperator projectors \cite{Araujo2015,MPM,Milz_2022,Milz_2024}. 

These projectors have a twofold advantage: first, they make the characterization more straightforward, as one can answer the first question simply by applying the projector corresponding to a given higher-order object on the operator. Second, they offer an intuitive explanation of the type-theoretic semantics of higher order: the algebraic rules for composing these projectors correspond to the semantic rules for forming new types.

The paper and its findings are organized as follows: 
First, the type-theoretic framework of Ref. \cite{Perinotti2016,Bisio2018} is reviewed in Section \ref{sec:types}. 
In Section \ref{sec:Proj_char}, types are put in correspondence with subsets of operators called \textit{state structures} that are characterized by a superoperator projector. Looking at these projectors instead of the state structures, and, by extension, the types they define, is the idea behind the projective characterization. The type connectives $\{\overline{\:\cdot\:},\otimes,\rightarrow\}$, which respectively correspond to the notions of discarding, parallel composition, and evolution of states, are shown to correspond to rules on how to derive new projectors from the ones characterizing states. It is then proven that this way of characterizing higher-order transformations recovers the type theory.
In Section \ref{sec:Projos_alg}, it is shown that the convenient aspect of expressing the type connectives as operations on projectors is that two new operations are obtained `for free': the intersection $\cap$ and union $\cup$ of projectors.  
This lifts the type system to a particular algebra over $\mathbb{C}$, which is shown in App. \ref{sec:projo} to be a Boolean lattice $\{\overline{\:\cdot\:},\cap,\cup\}$, equipped with two bilinear compositions $\{\otimes,\rightarrow\}$. Hence, the characterization of higher-order transformations is simplified into manipulations of formulae under Boolean-logic-like rules. In particular, two state structures are equivalent up to normalization if their projectors can be shown equivalent using the rewrite rules of the algebra. 
It is also observed in this section that the algebra of projectors, when interpreted as a logic, shares similarities with linear logic \cite{GIRARD1987}. 
In Section \ref{sec:NS}, the signaling structure of compositions is studied: a third composition of projectors, the `prec' $\prec$, is introduced in the algebra. It encodes the composition of two subsystems in a way that allows signaling in a single fixed direction. It is then shown that the other two, $\otimes$ and $\rightarrow$, correspond to, respectively, no signaling and two-way signaling compositions. 
In addition to its obvious role of making the underlying signaling structure of a set of transformations apparent, the prec has useful algebraic properties: as presented in Section \ref{sec:applications_NSandrelations}, these allow one to rewrite projectors into a normal form, which proves useful for comparing state structures.
Finally, in Section \ref{sec:applications_iso}, an example of the use of the projective characterization is given: we recover the result \cite[Prop. 6]{Bisio2018} that quantum combs are quantum networks, which have a single signaling direction.

Three additional examples of the use of the projective characterization are provided in the appendices. First, the characterizations of POVM, channel, and bipartite process matrix formalisms are presented in App. \ref{sec:examples_QT} as an example of what can be done with the results of Sec. \ref{sec:Proj_char}. Next, a channel theory based on a post-selected state structure is presented in App. \ref{sec:examples_biased_QT} as an example of what can be done in generality with the algebra, according to the results of Sec. \ref{sec:Projos_alg}. Lastly, in App. \ref{sec:examples_dynamics_constr}, the repeated nesting of channels (map, supermap, and super-supermap) is shown to result in ICO as an illustration of the results presented in Sec. \ref{sec:NS} and \ref{sec:applications}.

\section{Types\label{sec:types}}
The formalism of quantum channels and operations \cite{Kraus1983} describes transformations of quantum states into quantum states, but is unable to effectively describe the transformations whose inputs and outputs are transformations themselves \cite{Chiribella2013}. To overcome this issue, quantum supermaps were introduced \cite{Chiribella2008,Chiribella2009} and then Perinotti \cite{Perinotti2016} and Bisio \cite{Bisio2018} extended the concept into an axiomatic framework to classify any transformations of transformations, or \textit{higher-order quantum maps}. 

At the core of their work is the utilization of the Choi-Jamio{\l}kowski (CJ) isomorphism \cite{Jamiolkowski1972,Choi1975}. Define $\Hilb{A}$ to be a Hilbert space of finite dimension $d_A$ associated to system $A$, $\LinOp{\Hilb{A}}$ to be the space of linear operators on $\Hilb{A}$, and $\LinOpB{\LinOp{\Hilb{A}}}{\LinOp{\Hilb{B}}}$ to be the space of linear maps from $\LinOp{\Hilb{A}}$ to $\LinOp{\Hilb{B}}$. 
Let $\ket{\phi^+} = \sum_{i=0}^{d_A-1} \ket{i}^{A'}\otimes \ket{i}^{A}$ be a(n unnormalized) maximally entangled state on space $\Hilb{A'}\otimes\Hilb{A}$, with $A'$ a copy of $A$, and $\phi^+ \equiv \dyad{\phi^+}$ its density operator representation. Then,
\begin{equation}\label{eq:CJ}
    \begin{gathered}
        \mathcal{M} \in \LinOp{\LinOp{\Hilb{A}},\LinOp{\Hilb{B}}},  \\ \mathcal{M} \mapsto M \in \LinOp{\Hilb{A'} \otimes \Hilb{B}}:\\
        M \equiv \left[\left( \mathcal{I} \otimes \mathcal{M} \right)\left\{\phi^+\right\}\right]^T \:,
    \end{gathered}
\end{equation}
is an isomorphic bijective mapping. The correspondence sends linear maps between operator spaces to operators in tensor product spaces. It has the properties that a Hermitian-preserving (HP) map is mapped to a Hermitian operator and a completely positive (CP) map is mapped to a positive semi-definite operator. To recover the action of the map, the `reverse direction' of the CJ correspondence is used:
\begin{equation}\label{eq:CJ^-1}
    \mathcal{M}\left(V_A\right) = \left( \TrX{A}{M_{AB} \: \left(V_A \otimes \mathds{1}_B\right) }\right)^T \:,
\end{equation}
where $V_A \in \LinOp{\Hilb{A}}$ is an arbitrary operator, and $\mathds{1}_B$ the identity operator in $\LinOp{\Hilb{B}}$. (Capital letters $V$, $N$, and $U$ will be reserved to denote operators on single-partite space, while $W$ and $M$ will be for operators on bipartite ones. When needed, subscript letters will be used to clarify to which tensor factor an operator is associated.)

Using this correspondence, maps, maps of maps, and so on, can all be represented as sets of operators on some composite space. These sets resemble formally a constrained version of the set of quantum states on the composite space, and this is how one can compare them: by characterizing the constraints associated with a given set of supermaps. 

Then it is possible to define \textbf{types} of maps by associating to each of these constrained state spaces a given \textit{type} \cite{Perinotti2016,Bisio2018}. One first defines the base types $A,B,C,...$ associated with given systems $A,B,C,...$ upon which some parties A(lice), B(ob), C(harlie),... can act\footnote{Note that Refs. \cite{Perinotti2016,Bisio2018} use the terminology `elementary type' instead of `base type', reserve capital letters to base types, and refer to the `constrained state spaces' associated with a given type $A$ as the `set of deterministic events of type $A$', which is noted $\mathrm{T}_1(A)$.}. A base type, therefore, indicates the state space (in density matrix form) of a given finite-dimensional quantum system as well as its label\footnote{Remark that Ref. \cite{Kissinger_2019} provides a more general construction by allowing the state space of base types to be different than the one of quantum systems.}. (The base types will bear the same label as the system they are associated with. To avoid complicated notation when we describe types instantiated on different systems, we treat them as different, even though two systems $A$ and $B$ could have isomorphic state spaces. If needed, we will indicate by words when two types are isomorphic.)

Out of the base types, one can define transformations as a new type constructed from the base types. For example, if the operator $M$ in Eqs. \eqref{eq:CJ} and \eqref{eq:CJ^-1} above was the CJ representation of a quantum channel, it would transform every element of base type $A$ into an element of base type $B$, so its type would be `of a transformation from type $A$ to type $B$', noted as `type $A\rightarrow B$'. Other types can be defined by repeating this construction, yielding the hierarchy of higher-order transformations: if types $A$ and $B$ are valid types (not necessarily base types), then $A \rightarrow B$ is a valid type associated with the transformations having an input of type $A$ and an output of type $B$. The only required notion for defining the hierarchy is that of an \textit{admissible map}, which axiomatizes how the transformations are defined. That is, how (the CJ representations of) the admissible maps between any two types $A$ and $B$ are built from the knowledge of the sets of admissible\footnote{The admissibility axioms in Ref. \cite{Bisio2018} essentially generalize the definition of quantum channels by requiring the transformations to be linear maps that preserve the concepts corresponding to complete positivity and trace-preservation for type $A$ and type $B$. The notion of an admissible map will be formally restated in our formalism under the name of a structure-preserving map, Def. \ref{def:struc_pres}.} maps forming types $A$ and $B$.

The formalism is an instance of a \textit{type system}. This ``$\rightarrow$'' connector, here nicknamed \textbf{transformation}, is the key element of the type theory of higher order transformations: each set can be seen as an abstract type, and new types can be defined out of existing ones using the transformation connector as a semantic rule. Alongside the base types and the transformation, three other symbols must be added to the \textit{alphabet} for writing types: the first one is the trivial type, noted ``$1$''. This is the type associated with a one-dimensional Hilbert space $\mathbb{C}$ whose associated set is made of the number 1, thus corresponding to the absence of a system. The other two symbols are parentheses, ``$($'' and ``$)$'', needed because $\rightarrow$ is not associative. Indeed, $A \rightarrow (B \rightarrow C) \neq (A \rightarrow B) \rightarrow C$ as the first type transforms a `state' of type $A$ into a `channel' of $B \rightarrow C$, whereas the second transforms a `channel' of $A \rightarrow B$ into a `state' of $C$.

\begin{figure}[htb]
    \centering
     \subfloat[An element of type $A$ is represented as a bottom half-circle (left). An element of type $\overline{A}\equiv A \rightarrow 1$ is a top half-circle (center). The closed circuit means that every element of type $\overline{A}$ takes each element of type $A$ to the number 1 (right).\label{fig:basics_nA}]{
         \includegraphics[width=.9\linewidth]{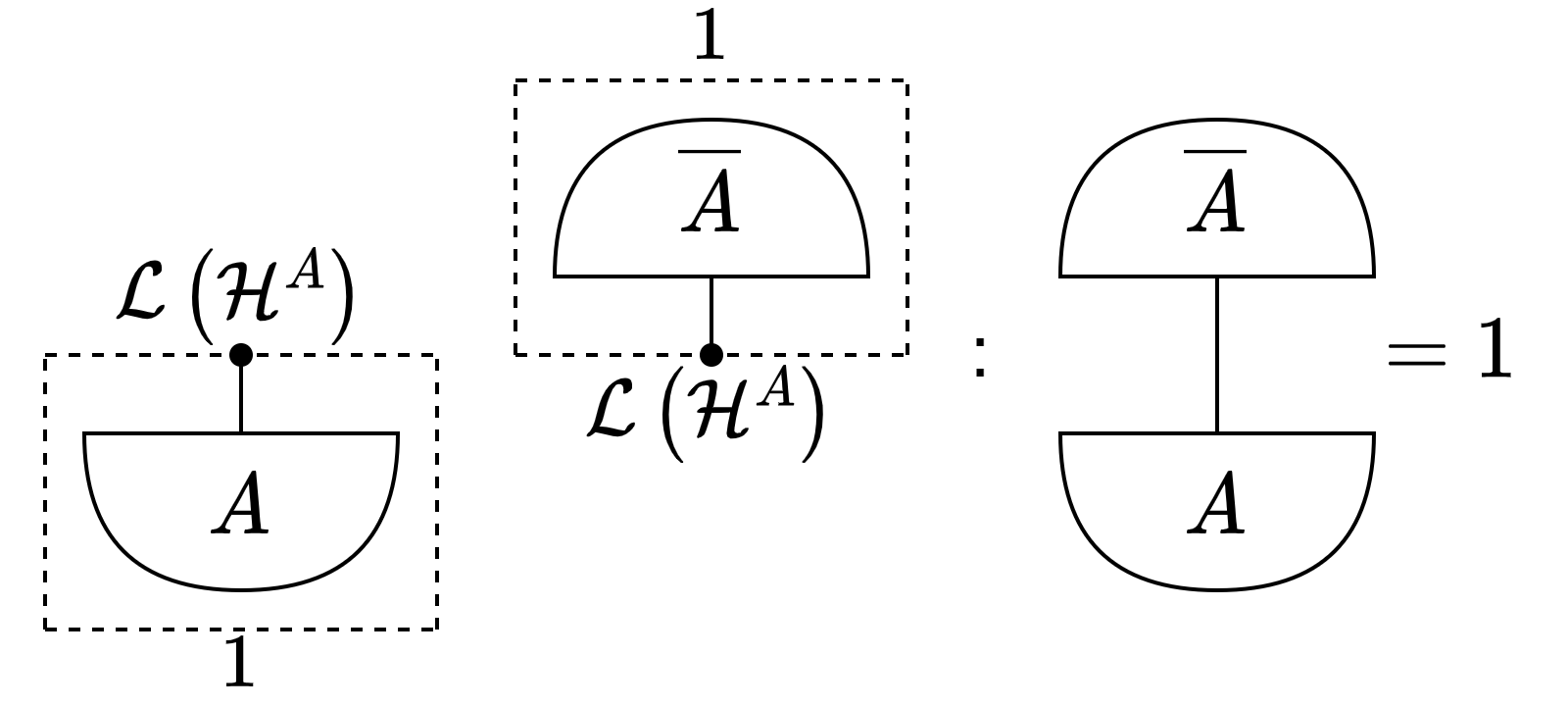}
     }\\
     \subfloat[An element of a type $A \rightarrow B$ is represented as a box. The half-closed circuit means that every element of type $A\rightarrow B$ transforms each element of $A$ into one of $B$. \label{fig:basics_A_to_B}]{
        \includegraphics[width=.9\linewidth]{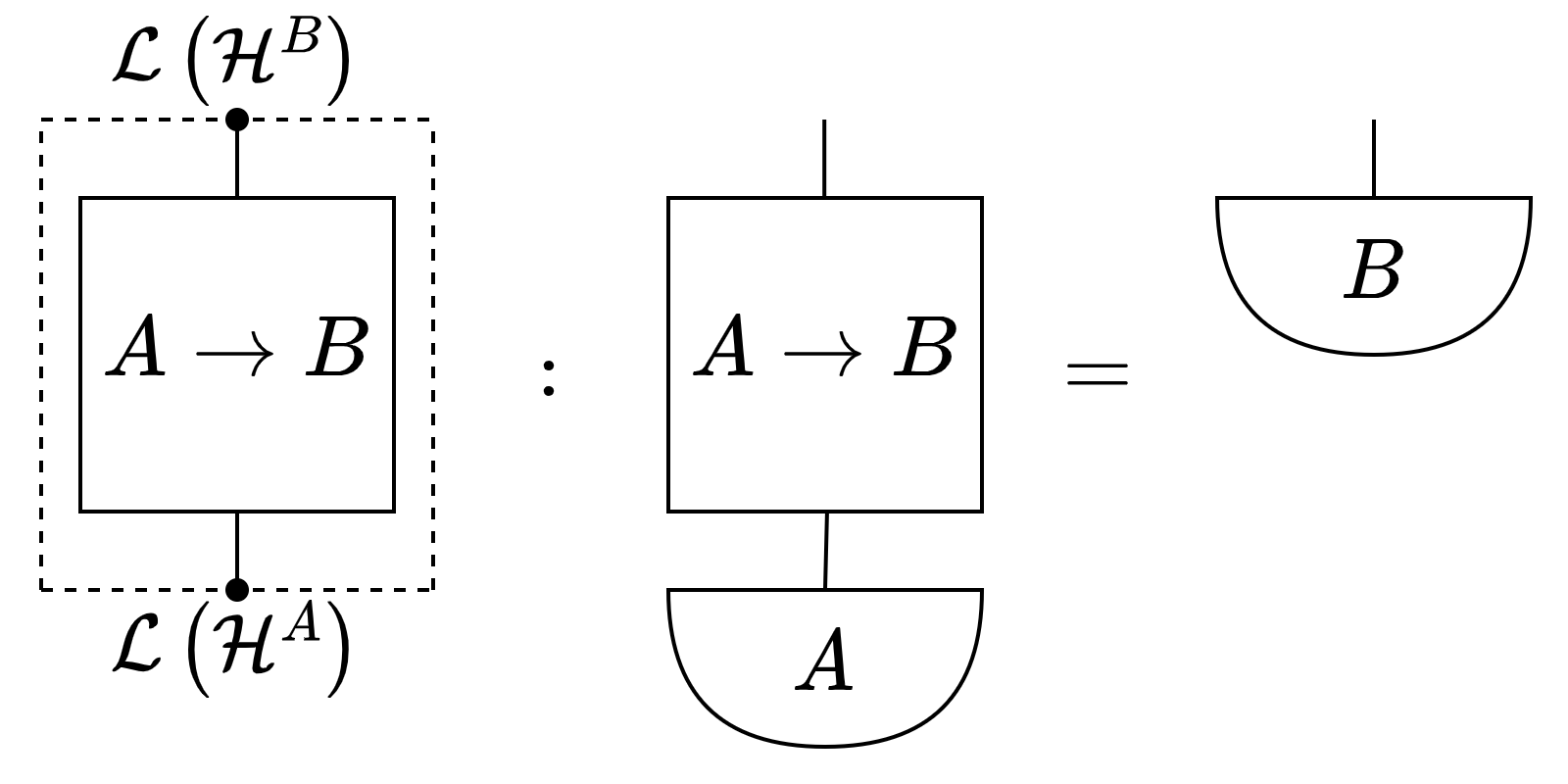}
     }\\
     \subfloat[An element of type $(A \rightarrow (B \rightarrow 1)) \rightarrow 1 = \overline{A \rightarrow \overline{B}}$ is defined as the parallel composition $A \otimes B$ of two types, so it is represented as its two constituents drawn side by side.\label{fig:basics_A_otimes_B}]{
        \includegraphics[width=.95\linewidth]{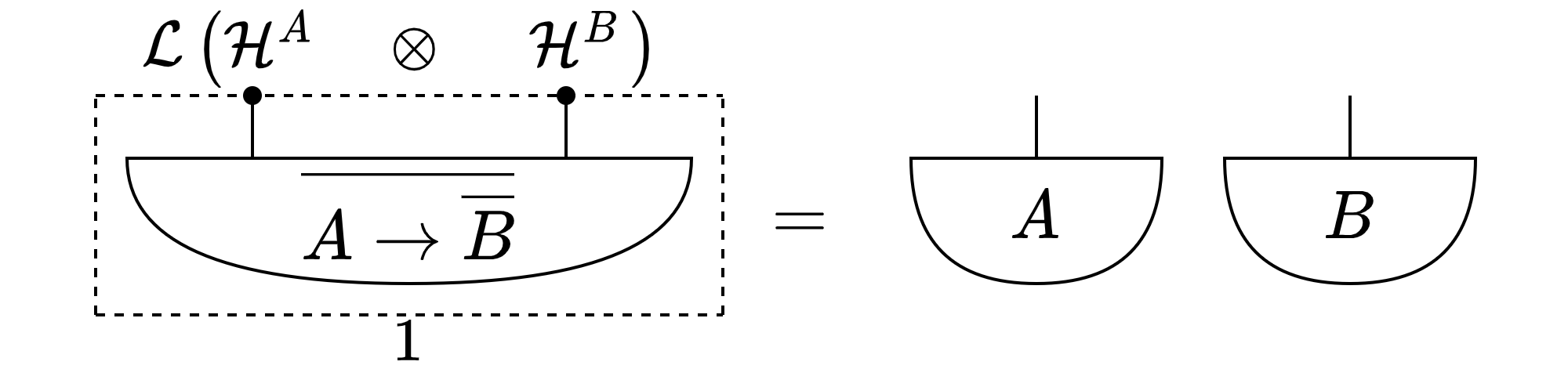}
    }
        \caption[Picturing types/state structures]{Picturing types graphically. 
        In the following, the notion of a type like $A$ will be subsumed by the one of a state structure $\Alg{A}{}$, but the diagrammatic notation will remain the same for simplicity.
        \label{fig:basics}}
\end{figure}
Some graphical notation will be used to represent the types, refer to Fig. \ref{fig:basics} in which diagrams are read from bottom to top: open wires are Hilbert spaces of dimension $>1$, and an absence of wire is a space of dimension one (i.e., a trivial system with trivial type $1$). 
A box between two open wires then represents an admissible map (a transformation type). For example, in the left part of subfigure \ref{fig:basics_A_to_B}, the box within dashed lines can be seen as (the CJ representation of) any map $M$ of type $A \rightarrow B$; it connects an `input' space (in this case, $\LinOp{\Hilb{A}}$) over which the operators of a certain type are defined (in this case, $A$) to an `output' space over which the operators of another type are defined (in this case, the output space is $\LinOp{\Hilb{B}}$ and the type is $B$). Note that the output type is possibly isomorphic to the input (meaning they are associated with the same set of operators) but they are defined on different spaces nonetheless. 

Every type can actually be seen as a transformation from the trivial type to itself, reflecting that any operator can be seen as a linear map from the base field to itself. In the type system, it means that the relation $A = 1 \rightarrow A$ holds. Graphically, it means that a bottom half-circle marked with the letter $A$ is any operator of type $A$. 
For example, the left part of subfigure \ref{fig:basics_nA} can be seen as any operator $V \in \LinOp{\Hilb{A}}$ of type $A$ but also as a transformation of type $1 \rightarrow A$, whence the absence of bottom wire. 

Two special ways of building new types have their own shorthand notation\footnote{Such inessential definitions introduced to make the syntax easier to interpret are often called `syntactic sugars' in computer sciences.} as well as graphical depiction because of their interpretation as sets of operators: the `\textit{functionals} on a type $A$', noted $\overline{A} \equiv A \rightarrow 1$, and the `\textit{tensor product} of types $A$ and $B$', noted $A \otimes B \equiv (A \rightarrow (B \rightarrow 1)) \rightarrow 1$. (Why this latter definition corresponds to a notion of tensor product is explained in Ref. \cite{Perinotti2016}; we shall revisit it in this article and explain it using the projector algebra.)

The functional generalizes the notion of discarding to higher-order maps: its type sends a type $A$ to the trivial type, meaning it sends the generalization of a state space to the single-element set $\{1\}$. In other words, $\overline{A}$ is the type of the \textit{effects} on $A$, mapping elements of type $A$ to the number 1. 

Such types are represented by a top half-circle, which `closes a wire' and therefore depicts the pairing of a state with an effect to produce a number. This pairing is represented in $\LinOp{\Hilb{A}}$ via the inner product so the central part of subfigure \ref{fig:basics_nA} can be seen as $\TrX{A}{N^\dag \: \cdot \: }$ where $N \in \LinOp{\Hilb{A}}$ is of type $A \rightarrow 1 = \overline{A}$. 
We take as a convention to represent negated types (i.e. with a bar on top of it) as such so they are interpreted as the functional $\TrX{A}{N^\dag \: \cdot \:}$ instead of the operator $N$. This is closer in spirit to how the type reads as `taking an element of type $A$ to the number 1'. The right part of the subfigure is then read as the equation $\TrX{A}{N^\dag \: V} = 1$ for all $V$ of type $A$ and $N$ of type $\overline{A}$. 
More generally, any closed wire can be seen as applying the reverse direction of the CJ isomorphism and any equation in graphical representation is true for all operators of the types on the left. For example, the right part of subfigure \ref{fig:basics_A_to_B} is read as the equation $(\TrX{A}{M_{AB}\cdot (V_A \otimes \mathds{1}_B)})^T = U_B$ for all $M$ and $V$ of types $A\rightarrow B$ and $A$, respectively, and for some $U$ of type $B$ (the choice of $U$ may depend on the choices of $M$ and $V$).

Notice that, as operations on types, the `bar' does not commute with the `arrow' i.e., $\overline{A \rightarrow B} \neq \overline{A} \rightarrow \overline{B}$. Indeed, discarding a channel is a different operation from transforming an effect into another one. However, the type $\overline{A} \rightarrow \overline{B}$ of transformations between effects is equivalent to the type $A \leftarrow B$ of reversed transformations between states \cite{Bisio2018}. (The notation $A \leftarrow B \equiv B \rightarrow A$ will often be used to keep the factors ordered alphabetically. In general, all equivalences of types are up to a reordering of the tensor factors.) This logic-like relation\footnote{If the bar is seen as negation and the arrow as implication, this equivalence is indeed the well-known `B implies A if and only if not-A implies not-B'.} encodes the fact that the action of a channel from state in $\LinOp{\Hilb{B}}$ to a state in $\LinOp{\Hilb{A}}$ is equivalent to the action of its adjoint from an effect (represented) in $\LinOp{\Hilb{A}}$ to one in $\LinOp{\Hilb{B}}$. 
This relation further implies that the functionals on functionals are equal to the base type, $\overline{\overline{A}} = A$, i.e. that the `bar' is an involution, since $\overline{A} \rightarrow 1 = A \leftarrow 1$.

The tensor product generalizes the notion of parallel composition to higher-order maps: it takes two types $A$ and $B$ and fuses them into a new bipartite type $A\otimes B$ behaving accordingly. 
Indeed, $A \rightarrow \overline{B}$ is naively a type that takes an input of a type $A$ and outputs one of type $\overline{B}$. However, this type can also be read in the other direction as $A \rightarrow \overline{B} = \overline{A} \leftarrow B$, so it is equivalently interpreted as a type having an input of type $B$. Hence, the functionals on such a type, i.e. of type $\overline{A \rightarrow \overline{B}}$, must be composed of a subsystem of type $A$ and one of type $B$, since both are valid inputs to $A \rightarrow \overline{B}$. In addition, these subsystems must be somewhat independent of each other as each can be interpreted as the first to be inputted into the objects of type $A \rightarrow \overline{B}$; the choice of an element of $A$ should not influence the choice of one in $B$ and vice-versa. Therefore, type $\overline{A \rightarrow \overline{B}}$ is a way to combine types $A$ and $B$ in a `parallel' fashion, motivating the notation $\overline{A \rightarrow \overline{B}} \equiv A \otimes B$. It can then be shown that the set of operators associated with $A\otimes B$ is indeed related to the set spanned by the tensor products of operators of type $A$ with those of type $B$\footnote{For now, this is only an intuitive explanation to motivate the definition of composition from Refs. \cite{Perinotti2016,Bisio2018} which was chosen because type $A \otimes B$ for $A$ and $B$ base types precisely corresponds to the state space of a composite quantum system in $\LinOp{\Hilb{A} \otimes \Hilb{B}}$. They then proceed to show that this definition of the tensor product for two arbitrary types results in a suitable notion of the tensor product of their associated state spaces, namely the intersection of the affine span of the tensor products with the cone of positive operators, see Sec. 5.e of Ref. \cite{Bisio2018}. We will properly restate this definition in Sec. \ref{sec:Proj_char_tensor}. In Sec. \ref{sec:NS_relations}, the intuition we presented in the text will be turned into a formal justification for choosing this definition of the tensor product of types; it will be shown to follow only from requirements on the signaling relations.}. 
For example, $A\otimes B$ corresponds to the type of all bipartite states in $\LinOp{\Hilb{A}\otimes\Hilb{B}}$ if $A$ and $B$ are base types, whereas $(A \rightarrow C) \otimes (B \rightarrow D)$ is the type of all no signaling channels \cite{Beckman2001,Piani2006} from $\LinOp{\Hilb{A}\otimes\Hilb{B}}$ to $\LinOp{\Hilb{C}\otimes\Hilb{D}}$; this is the set spanned by all tensor products of channels of type $A \rightarrow C$ with those of type $B \rightarrow D$. Contrastingly, $(A \otimes B) \rightarrow (C \otimes D)$ is the type of all bipartite channels from $\LinOp{\Hilb{A}\otimes\Hilb{B}}$ to $\LinOp{\Hilb{C}\otimes\Hilb{D}}$ as it transforms bipartite states of type $A \otimes B$ into ones of $C \otimes D$.

Multipartite types like $A \rightarrow B$ are accordingly represented by diagrams with multiple open wires -- one for each tensor factor of the underlying Hilbert space. The tensor product of two types is particular as it is represented as the diagrams for its constituents next to each other rather than a joint diagram, see subfigure \ref{fig:basics_A_otimes_B}. Be aware that this graphical representation does not imply that the set of operators of type $A \otimes B$ is separable. 
For example, a maximally entangled state in $\LinOp{\Hilb{A}\otimes \Hilb{B}}$ is of type $A \otimes B$ so it is one of the operators represented by the disjoint diagram on the right of the subfigure. 

As expected from a tensor product, this is a commutative and associative composition as $A \otimes B = B \otimes A$ (since $\overline{A \rightarrow\overline{B}} = \overline{B \rightarrow\overline{A}}$ up to a reordering of tensor factors) and $A\otimes (B \otimes C) = (A \otimes B) \otimes C$ (this stems from the `uncurrying' property of quantum (super)maps, $A \rightarrow (B \rightarrow C) = (A \otimes B)\rightarrow C$, which encodes the idea that transforming a state into a channel is equivalent to having a channel with two inputs\footnote{However, compared to $A\rightarrow B = \overline{A} \leftarrow \overline{B}$ which is a requirement that comes from the identification of a map with its adjoint, the uncurrying property $A \rightarrow (B \rightarrow C) = (A \otimes B)\rightarrow C$ is rather deduced from the definition of the tensor product; see Sec. 5.e of Ref. \cite{Bisio2018}.}). 

Thus, starting from some postulates, a trivial type 1, base types $A,B,C,...$, and the ``$\rightarrow$'' connector rule, all the higher-order generalizations of the quantum formalism can be defined using the type system. These in turn yield the constraints defining the set of operators representing the transformations of a given type \cite{Perinotti2016, Bisio2018}. 
For example, let types $A_0$ and $A_1$ be, respectively, the set of input and output quantum states of Alice, who applies some quantum operation in between (subsystems associated with the same party will be differentiated with a numeral index starting at 0, and so will be the tensor factors of the associated Hilbert space as well as the types). Then, one can infer that the set of allowed transformations to which she has access is of type $A_0 \rightarrow A_1$. This simple semantic statement is then translated into constraints to apply on the Hilbert space $\mathcal{L}\left(\mathcal{H}^{A_0} \otimes \mathcal{H}^{A_1} \right)$ and yields (the Choi-Jamio{\l}kowski representation of) the set of valid quantum channels for Alice. 
A more complex example is obtained by recovering the set of bipartite process matrices \cite{OCB2012}, whose type corresponds to a deterministic functional on the local quantum instruments of two parties, say Alice and Bob. Knowing that their local instruments sum up to quantum channels, \textit{i.e.} they belong to types $\left(A_0 \rightarrow A_1 \right)$ and $\left(B_0 \rightarrow B_1\right)$, the set of process matrices is the type that takes a composition of these two types as input and outputs a trivial system. In the semantics, this statement corresponds to forming type $\left(\left(A_0 \rightarrow A_1 \right)\otimes \left(B_0 \rightarrow B_1\right)\right) \rightarrow 1$, from which the constraints for the characterization directly ensue. According to what has been said so far, the type of bipartite process matrices rewritten as $\overline{\left(A_0 \rightarrow A_1 \right)\otimes \left(B_0 \rightarrow B_1\right)}$ informs us that that they can be seen as the set of functionals (indicated by the bar notation) on the no signaling channels shared between Alice and Bob (of type $\left(A_0 \rightarrow A_1 \right)\otimes \left(B_0 \rightarrow B_1\right)$).

In the following sections, we will develop new characterization methods that improve three aspects of the type system: first, it provides a way to obtain the characterization constraints directly from a type formula. Second, it allows to compare any type of higher-order transformations. 
This latter feat is key to understanding how ICO arises at higher-orders.
Third, it expands the type semantics so to add another connector to represent the type of transformations with a fixed signaling direction. Not only does this also help understanding ICO, but this connector can moreover be used to read the signaling relations that may feature any given kind of higher order directly from its type.

\section{Projective characterization of the type theory of higher-order quantum transformations\label{sec:Proj_char}}

We start with an example. In quantum theory, the most general representation of a destructive measurement on a state $\rho \in \LinOp{\Hilb{}}$ is done by a POVM \cite{Nielsen2009}: a collection of positive operators $\{E_i\}$, the effects, \textit{resolving} the identity, $\sum_i E_i = \mathds{1}$. The probability of observing outcome $i$ is given by the \textit{Born rule}:
\begin{equation}\label{eq:Born}
    p(i|\rho,\mathds{1}) = \InProd{E_i}{\rho} = \TrX{}{E_i^\dag \: \rho} \:.
\end{equation}
The effects then send each state $\rho$ to a probability $p(i|\rho,\mathds{1})\in [0,1]$ through the Hilbert-Schmidt inner product $\InProd{\cdot}{\cdot}$; they can be seen as \textit{probabilistic functionals} from the state space to a probability, $\InProd{E_i}{\cdot}=\TrX{}{E_i^\dag \: \cdot }: \LinOp{\Hilb{}} \rightarrow [0,1]$. These functionals sum up to a \textit{deterministic functional} (or \textit{unit effect}\footnote{The nomenclature `effect' and `unit effect' have been used in the literature to refer to the operators respectively associated with the probabilistic and deterministic functionals. In this article, we will treat both nomenclatures interchangeably so that the terms probabilistic and deterministic effects will also appear.}) $\TrX{}{\mathds{1} \cdot }: \LinOp{\Hilb{}} \rightarrow 1$ that sends each state to a probability of 1.

When one considers a non-destructive measurement instead, the most general representation is given by a quantum instrument \cite{Davies1970}: a collection of completely-positive (CP) trace-non-increasing maps $\{\mathcal{M}_i\}: \LinOp{\Hilb{A}} \rightarrow \LinOp{\Hilb{B}}$ \textit{resolving} a quantum channel $\sum_i \mathcal{M}_i = \mathcal{M}$. The probability of observing outcome $i$ together with the (unnormalized) output state $\MapOn{}{i}{\rho}\in \LinOp{\Hilb{B}}$ is given by
\begin{equation}
    p(i|\rho,\mathcal{M}) =\TrX{}{\MapOn{}{i}{\rho}}\:.
\end{equation}
The similarity with the destructive case is striking in the CJ picture \cite{Chiri2008memory,Shrapnel2018}: 
\begin{equation}\label{eq:Born_gen}
    p(i|W,M) =\InProd{M_i}{W}\:,
\end{equation}
where the `state' $W = \rho_A \otimes \mathds{1}_B \in \LinOp{\Hilb{A}\otimes\Hilb{B}}$ is destructively measured by a collection of positive operators $\{M_i\}$ resolving a channel $M$, which is the CJ representation of the instrument. 

The two formalisms look similar in the CJ picture, but the `effects' $\mathcal{M}_i$ in the quantum instrument formalism are mappings between the states of the POVM formalism, they are thus higher-order effects. Switching from POVMs to quantum instruments can therefore be seen as a construction of a higher-order theory (this example will be reformulated using the projective characterization in App. \ref{sec:examples_dynamics_constr_1}).

Notice the common threads that will be generalized in this article: the two formalisms involve `states' ($\rho$ and $W$, respectively) that are measured by a `deterministic functional' ($\mathds{1}$ and $M$), resolved into `probabilistic functionals' ($E_i$ and $M_i$). What changes between the two situations is that the set of deterministic functionals in the POVM case, $\{\mathds{1}\}$, has a single element, whereas the one in the quantum instrument case, $\{M\}$, has infinitely many. 
However, the sets of states and deterministic functionals have a similar structure of subspaces of trace-normalized positive operators, which we will abstract under the name `state structure'. And moreover, the \textit{state structures} of both cases are related: $W$ is the tensor product of a quantum state $\rho$ and the functional on it, $\mathds{1}$; the set of valid $W$'s is thus a state structure obtained by \textit{the tensor product of the state structures} of $\rho$ and $\mathds{1}$. Likewise, the state structure of the $M$'s maps the state structure of $\rho$ in space $A$ to a similar state structure but on space $B$; it is \textit{the transformation between two state structures}.

This example motivates the common features of higher-order objects and their projective characterization: the quantum comb formalism \cite{Chiribella2009} or the process matrix formalism \cite{OCB2012} are also characterized by a pair of `states' and `effects'. Their sets of states are state structures in the CJ picture, as they are sets of positive and normalized operators, each defined over a specific subspace. The projective characterization then consists of finding the projector to that subspace, and classifying how the state structures relate to each other. The three ways of being related that will be formalized in this section appeared in the example: \textit{(deterministic) functional on}, \textit{tensor of}, and \textit{transformation between} state structure(s). Similarly to the types reviewed in the previous section, whose characterization can be inferred from the base type and the rules to form new types, knowing the projector characterizing the `base' state structure will be enough to infer the characterization of every state structure through the rules to form new projectors.

\subsection{State structures\label{sec:state_structure}}
Central in this work will be the notion of a \textit{subset of trace-normalized positive elements of an operator system in} $\LinOp{\Hilb{A}}$, called a \textit{state structure}. This kind of set is the backbone of the characterization, as to each type of Ref. \cite{Perinotti2016,Bisio2018} --or ``level in the hierarchy of transformations''-- corresponds such a set. 
The characterization is thus focused on these sets of positive operators that are the images of the types of higher-order maps through the CJ isomorphism, and the abstract structure of these sets is axiomatized through the notion of state structures. A simplification will be to consider only the sets containing an element proportional to the identity operator $\mathds{1}$; an assumption closely related to the concept of \textit{flatness} defined in Ref. \cite{Kissinger_2019}. It is motivated by the fact that the identity is the CJ image of the `trace and replace by white noise' operation; always having an element proportional to the identity then implies that any higher-order map can consist of discarding the input and preparing a purely random output. 
With states being self-adjoint operators, it is moreover natural to ask that every higher-order map send self-adjoint operators to self-adjoint operators. In the CJ picture, this amounts to considering only self-adjoint Choi matrices. The linear span of a set of self-adjoint operators that contains the identity is called an \textit{operator system}.

\begin{defi}[Operator system \cite{ChoiEffros1977}.] \label{def:OS}
    For a given space of operators, an \textbf{operator system} is a linear subspace closed under the adjoint that contains the identity.
\end{defi}
We will refer to both an operator system and its subset of trace-normalized positive elements (for a given normalization) by using the calligraphic font of the letter associated with the subsystem it is defined upon (it should be clear from the context which one is under consideration) and we will use primes to differentiate between operator systems defined on the same space. For example, $\mathscr{A}$ and $\mathscr{A}'$ are two different operator systems over $\LinOp{\Hilb{A}}$. 

Since operator systems are linear subspaces, they can be characterized by linear superoperator projectors noted $\Proj{}{A} : \LinOp{\Hilb{A}} \rightarrow \LinOp{\Hilb{A}}$, see Ref. \cite{Araujo2015,MPM,Milz_2022}. 
Projectors on operator systems must be unital and Hermitian-preserving. Because these subspaces are closed, the projectors can be taken self-adjoint without loss of generality. 
\begin{defi}[Projector on operator system]\label{def:proj}
    A superopertator projector $\Proj{}{A} : \LinOp{\Hilb{A}} \rightarrow \LinOp{\Hilb{A}}$ obeying the following conditions for all $V,N \in \LinOp{\Hilb{A}}$:
    \begin{subequations}\label{eq:proj}
        \begin{gather}
            \ProjOn{}{A}{\mathds{1}} = \mathds{1} \:;\\
            \ProjOn{}{A}{V^\dag} = (\ProjOn{}{A}{V})^\dag \:;\\
            \TrX{}{N^\dag \: \ProjOn{}{A}{V}} = \TrX{}{(\ProjOn{}{A}{N})^\dag \: V} \:;
        \end{gather}
    \end{subequations}
    is called a \textbf{projector on (an) operator system}.
\end{defi}

In addition, higher-order transformations should preserve the positivity of the spectrum. This has to do with the probabilistic interpretation of state structures that have been obtained from a set of linear maps that underwent the CJ isomorphism (this will be explained in Sec. \ref{sec:Proj_char_trans} below). Briefly put, the state of a party cannot have negative eigenvalues after a transformation, as this would correspond to negative outcome probabilities. In addition, following the `admissibility criterion' of Ref. \cite{Bisio2018}, any higher-order transformation should, in general, be definable on a subsystem. Preservation of the positivity of the spectrum even when acting on a subsystem requires transformations to preserve complete positivity. In the CJ picture, this translates into the operators being positive semi-definite. 
In addition, as the probabilities must still sum up to one after any higher-order transformation and since positivity is preserved, the transformations can at most rescale the traces by a positive factor. In the CJ picture, this translates into the operators having a fixed trace norm. 

These two requirements restrict the operator system to trace-normalized positive operators, meaning that any of its elements $V_A$ have to respect the following set of conditions:
\begin{subequations}\label{eq:det_struct}
    \begin{gather}
        V_A \geq 0 \:, \label{eq:det_struct_pos}\\
        \TrX{}{V_A} = c_A \:, \label{eq:det_struct_norm}\\
        \ProjOn{}{A}{V_A} = V_A \:, \label{eq:det_struct_proj}
    \end{gather}
\end{subequations}
where $\Proj{}{A}$ satisfies conditions \eqref{eq:proj}. 
The abstract structure they define will play a central role in the characterization of higher-order quantum transformations. We name it a \textit{state structure}.
\begin{defi}[State structure] \label{def:struct}
    A \textbf{state structure} $\mathscr{A} \in \LinOp{\Hilb{A}}$ is a set obeying constraints \eqref{eq:det_struct}. 
\end{defi}

As mentioned in the introductory example, an instance of a state structure is the set of quantum states in density matrix form, $\rho \in \LinOp{\Hilb{}}$, characterized by
\begin{subequations}\label{eq:state_char}
    \begin{gather}
        \rho \geq 0 \:, \label{eq:state_char_pos}\\
        \TrX{}{\rho} = 1\:, \label{eq:state_char_norm}\\
        \mathcal{I}\{\rho\} = \rho \:. \label{eq:state_char_proj}
    \end{gather}
\end{subequations}
Here, $\mathcal{I}$ is the \textit{identity mapping}, defined as
\begin{equation}\label{eq:Id}
    \forall\, V \in \LinOp{\Hilb{A}}\::\: \mathcal{I}_A(V) = V \:.
\end{equation}

The elements of a POVM are actually operators summing up to (or \textit{resolving}) the single-element state structure defined by
\begin{subequations}\label{eq:effect_char}
    \begin{gather}
        \mathds{1} \geq 0 \:, \label{eq:effect_char_pos}\\
        \TrX{}{\mathds{1}} = d\:, \label{eq:effect_char_norm}\\
        \mathcal{D}\{\mathds{1}\} = \mathds{1} \:, \label{eq:effect_char_proj}
    \end{gather}
\end{subequations}
where $\mathcal{D}$ is the \textit{depolarizing superoperator}, obeying
\begin{equation}\label{eq:depolop}
    \forall\, V \in \LinOp{\Hilb{A}}\:: \mathcal{D}_A(V) = \MapX{A}{V} \:,
\end{equation}
which projects onto the span of the identity. 

Remark that the projectors $\mathcal{I}$ and $\mathcal{D}$ commute. On a given Hilbert space, we will assume all projectors under consideration to always commute pairwise. The reasons behind this assumption will become clearer in the following. 
Assuming that all projectors under consideration commute pairwise allows to use the freedom in the choice of basis for the CJ isomorphism to ensure that all of them commute with the transposition appearing in Eqs. \eqref{eq:CJ} and \eqref{eq:CJ^-1} when acting on self-adjoint operators. As a consequence, the operator systems under consideration will also be closed under this transposition. 
(If needed, the reader can refer to App. \ref{sec:projos} for an extended discussion of what is meant and assumed when referring to `projectors on operator systems' in this article. See in particular Eqs. \eqref{eq:projo} for an abstract formulation of Eqs. \eqref{eq:proj}, and Eqs. \eqref{eq:projos} for the properties assumed about every set of projectors on operator systems.)

In order to compute probabilities on a state structure, it becomes interesting to consider the \textit{probabilistic families resolving} its elements. 
This generalizes the idea of POVM elements in quantum theory, for which each POVM is a different resolution of the identity $\sum_i E_i = \mathds{1}$, as well as the idea of quantum instrument formalism, for which each quantum instrument is a particular family of CP trace-non-increasing maps $\mathcal{M}_i$ providing a resolution of a particular CPTP map $\sum_i \mathcal{M}_i=\mathcal{M}$.
Formally,
\begin{defi}[Resolution of a state structure]\label{def:resolution}
    Let $\mathscr{A}$ be a state structure in $\LinOp{\Hilb{A}}$. The set of operators \textbf{resolving} $\mathscr{A}$ is the set of all collections of positive operators summing up to an element of $\mathscr{A}$. That is, a set of operators $\{E_i\}$ is a \textbf{resolution} of $\mathscr{A}$ if and only if
    \begin{equation}
        \begin{gathered}
            \forall E_i \in \{E_i\} \subset \LinOp{\Hilb{A}},\: \exists V \in \mathscr{A} \::\\
            E_i \geq 0 \:,\\
            E_i + \sum_{j \neq i } E_j = V \:.
        \end{gathered}
    \end{equation}
\end{defi}

\subsection{Functionals} 
Type $A$ has been generalized into state structure $\mathscr{A}$, we now want to generalize the notion of a set of \textit{deterministic effects} from the type $\overline{A}$ to a state structure. It is the set of operators $N$ representing all the functionals mapping each element $V$ of a state structure to the number 1 via the Hilbert-Schmidt inner product for self-adjoint operators, $\TrX{}{N\cdot V}=1$ (note that this inner product features a dagger on the left member, which we omit since all operators are positive). 
In the following, we will always mean deterministic effect when we use the word `effect' alone. 

In a previous work \cite{MPM}, it was noticed that since both the set of states and effects must contain an element proportional to the identity, the two subspaces on which the states and effects are respectively defined must be \textit{quasi-orthogonal}\footnote{The terminology ``quasi-orthogonal'' originated in the study of maximally Abelian subalgebras (MASAs), see Ref. \cite{Hiai2014}}. 
This means that operators $V$ and $N$ must be orthogonal everywhere but at the span of the identity. If $\CompProj{}{A}$ is the projector for $\overline{\mathscr{A}}$, it is related to $\Proj{}{A}$ by $\CompProj{}{A}\equiv \mathcal{I}_A - \Proj{}{A} +\mathcal{D}_A$. This leads to the rephrasing in terms of projectors of the characterization of $\overline{\mathscr{A}}$ from type theory \cite[Lemma 4]{Bisio2018}.
\begin{prop}[Functional]\label{theo:det_fctal}
    Let $\mathscr{A}$ be a state structure. Let $\{E_i\}$ a resolution of an element of $\mathscr{A}$ as in Def. \ref{def:resolution}.
    Let $\overline{\mathscr{A}}$ be the set of operators taking each element $V$ of $\mathscr{A}$ to the number 1 through the inner product,
    \begin{equation}\label{eq:fctal_def}
        N \in \overline{\mathscr{A}} \Rightarrow \forall V \in \mathscr{A}\: : \:\TrX{}{N \cdot V} = 1\:,
    \end{equation}
    and taking each element $E_i$ of every resolution of $V$ to a positive number between 0 and 1, i.e.,
    \begin{equation}\label{eq:fctal_def_resol}
        \begin{gathered}
           \forall V \in \mathscr{A}\:,\:\forall \{E_i\} \:: E_i\geq 0\:,\: \sum_i E_i = V \:,\\
            \TrX{}{N\cdot E_i} \in [0,1]\:.
        \end{gathered}
    \end{equation}
    Then $\overline{\mathscr{A}}$ is a state structure characterized by the following conditions:
    \begin{subequations}\label{eq:det_fctal}
        \begin{gather}
            N \in \overline{\mathscr{A}} \:\iff\notag\\
            N \geq 0 \:, \label{eq:det_fctal_pos}\\
            \TrX{}{N} = \frac{d_A}{c_A} \equiv c_{\overline{A}}\:,\label{eq:det_fctal_norm}\\
            \left\{\mathcal{I}_A - \Proj{}{A} +\mathcal{D}_A\right\}(N)\equiv \CompProjOn{}{A}{N}= N \label{eq:det_fctal_proj} \:.
        \end{gather}
    \end{subequations}
\end{prop}
\begin{figure}[htb]
    \centering
     \subfloat[Defining relation \eqref{eq:fctal_def}.\label{fig:Tr_A_notA}]{
         \centering
         \includegraphics[height=0.12\textheight]{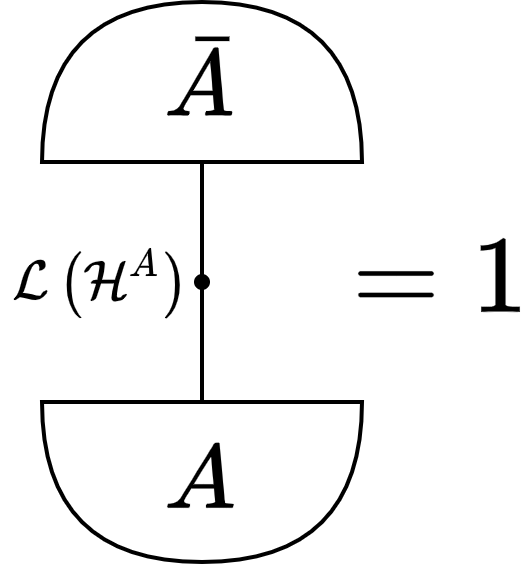}
         }~\hfill
         \subfloat[Diagrammatic representation of their spans.\label{diag:A_notA}]{
         \includegraphics[height=.12\textheight]{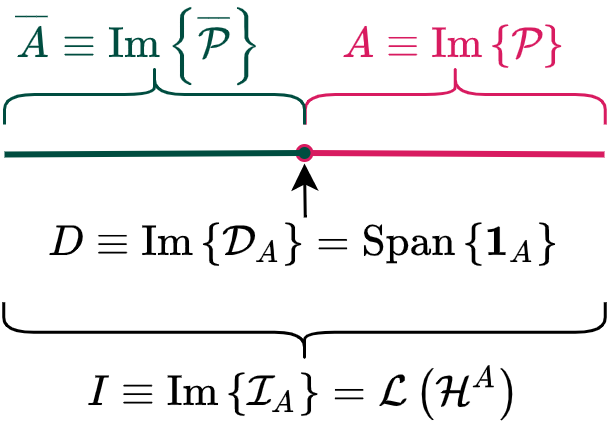}
     }
        \caption{Picturing the characterizing equations of an effect state structure, Proposition \ref{theo:det_fctal}. Fig. \ref{fig:Tr_A_notA} is the graphical representation of the defining relation \eqref{eq:fctal_def}. Fig. \ref{diag:A_notA} is a diagrammatic representation of how the Hilbert space splits between the operator system $\Alg{A}{}$ (pink) and $\CompAlg{A}{}$ (green). (In the diagram, the regular fonts $A$ and $\overline{A}$ have been used instead of script-style because of software limitations; `Im' means the image of a linear map). Together they span the full space (whole line), yet they only intersect at identity (central dot).}
        \label{fig:fctal}
\end{figure}
The proof is provided in App. \ref{app:fctal_proof}, see Fig. \ref{fig:Tr_A_notA} for an illustration. This proposition was first obtained as a theorem in our earlier work on Multi-round Process Matrices (MPM) \cite{MPM}. 
A concrete example of a state structure characterized by Prop. \ref{theo:det_fctal} are Eqs. \eqref{eq:effect_char} obtained from Eqs. \eqref{eq:state_char}. Other examples feature the quantum testers \cite{Chiribella2009} whose characterization follows from the one of quantum combs using Prop. \ref{theo:det_fctal}, and process matrices \cite{OCB2012}, whose characterization follows from the one of the tensor product of channels. The projective characterizations of single and bipartite process matrices are respectively presented in App. \ref{sec:examples_single_channel} and \ref{sec:examples_bipartite_PM}. 

Notice how Prop. \ref{theo:det_fctal} splits the space of operators: the projectors $\Proj{}{A}$ and $\CompProj{}{A}$ obey the relation
\begin{equation}\label{eq:decomposition_I}
    \mathcal{I}_A = \Proj{}{A} + \CompProj{}{A} - \mathcal{D}_A \:;
\end{equation}
this means that the space of all operators is divided between operators $V$ obeying $\ProjOn{}{A}{V} = V$ and operators $N$ obeying  $\CompProjOn{}{A}{N} = N$; the only operators $X$ allowed to live on the intersection of the two subspaces are those obeying $\mathcal{D}_A(X) = X$ since $\Proj{}{A} \circ \CompProj{}{A} = \mathcal{D}_A$. Consequently, the only elements in $\Alg{A}{} \cap \CompAlg{A}$ are multiples of the identity. See Fig. \ref{diag:A_notA} for a diagrammatic illustration of these relations.

Here, the bar in ``$\CompProj{}{A} \equiv \mathcal{I} -  \Proj{}{A} +\mathcal{D}$'' can be seen as an operation at the level of projectors that was applied on $\Proj{}{A}$. 
As such, the projector $\CompProj{}{A}$ will be called the \textbf{negation} of $\Proj{}{A}$ because the operation has properties similar to a logical \textit{not}; see  App. \ref{sec:projo_Boolean} for the details. This algebraic interpretation of projectors will be developed in Sec. \ref{sec:Projos_alg} below.

\begin{figure}[thb]
    \centering
    \includegraphics[height=0.12\textheight]{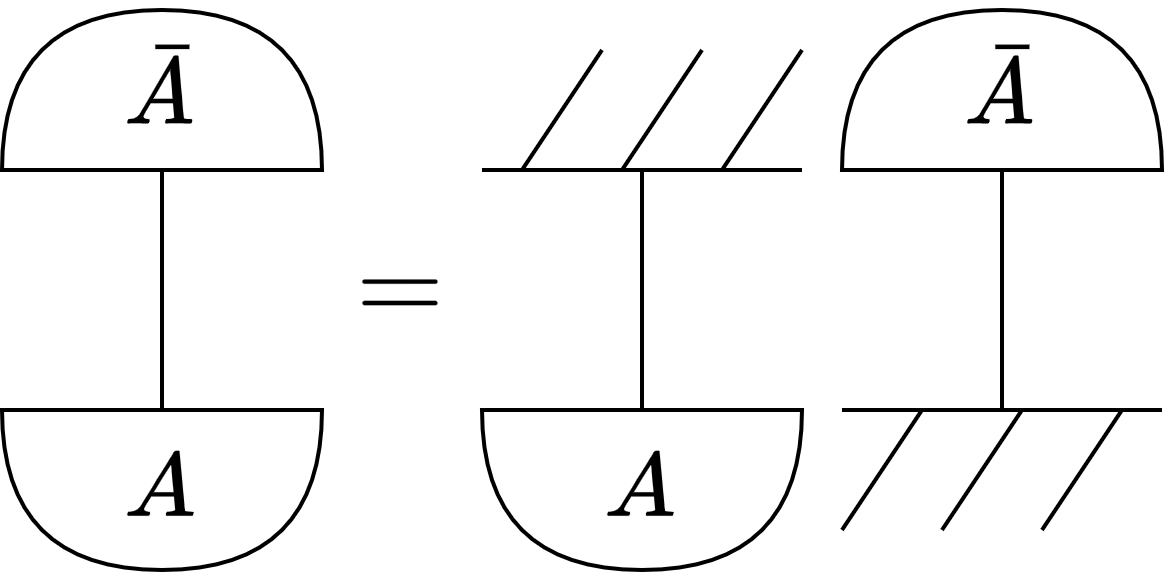}
    \caption{Quasi-orthogonality condition \eqref{eq:QO} is equivalent to the factorization \eqref{eq:fctal_indep} of the inner product. Graphically, the element of a state structure proportional to $\mathds{1}$ is by convention represented by the symbol /// (this is to single it out from arbitrary elements, which are represented by a half-circle). This diagram should then be read as $ \forall V \in \Alg{A}{}, \forall N \in \CompAlg{A}{},\, \TrX{}{N \cdot V} = \TrX{}{\frac{\mathds{1}}{c_A} \cdot V} \TrX{}{N\cdot \frac{\mathds{1}}{c_{\overline{A}}}}$. The diagrammatic translation of quasi-orthogonality between two state structures is thus the possibility of `cutting' the wire between them; this disconnection conveys the idea that these two state structures cannot influence each other.}
    \label{fig:QO}
\end{figure}
Quasi-orthogonality implies the following property \cite[Theorem 2.37 iii)]{Hiai2014}:
\begin{equation}\label{eq:QO}
    \forall V\in \mathscr{A}, \forall N\in \overline{\mathscr{A}}\:,\: \TrX{}{N\cdot V} = \frac{1}{d_A}\TrX{}{N}\TrX{}{V} \: .
\end{equation}
Combining the trace conditions \eqref{eq:det_struct_norm} and \eqref{eq:det_fctal_norm} with equation \eqref{eq:QO} yields the following.
\begin{equation}\label{eq:fctal_indep}
     \TrX{}{N \cdot V } = \TrX{}{N\cdot \frac{\mathds{1}}{c_{\overline{A}}}}\TrX{}{\frac{\mathds{1}}{c_A} \cdot V} \: .
\end{equation}
Fig. \ref{fig:QO} provides a graphical interpretation of this relation: a wire connecting a state $V$ with a `deterministic functional' (a trace with unit effect $N$; $\TrX{}{N \cdot V}$) can be `cut' in two pieces (two traces, $\TrX{}{N \cdot \bullet }\TrX{}{\bullet \cdot V} $) by adding two `discarding effects' on each loose end ($\frac{\mathds{1}}{c_{\overline{A}}}$ and $\frac{\mathds{1}}{c_A}$ replace their respective $\bullet$). 
This hints the physical content of Proposition \ref{theo:det_fctal}: the generalized Born rule is independent of \textit{which} \textit{particular} deterministic state $V$ and effect $N$ are getting connected together; both the state and the effect can be replaced by the identity -- or any other element of their respective state structures -- without altering the outcome probability of 1. Put differently, even when replacing either the state or the effect by pure noise (that is, the discarding effect), it is always a certainty that \textit{something has happened} at the end of the protocol. 
The connection between quasi-orthogonality and the independence between the choice of (deterministic) state and effect is what allows to generalize the notion of \textit{no signaling} to state structures. This will be developed in Sec. \ref{sec:NS}.

\subsection{Tensor product \label{sec:Proj_char_tensor}}
The tensor product of two types as a parallel composition is an abstraction that encompasses both the notions of \textit{bipartite quantum state} as well as \textit{no signaling channel} \cite{Beckman2001,Piani2006}. 
In terms of projectors, the characterization relies on raising the tensor product operation at the level of superoperators in a natural fashion. Let $\Proj{}{A}$ and $\Proj{}{B}$ be respectively projectors on operator systems $\Alg{A}{}$ and $\Alg{B}{}$. Then $\Proj{}{A} \otimes \Proj{}{B}$ acts on $\LinOp{\Hilb{A}\otimes \Hilb{B}}$ so that $\forall q_i \in \mathbb{C}, V_i\in \LinOp{\Hilb{A}}, U_i \in \LinOp{\Hilb{B}}$,
\begin{multline}\label{eq:tensor}
    \left(\Proj{}{A} \otimes \Proj{}{B} \right) \left\{\sum_i \; q_i \; \left(V_i \otimes U_i\right)\right\} \equiv\\ \sum_i\;q_i\;\left(\ProjOn{}{A}{V_i}\otimes \ProjOn{}{B}{U_i}\right) \:.
\end{multline}
This straightforward definition of the tensor product of two projectors on operator systems results in a projector on operator system in $\LinOp{\Hilb{A}\otimes \Hilb{B}}$ (see App. \ref{sec:projo_prop_tensor}). Hence $ \Proj{}{A} \otimes \Proj{}{B}$ is an operation on projectors $\Proj{}{A}$ and $\Proj{}{B}$ that will be called their \textbf{tensor product} (or tensor in short). A diagrammatic depiction of its span with respect to the bipartition of the Hilbert space is given in Fig. \ref{fig:tensor}.
\begin{defi}[No signaling (bipartite) composition]\label{prop:tensor}
    Let $\mathscr{A}$ and $\mathscr{B}$ be two state structures as in Eqs. \eqref{eq:det_struct}. Their no signaling composition $\mathscr{A}\otimes \mathscr{B} \subset \LinOp{\Hilb{A}\otimes \Hilb{B}}$ is the set of all operators $W$ characterized by the following constraints:
    \begin{subequations}\label{eq:det_tensor}
    \begin{gather}
        W\in \mathscr{A}\otimes \mathscr{B} \iff \notag \\
        W \geq 0 \:,\label{eq:det_tensor_pos}\\
        \TrX{}{W} = c_Ac_B \label{eq:det_tensor_norm}\:,\\
        \left(\Proj{}{A}\otimes\Proj{}{B}\right)\{W\} = W\:. \label{eq:det_tensor_proj}
    \end{gather}
    \end{subequations}
    Consequently, $\mathscr{A}\otimes \mathscr{B}$ is the intersection of the affine span of tensor products of operators from $\mathscr{A}$ and $\mathscr{B}$ with the cone of positive operators.
\end{defi}
Def. \ref{prop:tensor} relies on the observation that the positive elements of the affine hull of the tensor product of state structures form a state structure \cite[Lemma 5]{Bisio2018} to define a way to compose two state structures.
It recovers reasonable notions of parallel composition: the set of bipartite states is obtained by setting $\mathscr{A}$ and $\mathscr{B}$ to be the state structures of quantum states, whereas the set of no signaling (sometimes called \textit{causal} \cite{Beckman2001}) bipartite channels is obtained by setting $\mathscr{A}$ and $\mathscr{B}$ to be the state structures of channels. (These two examples are respectively presented in App. \ref{sec:examples_bipartite_QT} and \ref{sec:examples_no_signaling_channel}.) 
We have called it `the no signaling composition' for reasons that will become clear in Sec. \ref{sec:NS}. 
For now, this definition will help us for the characterization of the CJ representation of transformations between state structures.

Notice that the decomposition \eqref{eq:decomposition_I} of the identity projector in terms of $\Proj{}{A}$ and $\CompProj{}{A}$ can be generalized to the identity projectors on two systems. Since $\mathcal{I}_{AB} = \mathcal{I}_A \otimes \mathcal{I}_B$ by linearity, taking the tensor product of Eq. \eqref{eq:decomposition_I} with itself gives
\begin{multline}\label{eq:decomposition_II}
    \mathcal{I}_A \otimes \mathcal{I}_B =\\ \Proj{}{A} \otimes \Proj{}{B} + \Proj{}{A} \otimes \CompProj{}{B} + \CompProj{}{A} \otimes \CompProj{}{B} + \CompProj{}{A} \otimes \Proj{}{B} \\
        -\Proj{}{A} \otimes \mathcal{D}_B - \mathcal{D}_A \otimes \CompProj{}{B} - \CompProj{}{A} \otimes \mathcal{D}_B - \mathcal{D}_A \otimes \Proj{}{B} \\
        + \mathcal{D}_A \otimes \mathcal{D}_B \:.
\end{multline}
This relation is diagrammatically represented in Fig. \ref{fig:tensor}; it implies that the space is divided between four main parts depending on $\Proj{}{A}$, $\Proj{}{B}$, and their negations. Each of these parts might share a common `border', which can be non-trivial, e.g., $\Proj{}{A} \otimes \mathcal{D}_B$ is at the intersection of $\Proj{}{A} \otimes \Proj{}{B}$ and $\Proj{}{A} \otimes \CompProj{}{B}$ since $(\Proj{}{A} \otimes \Proj{}{B}) \circ (\Proj{}{A} \otimes \CompProj{}{B}) = \Proj{}{A} \otimes \mathcal{D}_B$, or not, e.g., $\mathcal{D}_A \otimes \mathcal{D}_B$ is at the intersection of $\Proj{}{A} \otimes \Proj{}{B}$ and $\CompProj{}{A} \otimes \CompProj{}{B}$ since $(\Proj{}{A} \otimes \Proj{}{B}) \circ (\CompProj{}{A} \otimes \CompProj{}{B}) = \mathcal{D}_A \otimes \mathcal{D}_B$. In this latter case, the `border' is trivial as the only elements contained in it are proportional to the identity operator $\mathds{1}_A \otimes \mathds{1}_B$. 
The systematic way to discuss the splitting of the space of operators induced by the projectors and how to find the `borders', i.e. to obtain intersections of operator systems as compositions of projectors, will be presented in Sec. \ref{sec:Projos_alg}.

\begin{figure}[htb]
    \centering
        \includegraphics[width=.8\linewidth]{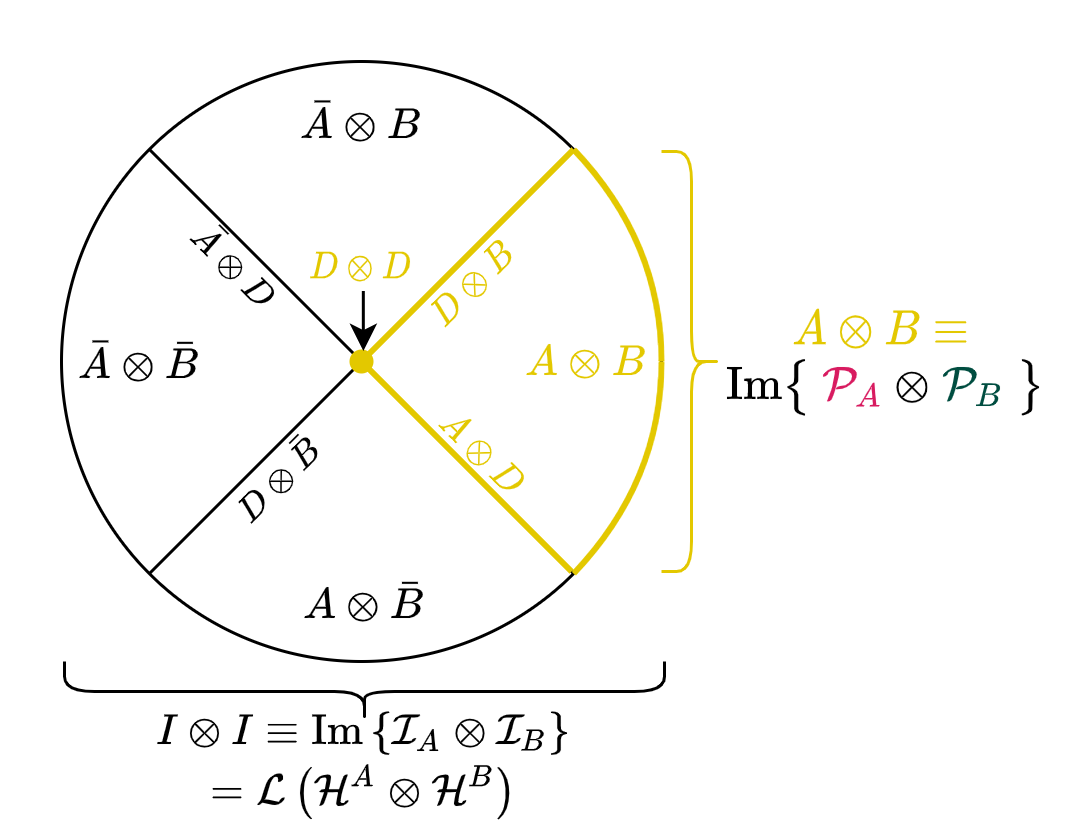}
    \caption{Diagrammatic representation of the image of the tensor product of projectors, as in Def. \ref{prop:tensor}. The full wheel represents $\LinOp{\Hilb{A}\otimes \Hilb{B}}$ while its parts represent its tensor factors, Eq. \eqref{eq:decomposition_II}. (Like in Fig. \ref{diag:A_notA}, `Im' means `Image' and e.g. $D$ is a shortcut for $\mathrm{Im}\{\mathcal{D}\} = \Span{\mathds{1}}$). By convention, the projectors acting on $\LinOp{\Hilb{A}}$ are always on the left-hand side of tensor products and vice-versa for $B$, so that subscripts can be omitted. Note that the intersections are well defined; for example the line `$A \otimes D$', which represents $\mathrm{Im}\{\Proj{}{A} \otimes \mathcal{D}_B\}$ is indeed the intersection $A \otimes B \cap A \otimes \overline{B}$.}
    \label{fig:tensor}
\end{figure}

\subsection{Transformations \label{sec:Proj_char_trans}}
In the type theory of higher-order transformations, every type is an instance of a transformation, so the previous characterization of functionals (type $\overline{A}$) and bipartite states (type $A \otimes B$) are but special cases of a more general rule. This rule says that if $A$ and $B$ are types, then $A \rightarrow B$ is itself a type that takes elements of $A$ to $B$. A state of type $A$ is actually a \textit{transformation} of the trivial type -- the number 1 -- into $A$, noted $1\rightarrow A$. Accordingly, the effect of type $\overline{A}$ is the transformation of type $A$ into the trivial type, $A\rightarrow 1$. No signaling composition is the type $A\otimes B \equiv \overline{A \rightarrow \overline{B}}$. 

As we have seen, the set of elements of a given type corresponds to a state structure, ultimately determined by its projector. Thus, to type $A \rightarrow B$ corresponds a state structure, thereafter noted $\mathscr{A} \rightarrow \mathscr{B}$, which is obtained by combining state structures $\mathscr{A}$ and $\mathscr{B}$ in a certain way. This structure should then be characterized by a composite projector built by combining $\Proj{}{A}$ and $\Proj{}{B}$ using the two previously introduced rules.

Following previous approaches \cite{Kraus1983,Perinotti2016,Bisio2018,Dynamics}, it should be possible to randomize the dynamics --to pick a transformation at random-- without losing the probabilistic interpretation, implying that the transformation must conserve the total probability as well as probabilistic mixtures. And it should be possible to do so even when acting locally on a tensor factor, i.e. even when it acts on a single part of a no signaling composite state structure \cite[Definition 7]{Bisio2018}. The preservation of probabilistic mixtures implies that the maps representing such dynamics are linear (see Ref. \cite{Kraus1983,OCB2012,Perinotti2016,Bisio2018} for the proof). 
It should also be possible to consider entangled input and output pairs \cite{Chiribella2008,Chiribella2009,Perinotti2016,Bisio2018} --i.e., to have correlated inputs and outputs obtained via the exchange of a side system-- without losing the probabilistic interpretation. This generalizes the idea of a CP map and implies that the transformation must be normalized on any state in the tensor product of its input and output and any resolution thereof. This last bit is important as it constrains the CJ representations of transformations to positive semi-definite operators (see Ref. \cite{Chiribella2008,Bisio2018}).
\begin{defi}[Structure-preserving map] \label{def:struc_pres}
    Let $\mathcal{M}$ be a map from $\LinOp{\Hilb{A}}$ to $\LinOp{\Hilb{B}}$. This map is \textbf{structure-preserving} between $\mathscr{A}$ and $\mathscr{B}$ if 1) it is linear; 2) it maps any element of state structure $\mathscr{A}$ to one in $\mathscr{B}$, that is,
    \begin{equation}\label{eq:A_trasnformsinto_B}
        \forall V \in \mathscr{A},\quad\mathcal{M}(V) \in \mathscr{B}\:;
    \end{equation}
    3) its CJ operator is positive.
\end{defi}

\begin{prop}[Transformation between state structures]\label{theo:det_map}
    Let $\mathcal{M} \in \LinOpB{\LinOp{\Hilb{A}}}{\LinOp{\Hilb{B}}}$ be a structure-preserving map between state structures $\mathscr{A}$ and $\mathscr{B}$. Call $M \in \LinOp{\Hilb{A}\otimes \Hilb{B}}$ the Choi-Jamio{\l}kowski representation of this map, as in Eq. \eqref{eq:CJ}. Then, the set $\{M\}$ of all such maps belongs to the state structure $\Alg{A}{} \rightarrow \Alg{B}{}$ defined by the following conditions:
    \begin{subequations} \label{eq:det_map}
    \begin{gather}
        M\in \mathscr{A}\rightarrow \mathscr{B} \: \iff\notag \\
        M \geq 0 \:,\label{eq:det_map_pos}\\
        \TrX{}{M} = c_{\overline{A}}c_B = \frac{c_B}{c_A}d_A \:,\label{eq:det_map_norm}\\
        \left(\Proj{}{A}\rightarrow \Proj{}{B}\right)\{M\} \equiv \left(\overline{\Proj{}{A}\otimes\CompProj{}{B}}\right)\{M\} = M \:.\label{eq:det_map_proj}
    \end{gather}
    \end{subequations}
\end{prop}
The proof is presented in App. \ref{sec:det_map_proof}; see Fig. \ref{fig:det_map} for an illustration. An example of a state structure characterized by this proposition is the set of quantum channels, obtained by setting $\mathscr{A}$ and $\mathscr{B}$ to be state structures of quantum states (the examples of single and bipartite channels are respectively presented in App. \ref{sec:examples_single_channel} and \ref{sec:examples_bipartite_channel}; note that App. \ref{sec:examples_single_channel} is but the introductory example of this section rephrased in the formalism developed in this paper). Another example is the set of quantum $n$-combs \cite{Chiribella2009} in which $\mathscr{A}$ is the set of $(n-1)$-combs and $\mathscr{B}$ the one of channels (this example is revisited in Sec. \ref{sec:applications_iso}).
\begin{figure}[thb]
    \centering
    \subfloat[Graphical interpretation of the defining condition of the set $\mathscr{A}\rightarrow\mathscr{B}$: $M \in \mathscr{A}\rightarrow\mathscr{B} \Rightarrow \TrX{A}{M \cdot (A\otimes \mathds{1}_B)})^T \in \mathscr{B}$.\label{fig:A_to_B}]{
        \includegraphics[width=.5\linewidth]{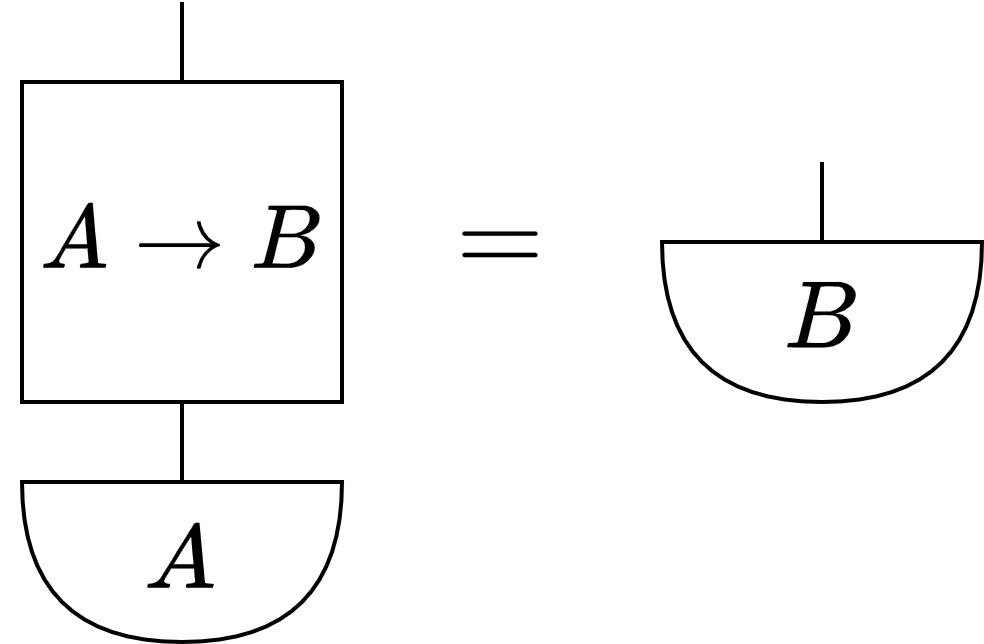}
    }~\hfill
    \subfloat[Diagram for the subspaces defined by Eq. \eqref{eq:det_map_proj} (blue) and $\Proj{}{A}\otimes \CompProj{}{B}$ (yellow)\label{diag:A_to_B}]{
        \includegraphics[width=.40\linewidth]{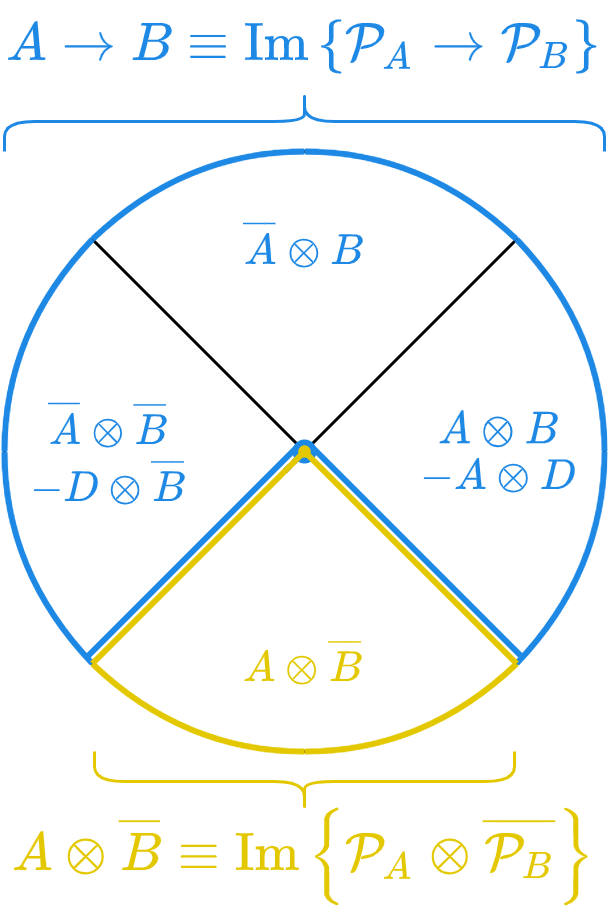}
    }
    \caption{Picturing the characterizing equations of a transformation between state structures, Proposition \ref{theo:det_map}. Fig. \ref{fig:A_to_B} is the graphical representation of the defining relation \eqref{eq:A_trasnformsinto_B}. Fig. \ref{diag:A_to_B} is a diagrammatic representation of how the Hilbert space splits between the span of the state structure of transformations, $\Alg{A}{} \rightarrow \Alg{B}{}$ (blue), and the span of the deterministic functionals on it, $\Alg{A}{}\otimes \CompAlg{B}{}$ (yellow). As in Fig. \ref{diag:A_notA}, they span the full space together (i.e., they cover the whole disk), yet they only intersect at the identity (central dot). Notice how the transformation does not contain the boundaries $\Im{\mathcal{D}_A \otimes \CompProj{}{B}}$ and $\Im{\Proj{}{A}\otimes \mathcal{D}_B}$ as they belong to $\Alg{A}{}\otimes \CompAlg{B}{}$ already.}
    \label{fig:det_map}
\end{figure}

In Proposition \ref{theo:det_map}, a new operation on the projectors has been defined, the \textbf{transformation} $\rightarrow$. 
It combines the tensor and negation operations, $\Proj{}{A} \rightarrow \Proj{}{B} \equiv \overline{\Proj{}{A} \otimes \CompProj{}{B}}$ and its full form is the following projector:
\begin{multline}\label{eq:proj_A_to_B}
    \Proj{}{A} \rightarrow \Proj{}{B} \equiv 
    \mathcal{I}_A\otimes \mathcal{I}_B - \mathcal{P}_A \otimes \mathcal{I}_B \\ + \Proj{}{A} \otimes \Proj{}{B} - \mathcal{P}_A \otimes \mathcal{D}_B + \mathcal{D}_A\otimes \mathcal{D}_B \:.
\end{multline}
This is a connector whose direction matters. As we did with types, we define a reversed symbol, $\leftarrow$, by $\Proj{}{A} \leftarrow \Proj{}{B} \equiv \overline{\CompProj{}{A} \otimes \Proj{}{B}}$ in order to use it without changing the ordering of the tensor factors.

\paragraph*{Remark.} The above derivation is very close in spirit to Ref. \cite{Dynamics} where the projective characterization of transformations was first made in the context of transformation between process matrices, but the obtained projector was missing two terms\footnote{What the authors had obtained in the context of process matrix to process matrix transformation was $\mathcal{I}_A\otimes \mathcal{I}_B - \mathcal{P}_A \otimes \mathcal{I}_B + \Proj{}{A} \otimes \Proj{}{B}$, where $\Proj{}{A}$ is the projector characterizing the input process matrix state structure and $\Proj{}{B}$ the output. Compared to equation \eqref{eq:det_map_proj}, here written fully without negation, 
$\overline{\Proj{}{A}\otimes\CompProj{}{B}} = \mathcal{I}_A\otimes \mathcal{I}_B - \mathcal{P}_A \otimes \mathcal{I}_B + \Proj{}{A} \otimes \Proj{}{B} - \mathcal{P}_A \otimes \mathcal{D}_B + \mathcal{D}_A\otimes \mathcal{D}_B$, 
one sees that two terms are missing. In particular, they are the ones that forbid to postselect on a process matrix before preparing a new one that is purely noisy.}. After verification, it appears that this omission has not hindered the other conclusions of this article \cite{ECR_comm}.

\subsection{The algebra of projectors recovers the type theory}
We started from the type theory of higher-order transformations for which $ \overline{A}, A \otimes B,$ and $ A\rightarrow B$ are the new types that one can obtain from base types $A$ and $B$ under the semantics $\{1,(,),\rightarrow\}$.

We observed that types define \textit{state structures}, which are subsets of operator systems made of trace-normalized positive elements. These state structures are the sets of CJ operators representing higher-order maps for which $\overline{\mathscr{A}}, \mathscr{A} \otimes \mathscr{B}, $ and $ \mathscr{A} \rightarrow \mathscr{B}$ are the new state structures that one can obtain from bases $\mathscr{A}$ and $\mathscr{B}$. 

We thus refined the type theory by relaxing the requirement that the base types are quantum states into the base type being a valid state structure. 
As these state structures are defined on linear subspaces, most of their properties are encoded by the superoperator projector that defines them. These projectors are an abstraction of states structures for which $\CompProj{}{A}, \Proj{}{A}\otimes \Proj{}{B},$ and $ \Proj{}{A}\rightarrow\Proj{}{B}$ are the new projectors that one can obtain from bases $\Proj{}{A}$ and $\Proj{}{B}$. In that respect, the projective characterization amounts to the three algebraic rules $\{\overline{\:\cdot\:},\otimes,\rightarrow\}$ defined in Eqs. \eqref{eq:det_fctal_proj}, \eqref{eq:det_tensor_proj}, and \eqref{eq:det_map_proj}.

\begin{lemm}\label{lem:Boolean_lattice}
    Let $\Proj{}{A}$ be an arbitrary projector on an operator system (see Def. \ref{def:proj}) and acting on space $\LinOp{\Hilb{A}}$. Let $\Proj{}{B},\ldots,\Proj{}{K}$ be similarly defined projectors, respectively acting on $\LinOp{\Hilb{B}}, \ldots, \LinOp{\Hilb{K}}$. 
    Then, every superoperator acting on $\LinOp{\Hilb{A}\otimes \Hilb{B}\otimes\ldots\Hilb{K}}$ obtained by combining these projectors using any permutation of the rules $\{\overline{\:\cdot\:},\otimes,\rightarrow\}$ is a projector on operator system.
\end{lemm}
\begin{proof}
    Each rule individually results in a projector on an operator system as can be shown by direct computation, meaning that if $\Proj{}{A}$ and $\Proj{}{B}$ are projectors on operator systems, then so are $\CompProj{}{A}$, $\Proj{}{A} \otimes \Proj{}{B}$, and $\Proj{}{A} \rightarrow \Proj{}{B}$, see Appendices \ref{sec:projo_Boolean} and \ref{sec:projo_prop_tensor} for the full proof of this statement. 
    It remains to show that any combination of these rules are also giving projectors also obeying Eqs. \eqref{eq:proj}. Using that $\Proj{}{A}\rightarrow\Proj{}{B} = \overline{\Proj{}{A}\otimes \CompProj{}{B}}$, any expression is either an overall negation of a projector, $\CompProj{}{}$, or it is an overall tensor of k factors, $\Proj{}{A}\otimes \ldots \otimes \Proj{}{K}$. In the first case, it is a valid projector provided that the projector being negated is valid. In the second case, the k factors can be seen as a tensor product of two factors using the associativity of the tensor, $\left(\Proj{}{A}\otimes \ldots \right)\otimes \left(\Proj{}{K}\right)$, which is a valid projector provided $\Proj{}{A}\otimes \ldots$ and $\Proj{}{K}$ are. This reasoning repeats inductively until one reaches each individual projector $\Proj{}{A},\Proj{}{B},\ldots,\Proj{}{K}$ which are projectors on operator systems by assumption.
\end{proof}
And because of that, they define valid state structures globally.
\begin{coro}
    Any set build from state structures using the rules $\{\overline{\:\cdot\:},\otimes,\rightarrow\}$ (characterized by, respectively, Prop. \ref{theo:det_fctal}, Def. \ref{prop:tensor}, and Prop. \ref{theo:det_map}) is a state structure itself.
\end{coro}

Thus, when a set $\left\{\Proj{}{A}, \Proj{}{B}, ..., \Proj{}{K} \right\}$ of $k$ `base' projectors has been defined and associated with $k$ systems $A,B,\ldots,K$, one can form arbitrary combinations (for example $(\Proj{}{A} \rightarrow \Proj{}{B}) \otimes \CompProj{}{F}$) and this is guaranteed to define a new composite state structure on the joint space once the normalizations have been fixed (in the example $(\Alg{A}{} \rightarrow \Alg{B}{}) \otimes \CompAlg{F}{} \subset \LinOp{\Hilb{A} \otimes \Hilb{B} \otimes \Hilb{F}}$ is a valid state structure once the constants $c_A$, $c_B$, and $c_{\overline{F}}$ are fixed). 
Now, the characterization has shown that the obtained rules are redundant: instead of using $\{\overline{\:\cdot\:},\otimes,\rightarrow\}$, one can use $\{\overline{\:\cdot\:},\rightarrow\}$ only, as Eq. \eqref{eq:det_map_proj} shows that any occurrence of the no-signaling composition can be replaced by a combination of the negation and the transformation. 
Similarly, instead of using $\{\overline{\:\cdot\:},\rightarrow\}$, one can use the transformation $\rightarrow$ only. To do so one can leverage the trivial system $1$, which is the single-element state structure composed of the number one and characterized by a projector which is also the number one (which is the only projector on $\mathbb{C}$). Using such a rule, one can see that $\CompProj{}{A} = \Proj{}{A} \rightarrow 1$ since $\Proj{}{A} \rightarrow 1 = \mathcal{I}_A - \Proj{}{A} + \mathcal{D}_A$ using Eq. \eqref{eq:proj_A_to_B}. 
Hence, the operations on projectors can be reduced to only the transformation, and together with the parentheses and the trivial projector $1$, they give the same semantics $\{1,(,),\rightarrow\}$ as the type theory of Refs. \cite{Perinotti2016,Bisio2018}. 
This is not a coincidence: when the base state structures are chosen to be the state spaces of quantum states, i.e., when the base projectors are chosen to be the identity projector $\mathcal{I}$ and the associated traces are set to 1, the type formalism is recovered. The next theorem is a technically precise version of this statement, given for completeness only and aimed at a reader familiar with Ref. \cite{Bisio2018}. 
\begin{prop}\label{theo:Bisio_equivalence}
    When the base state structures are set to be the sets of density matrices, the state structures characterized by the projective characterization under the rules $\{\overline{\:\cdot\:},\otimes,\rightarrow\}$ are in one-to-one correspondence with the sets of `deterministic events' characterized using the type-theoretic approach of Refs. \cite{Perinotti2016,Bisio2018}.
\end{prop}
The proof is provided in App. \ref{sec:Bisio_equivalence_proof}. 

Therefore, working with projectors is not only a handy way to define the characterization constraints implied by type theory: There is a correspondence between the semantic definition of new types using $\{1,(,),\rightarrow\}$ and the construction of new projectors using $\{\overline{\:\cdot\:},\otimes,\rightarrow\}$. 
If the base state structures are chosen to be quantum states, then the type theory of Ref. \cite{Perinotti2016,Bisio2018} is recovered. In Sec. \ref{sec:applications_iso}, we will provide an example in which the results of these papers are recovered from the projector algebra.

\section{Implications of replacing types with projectors\label{sec:Projos_alg}}
We showed in the previous section that we can recast the type-theoretic characterization of higher-order quantum transformations into a projective formulation. In this section, we present the two advantages of doing so. The first one is the flexibility when defining the base types, whereas the second one is a new tool for comparing any two types. While the first advantage is `only' a formalisation of something that could have been done using the type-theoretic approach, the second one is intrinsic to working with projectors. It further reveals that the algebraic rules obeyed by the projectors are not arbitrary. They actually form an algebra, as well as a model of logic (first considered in Ref. \cite{Kissinger_2019}). This, on the one hand, provides a mathematical reason why the type system works the way it does. Moreover, it offers a new way of bridging between the type-theoretical approach \cite{Perinotti2016,Bisio2018} and the categorical one \cite{Kissinger_2019,Simmons2022}. 
On the other hand, it allows for a classification of higher-order quantum processes, which will be shown to be based on the notion of signaling in Sec. \ref{sec:algebra_the_algebra}, after a new way of combining projectors is introduced in Prop. \ref{prop:semi-causal_trans}.

\subsection{Extending the hierarchy\label{sec:base_sets}}
Type theory is dependent on a set of base types. In Ref. \cite{Bisio2018}, the first nontrivial types in the hierarchy, called the \textit{elementary types}, are taken as the set of quantum states as in Eqs. \eqref{eq:state_char}. 
In terms of the projective characterization, this means that each base type is characterized with projector $\mathcal{I}$, and the elementary type on $k$ subsystems is the set of quantum states characterized by projector $\mathcal{I_A} \otimes \mathcal{I_B} \otimes \ldots \otimes \mathcal{I_K}$. The other types of higher-order transformations are subsequently obtained by using different sets of connectors to combine the $k$ base projectors.

Within the framework of state structures, instead of having a fixed base type, one can consider any state structure as a base. In fact, any set of higher-order transformations can be coarse-grained into a base state structure by seeing all the systems they act on as a single overall system. With such a procedure, one can obtain many inequivalent base state structures, so the formalism can be further broadened to allow the parties to choose between several base state structures. For example, Alice could be choosing between inputting a quantum or a classical state, and the formalism would accommodate it because these two choices are valid state structures whose projectors commute with one another. In terms of projectors, it means that, for a given subsystem, the base projector can be picked from a set of pairwise commuting projectors on operator systems $\{\Proj{}{},\Proj{'}{},...\}$, instead of being $\mathcal{I}$. 
Indeed, as shown in the following, building multipartite projectors from any set of commuting base projectors will lead to a set of commuting multipartite projectors. Thus, one can always coarse-grain these new multipartite projectors into a new set of base projectors, build new higher-order objects using them, and then repeat the operation to obtain a recursively more complex higher-order transformation (an example of such a procedure is presented in App. \ref{sec:examples_dynamics_constr}). 

Yet, and this is a difference with the type framework, the base projectors do not have to be built recursively by starting with $\mathcal{I}$. Combining several such projectors using $\{\overline{\:\cdot\:}, \otimes, \rightarrow\}$ to obtain a new set of multipartite projectors and coarse-graining into a new set of base projectors is not the only way of obtaining base projectors that are different from $\mathcal{I}$. Actually, any set of pairwise commuting projectors obeying Def. \ref{def:proj} can be considered. 
Therefore, different higher-order hierarchies than the one of quantum processes can be considered in principle; one only needs to start from a base projector characterizing a different set than the one of density matrices. 
For example, the projector to diagonal matrices in a fixed basis, corresponding to the state structure of classical states we mentioned above, is an example of a projector that cannot be obtained from $\mathcal{I}$ using $\{\overline{\:\cdot\:}, \otimes, \rightarrow\}$. 
We will come back to this at the end of the section, leading to the example presented in App. \ref{sec:examples_biased_QT}. It is also discussed in more depth in App. \ref{sec:projos_base_sets}. 
Still, we emphasize that the base projectors one can consider are limited to sets of pairwise commuting elements. 

\subsection{The algebra of projectors\label{sec:projos_alg_def}}
For any hierarchy of higher-order transformations, a key question is how to classify the non-equivalent kinds of transformations, and how to compare them. For example, it is known that the quantum supermaps happen to be a special kind of bipartite quantum channels \cite{Chiribella2009,Perinotti2016}, but is there a systematic way to infer such a conclusion? 

This is where characterizing using projectors helps: instead of focusing on types, the comparison can be made at the level of the linear subspaces defined by the projectors. Comparing now becomes straightforward as it amounts to computing the \textit{overlap} between two subspaces. 
Moreover, this overlap does not require to be computed explicitly; the projector characterizing it -- which is a function of the projectors associated with the compared subspaces -- is enough. Hence the comparison of two types, simplified into finding the overlap of their associated subspaces, can be further simplified into algebraic operations on their associated projectors. For any two projectors on operator systems $\Proj{}{}$ and $\Proj{'}{}$, we define the operation $\cap$ (nicknamed `cap') as a bilinear function resulting in a new projector $\Proj{}{} \cap \Proj{'}{}$ called their \textit{intersection} and characterizing the overlap of their two associated subspaces. Obviously, this is also a projector on an operator system. 

\subsubsection{Intersection, union, and lattice of projectors}
If the projectors $\Proj{}{}$ and $\Proj{'}{}$ are commuting, the intersection is simply their composition:
\begin{equation}\label{eq:cap}
	\Proj{}{}\cap \Proj{'}{} \equiv \Proj{}{}\circ \Proj{'}{} \:,
\end{equation}
and it commutes with both $\Proj{}{}$ and $\Proj{'}{}$. 
Similarly, the following will be called the \textit{union} of $\Proj{}{}$ and $\Proj{'}{}$:
\begin{equation}\label{eq:cup}
	\Proj{}{}\cup \Proj{'}{}\equiv \Proj{}{} + \Proj{'}{} - \Proj{}{}\cap \Proj{'}{}\:.
\end{equation}
It characterizes the linear span of their two subspaces and the connector $\cup$ will be nicknamed `cup'. 
(See App. \ref{sec:projos_prop} for the mathematical details of $\cap$ and $\cup$.)

With these, one can prove that an operator system characterized by $\Proj{'}{}$ is contained within another characterized by $\Proj{}{}$ by showing either of the following:
\begin{subequations}\label{eq:inclusion_condition}
    \begin{align}
        \Proj{}{} \cap \Proj{'}{} = \Proj{'}{} \:;\\
        \Proj{}{} \cup \Proj{'}{} = \Proj{}{} \:.
    \end{align}
\end{subequations}
In terms of projectors, these conditions will be concisely noted
\begin{equation}\label{eq:subset_def}
    \Proj{'}{} \subseteq \Proj{}{} \:.
\end{equation}
In the following, any such equation comparing two projectors (i.e., inclusions and equivalences) should be understood as a statement about their images.
When comparing state structures, the inclusion of projectors is sufficient to show inclusion (up to normalization) since the other constraint, positivity, is common to all state structures.

As it turns out, the projectors form an algebra over $\mathbb{C}$ under the connectives $\{\cap,\cup\}$ (see App. \ref{sec:projos_prop} for the proof). And because every operator system must contain the identity, every projector in the algebra is contained between the depolarizing and identity projectors,
\begin{equation}\label{eq:DsubsetPsubsetI}
    \mathcal{D} \subseteq \Proj{}{} \subseteq \mathcal{I}.
\end{equation}
Actually, this is but the statement that a set of commuting projectors closed under a `join' and `meet' operations (intersection and union, respectively) forms a specific kind of algebra called a bounded lattice. Conditions \eqref{eq:inclusion_condition} define the partial order \eqref{eq:subset_def}, and Eq. \eqref{eq:DsubsetPsubsetI} states that $\mathcal{D}$ and $\mathcal{I}$ are respectively the least and greatest elements in the lattice. 
Hence, the projectors can be compared using the partial order, which corresponds to the embedding of an operator system into another one. 
In terms of state structures, Eqs. \eqref{eq:inclusion_condition} means that any state structure contains or is contained in another one (they can also both contain and be contained in another one, in which case the state structure are equivalent up to normalization); Eq. \eqref{eq:subset_def} then implies that the state structures can be classified using this inclusion relation as it defines a partial ordering; and Eq. \eqref{eq:DsubsetPsubsetI} implies that the existence of a common least and greatest element in the ordering: every state structure is a subset of the set of density matrices (up to normalization) and a superset of the single-element set $\{\mathds{1}\}$. In addition, it also implies that new state structures can be found by taking the intersections and unions of those that are known, and that the number of new state structures found using this procedure is finite.

Nevertheless, compared to $\otimes$ and $\overline{\:\cdot\:}$, the cap and the cup are new rules that cannot be reexpressed using the $\rightarrow$ connector alone. For example, the projector $\mathcal{I}_A$ can be obtained from $\Proj{}{A}$ by using $\{\rightarrow,1,(,),\cap,\cup\}$ since $\mathcal{I}_A=\Proj{}{A} \cup \CompProj{}{A}$ (see below), whereas by using $\{\rightarrow,1,(,)\}$ one can only obtain $\Proj{}{A}$ and $\CompProj{}{A}$. So the downside of the algebra is that it goes outside of the type-theoretic framework of Ref. \cite{Perinotti2016,Bisio2018}. 

Adding $\overline{\:\cdot\:}$ as an operation in the algebra promotes it into a \textit{Boolean algebra} (or \textit{Boolean lattice}), i.e. an algebra of idempotent elements which possess a \textit{negation}.
$\overline{\:\cdot\:}$ acts as the negation since it makes $\CompProj{}{}$ complementary to $\Proj{}{}$, $\Proj{}{}\cup \CompProj{}{} = \mathcal{I}$, and since it is an involution, $\overline{\overline{\Proj{}{}}}= \Proj{}{}$. 
A Boolean algebra obeys the De Morgan law:
\begin{equation}\label{eq:deMorgan_add}
    \overline{\Proj{}{} \cap \Proj{'}{}} = \CompProj{}{} \cup \overline{\Proj{'}{}}  \:,
\end{equation}
which induces the following duality:
\begin{equation}\label{eq:inclusion_duality}
    \Proj{'}{} \subseteq \Proj{}{} \:\iff\: \CompProj{}{} \subseteq \overline{\Proj{'}{}}\:,
\end{equation}
(see App. \ref{sec:projo_Boolean} for the proofs). These two rules greatly reduce the number of computations required to determine the relations between state structures, as inclusions between functionals are dual to inclusions between the corresponding states.

\subsubsection{The sub-lattice over a no signaling subset\label{sec:Projos_NS_lattice}}
The Boolean algebra of projectors is a lattice, meaning that every projector is either the union or the intersection of some other projectors. One can then characterize every other state structure related to the one(s) under consideration by enumerating all their possible intersections, unions, and negations. 
However, such an approach only considers each state structure as a global thing, characterized by a global projector $\Proj{}{}$, resulting in a simplistic lattice $\{\Proj{}{},\CompProj{}{},\mathcal{I},\mathcal{D}\}$. But this is ignoring the finer structure of the projectors: often, they represent state structures defined on several subsystems, in which case they split as some combination of the projectors associated to each subsystem, say $\Proj{}{A}, \Proj{}{B}, \ldots, \Proj{}{K}$, linked together to form $\Proj{}{}$ by using any mixture of the rules $\{\overline{\:\cdot\:},\otimes, \rightarrow\}$, say $\Proj{}{} = (\Proj{}{A} \otimes \CompProj{}{B} \otimes \ldots )\rightarrow\Proj{}{K}$ for instance. 
It is then interesting to understand how this finer structure influences the overall shape of the lattice it generates: given a set $\{\Proj{}{A}, \Proj{}{B}, \ldots, \Proj{}{K}\}$ of $k$ base projectors, what can we say about the lattice of projectors acting on $\LinOp{\Hilb{A} \otimes \Hilb{B} \otimes \ldots \otimes \Hilb{K}}$ obtained from every possible way to combine of these $k$ projectors under $\{\overline{\:\cdot\:},\otimes, \rightarrow\}$?

Before we begin, note that the assumption that all such projectors commute with each other becomes justified by the following lemma:
\begin{lemm}\label{lem:Boolean_lattice_comm}
    The rules $\{\overline{\:\cdot\:},\otimes,\rightarrow\}$ preserve commutation.
\end{lemm}
This is proven by direct computation, and so are all the relations presented in this subsection. (These proofs are all gathered in App. \ref{sec:projo_Boolean} and \ref{sec:projo_LL}.)

Note also that negating the input of a transformation like $\CompAlg{A}\rightarrow\Alg{B}{}$ allows one to treat it as a composition akin to $\Alg{A}{}\otimes \Alg{B}{}$. Simply put, if $\Alg{A}{}\rightarrow\Alg{B}{}$ is a transformation from an input in $\LinOp{\Hilb{A}}$ to an output in $\LinOp{\Hilb{B}}$ (diagrammatically, it has an input wire at $A$ and an output wire at $B$), negating the input turns it into a bipartite state preparation since it now has an output in $\LinOp{\Hilb{A}}$ and one in $\LinOp{\Hilb{B}}$ (diagrammatically, it now has an output wire at $A$ as well as one at $B$). 
Indeed, while an element of $\Alg{A}{}\rightarrow\Alg{B}{}$ can be seen as some composition of an effect at $A$ (its input) together with a state at $B$ (its output), an element of $\CompAlg{A}\rightarrow\Alg{B}{}$ is accordingly the composition of the dual of an effect, i.e., a state, together with a state, yielding a bipartite state (with outputs at $A$ and $B$). 
This composition is however different than the no signaling one $\otimes$. Thus, there are two different ways of composing the state structures $\Alg{A}{}$ and $\Alg{B}{}$ into a bipartite state structure: $\Alg{A}{}\otimes \Alg{B}{}$ and $\CompAlg{A}\rightarrow\Alg{B}{}$. 
Under the same logic, taking the dual (negation) of one of the state structures in a no signaling composition, like $\CompAlg{A}\otimes \Alg{B}{}$ for instance, allows to treat it as a transformation: it has an input (effect of $\CompAlg{A}$) and an output (state of $\Alg{B}{}$), so it can be interpreted as a map from a valid state in $\Alg{A}{}$ to one in $\Alg{B}{}$.

In terms of projectors, the operation $\overline{\:\cdot\:}\rightarrow\cdot = \overline{\overline{\:\cdot\:} \otimes \overline{\:\cdot\:}}$ is then the `compositional' version of the transformation $\rightarrow$,
\begin{equation}
    \CompProj{}{A} \rightarrow \Proj{}{B} = \overline{\CompProj{}{A} \otimes \CompProj{}{B}} \:.
\end{equation}
It is a `larger' composition than the tensor (see App. \ref{sec:projo_prop_tensor} for the proof) as $\overline{\CompProj{}{A} \otimes \CompProj{}{B}}  \neq \Proj{}{A} \otimes \Proj{}{B}$ and 
\begin{equation}\label{eq:tensor_in_par}
    \Proj{}{A} \otimes \Proj{}{B} \subseteq \overline{\CompProj{}{A} \otimes \CompProj{}{B}}.
\end{equation}
Moreover, like the tensor, this $\overline{\:\cdot\:}\rightarrow\cdot$ composition is associative since
\begin{equation}\label{eq:gloubi_3}
    \overline{\left(\CompProj{}{A} \rightarrow \Proj{}{B}\right)} \rightarrow \Proj{}{C} = \CompProj{}{A} \rightarrow \left( \CompProj{}{B} \rightarrow \Proj{}{C}\right)\:.
\end{equation}
Indeed, the left-hand side of the above corresponds to the partition $(AB)C$, i.e. first applying $\overline{\:\cdot\:}\rightarrow\cdot$ between $A$ and $B$, then between $(AB)$ and $C$; whereas the right-hand side corresponds to the partition $A(BC)$. 
This is readily proven by expressing $\cdot\rightarrow\cdot$ as $\overline{\cdot\otimes\overline{\:\cdot\:}}$ in Eq. \eqref{eq:gloubi_3}: the left-hand side becomes $\overline{\left(\CompProj{}{A} \rightarrow \Proj{}{B}\right)} \rightarrow \Proj{}{C} = \overline{\left(\overline{\CompProj{}{A} \otimes \CompProj{}{B}}\right)} \rightarrow \Proj{}{C} = \overline{\CompProj{}{A} \otimes \CompProj{}{B} \otimes \CompProj{}{C}}$ and, similarly, the right-hand side becomes $\CompProj{}{A} \rightarrow \left( \CompProj{}{B} \rightarrow \Proj{}{C}\right) = \CompProj{}{A} \rightarrow \overline{\left( \CompProj{}{B} \otimes \CompProj{}{C}\right) } = \overline{\CompProj{}{A} \otimes  \CompProj{}{B} \otimes \CompProj{}{C} }$. This proof is direct to generalize by induction: no matter the ordering of parentheses and the number of parties, one always ends up with the same expression when expressed using the negation and the tensor, e.g.
\begin{multline}
    \overline{\Proj{}{A}} \rightarrow \left( \overline{\Proj{}{B}} \rightarrow  \ldots \left(\overline{\Proj{}{J}} \rightarrow \Proj{}{K}\right)\ldots\right)\\ 
    =\overline{(\CompProj{}{A} \rightarrow \Proj{}{B})} \rightarrow \left(   \overline{\vphantom{\Proj{}{A}}\ldots}  \rightarrow \left(\overline{\Proj{}{J}} \rightarrow \Proj{}{K}\right)\ldots\right) \\ 
    = ... \\
    =  \overline{\overline{(...\overset{...}{\overline{(\CompProj{}{A} \rightarrow \Proj{}{B})}} \rightarrow \ldots )} \rightarrow \Proj{}{J})} \rightarrow \Proj{}{K} \\  
    = \overline{\CompProj{}{A} \otimes \CompProj{}{B} \otimes \ldots \CompProj{}{J}\otimes \CompProj{}{K}}\:.
\end{multline}

Going back to the comparison of higher-order transformations, we have so far argued that it is reduced to the study of lattices of commuting projectors on operator systems. For $k$ parties, any expression $\Proj{}{}$ built using $\{\overline{\:\cdot\:},\otimes,\rightarrow\}$ is contained in a Boolean lattice closed under operations $\{\cap,\cup\}$ with smallest element $\mathcal{D}_A\otimes\ldots \otimes \mathcal{D}_K$ and biggest element $\mathcal{I}_A\otimes\ldots \otimes \mathcal{I}_K$ in accordance with Eq. \eqref{eq:DsubsetPsubsetI}. 
Of course, with an increasing number of parties these lattices quickly get huge. However, any projector in the lattice has natural subset and superset that are in general different than the tensor product of the depolarizing and of the identity projectors, meaning that projectors usually belong to a sublattice. These sets are built using exclusively one of the two possible combination rules we just discussed.

\begin{defi}[No signaling subset / Fully signaling superset]\label{def:NS_FS_sets}
    Let $\Proj{}{}$ be a projector on an operator system built using $k$ base projectors $\Proj{}{A},\Proj{}{B}\ldots\Proj{}{K}$ composed together using operations $\{\overline{\:\cdot\:},\cap,\cup,\otimes,\rightarrow\}$.\\
    The projector $\Proj{}{NS}$ to its \textbf{no signaling subset} is the largest projector contained in $\Proj{}{}$ obtained as the tensor composition of single-partite projectors. By construction, the single-partite projector associated with party $X$ can only be the base projector, its negation, or the depolarizing superoperator $\mathcal{D}_X$ so that $\Proj{}{NS}$ reads
    \begin{equation}\label{eq:no-sig_subset}
        \Proj{}{NS}\equiv \TProj{\mathcal{D}}{A} \otimes \TProj{\mathcal{D}}{B}\otimes \ldots \otimes \TProj{\mathcal{D}}{K} \subseteq \Proj{}{} \:,
    \end{equation}
    where the $\TProj{\mathcal{D}}{X}$ notation means an element chosen among $\left\{\Proj{}{X},\CompProj{}{X},\mathcal{D}_X\right\}$ depending on the exact form of $\Proj{}{}$.\\
    The projector $\Proj{}{FS}$ to its \textbf{fully signaling superset} is the smallest projector containing $\Proj{}{}$ obtained as the ``\,$\overline{\:\cdot\:}\rightarrow\cdot$'' composition of single-partite projectors. By construction, the single-partite projector associated with party $X$ can only be the base projector, its negation, or the identity superoperator $\mathcal{I}_X$ so that $\Proj{}{FS}$ reads
    \begin{equation}\label{eq:full-sig_supset}
        \Proj{}{} \subseteq \overline{\overline{\TProj{\mathcal{I}}{A}} \otimes \overline{\TProj{\mathcal{I}}{B}} \otimes  \ldots \otimes \overline{\TProj{\mathcal{I}}{J}} \otimes \overline{\TProj{\mathcal{I}}{K}}} \equiv \Proj{}{FS} \:,
    \end{equation}
    where the $\TProj{\mathcal{I}}{X}$ notation means an element chosen among $\left\{\Proj{}{X},\CompProj{}{X},\mathcal{I}_X\right\}$ depending on the exact form of $\Proj{}{}$.
\end{defi}

The reason behind the names of these sets will be explained in the next section. To see why these sets always exist, notice that the definition allows $\Proj{}{NS}$ to be equal to $\mathcal{D}_A \otimes \ldots \otimes \mathcal{D}_K$ and $\Proj{}{FS}$ to be $\mathcal{I}_A \otimes \ldots \otimes \mathcal{I}_K$ in the worst-case scenario. In general, these define a sub-lattice 
\begin{equation}
    \mathcal{D}_A \otimes \ldots \otimes \mathcal{D}_K \subseteq \Proj{}{NS} \subseteq \Proj{}{} \subseteq \Proj{}{FS} \subseteq \mathcal{I}_A \otimes \ldots \otimes \mathcal{I}_K\:,
\end{equation}
around each projector $\Proj{}{}$. 

The point is that a projector $\Proj{}{}$ is always embedded in a lattice of projectors that use the same base projectors and is contained between $\mathcal{I}$ and $\mathcal{D}$. But it is often uninteresting to compare this projector to elements that have a different interpretation. In the single-partite case for instance, the lattice is $\{\Proj{}{},\CompProj{}{},\mathcal{I},\mathcal{D}\}$, but it is not very enlightening to compare $\Proj{}{}$ with $\CompProj{}{}$ as one is the dual of the other, meaning that while one is related to state preparation, the other is related to measurements (effects), which are qualitatively different. 
The following lemma tells us that a projector $\Proj{}{}$ is also embedded in a smaller lattice of only the projectors that have the same interpretation (in terms of seeing a subsystem as either state or effect) as $\Proj{}{}$.
\begin{lemm}[No signaling sub-lattice around a projector]\label{lem:NS_lattice}
    Let $\Proj{}{}$ be a projector on an operator system built using $k$ base projectors $\Proj{}{A},\Proj{}{B}\ldots\Proj{}{K}$ composed together using operations $\{\overline{\:\cdot\:},\otimes,\rightarrow\}$. Then, it belongs to a sub-lattice closed under operations $\{\cap,\cup\}$ so that
    \begin{equation}
        \TProj{}{A} \otimes \ldots \otimes  \TProj{}{K} \subseteq \Proj{}{} \subseteq \overline{\overline{\TProj{}{A}} \otimes \ldots \otimes  \overline{\TProj{}{K}}}\:,
    \end{equation}
    where the $\TProj{}{X}$ notation means an element chosen among $\{\Proj{}{X},\CompProj{}{X}\}$ depending on $\Proj{}{}$. Importantly, the choice is the same on both sides of the equation.
\end{lemm}
This follows mainly from the property \eqref{eq:subset_tensor} that $\Proj{}{A}\subseteq\Proj{'}{A}\subseteq\Proj{''}{A} \Rightarrow \Proj{}{A}\otimes \Proj{}{B}\subseteq\Proj{'}{A} \otimes \Proj{}{B} \subseteq \Proj{''}{A}\otimes \Proj{}{B} \subseteq \overline{\CompProj{''}{A}\otimes \CompProj{}{B} }$, and from property \eqref{eq:inclusion_duality}. The complete proof is provided in App. \ref{sec:NS_lattice_proof}.

A rule of thumb for finding the no signaling projectors is to `count the number of bars' above each base projector in a given expression $\Proj{}{}$. Since negation is an involution, an odd (respectively, even) amount will indicate that the projector is (not) negated in the no signaling subset. The fully signaling projector is then obtained by taking the negation of the no signaling one. More rigorously, the projectors of Lemma \ref{lem:NS_lattice} are found by repetitively using Eq. \eqref{eq:tensor_in_par}. 
For example, to find the no signaling subset of $\left(\Proj{}{A} \otimes \left(\Proj{}{B} \rightarrow \Proj{}{C} \right)\right)\rightarrow \Proj{}{D}$, one first expresses it using negations and tensors products, $\left(\Proj{}{A} \otimes \left(\Proj{}{B} \rightarrow \Proj{}{C} \right)\right)\rightarrow \Proj{}{D} = \overline{\Proj{}{A} \otimes \overline{\Proj{}{B} \otimes \CompProj{}{C}} \otimes \CompProj{}{D}}$, then by inspection $\Proj{}{A}$ and $\Proj{}{C}$ have an odd number of negations above them (one and three, respectively) so they will be negated, whereas $\Proj{}{B}$ and $\Proj{}{D}$ have an even number (two and two) so they will not be. On that account, the no signaling subset of $\left(\Proj{}{A} \otimes \left(\Proj{}{B} \rightarrow \Proj{}{C} \right)\right)\rightarrow \Proj{}{D}$ is $\CompProj{}{A} \otimes \Proj{}{B} \otimes \CompProj{}{C} \otimes \Proj{}{D}$ and its fully signaling superset is $\overline{\Proj{}{A} \otimes \CompProj{}{B} \otimes \Proj{}{C} \otimes \CompProj{}{D}}$

Thus, for a given $\Proj{}{}$, the state structure it defines is always contained between a no signaling one and a fully signaling one. 
Note that the position of the negations in Eqs. \eqref{eq:no-sig_subset} and \eqref{eq:full-sig_supset} are the same because of Eq. \eqref{eq:tensor_in_par}. Lemma \ref{lem:NS_lattice} indeed states that the projectors to these subspaces are actually the least and greatest elements of a sub-algebra stable under $\cap$ and $\cup$, i.e. a sub-lattice.
A rule of thumb to determine if a projector is in the no signaling sub-lattice of another one is again to `count the number of bars' above each base projector. If each number is the same modulo two, then it belongs to it, otherwise it does not. For example, the no-signaling lattice around projector $\Proj{}{A}$ is $\{\Proj{}{A}\}$ itself (in which case it is both it no signaling sub- and fully signaling super-sets), whereas the one around $\Proj{}{A} \otimes \Proj{}{B}$ is $\{\Proj{}{A} \otimes \Proj{}{B}, \CompProj{}{A} \rightarrow \Proj{}{B}\}$ because in $\CompProj{}{A} \rightarrow \Proj{}{B} = \overline{\CompProj{}{A} \otimes \CompProj{}{B}}$ there are `two bars' over both $\Proj{}{A}$ and $\Proj{}{B}$ while in $\Proj{}{A} \otimes \Proj{}{B}$ there are `zero bars' above the two. 

Lemma \ref{lem:NS_lattice} is then a `decompositional' approach to the no signaling sub-lattice: given a multipartite projector, one considers every other projector in its sub-lattice. 
This sub-lattice w.r.t. a given $\Proj{}{}$ can be obtained by defining the projectors $\TProj{}{X}$ involved in its no signaling subset as a new set of base projectors, $\TProj{}{X}\mapsto\Proj{}{X}$ (i.e., by redefining the base projectors so to incorporate the negations), and by restricting the choice of operations $\{\overline{\:\cdot\:},\otimes,\rightarrow\}$ to a choice that preserves the same no signaling subset, which, from of Eq. \eqref{eq:tensor_in_par}, is $\{\cdot\otimes\cdot,\overline{\:\cdot\:}\rightarrow\cdot\}$ (this restiction is the formal version of the rule of thumb mentioned above, since it restrict the operations to only those `adding an even number of bars over each projector'). 
The procedure we just described is then the `compositional' approach to the no signaling sub-lattice: given a set of base projectors, one constructs the no signaling sub-lattice out of it. This leads to the following counterpart to Lemma \ref{lem:NS_lattice}.
\begin{lemm}[Building a no signaling sub-lattice.]\label{theo:NS_lattice}
    Let $\{\Proj{}{A}, \Proj{}{B}, \ldots, \Proj{}{K}\}$ be a set of $k$ projectors on operator systems each associated with a specific Hilbert space $\Hilb{A},\Hilb{B}\ldots, \Hilb{K}$, i.e., a set of base projectors. 
    Any expression $\Proj{}{}$ built using each element of this set once and under the rules $\{\cdot\otimes\cdot,\overline{\:\cdot\:}\rightarrow\cdot\}$ is a projector on an operator system over $\LinOp{\Hilb{A}\otimes \ldots \otimes\Hilb{K}}$ and belongs to the no signaling lattice
    \begin{equation}
        \Proj{}{A} \otimes \ldots \otimes  \Proj{}{K} \subseteq \Proj{}{} \subseteq \overline{\CompProj{}{A} \otimes \ldots \otimes  \CompProj{}{K}}\:.
    \end{equation}
\end{lemm}
This will be proven for a more general setting in Proposition \ref{prop:lattice}, but we first need to discuss the notion of no signaling to better understand why these sub-lattices are called `no signaling'. This will be the topic of the next section. 

If needed, concrete examples of the use of the algebra of projectors developed in this section are provided in App. \ref{sec:examples_QT}. As we mentioned, the algebra is well-defined for any base projector commuting with the identity and depolarizing projectors. Using base projectors different from ones that can be obtained from $\mathcal{I}$ and the operations of the algebra amounts to building a theory that is different than quantum theory and its higher-order generalizations. An example of such a toy model, which we call `biased quantum theory', is presented in App. \ref{sec:examples_biased_QT}. 
These two appendices come with many remarks providing intuition on the notion of no signaling at the level of state structures so to motivate the next section.

\paragraph*{Remark: the algebra of projectors is (almost) Linear Logic.} With the further additions of the tensor and of the transformation operations, the Boolean algebra of projectors is lifted to an abstract lattice-like structure. Such a structure is always associated with a model of formal logic. Observe that the transformation between states is equivalent to the reverse transformation between effects,
\begin{equation}\label{eq:transfoduality}
    \Proj{}{A} \rightarrow \Proj{}{B} = \CompProj{}{A} \leftarrow \CompProj{}{B} = \overline{\Proj{}{A} \otimes \CompProj{}{B}} \:.
\end{equation}
This is a sign of the algebra being an instance of logic: it follows the logic rule that if A implies B, then not B must imply not A. `Propositions' $\Proj{}{}$ consist of $k$ `sub-propositions' $\{\Proj{}{A},\ldots\Proj{}{K}\}$ that can be composed in two different manners: using $\otimes$ or $\overline{\:\cdot\:}\rightarrow\cdot$; and each (sub)proposition can be negated using $\overline{\:\cdot\:}$. To compare two propositions, the connectives $\cap$ and $\cup$ are used, and the comparison of two propositions is itself a valid proposition. The lattice is then the set of all propositions. 
As a model of logic, the projector algebra happens to form a degenerate instance of \textit{multiplicative additive linear logic} (MALL) \cite{GIRARD1987} because its connectives obey the following De Morgan dualities:
\begin{subequations}\label{eq:deMorgan}
\begin{gather}
    \overline{\overline{\mathcal{P}}}_A = \Proj{}{A} \:;\\
    \overline{\Proj{}{A}\cap\Proj{'}{A}} = \CompProj{}{A} \cup \overline{\Proj{'}{A}} \:; \label{eq:deMorgan_addbis}\\
    \overline{\CompProj{}{A} \otimes \Proj{}{B}} = \CompProj{}{A} \rightarrow \CompProj{}{B} \label{eq:deMorgan_mult}\:.
\end{gather}
\end{subequations}
(See App. \ref{sec:projo_prop_LL} for more details.) This observation connects our formalism to the categorical approach of Ref. \cite{Kissinger_2019}, in which they noticed that the logic of higher-order quantum transformations form an instance of \textit{multiplicative linear logic} (MLL)\footnote{In our language, the observation is that the projector rules $\{\overline{\:\cdot\:}, \cdot \otimes\cdot, \overline{\:\cdot\:}\rightarrow\cdot\}$ is an instance of MLL. The fact that the algebra of the projectors characterizing quantum higher-order transformations is an instance of MLL can additionally be understood as the consequence of the correspondence between the types of Ref. \cite{Bisio2018} and $*-$autonomous categories as in Ref. \cite{Kissinger_2019}.}.  The interested reader is invited to consult the categorical counterpart of the present article \cite{Simmons2022}, in which they build a non-degenerate MALL out of higher-order quantum theory and discuss the internal logic of the characterization in a rigorous way.

\section{No signaling\label{sec:NS}}
In a previous work, the projective characterization was used to prove that the multi-round process matrix, an extension of the process matrix \cite{OCB2012}, is a linear sum of quantum combs \cite[Theorem 2]{MPM}. 
By doing so, it relates an object that typically involves indefinite causal order (the multi-round process matrix, and as a consequence the process matrix) with \textit{several} that do not (a quantum comb represents a causally ordered succession of quantum channels sharing a side channel; in the literature, this structure is usually referred to as a quantum network or as a quantum channel with memory), thus formalizing the intuitive idea that `an indefinite causal order is a superposition of \textit{several} causal orderings'. 
Formally, the result states that the projector characterizing the MPM is exactly a sum of projectors characterizing quantum combs, with each term in the sum corresponding to each different causal ordering that the MPM can feature. 

As the MPM and the quantum comb are typical instances of objects in structures that can be characterized using the projective characterization, this suggests that the projectors built using the previously defined operations can be split into terms that have a fixed signaling direction, so that the causal structure can be inferred from the analysis of the projector alone. 
Therefore, in this section, the notion of no signaling is extended to state structures in order for this decomposition to be systematized. Quasi-orthogonality will be shown to be the relevant notion for no signaling at the level of projectors, allowing in turn to define the \textit{prec}, an algebraic rule for composing projectors that encodes a one-way signaling constraint.

\subsection{Definition of no signaling for state structures \label{sec:NS_def}} 

No signaling from the output to the input can be imposed on a transformation as an extra condition. This translates the physical idea that, if the output is assumed to be in the causal future of the input, no (deterministic) operation done in the future can influence, or signal to, the past \cite{Chiribella_2011}.

The subset obeying this constraint will be noted using a ``$\prec$'' symbol, nicknamed the \textbf{prec}. 
By convention, this connector will be seen as a composition, like $\otimes$, rather than a transformation, like $\rightarrow$. 
As explained in last section, this means that the two parts of the connector are treated on equal footing; an object of $\mathscr{A} \otimes \mathscr{B}$ or of $\mathscr{A}\prec \mathscr{B}$ \textit{is composed} of objects in $\mathscr{A}$ \textit{and} in $\mathscr{B}$. This implies that $\mathscr{A}\prec \mathscr{B}$ should contain $\mathscr{A} \otimes \mathscr{B}$ as its no signaling subset. 
On the contrary, an object of $\mathscr{A} \rightarrow \mathscr{B}$ \textit{transforms} an input in $\mathscr{A}$ \textit{into} an output in $\mathscr{B}$; it needs an object of $\mathscr{A}$, and therefore it is composed of objects in $\overline{\mathscr{A}}$ and in $\mathscr{B}$. This implies that $\mathscr{A} \rightarrow \mathscr{B}$ contains $\overline{\mathscr{A}} \otimes \mathscr{B}$ as its no signaling subspace instead.
With respect to the sub-lattice of Lemma \ref{theo:NS_lattice}, we want to introduce a third connective rule, $\prec$, as an in-between $\otimes$ and $\overline{\:\cdot\:}\rightarrow\cdot$. This rule will encode a \textit{composition} that allows signaling from the state structure on its left to the state structure on its right. 
In order to define our use of the term signaling with respect to \textit{transformations}, the left-hand side of connectives $\otimes$ and $\prec$ will therefore be negated; we are looking for the particular subset $\CompAlg{A}\prec \Alg{B}{}$ of the transformation $\Alg{A}{}\rightarrow\Alg{B}{}$ that forbids signaling from output to input. In other words, we want to induce a fine-graining of the order relation \eqref{eq:tensor_in_par} into $\overline{\mathscr{A}} \otimes  \Alg{B}{} \subseteq\CompAlg{A}\prec \Alg{B}{} \subseteq \Alg{A}{}\rightarrow\Alg{B}{}$. 

The proposed condition for defining the set $\CompAlg{A}\prec \Alg{B}{}$ is the requirement that the input of the transformation should be independent of the \textit{particular choice} of effect applied at its output:
\begin{multline}\label{eq:a}
    M\in \CompAlg{A}\prec \Alg{B}{} \subseteq \ChanScr{A}{B}: \quad \forall N, N' \in \overline{\mathscr{B}},\\
    \TrX{B}{M\cdot \left(\mathds{1} \otimes N \right) } = \TrX{B}{M\cdot \left(\mathds{1} \otimes N' \right)}\:.
\end{multline}
As the identity matrix $\mathds{1}$ is always both a valid state and effect (with suitable normalization) in this work, this requirement is rephrased as the following constraint,
\begin{multline}\label{eq:caus}
 M\in \CompAlg{A}\prec \Alg{B}{}:\quad \forall N \in \overline{\mathscr{B}},\\ \TrX{B}{M\cdot \left(\mathds{1} \otimes N \right) } =\TrX{B}{M\cdot \left(\mathds{1} \otimes \frac{\mathds{1}}{c_B}\right)}\:.
\end{multline}
This is depicted in Fig. \ref{fig:notA_prec_B}. When the above is satisfied, we will use it as a definition to say that $M$ is \textit{no signaling} from $B$ to $A$ \cite{Beckman2001,Piani2006}. In the same way, one can define the no input-to-output signaling subset, $\CompAlg{A}\succ \Alg{B}{}$, by requiring the reverse condition, as drawn in Fig. \ref{fig:notA_succ_B}.

\begin{figure}[th!]
    \centering
    \subfloat[$\CompAlg{A}\leftarrow\CompAlg{B}$; by Eq. \eqref{eq:transfoduality} it is equivalent to $\Alg{A}{}\rightarrow\Alg{B}{}$ drawn in Fig. \ref{fig:A_to_B}. \label{fig:notB_to_notA}]{\includegraphics[width=.45\linewidth]{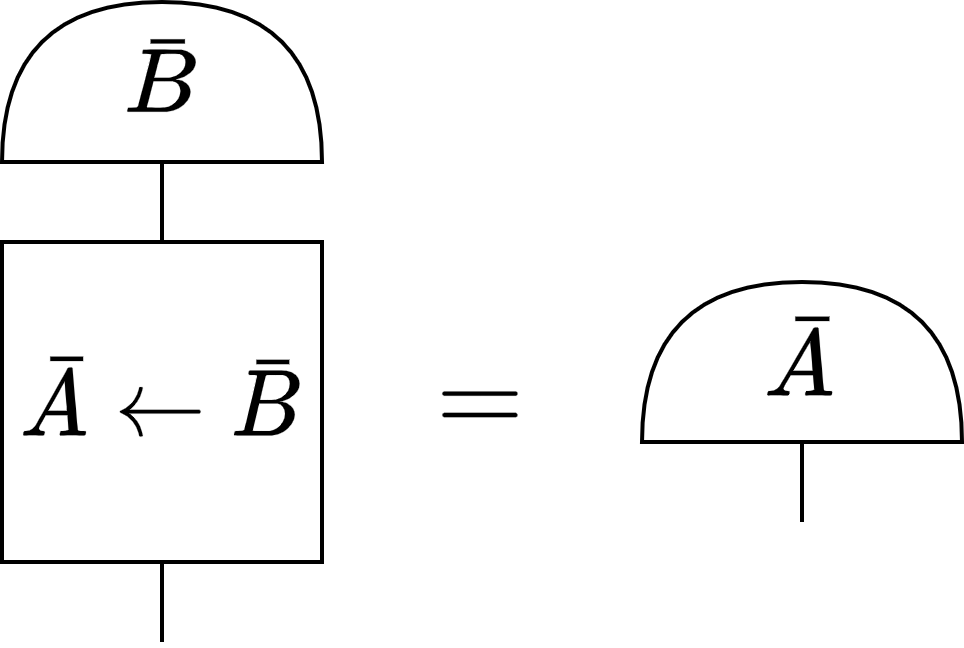}} \hfill
    \subfloat[$\CompAlg{A}\prec \Alg{B}{}$, no signaling from output $B$ to input $A$, Eq. \eqref{eq:caus} and Eq. \eqref{eq:nA_prec_B}. \label{fig:notA_prec_B}]{\includegraphics[width=.45\linewidth]{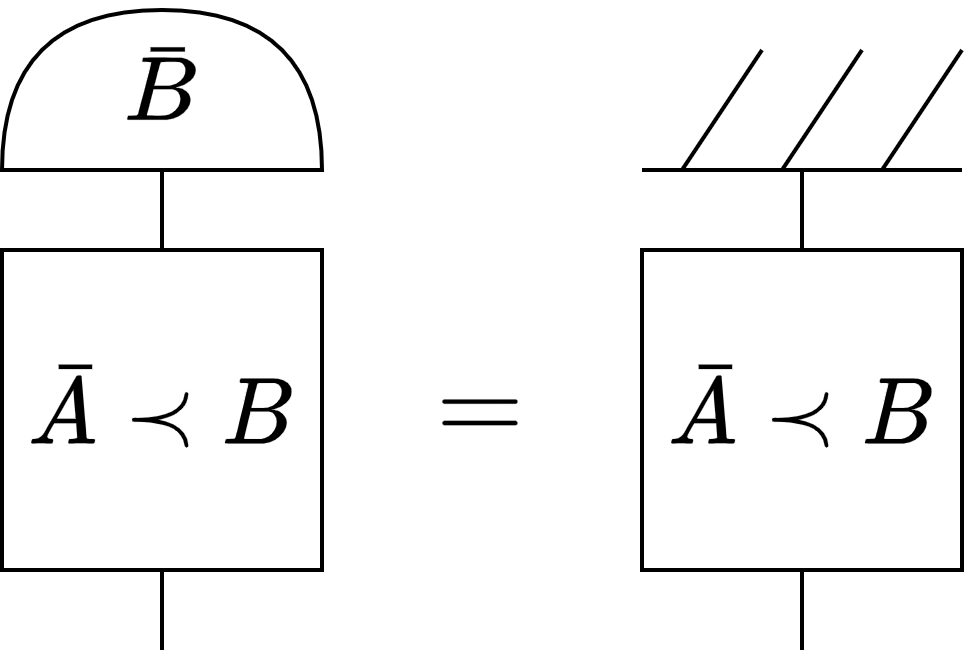}} \\
    \subfloat[$\CompAlg{A}\leftarrow\CompAlg{B}$ in the case of quantum channels where $\CompProj{}{B} = \mathcal{D}_B$. This is the usual ``no signaling from output to input'' condition used in the literature...\label{fig:StrongCausality}]{\includegraphics[width=.45\linewidth]{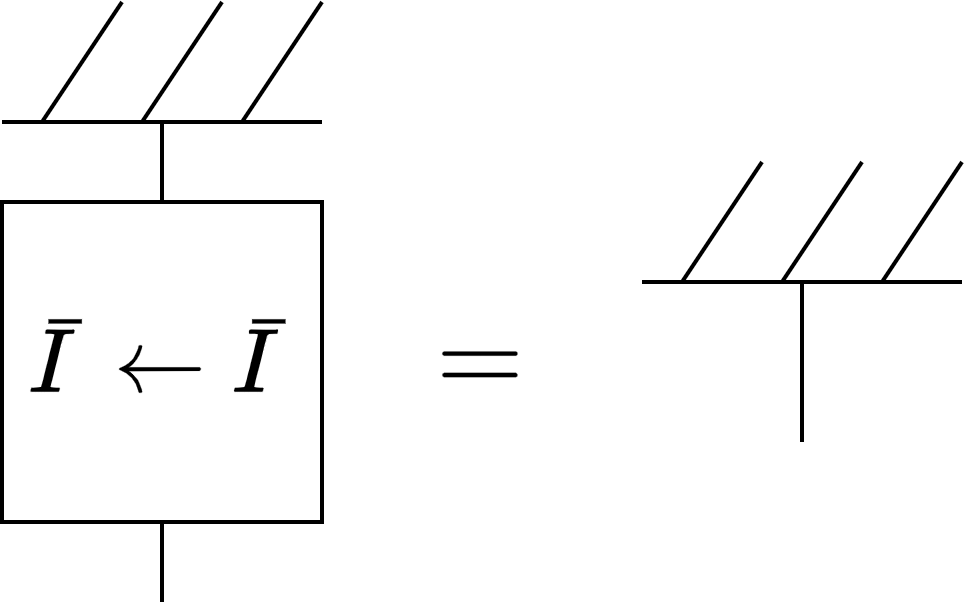}} \hfill
    \subfloat[... because Eq. \eqref{eq:caus} is trivial for the case $\CompProj{}{B} = \mathcal{D}_B$; this is captured by Lemma \ref{lem:accidental} below.\label{fig:StrongCausality_trivial}]{\includegraphics[width=.45\linewidth]{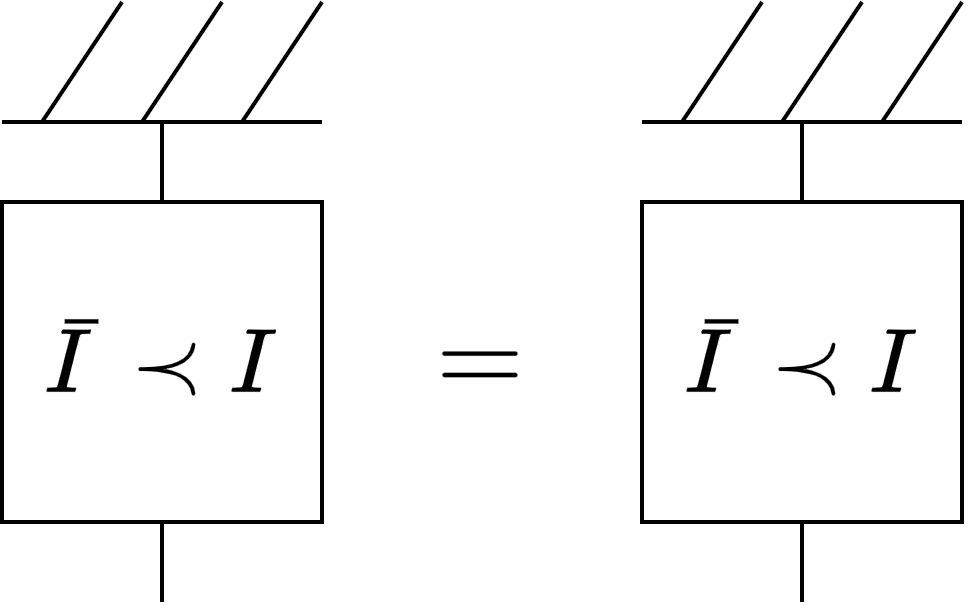}}\\
    \subfloat[Diagram for the subspace defined by Eq. \eqref{eq:caus}.\label{fig:nA_prec_B_diag}]{\includegraphics[width=.4\linewidth]{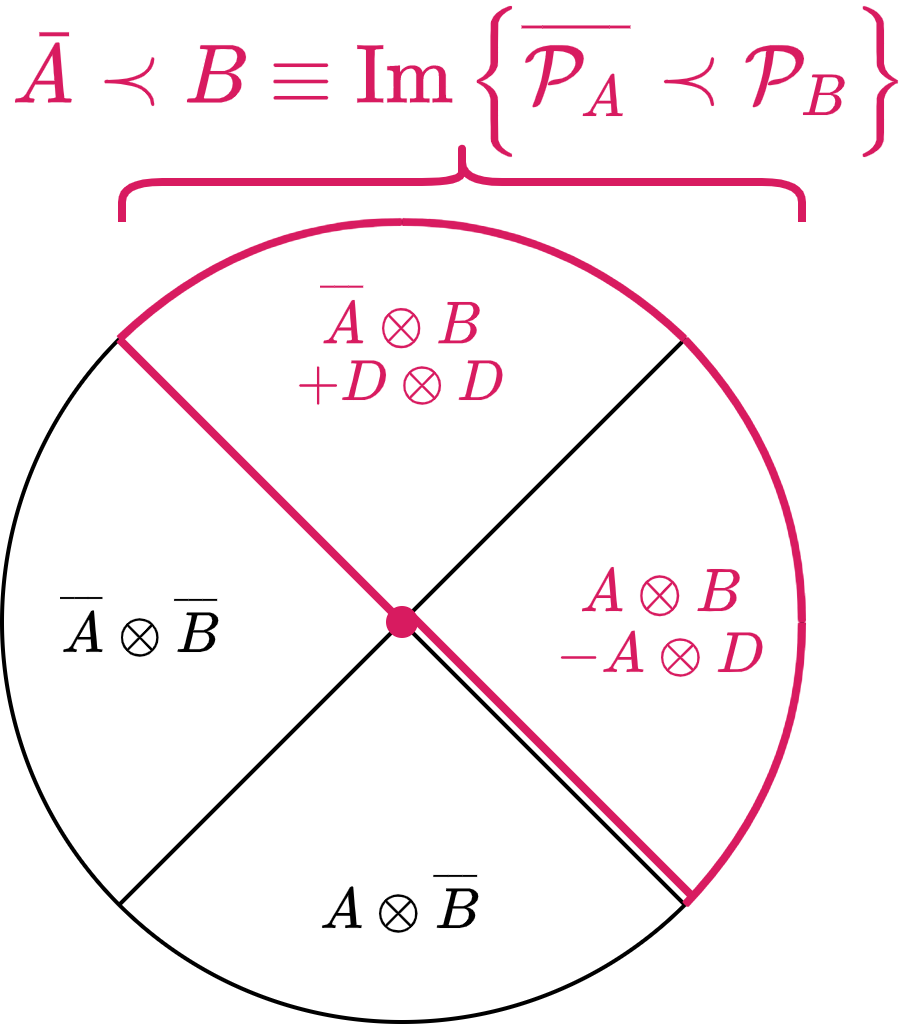}}\hfill
    \subfloat[$\CompAlg{A}\succ \Alg{B}{}$, no signaling from input to output, Eq. \eqref{eq:nA_succ_B}.\label{fig:notA_succ_B}]{\includegraphics[width=.45\linewidth]{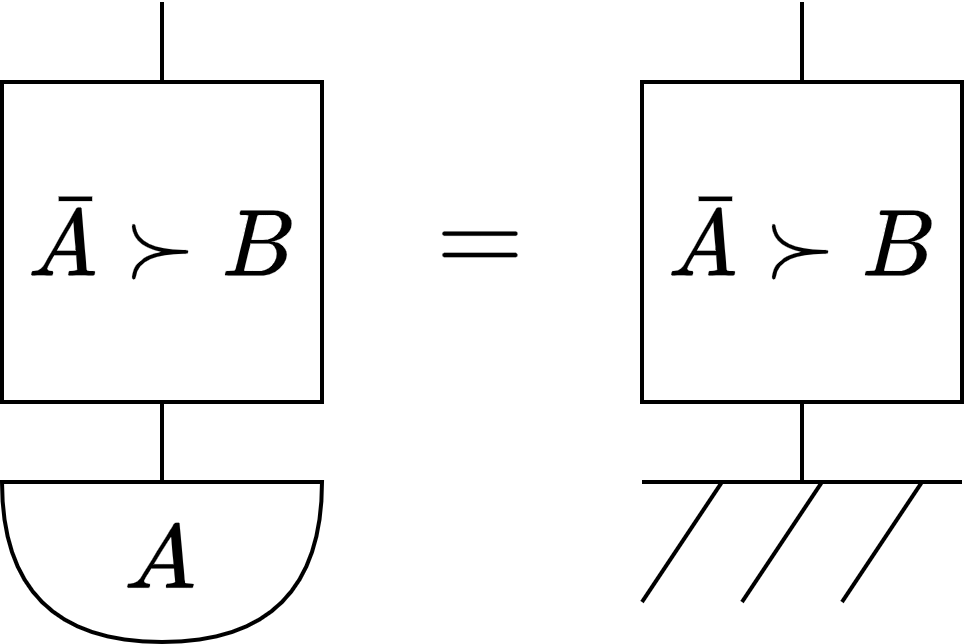}}
    \caption{Graphical interpretation of the defining condition of a transformation and of its no signaling from output to input subset. On the first line, the condition to be a valid element of $\mathscr{A} \rightarrow \mathscr{B}$ is represented as the equivalent condition $\CompAlg{A} \leftarrow \CompAlg{B}$ in \ref{fig:notB_to_notA}. The extra condition that the output cannot use the transformation to signal to the input, Eq. \eqref{eq:caus}, is represented in \ref{fig:notA_prec_B}. On the second line, the same two conditions are represented in the case where $\CompProj{}{B} = \mathcal{D}_B$. On the last line, the diagram representing the subspace defined by condition \eqref{eq:nA_prec_B} of Fig. \ref{fig:notA_prec_B} is drawn in Fig. \ref{fig:nA_prec_B_diag} as an illustration of Prop. \ref{prop:semi-causal_compo}. Finally, the reverse condition, no signaling from input to output, Eq. \eqref{eq:nA_succ_B}, has its graphical representation in Fig. \ref{fig:notA_succ_B}.}
    \label{fig:oui}
\end{figure}

For example, if $M \in \LinOp{\Hilb{A_0} \otimes \Hilb{A_1} \otimes \Hilb{B_0} \otimes \Hilb{B_1}}$ is a bipartite quantum channels, the local inputs on Alice's side has the form $\rho_{A_0} \otimes \mathds{1}_{A_1} \in \LinOp{\Hilb{A_0} \otimes \Hilb{A_1}}$ where $\rho$ is a quantum state, and the reduced single-partite channel seen by Bob after Alice has acted is given by $\TrX{A_0A_1}{M\: \cdot\: (\rho_{A_0} \otimes \mathds{1}_{A_1} \otimes \mathds{1}_{B_0} \otimes \mathds{1}_{B_1})}$ (see examples \ref{sec:examples_single_channel} and \ref{sec:examples_bipartite_channel} if this is unclear). We then say that the channel $M$ is no signaling from Alice to Bob, if, no matter which input $\rho \otimes \mathds{1}$ Alice chooses, Bob always sees the same channel on his side. In this case, she has a single choice of input at $A_1$, $\mathds{1}$, but she can choose any quantum state she wants at $A_0$. The channel is therefore no signaling from Alice to Bob if the following holds for any two states $\rho, \sigma$, 
\begin{multline}
    \TrX{A_0A_1}{M\: \cdot\: (\rho_{A_0} \otimes \mathds{1}_{A_1} \otimes \mathds{1}_{B_0} \otimes \mathds{1}_{B_1})} =\\ \TrX{A_0A_1}{M\: \cdot\: (\sigma_{A_0} \otimes \mathds{1}_{A_1} \otimes \mathds{1}_{B_0} \otimes \mathds{1}_{B_1})} \:.
\end{multline}

\paragraph*{Remark: relation with the previous definitions (see Refs. \cite{Beckman2001,Piani2006,Chiribella2009,Kissinger_2019} ).}In App. \ref{app:NS=QO}, condition \eqref{eq:a} is rephrased to correspond to a constraint on the probability distribution of the outcomes associated with the resolutions of some $N\in \CompAlg{B}$ acting on a bipartite state $\Alg{A}{}\otimes \Alg{B}{}$. This is done in order to link it to the usual theory-independent notion of no signaling used in nonlocality. 
With respect to the notion used in the case of quantum channel formalism (which gives Eq. \eqref{eq:NS_usual} below), our notion of no signaling actually requires more conditions, but recovers the quantum channel notion as a special case. Indeed, observe that due to Eq. \eqref{eq:transfoduality}, a map $M \in \mathscr{A}\rightarrow \mathscr{B}$ can be equivalently interpreted as $\CompAlg{A} \leftarrow \CompAlg{B}$. Consequently, $M$ must obey $\forall N_B \in \CompAlg{B}$, $\TrX{B}{M\left(\mathds{1}_A\otimes N_B\right)} \in \CompAlg{A}$ by definition (see Fig. \ref{fig:notB_to_notA}), and since $\CompAlg{A} = \{\mathds{1}_A\}$, we attain the usual $M \in \mathscr{A}\rightarrow \mathscr{B} \iff$
\begin{equation}\label{eq:NS_usual}
    \TrX{B}{M\left(\mathds{1}_A\otimes \mathds{1}_B\right)} = \TrX{B}{M} = \mathds{1}_A \:,
\end{equation}
which is used in previous works \cite{Chiribella2009,Kissinger_2019} as a condition to check if $M$ is a valid quantum channel \textit{but at the same time} as one to check if there is no signaling from output to input (depending on the source, this condition, Eq. \eqref{eq:NS_usual}, which is graphically depicted in Fig. \ref{fig:StrongCausality}, is also called ``causality condition'' \cite{Kissinger_2019}). 
The reason why using Eq. \eqref{eq:NS_usual} for the definition of the quantum channel, as well as the extra condition that the output is no signaling to the input, instead of using Eqs. \eqref{eq:NS_usual} \textit{and} \eqref{eq:caus}, is that Eq. \eqref{eq:caus} becomes tautological in that case, so it is automatically satisfied (see Fig. \ref{fig:StrongCausality_trivial}). 
However, this equivalence between a general transformation and a no signaling one is only valid for the case of density matrices (as proven explicitly in Lemma \ref{lem:accidental} below) and hides the subtle difference between the two definitions in the general case: being a valid transformation between state structures does not guarantee that the output cannot be used to signal to the input. Graphically, the condition represented in Fig. \ref{fig:notB_to_notA}) does not imply the one in Fig. \ref{fig:notA_prec_B}). 

\subsection{Projective characterization}
Equation \eqref{eq:caus} can be rewritten as
\begin{multline}\label{eq:causal_notA_to_B}
    \forall N \in \overline{\mathscr{B}},\\ \TrX{B}{M\cdot \left(\mathds{1} \otimes N \right) } =\frac{1}{d_B}\TrX{B}{M}\TrX{}{N}\:.   
\end{multline}
This `factorization' of the partial trace is reminiscent of the property \eqref{eq:QO} of quasi-orthogonal subspaces. In App. \ref{app:NS=QO}, it is proven that our definition of no signaling is indeed enforced by the following subsystem-wise notion of quasi-orthogonality.
\begin{lemm} \label{theo:causal_map}
    Let $M$ be an operator in $\LinOp{\Hilb{A}\otimes \Hilb{B}}$, let $V$ be one in $\mathscr{A} \subset \LinOp{\Hilb{A}}$, and let $N$ be one in $\LinOp{\Hilb{B}}$. Then a necessary and sufficient condition for:
    \begin{equation}\label{eq:partTrAB=TrATrB}
        \TrX{A}{M \cdot \left(V \otimes N\right)} = \frac{\TrX{A}{M}\TrX{}{V}}{d_A} \cdot N \:,
    \end{equation}
    to hold for all $V$ and $N$ is that 
    \begin{equation}\label{eq:nAotimesiB}
        \left(\CompProj{}{A} \otimes \mathcal{I}_B \right)\left\{M\right\} = M \:. 
    \end{equation}
    That is to say, that $M$ belongs to $\overline{\mathscr{A}} \otimes \LinOp{\Hilb{B}}$.
\end{lemm}
The proof is presented in App. \ref{sec:proof_causal_map}.
The above lemma states that since no signaling condition on $M$ --here encoded as Eq. \eqref{eq:partTrAB=TrATrB}-- is a linear constraint, there exists a subspace of all the operators obeying this conditions, and this subspace is the operator system $\overline{\mathscr{A}} \otimes \LinOp{\Hilb{B}}$. 

Now, the state structure $\overline{\mathscr{A}} \prec \mathscr{B}$ has been defined as the subset of the state structure $\mathscr{A}\rightarrow\mathscr{B}$ whose elements are no signaling from the output to the input. So, alongside the valid transformation constraint, encoded by the projector $\Proj{}{A} \rightarrow \Proj{}{B}$ (according to Prop. \ref{theo:det_map}), the extra constraint \eqref{eq:a} must be added, and it is encoded by the projector $\mathcal{I}_A \otimes \Proj{}{B}$ (according to Lemma \ref{theo:causal_map}). As two projective constraints must be simultaneously satisfied, we can use the $\cap$ operation to combine them into a unique projector:
\begin{multline}\label{eq:transfo_cap_lemma}
    \left(\Proj{}{A}\rightarrow\Proj{}{B}\right)\cap\left(\mathcal{I}_A \otimes \Proj{}{B}\right)=\\
     \left(\mathcal{I}_A\otimes\mathcal{I}_B - \Proj{}{A}\otimes\CompProj{}{B} + \mathcal{D}_A \otimes \mathcal{D}_B\right) \cap \left(\mathcal{I}_A \otimes \Proj{}{B}\right)\\
    = \mathcal{I}_A \otimes \Proj{}{B} - \Proj{}{A} \otimes \mathcal{D}_B + \mathcal{D}_A \otimes \mathcal{D}_B\:.
\end{multline}
This projector, as it turns out, is a projector on operator system. Besides, neither positivity nor normalization has been affected by the addition of the new constraint. As a consequence, $\overline{\mathscr{A}} \prec \mathscr{B}$ is a state structure.
\begin{prop}[One-way signaling transformation] \label{prop:semi-causal_trans}
Let $\mathcal{M} \in \LinOpB{\LinOp{\Hilb{A}}}{\LinOp{\Hilb{B}}}$ be a structure-preserving map between state structures $\mathscr{A}$ and $\mathscr{B}$ as in Def. \ref{def:struc_pres}. 
Let $\{M\}\equiv \CompAlg{A}\prec \Alg{B}{}$ be the subset of all $\mathcal{M}$ in CJ representation that obey the no signaling from output to input constraint, Eq. \eqref{eq:a}, $\forall N,N' \in \CompAlg{B}:$ $\TrX{B}{M\cdot \left(\mathds{1} \otimes N \right) } = \TrX{B}{M\cdot \left(\mathds{1} \otimes N'\right)}$.
Then, $\CompAlg{A}\prec \Alg{B}{}$ is characterized by
\begin{subequations}\label{eq:det_semi-causal_B!<nA}
    \begin{gather}
        M\in \CompAlg{A}\prec \mathscr{B} \iff \notag \\
        M \geq 0 \:,\label{eq:det_semi-causal_B!<nA_pos}\\
        \TrX{}{M} = c_{\overline{A}}c_B \label{eq:det_semi-causal_B!<nA_norm}\:,\\
        \left(\mathcal{I}_A \otimes \Proj{}{B} - \Proj{}{A} \otimes \mathcal{D}_B + \mathcal{D}_A \otimes \mathcal{D}_B\right)\{M\} = M\:, \label{eq:det_semi-causal_B!<nA_proj}
    \end{gather}
    \end{subequations}
    and it is a state structure as in Def. \ref{def:struct} since
    \begin{equation}\label{eq:nA_prec_B_proj}
        \mathcal{I}_A \otimes \Proj{}{B} - \Proj{}{A} \otimes \mathcal{D}_B + \mathcal{D}_A \otimes \mathcal{D}_B \equiv \CompProj{}{A}\prec\Proj{}{B}
    \end{equation}
    is a projector on an operator system as in Def. \ref{def:proj}.
\end{prop}
\begin{proof}
    Positivity and trace conditions are inherited from $M$ being a valid transformation. The projector condition is obtained directly from the above discussion, i.e. by taking the intersection of the projector on valid transformations, which is $\Proj{}{A}\rightarrow\Proj{}{B}$ by Prop. \ref{theo:det_map}, with the projector on the set of operators obeying the no signaling condition, which is $\mathcal{I}_A \otimes \Proj{}{B}$ by Lem. \ref{theo:causal_map}. That \eqref{eq:nA_prec_B_proj} is a projector on an operator system comes from the fact that it is obtained from two valid projectors on operator systems, $\Proj{}{A}\rightarrow\Proj{}{B}$ and $\mathcal{I}_A \otimes \Proj{}{B}$, composed using the cap operation, $\cap$, which preserves this property (this is proven in App. \ref{sec:projo_prec}).
\end{proof}
See Fig. \ref{fig:nA_prec_B_diag} for a diagrammatic depiction of the subspace associated with projector \eqref{eq:det_semi-causal_B!<nA_proj}. 

As is the case for the transformation, the one-way signaling transformation can be seen as the composition of the state structures $\CompAlg{A}$ and $\Alg{B}{}$ whose characteristics are encoded in a specific algebraic operation that combines projectors $\CompProj{}{A}$ and $\Proj{}{B}$. Formally,
\begin{defi}[One-way signaling composition] \label{prop:semi-causal_compo}
Let $\mathscr{A}$ and $\mathscr{B}$ be two state structures as in Eqs. \eqref{eq:det_struct}, their A-to-B one-way signaling composition is the state structure $\mathscr{A}\prec \mathscr{B} \subset \LinOp{\Hilb{A}\otimes\Hilb{B}}$ consisting of all operators $W$ characterized by the following conditions:
\begin{subequations}\label{eq:det_semi-causal_B!<A}
    \begin{gather}
        W \geq 0 \:,\label{eq:det_semi-causal_B!<A_pos}\\
        \TrX{}{W} = c_Ac_B \label{eq:det_semi-causal_B!<A_norm}\:,\\
        \left(\Proj{}{A}\prec\Proj{}{B}\right)\{W\} = W\:. \label{eq:det_semi-causal_B!<A_proj}
    \end{gather}
    \end{subequations}
where
\begin{equation}\label{eq:semi-causal_comp_A<B}
    \Proj{}{A} \prec \Proj{}{B} \equiv \mathcal{I}_A \otimes \Proj{}{B} - \CompProj{}{A} \otimes \mathcal{D}_B + \mathcal{D}_A \otimes \mathcal{D}_B\;.
\end{equation}
The same way, their B-to-A one-way signaling composition is the analogously defined state structure $\mathscr{A}\succ \mathscr{B} \subset \LinOp{\Hilb{A}\otimes\Hilb{B}}$ but instead using the projector 
\begin{equation}\label{eq:semi-causal_comp_A>B}
    \Proj{}{A} \succ \Proj{}{B} \equiv \Proj{}{A} \otimes \mathcal{I}_B - \mathcal{D}_A \otimes  \CompProj{}{B} + \mathcal{D}_A \otimes \mathcal{D}_B\;.
\end{equation}
\end{defi}

\subsection{Relating the compositions}\label{sec:NS_relations}

The main rules for the characterization of state structures through their projector rules that have been defined so far are summarized in Table \ref{tab:sum_char}. Using the new connectives introduced in Sec. \ref{sec:Projos_alg}, i.e. $\cap$ and $\cup$ (defined by Eqs. \eqref{eq:cap} and \eqref{eq:cup}, respectively), the three bipartite connectors that have been introduced so far, i.e. $\otimes, \prec,$ and $\rightarrow$ (defined by Eqs. \eqref{eq:tensor}, \eqref{eq:semi-causal_comp_A<B}) and \eqref{eq:proj_A_to_B}) can be related together in the algebra of projectors. 
\begin{table*}[htb]
    \centering
    \begin{tabular}{|l|l|c|c|}
        \hline
        Name & Characterization & Proj. rule & Rule nickname \\
        \hline
        State & Definition \ref{def:struct}; Eqs. \eqref{eq:det_struct} & $\Proj{}{A}$ & / \\
        Effect & Proposition \ref{theo:det_fctal}; Eqs. \eqref{eq:det_fctal} & $\CompProj{}{A}$ & Negation \\
        \hline
        2-ways Sign. & Proposition \ref{theo:det_map}; Eqs. \eqref{eq:det_map} & $\Proj{}{A} \rightarrow \Proj{}{B}$ & Transformation \\
        1-way Sign. & Proposition \ref{prop:semi-causal_trans}; Eqs. \eqref{eq:det_semi-causal_B!<nA} & $\CompProj{}{A} \prec \Proj{}{B}$ & Prec\\
        No Sign. & Proposition \ref{prop:causal_trans}; Eqs. \eqref{eq:causal_trans} & $\CompProj{}{A} \otimes \Proj{}{B}$ & Tensor \\
        \hline
    \end{tabular} 
    \caption{Summary of the characterization rules for state structure of states, effects, and the three ways to transform from state structure $\mathscr{A}$ to $\mathscr{B}$.\label{tab:sum_char}}
\end{table*}

For instance, consider the four ways of composing a state in $\Alg{A}{}$ with an effect in $\Alg{B}{}$. They lead to four projectors all having $\CompProj{}{A} \otimes \Proj{}{B}$ as their no signaling subset and the following relations can be inferred:
\begin{subequations}\label{eq:relations}
\begin{align}
    \CompProj{}{A} \otimes \Proj{}{B} &= (\CompProj{}{A} \prec \Proj{}{B}) \cap (\CompProj{}{A} \succ \Proj{}{B})\:, \label{eq:relations_1}\\
    \Proj{}{A} \rightarrow \Proj{}{B} &= (\CompProj{}{A} \prec \Proj{}{B}) \cup (\CompProj{}{A} \succ \Proj{}{B}). \label{eq:relations_2}
\end{align}
\end{subequations}
See App. \ref{sec:projo_prec_inclusion} for the proof.

Using that $\prec$ represents one-way signaling composition, these have a concrete physical interpretation. Consider the set $\CompAlg{A} \otimes \mathscr{B}$ first. Its elements are transformations from $\Alg{A}{}$ to $\Alg{B}{}$ as each $M \in \CompAlg{A} \otimes \mathscr{B}$ satisfies $\forall V \in \Alg{A}{}$:
\begin{equation}\label{eq:A_to_B}
    \TrX{A}{M\cdot (V\otimes \mathds{1})} \in \mathscr{B} \:.
\end{equation}
This is because its linear span is included in the one of the transformations; projector-wise, this is the relation $\CompProj{}{A} \otimes \Proj{}{B} \subseteq \Proj{}{A}\rightarrow \Proj{}{B}$. As the tensor product lives in a subspace of the transformation, each linear constraint obeyed by the latter must be obeyed by the former as well, so this includes Eq. \eqref{eq:A_to_B}. 
In addition to that, the right-hand side of Eq. \eqref{eq:relations_1} puts two extra conditions similar to Eq. \eqref{eq:caus} on $\CompAlg{A} \otimes \mathscr{B}$ (see Figs. \ref{fig:notA_succ_B} and \ref{fig:notA_prec_B}):
\begin{prop}[No signaling transformation] \label{prop:causal_trans}
Let $\mathcal{M} \in \LinOpB{\LinOp{\Hilb{A}}}{\LinOp{\Hilb{B}}}$ be a structure-preserving map between state structures $\mathscr{A}$ and $\mathscr{B}$ as in Def. \ref{def:struc_pres}. 
Let $\{M\}\equiv \CompAlg{A}\otimes \Alg{B}{}$ be the subset of all $\mathcal{M}$ in CJ representation that obeys the no signaling from input to output as well as no signaling from output to input constraints: $\forall V \in \Alg{A}{}, \forall N \in \overline{\mathscr{B}}$,
\begin{subequations}\label{eq:A_to_B_one_way}
\begin{align}
     &\TrX{A}{M\cdot \left(V \otimes \mathds{1}\right) } =\TrX{A}{M\cdot \left(\frac{\mathds{1}}{c_{\overline{A}}} \otimes \mathds{1}\right)}\:; \label{eq:nA_succ_B}\\
     & \TrX{B}{M\cdot \left(\mathds{1} \otimes N \right) } =\TrX{B}{M\cdot \left(\mathds{1} \otimes \frac{\mathds{1}}{c_B}\right)}\:. \label{eq:nA_prec_B}
\end{align}
\end{subequations}

Then, $\CompAlg{A}\otimes \Alg{B}{}$ is a state structure characterized by
\begin{subequations}\label{eq:causal_trans}
    \begin{gather}
        M\in \CompAlg{A}\otimes \mathscr{B} \iff \notag \\
        M \geq 0 \:,\label{eq:causal_trans_pos}\\
        \TrX{}{M} = c_{\overline{A}}c_B \label{eq:causal_trans_norm}\:,\\
        \left(\CompProj{}{A}\otimes \Proj{}{B}\right)\{M\} = M\:. \label{eq:causal_trans_proj}
    \end{gather}
    \end{subequations}
\end{prop}
\begin{proof}
The only difference with Prop. \ref{prop:semi-causal_trans} is the projector condition, $\left(\CompProj{}{A} \otimes \Proj{}{B}\right)\{M\} = M $. By using twice Lemma \ref{theo:causal_map}, Eqs. \eqref{eq:causal_trans} constraint the subspace to the intersection $\left(\Proj{}{A}\rightarrow\Proj{}{B}\right)\cap\left(\mathcal{I}_A \otimes \Proj{}{B}\right)\cap\left(\CompProj{}{A} \otimes \mathcal{I}_B\right)$. Grouping the two terms to the left, this is equal to $\left(\Proj{}{A}\rightarrow\Proj{}{B}\right)\cap\left(\CompProj{}{A} \otimes \Proj{}{B}\right)$. From Eq. \eqref{eq:tensor_in_par}, this further simplifies into $\CompProj{}{A} \otimes \Proj{}{B}$.
\end{proof}

Going back to the physical interpretation of \eqref{eq:relations_1}, the pair of conditions \eqref{eq:A_to_B_one_way} each corresponds to a term in the right-hand side of the equation. As these are linked by a `cap', this reveals the set $\CompAlg{A} \otimes \mathscr{B}$ as the set of transformations that are compatible with no signaling from system A to system B \textbf{and} from B to A at the same time; it is no signaling in both directions, whence the name. 
The same way, the set $\Alg{A}{} \rightarrow \mathscr{B}$ is the one of transformations that respect no signaling from A to B \textbf{or} B to A. Thus, it may allow linear combinations of terms allowing signaling in each direction and consequently to indefinite causal order. 
The sets $\CompAlg{A}\prec \Alg{B}{}$ and $\CompAlg{A}\succ \Alg{B}{}$ lie in between, as they permit signaling in only one direction. Their elements indeed obey condition \eqref{eq:A_to_B} but only one of the conditions \eqref{eq:A_to_B_one_way}. 
This discussion underlies the following chain of inclusions:
\begin{subequations}\label{eq:inclusions}
    \begin{gather}
        \CompProj{}{A} \otimes \Proj{}{B} \subseteq \CompProj{}{A} \prec \Proj{}{B} \subseteq \Proj{}{A}\rightarrow \Proj{}{B}\:,\\
        \CompProj{}{A} \otimes \Proj{}{B} \subseteq \CompProj{}{A} \succ \Proj{}{B} \subseteq \Proj{}{A}\rightarrow \Proj{}{B} \:.
    \end{gather}
\end{subequations}
See Fig. \ref{fig:compos} for a diagrammatic interpretation.

\begin{figure}[thb]
    \centering
    \includegraphics[width=\linewidth]{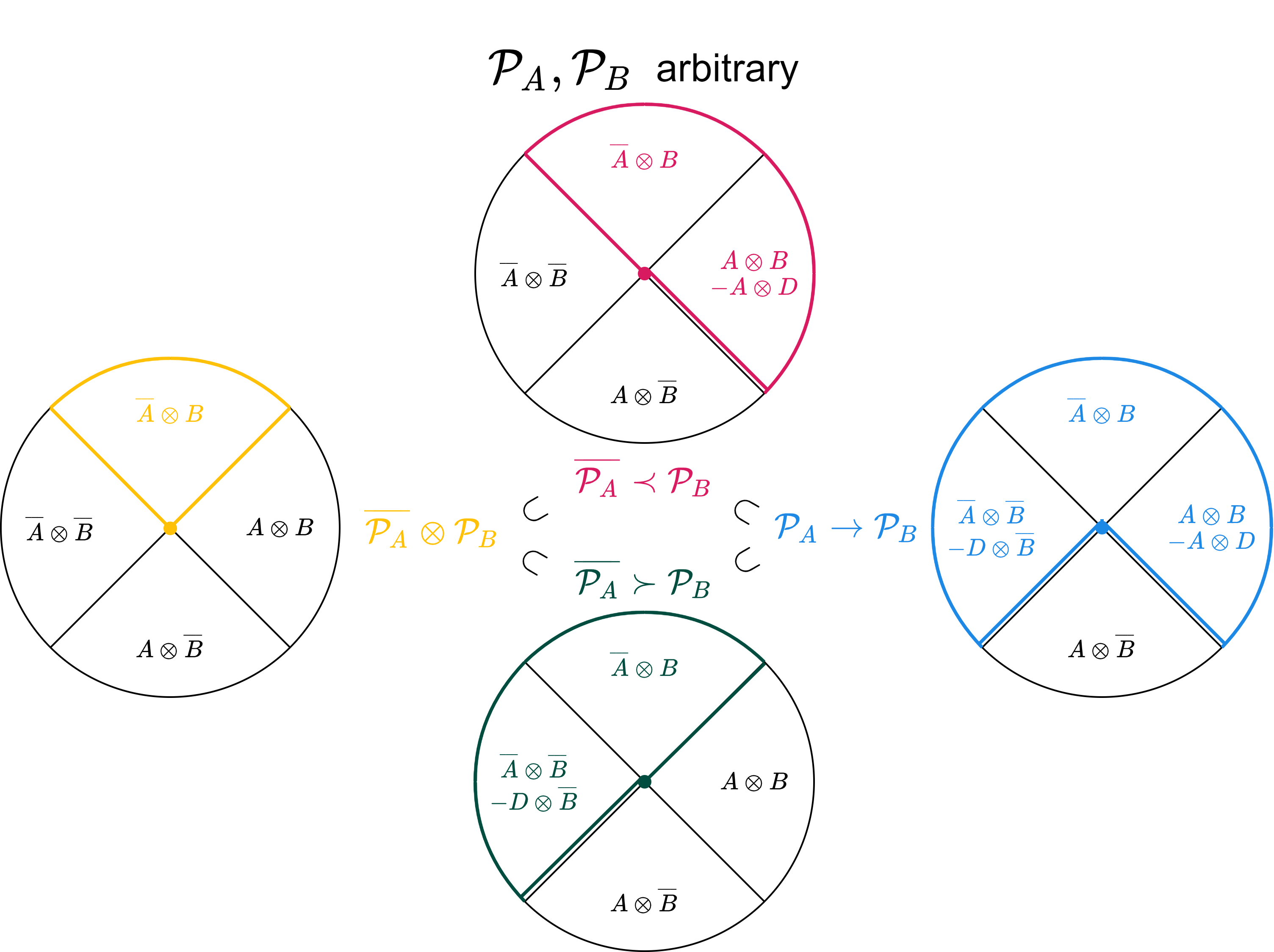}
    \caption{Inclusions \eqref{eq:inclusions} between the supports of the four bipartite composite projectors. These are the compositions involved in the characterization of the three kinds of transformations, Props. \ref{theo:det_map}, \ref{prop:semi-causal_trans}, and \ref{prop:causal_trans} (in, respectively, blue, pink, and yellow). The blue part corresponds to the span of operators obeying condition \eqref{eq:A_to_B}, the pink (resp., green) part to the one obeying \eqref{eq:A_to_B} and \eqref{eq:nA_prec_B} (resp., \eqref{eq:nA_succ_B}), and the yellow part to the one obeying \eqref{eq:A_to_B}, \eqref{eq:nA_succ_B} and \eqref{eq:nA_prec_B}. See also the discussion around Fig. \ref{fig:diag_compo} in App. \ref{sec:projo_prec_inclusion}.}
    \label{fig:compos}
\end{figure}

\paragraph*{Remark: defining the tensor product as no signaling.} It can be noticed that the characterization of the no signaling transformation did not require $M$ to be a valid transformation, i.e. an element of $\Proj{}{A}\rightarrow\Proj{}{B}$, only that it is a state structure. The no signaling composition of two state structures can be characterized as a corollary of Lemma \ref{theo:causal_map} solely by requiring the two one-way signaling conditions to be valid simultaneously.
\begin{coro}\label{theo:tensor}
The set of trace-normalised positive operators $W\in \LinOp{\Hilb{A}\otimes\Hilb{B}}$ obeying the following two no signaling conditions: $\forall N_A\in \CompAlg{A}{}, \forall N_B \in \overline{\mathscr{B}}$,
\begin{subequations}\label{eq:A_otimes_B_def}
\begin{align}
     &\TrX{A}{\left(N_A \otimes \mathds{1}\right)\cdot W } =\TrX{A}{ \left(\frac{\mathds{1}}{c_{A}} \otimes \mathds{1}\right) \cdot W}\:, \\
     & \TrX{B}{\left(\mathds{1} \otimes N_B \right) \cdot W } =\TrX{B}{ \left(\mathds{1} \otimes \frac{\mathds{1}}{c_B}\right)\cdot W}\:,
\end{align}
\end{subequations}
is the state structure of Definition \ref{prop:tensor} characterized by Eqs. \eqref{eq:det_tensor}.
\end{coro}
\begin{proof}
    A set of trace-normalized positive operators belonging to the space $\LinOp{\Hilb{A}\otimes\Hilb{B}}$ is characterized by the projector $\mathcal{I}_A\otimes \mathcal{I}_B$. Making it obey conditions \eqref{eq:A_otimes_B_def} restricts its projector to $\left(\mathcal{I}_A\otimes \mathcal{I}_B\right)\cap\left(\mathcal{I}_A \otimes \Proj{}{B}\right)\cap\left(\Proj{}{A} \otimes \mathcal{I}_B\right)$ by Lemma \ref{theo:causal_map}. Computing the intersections yields the projector $\Proj{}{A} \otimes \Proj{}{B}$ and the definition is reached. 
\end{proof}
This observation has two implications: on the one hand, it justifies its nickname of `no signaling composition'. On the other hand, it provides a `transformation-free' definition of the tensor product based only on the validity of Lemma \ref{theo:causal_map}. This is an important result for the characterization since it avoids cyclic reasoning: the characterization of the transformation in Prop. \ref{theo:det_map} requires the characterization of the tensor from Def. \ref{prop:tensor}. But without Corollary \ref{theo:tensor}, this definition is either assumed \textit{by fiat} or else it comes from Prop. \ref{prop:causal_trans} whose proof requires Prop. \ref{theo:det_map}. 

\subsection{The full algebra of projectors}\label{sec:algebra_the_algebra}
With respect to the discussion in Section \ref{sec:Projos_NS_lattice}, the further addition of the prec connector to the algebra of projectors increases the number of projectors in the Boolean lattice $\mathcal{D} \subseteq \Proj{}{} \subseteq \mathcal{I}$ under operations $\cap,\cup$. 
Any $\Proj{}{}$ can now be decomposed into base projectors $\Proj{}{A},\Proj{}{B},\ldots$ combined using any sequence of operations $\{\overline{\:\cdot\:},\otimes,\prec,\rightarrow\}$ instead of $\{\overline{\:\cdot\:},\otimes,\rightarrow\}$. 
As is the case for all operations on projectors defined so far (see Lemmas \ref{lem:Boolean_lattice} and \ref{lem:Boolean_lattice_comm}), composing using the prec yields a projector onto an operator system that preserves commutation; see App. \ref{sec:projo_prec} for the proof. Summarizing the discussion of App. \ref{sec:projo}, not only the six operations $\{\overline{\:\cdot\:},\cap,\cup,\otimes,\prec,\rightarrow\}$ studied in this paper each preserve commuting projectors on operator systems, but any combination of these operations also does, leading to the following:
\begin{theo}\label{prop:algebra}
    Let $\{\Proj{}{A},\Proj{'}{A},\Proj{''}{A},\dots\}$ be a set of commuting projectors on operator systems on a Hilbert space of operators $\LinOp{\Hilb{A}}$. Let $\{\Proj{}{B},\Proj{'}{B},\Proj{''}{B},\dots\}\in \LinOp{\Hilb{B}}$
    be a similarly defined set. Let $\Proj{}{}$ be an expression built using a combination of elements from these two sets under the rules $\{\overline{\:\cdot\:},\cap,\cup,\otimes,\prec,\rightarrow\}$ so that $\Proj{}{}$ is a projector on $\LinOp{\Hilb{A}\otimes\Hilb{B}}$. Let $\Proj{'}{}$ be another such expression. 
    Then $\Proj{}{}$ and $\Proj{'}{}$ are a pair of projectors on operator systems of $\LinOp{\Hilb{A} \otimes\Hilb{B}}$ 
    that commute.
\end{theo}
    The proof is straightforward. Its details are contained in App. \ref{sec:projo} where each individual operation in $\{\overline{\:\cdot\:},\cap,\cup,\otimes,\prec,\rightarrow\}$ is shown to map commuting projectors on operator systems to commuting projectors on operator systems, and so do their combinations.
    
When the base projectors associated with each tensor factor are fixed, the composite projectors are classified by their inclusion relations using $\cap,\cup$. As explained in Sec. \ref{sec:Projos_NS_lattice}, they form a bounded lattice, meaning that there is a finite number of possible sets of higher-order objects for a fixed number of parties, and that each set is either contained in or contains other sets, with the largest set being the fully signaling superset and the smallest being the no signaling subset. With respect to that, Eqs. \eqref{eq:relations} show that these lattices (as in Definition  \ref{def:NS_FS_sets}) can be extended to contain new, finer-grained elements built using the prec. In other words, Lemma \ref{theo:NS_lattice} can be generalized to include elements built using the prec connector, without changing its top or bottom element. This extended lattice structure of projectors is summarised in the following:
\begin{prop}[Properties of the lattice of projectors]\label{prop:lattice}
    Let $\{\Proj{}{A}, \Proj{}{B}, \ldots, \Proj{}{K}\}$ be a set of $k$ projectors on operator systems, each associated with a specific Hilbert space $\Hilb{A},\Hilb{B}\ldots, \Hilb{K}$, i.e. a set of base projectors.\\ 
    General case: Any expression $\Gamma$ built using each element of this set once and under the set of rules $\{\overline{\:\cdot\:},\cap,\cup,\otimes,\prec,\rightarrow\}$ is a projector on an operator system over $\LinOp{\Hilb{A}\otimes \ldots \otimes\Hilb{K}}$. Moreover, it belongs to a lattice closed under operations $\{\cap,\cup\}$ such that
    \begin{equation}
        \mathcal{D}_{A} \otimes \ldots \otimes  \mathcal{D}_{K} \subseteq \Gamma \subseteq \mathcal{I}_{A} \otimes \ldots \otimes  \mathcal{I}_{K}\:.
    \end{equation}
    Restricted case: Any expression $\Gamma'$ built using each element of this set once and under the restricted set of rules $\{\otimes,\prec,\overline{\:\cdot\:}\rightarrow \cdot\}$ belongs to a sub-lattice such that
    \begin{equation}
        \Proj{}{A} \otimes \ldots \otimes  \Proj{}{K} \subseteq \Gamma' \subseteq \overline{\CompProj{}{A} \otimes \ldots \otimes  \CompProj{}{K}}\:.
    \end{equation}
\end{prop}
The general case follows as a corollary of Theorem \ref{prop:algebra}, whereas the restricted one is one of Lemma \ref{lem:NS_lattice} under Eq. \eqref{eq:inclusions}. A proof is provided in App. \ref{sec:lattice_proof}. Notice
that Eqs. \eqref{eq:inclusions} characterizes the simple example of a no signaling sub-lattice (as in the restricted case) built using the base projectors $\{\CompProj{}{A},\Proj{}{B}\}$. 

The fact there can be more than one base projector associated with each subsystem is what makes the algebra versatile. Each subsystem can indeed be itself composite, for example $\LinOp{\Hilb{A}} = \LinOp{\Hilb{A_0}\otimes \Hilb{A_1}}$; in that case, the various commuting projectors $\{\Proj{}{A},\Proj{'}{A},\Proj{''}{A},\ldots\}$ to consider as bases can be for example the elementary one, $\Proj{}{A}= \Proj{}{A_0}\otimes \Proj{}{A_1}$, the one-way signaling from $A_0$ to $A_1$ one, $\Proj{'}{A}= \Proj{}{A_0}\prec \Proj{}{A_1}$, the two-way signaling one, $\Proj{''}{A}= \CompProj{}{A_0}\rightarrow \Proj{}{A_1}$, etc. 

As mentioned in Sec. \ref{sec:base_sets}, nothing forbids \textit{a priori} to consider \textit{any} family of commuting projectors on operator system as possible choices of bases for a subsystem --not only the ones built from $\mathcal{I}$ and the algebra rules $\{\overline{\;\cdot\;},\otimes, \prec, \rightarrow\}$. The biased quantum theory example of Sec. \ref{sec:examples_biased_QT} is an example of such a choice of arbitrary projectors as base projectors. In this paper, and especially in the applications of the next section, we focus on families built from $\mathcal{I}$ using the rules of the algebra only because of their interpretation as (structure-preserving) quantum supermaps with constraints on their signaling structure. 
Yet, other physically-relevant choices can be made, for example building from the family of bases $\{\mathcal{I}, \Delta\}$, where $\Delta$ projects onto diagonal operators w.r.t. a fixed basis. In that case, $\mathcal{I}$ defines the states of quantum theory and $\Delta$ its `dephased' subspace, i.e. classical states. This direction is left open for future research (a longer discussion is provided in App. \ref{sec:projos_base_sets}).

\paragraph*{Remark: the algebra of projectors is (almost) BV/R{\'e}tor{\'e} logic.} We pointed out the connection with MALL in the remark in Sec. \ref{sec:Projos_NS_lattice} and in App. \ref{sec:projo_prop_LL}. With the further addition of the prec, the algebra defined by the above theorem can be seen as an instance of BV- or R{\'e}tor{\'e} logic, as first noticed in Ref. \cite{Simmons2022} (although their choice of additives connectives is different from ours).

\subsection{The algebra of projectors recovers no signaling}

In this section, we observed that the notion of no signaling could be generalized for the case of state structures. In particular, it was noticed that the notion could be encoded by a new rule $\prec$ in the algebra of projectors, introduced under the name prec.

Bipartite state structures obtained by combining $\CompAlg{A}$ and $\Alg{B}{}$ fall into four classes, essentially determined by conditions \eqref{eq:A_to_B}, \eqref{eq:nA_prec_B}, and \eqref{eq:nA_succ_B} represented graphically in, respectively, Figs. \ref{fig:notB_to_notA}, \ref{fig:notA_prec_B}, and \ref{fig:notA_succ_B}. All of these classes must satisfy the first condition, which amounts to being a valid transformation and to which corresponds the projector $\Proj{}{A}\rightarrow \Proj{}{B}$.
If in addition condition \eqref{eq:nA_prec_B} (respectively, \eqref{eq:nA_succ_B}) is satisfied, the set is restricted to a transformation that is no signaling from the output to the input (resp., from the input to the output), to which corresponds the projector $\CompProj{}{A}\prec \Proj{}{B}$ (resp., $\CompProj{}{A}\succ \Proj{}{B}$). 
If both conditions are satisfied, the set is no signaling, to which corresponds the projector $\CompProj{}{A} \otimes \Proj{}{B}$.

Therefore, working with projectors provides a handy way of assessing the signaling structure within a state structure.

\section{Implications of characterizing signaling\label{sec:applications}}
With the formalization of the projective characterization and the introduction of the one-way signaling subsets using the prec, we now have the tools to address signaling in higher-order quantum transformations.

\subsection{Algebraic manipulations using the prec: two birds with one stone \label{sec:applications_NSandrelations}}
On purely algebraic grounds, the transformation $\rightarrow$ has two defining properties that make it inconvenient to work with. 
The first is not commuting with negation\footnote{In categorical terms, this is a manifestation of the $*-$autonomous character of the category of higher-order transformations \cite{Kissinger_2019}. If the $\rightarrow$ were to commute with the negation, the category would be compact closed instead of $*-$autonomous. This is what happens in the case of quantum theory where the projectors on both sides are the identity: $\overline{\mathcal{I}_A\otimes\mathcal{I}_B} = \overline{\mathcal{I}}_A \otimes \overline{\mathcal{I}}_B = \mathcal{D}_A\otimes\mathcal{D}_B$ and which explains why the bipartite deterministic effects of quantum theory are not two-way signaling. See the discussion in Sec. \ref{sec:applications_iso} and in App. \ref{sec:projo_prec_iso}.}: 
\begin{equation}\label{eq:neg_dont_comm}
    \overline{\Proj{}{A}\rightarrow \Proj{}{B}} \neq \CompProj{}{A} \rightarrow \CompProj{}{B}\:.
\end{equation}
This leads to difficulties in relating seemingly similar sets together using the algebra alone. Two expressions involving the same states and effects may be incomparable because one features a negation on several subsystems at once while the other does not. For example, $\overline{\Proj{}{A}\rightarrow \Proj{}{B}}= \Proj{}{A} \otimes \CompProj{}{B}$ cannot be compared with $\CompProj{}{A} \rightarrow \CompProj{}{B} = \overline{\CompProj{}{A} \otimes \Proj{}{B}}$ without explicitly doing a computation like \eqref{eq:proof_tensor_in_par} to show set inclusion.

The second inconvenient property is that $\rightarrow$ is not associative:
\begin{equation}
    \Proj{}{A} \rightarrow (\Proj{}{B} \rightarrow \Proj{}{C}) \neq (\Proj{}{A} \rightarrow \Proj{}{B}) \rightarrow \Proj{}{C} \:.
\end{equation}
While the algebraic difficulties of non-associativity can be alleviated by replacing the transformation connector ``$\rightarrow$'' with the associative ``$\overline{\:\cdot\:} \rightarrow \:\cdot\:$'', this issue still impedes the interpretation of type formulae. A naive interpretation of the arrow as a transformation may lead an unsuspecting reader to miss comparable state structures. For example, consider types $((A_0 \rightarrow A_1) \rightarrow A_2) \rightarrow A_3$ and $(A_1\rightarrow A_2) \rightarrow (A_0 \rightarrow A_3)$, the first formula depicts two cascaded transformations whose input is a channel from $A_0$ to $A_1$, whereas the second is a single transformation whose input is a channel from $A_1$ to $A_2$. It is tempting to conclude they have nothing in common. However, some computations will reveal that they both feature $\CompProj{}{A_0} \otimes \Proj{}{A_1} \otimes \CompProj{}{A_2} \otimes \Proj{}{A_3}$ as their no signaling subset (as in Def. \ref{def:NS_FS_sets}), and accordingly they should be comparable. If the base types in this example are the quantum states, then the two expressions are more than comparable: they describe the same set of \textit{quantum 2-combs} \cite{Chiribella2009}! Two expressions involving the same no signaling subset may be incomparable because there is no way to `shuffle the parentheses' to put their constituents in the same order.

Surprisingly, the prec does not lead to these issues, as it commutes with the negation:
\begin{equation}\label{eq:prec_neg}
    \overline{\Proj{}{A} \prec \Proj{}{B}} = \CompProj{}{A} \prec \CompProj{}{B} \:.
\end{equation}
And, like the tensor, it is associative:
\begin{multline}\label{eq:prec_assoc}
    (\Proj{}{A} \prec \Proj{}{B}) \prec \Proj{}{C} = \Proj{}{A} \prec (\Proj{}{B} \prec \Proj{}{C})\\ = \Proj{}{A} \prec \Proj{}{B} \prec \Proj{}{C} \:.
\end{multline}
(See Eqs. \eqref{eq:proof_commute_neg} and \eqref{eq:proof_assoc} in App. \ref{sec:projo_prec_assoc} for the proofs.) 
With these two obstacles gone, it becomes possible to compare state structures simply by algebraic manipulations of their projectors. As the tensor and the transformation split into precs, Eqs. \eqref{eq:relations}, and since the algebra obeys the De Morgan relations \eqref{eq:deMorgan}, the `impractical' tensor and transformation connectors can be expressed as precs instead. Any projector $\Proj{}{}$ built as in Theorem \ref{prop:algebra} and whose decomposition into base projectors features $\otimes$ and $\rightarrow$ can be rewritten into unions and intersections of expressions built using the prec alone.

Now, in the algebra, the cap distributes over the cup,
\begin{equation}\label{eq:cap_cup}
    (\Proj{}{A}\cup \Proj{'}{A}) \cap \Proj{''}{A} = (\Proj{}{A}\cap \Proj{''}{A}) \cup (\Proj{'}{A} \cap \Proj{''}{A}) \:.
\end{equation}
(See Eq. \eqref{eq:cap_cup_dist} in App. \ref{sec:projos_prop}.) 
Furthermore, the cap and the cup both distribute over the prec,
\begin{subequations}\label{eq:prec_cap_cup}
    \begin{gather}
        (\Proj{}{A}\cup \Proj{'}{A}) \prec \Proj{}{B} = (\Proj{}{A} \prec \Proj{}{B}) \cup (\Proj{'}{A} \prec \Proj{}{B}) \:,\\
        \Proj{}{A} \prec (\Proj{}{B}\cap \Proj{'}{B}) = (\Proj{}{A} \prec \Proj{}{B}) \cap (\Proj{}{A}\prec \Proj{'}{B}) \:,
    \end{gather}
\end{subequations}
where the above are also satisfied when caps and cups are switched, $\cap \leftrightarrow \cup$. (See App. \ref{sec:projo_prec_dist}.) Consequently, distributivity can be used to define a normal form: expressions involving these three connectors can always be rewritten so that they are put into a union of intersections of one-way signaling compositions. As these three operations are associative, the intermediate parentheses can be dropped everywhere. 

For example, in the bipartite case, there are eight possible one-way signaling compositions, or `prec chains': $\Proj{}{A}\prec\Proj{}{B}$, $\Proj{}{A}\succ\Proj{}{B}$, $\CompProj{}{A}\prec\Proj{}{B}$, etc. The claim is that expressions built using $\{\overline{\:\cdot\:},\cap,\cup,\prec\}$ can be rewritten using the distribution rules into unions of intersections of these eight prec chains like e.g., $\Proj{}{A}\succ \CompProj{}{B}$ or $\left(\Proj{}{A}\prec \Proj{}{B}\right)\cup \left(\left(\CompProj{}{A}\prec \Proj{}{B}\right)\cap \left(\CompProj{}{A}\succ \CompProj{}{B}\right)\right)$, and that expressions not in this form, like e.g., $ \left(\left(\Proj{}{A}\prec \Proj{}{B}\right)\cup \left(\CompProj{}{A}\prec \Proj{}{B}\right)\right) \cap \left(\CompProj{}{A}\succ \Proj{}{B}\right)$, can be `simplified' through distribution until it reaches a `union of intersections of prec chains', e.g., $ \left(\left(\Proj{}{A}\prec \Proj{}{B}\right)\cup \left(\CompProj{}{A}\prec \Proj{}{B}\right)\right) \cap \left(\CompProj{}{A}\succ \Proj{}{B}\right) $ is rewritten into $ \left(\left(\Proj{}{A}\prec \Proj{}{B}\right) \cap \left(\CompProj{}{A}\succ \Proj{}{B}\right) \right) \cup \left(\left(\CompProj{}{A}\prec \Proj{}{B}\right) \cap \left(\CompProj{}{A}\succ \Proj{}{B}\right) \right)$.

\begin{defi}[Normal form]\label{def:normal_form}
    Let $\Proj{}{}$ be a projector on an operator system built from a set of $k$ base projectors $\{\Proj{}{A}, \Proj{}{B}, \ldots, \Proj{}{K}\}$ composed together under operations $\{\overline{\:\cdot\:},\cap,\cup,\otimes,\prec,\rightarrow\}$ so that $\Proj{}{}$ acts on  $ \LinOp{\Hilb{A}\otimes \ldots \otimes\Hilb{K}}$. 
    Then, a \textbf{normal form} $\Gamma$ of $\Proj{}{}$ is a projector equivalent to $\Proj{}{}$ obtained as unions of intersections of expressions built from operations $\{\overline{\:\cdot\:},\prec\}$ alone:
    \begin{equation}
    \Gamma=\bigcup_{i=1}^{x}\left(\bigcap_{j=1}^{y_i} \TProj{}{{\sigma_{ij}(A)}} \prec \TProj{}{\sigma_{ij}(B)} \prec \ldots \prec \TProj{}{\sigma_{ij}(K)} \right) ,
    \end{equation}
    in which there are $x$ unions of expressions labeled by index $i$, and each expression involves $y_i$ intersections of sub-expressions labelled by index $j$, where $\sigma_{ij}$ is an element of the permutation group on $k$ element, like e.g. $\sigma_{00}(A)=A, \sigma_{01}(A)=B,\ldots$, so that each $\TProj{}{{\sigma_{ij}(A)}}$ is a choice of a base projector which is potentially negated depending on indices $i$ and $j$. Each sub-expression is thus a permutation of $\TProj{}{A} \prec \TProj{}{B} \prec \ldots \prec \TProj{}{K}$ where the position of the negations depends on $i$ and $j$ (note that the indices do not necessarily run over the full permutation group).
\end{defi} 
Remark that the convention of ordering expressions in alphabetical label order is dropped in the normal form. It is replaced by an ordering that allows reading prec chains from left to right using associativity. In other words, normal forms only use the $\prec$; its reversed version, $\succ$, must also be dropped. For example, the normal form of $(\Proj{}{A} \prec \Proj{}{B})\succ \Proj{}{C}$ is written as $\Proj{}{C} \prec \Proj{}{A} \prec \Proj{}{B}$.
\begin{theo}\label{theo:normal_form}
    Any projector $\Proj{}{}$ as in the general case of Prop. \ref{prop:lattice} has a normal form.
\end{theo}
\begin{proof}
    By Proposition \ref{prop:lattice}, it is a projector on operator system. By Eqs. \eqref{eq:relations}, the projector can be put in a form involving operations $\{\overline{\:\cdot\:},\cap,\cup,\prec\}$ only. The normal form can always be reached by first distributing negations over intersections and unions using De Morgan law \eqref{eq:deMorgan_add} and over the prec using Eq. \eqref{eq:prec_neg}, then by distributing the prec over unions and intersections using \eqref{eq:prec_cap_cup}, and finally by distributing the caps over the cups using \eqref{eq:cap_cup}. This procedure is explicitly proven in App. \ref{sec:normal_form_proof}.
\end{proof}

This implies that any expression can be brought down to the union of several one-way signaling expressions (which further decompose into intersections if needed). This way of rewriting exists for \textit{any} state structure, lifting any ambiguity induced by the negation or non-associativity. 
In the first example discussed above, the expressions reduce to $\overline{\Proj{}{A}\rightarrow \Proj{}{B}}= \left(\Proj{}{A} \prec \CompProj{}{B}\right) \cap \left(\Proj{}{A} \succ \CompProj{}{B}\right)$ and $\CompProj{}{A} \rightarrow \CompProj{}{B} = \CompProj{}{A} \rightarrow \CompProj{}{B} = \left(\Proj{}{A} \prec \CompProj{}{B}\right) \cup \left(\Proj{}{A} \succ \CompProj{}{B}\right)$. Now one can meaningfully compare them: they are both built from a combination of the same two one-way signaling transformations, but one set is the intersection and the other is the union, we thus conclude that $\overline{\Alg{A}{}\rightarrow \Alg{B}{}} \subseteq \CompAlg{A}\rightarrow \CompAlg{B}$ simply by inspecting the normal form of their respective projectors.

Moreover, the prec connector and the decomposition it induces also allow for an identification of the possible signaling directions these two sets may feature. Indeed, since the prec is exactly the one-way signaling condition, the above normal form is a breakdown of a general expression into several fully ordered expressions, or `prec chains'. 
That is, expressions built using only the prec like $\Proj{}{A}\prec\Proj{}{B}\prec\Proj{}{C}$ characterizes a set of objects with a fixed signaling order. Indeed, since $\Proj{}{A}\prec\Proj{}{B}\prec\Proj{}{C} = \left(\Proj{}{A}\prec\Proj{}{B}\right)\prec\Proj{}{C}$, $C$ cannot signal to $A$ and to $B$, and $B$ cannot signal to $A$; in the same time since $\Proj{}{A}\prec\Proj{}{B}\prec\Proj{}{C} = \Proj{}{A}\prec\left(\Proj{}{B}\prec\Proj{}{C}\right)$, $A$ might signal to $B$ and $C$ while $B$ might signal to $C$\footnote{More precisely: certain operators $W_{ABC}$ in $\Alg{A}{}\prec\Alg{B}{}\prec\Alg{C}{}$ may allow $A$ to signal \textit{deterministically} to $B$ and to $C$ by suitably choosing the (\textit{local} \textit{deteministic}) effect $N_A \in \CompAlg{A}$ she will apply on $W_{ABC}$. And the same way, $B$ may signal to $C$ by suitably choosing his effect $N_B\in \CompAlg{B}$, but he will never be able to signal to $A$, no matter his choice of $N_B $.}. 
These chains are then combined first by requiring no signaling between some subsystems, i.e. using the cap, and then by allowing signaling in both directions, i.e. using the cup. For example, to further impose no signaling between $A$ and $B$ in $\Proj{}{A}\prec\Proj{}{B}\prec\Proj{}{C}$, one intersects it with $\left(\Proj{}{A}\succ\Proj{}{B}\right)\prec\Proj{}{C}$ so that $\left(\Proj{}{A}\prec\Proj{}{B}\prec\Proj{}{C}\right) \cap \left(\Proj{}{B}\prec\Proj{}{A}\prec\Proj{}{C}\right) = \left(\Proj{}{A}\otimes\Proj{}{B}\right) \prec\Proj{}{C}$. Then, to loosen the requirement that $C$ can signal to neither $A$ nor $B$, one takes the union of it with $\left(\Proj{}{C} \prec \Proj{}{A}\prec\Proj{}{B}\right) \cap \left(\Proj{}{C} \prec \Proj{}{B}\prec\Proj{}{A}\right) = \left(\Proj{}{A}\otimes\Proj{}{B}\right)\succ\Proj{}{C}$. The obtained projector, $\overline{\Proj{}{A}\otimes\Proj{}{B}}\rightarrow\Proj{}{C}$, has the normal form $\left(\left(\Proj{}{A}\prec\Proj{}{B}\prec\Proj{}{C}\right) \cap \left(\Proj{}{B}\prec\Proj{}{A}\prec\Proj{}{C}\right)\right)\cup \left(\left(\Proj{}{C}\prec \Proj{}{A}\prec\Proj{}{B}\right) \cap \left(\Proj{}{C}\prec \Proj{}{B}\prec\Proj{}{A}\right)\right)$. This indeed defines a set of objects in which $A$ and $B$ cannot signal to each other, but both might signal to and receive signaling from some party $C$, and all this information can be read directly from its normal form.

Nevertheless, because the normal form encodes the possible signaling directions, it is not unique. Indeed, as the prec is associative, this means that the impossibility to signal is transitive, so that in an expression like $\Proj{}{A}\prec\Proj{}{B}\prec\Proj{}{C}$, $C$ cannot signal to $A$. This implies that some formulae are redundant. For example, the following can be shown in the algebra:
\begin{multline}\label{eq:example_ineq}
    \left(\Proj{}{A}\prec\Proj{}{B}\prec\Proj{}{C}\right) \cap \left(\Proj{}{C}\prec \Proj{}{A}\prec\Proj{}{B}\right) = \\
    \left(\Proj{}{A}\prec\Proj{}{B}\prec\Proj{}{C}\right) \cap \left(\Proj{}{C}\prec \Proj{}{A}\prec\Proj{}{B}\right)\\
    \cap \left(\Proj{}{A}\prec \Proj{}{C}\prec\Proj{}{B}\right)\:.
\end{multline}
This is intuitively straightforward: if $C$ cannot signal to $A$ and $B$ and vice-versa, this also includes the situation in which $C$ cannot signal to $A$ while $B$ cannot signal to $C$. Therefore, the third term on the right-hand side of the above is redundant, and these two normal forms are equivalent. Improving the definition of the normal form so as to incorporate this redundancy would hopefully make the normal form unique. This is left open as a future research direction.

Hence, putting a projector in normal form provides a direct way to read the possible signaling structures of the set of higher-order transformations it represents. Moreover, putting every projector under consideration into normal form provides a quick way to spot equivalent state structures. Nevertheless, having equivalent normal forms is but a sufficient condition for the equivalence of projectors. One should complete the test by actually computing the intersection (or union) of projectors with inequivalent normal forms to reach a definitive conclusion. For example, in Eq. \eqref{eq:example_ineq} the normal forms allows one to regroups the terms into two projectors $\Proj{}{1} := \left(\Proj{}{A}\prec\Proj{}{B}\prec\Proj{}{C}\right) \cap \left(\Proj{}{C}\prec \Proj{}{A}\prec\Proj{}{B}\right)$ and $\Proj{}{2} := \Proj{}{A}\prec \Proj{}{C}\prec\Proj{}{B}$. From there, it is obvious that the right-hand side is contained in the left-hand side since this equation has the form $\Proj{}{1} \subseteq \Proj{}{1} \cap \Proj{}{2}$ and it remains to show that $\Proj{}{1} = \Proj{}{1} \cap \Proj{}{2}$ by explicit computation, which is also made simpler since one knows that they do not have to explicitly compute the intersection contained in $\Proj{}{1}$ as it appears on both side of the equation.

A concrete example of the use of the normal form, as well as of an equivalence of normal form, is provided in App. \ref{sec:examples_dynamics_constr}.

Summarizing, the algebraic fact that the prec is associative and commutes with the negation allows one to write a `normal' form useful to compare types of transformations.
This gives a tool to determine whether two types of higher-order transformations are equivalent by sole inspection of their projectors.
In addition, the physical fact that the prec is a one-way signaling composition, combined with the algebraic rules, allows one to split projectors into several terms with a fixed signaling direction between their base state structures.
This implies that one can know the possible signaling structure that a higher-order theory may feature by sole inspection of its projector.

\subsection{When quantum combs are isomorphic to quantum networks \label{sec:applications_iso}}
The next part of the results is concerned with the relationship between two different ways of building higher-order objects using the developed formalism of projectors. This will be shown to be an example for which the normal form allows to swiftly recover results of Ref. \cite{Bisio2018}.

Following Ref. \cite{Chiribella2009}, a \textbf{network} is defined as the causally ordered (i.e. one-way signaling) succession of `nodes' of the same state structure. A `1-network of base $\Alg{A}{}$' will be the set $\Alg{A}{}\subset \LinOp{\Hilb{A}}$ itself, thereafter noted with an index as $\Alg{A}{0}$ to distinguish between the multiple copies of the state structure $\Alg{A}{}$ instantiated on an increasingly larger number of systems. The `2-network of base $\Alg{A}{}$' will be the set $\Alg{A}{0} \prec \Alg{A}{1} \subset \LinOp{\Hilb{A_0}\otimes\Hilb{A_1}}$, the `3-network' will be $\Alg{A}{0} \prec \Alg{A}{1} \prec \Alg{A}{2}$, etc. up to the `$n$-network' defined as $\Alg{A}{0} \prec \Alg{A}{1} \prec \ldots \prec \Alg{A}{n-1} \subset \LinOp{\Hilb{A_0}\otimes \Hilb{A_1} \otimes \ldots \Hilb{A_{n-1}}}$.
A common occurrence of this structure is the network whose base is a quantum channel so that $\Alg{A}{}$ is characterized by a projector $\Proj{}{\Alg{A}{}} = \mathcal{I}_{A_0} \rightarrow \mathcal{I}_{A_1}$; it is called a \textbf{quantum network} (notice that it will require twice as many systems since the base is defined as a state structure involving two subsystems). This quantum network, here associated with some party that will be called Alice, represents the successive operations of that party. If Alice has a single-node quantum network, it means that Alice acts once on subsystem $A_0$ with a quantum channel and outputs a subsystem $A_1$ in a quantum state defined in space $\LinOp{\Hilb{A_1}}$. If she has a network with two nodes, she will act a first time on $A_0$ and output a first state at $A_1$, then a second time on $A_2$, now potentially using any size of ancillary qudit as a memory register she preserved from her first operation, and outputs a second state in $A_3$. And so on for all numbers of nodes, as defined recursively.

Another way of building a higher-order state structure is the \textbf{comb}, which consists of recursively transforming into a base type: the `1-comb of base $\Alg{A}{}$' is again $\Alg{A}{0}$, then the $2-$comb is $\Alg{A}{0} \rightarrow \Alg{A}{1}$, the $3-$comb is $(\Alg{A}{0} \rightarrow \Alg{A}{1}) \rightarrow \Alg{A}{2}$, etc up to the $n-$comb defined as $(...(\Alg{A}{0}\rightarrow \Alg{A}{1}) \rightarrow ...)\rightarrow \Alg{A}{n} \subset \LinOp{\Hilb{A_0}\otimes \Hilb{A_1} \otimes \ldots \Hilb{A_{n-1}}}$. 
As is the case for the network case, a common occurrence of this structure is the comb whose base is a quantum channel, called the \textbf{quantum comb}.

When the two constructions were introduced in Ref. \cite{Chiribella2009}, it was proven that a quantum network \textit{is} a quantum comb.
When treated using the formalism developed here, there is a stark contradiction. All quantum combs are built using the transformation, $\rightarrow$, which permits two-way signaling. How could it be that they are all equivalent to networks which are built using the prec -- that is, objects featuring a single direction of signaling? 
Besides, why is the 1-comb (built using the two-way signaling transformation) equivalent to the quantum channel (which is causal)? Why does the quantum 2-comb reduce to a two-node quantum network, i.e. a map acting on two quantum states, when by definition it should be a supermap, i.e. a map acting on a quantum channel? And why does it reduce to a succession of two operations that have a well-defined direction of signaling between parties, and not, as its projector suggests, a nesting of two two-way signaling compositions, resulting in four possible signaling directions?

To phrase this issue formally, some notation is required. Let $\Alg{A}{i}^{state} \subset \LinOp{\Hilb{A_i}} $ be the state structure of states on system $A_i$ so that its projector is $\mathcal{I}_{A_i}$, and let $\Alg{A}{i,j}^{channel} \subset \LinOp{\Hilb{A_i}\otimes\Hilb{A_j}}$ be the state structure of channels between systems $A_i$ and $A_j$ so that its projector is $\mathcal{I}_{A_i}\rightarrow \mathcal{I}_{A_j}$. Writing $\Alg{A}{\text{base}}$ as either of these two structures used as the base of a larger state structure (a name appearing in the index indicates that the labeling of the subsystems has not been specified yet), the network and comb can be abstractly defined using recursive formulae. Let $\Alg{A}{\text{1-network}} = \Alg{A}{\text{1-comb}} = \Alg{A}{\text{base}}$. Then, the definitions of Ref. \cite{Chiribella2009} can be abstracted as $\Alg{A}{\text{n-network}} \equiv \Alg{A}{\text{(n-1)-network}} \prec \Alg{A}{\text{base}}$ and $\Alg{A}{\text{n-comb}} \equiv \Alg{A}{\text{(n-1)-comb}} \rightarrow \Alg{A}{\text{1-comb}}$. 
Finally, let the labeling of the base state structures be chosen such that the subsystems in each state structure coincide. 

The structures we want to compare are the quantum networks, the quantum combs, and the combs of quantum states with twice as many nodes.
In symbols, the state structure of quantum networks (i.e. of networks whose base is quantum channels) with $n$ nodes reads $\Alg{A}{0,1}^{channel} \prec \ldots \prec \ldots \prec \Alg{A}{2n-2,2n-1}^{channel}$ which, in the language of Refs. \cite{Perinotti2016,Bisio2018}, carries the type $\left(A_0 \rightarrow A_1\right) \prec \ldots \prec \left(A_{2n-2} \rightarrow A_{2n-1}\right)$ (where each $A$'s are of the type of quantum states), and it is associated with the projector
\begin{multline}\label{eq:proj_n_network_chan}
    \Proj{\text{(n-network)}}{\mathscr{A}_{\text{channel}}} \equiv\\ \left(\mathcal{I}_{A_0} \rightarrow \mathcal{I}_{A_1}\right) \prec \ldots \prec \left( \mathcal{I}_{A_{2n-2}} \rightarrow \mathcal{I}_{A_{2n-1}}\right) \:.
\end{multline}
The state structure of quantum combs (i.e. of combs whose base are quantum channels) with $n$ nodes reads $\left(\ldots\left(\Alg{A}{n-1,n}^{channel} \rightarrow \Alg{A}{n-2,n+1}^{channel}  \right)\ldots\right) \rightarrow \Alg{A}{0,2n-1}^{channel}$, it carries the type $(\ldots((A_{n-1} \rightarrow A_n)\rightarrow (A_{n-2} \rightarrow A_{n+1}))\rightarrow \ldots)\rightarrow (A_0 \rightarrow A_{2n-1})$, and it is associated with the projector
\begin{multline}\label{eq:proj_n_comb_chan}
    \Proj{\text{(n-comb)}}{\mathscr{A}_{\text{channel}}} \equiv\\ (\ldots(\mathcal{I}_{A_{n-1}} \rightarrow \mathcal{I}_{A_{n}}) \rightarrow \ldots ) \rightarrow \left(\mathcal{I}_{A_0} \rightarrow \mathcal{I}_{A_{2n-1}}\right)\:.
\end{multline}%
The state structure of combs based on states with $2n$ nodes reads $\left(\ldots\left(\left(\Alg{A}{0}^{state}\rightarrow \Alg{A}{1}^{state}\right) \rightarrow \Alg{A}{2}^{state}\right)\ldots\right)\rightarrow \Alg{A}{2n-1}^{state}$; it carries the type $ (\ldots((A_0\rightarrow A_1) \rightarrow A_2 ) \rightarrow \ldots)\rightarrow A_{2n-1}$; it is associated with the projector
\begin{equation}\label{eq:proj_2n_comb_states}
    \Proj{\text{(2n-comb)}}{\mathscr{A}_{\text{state}}} \equiv (\ldots (\mathcal{I}_{A_0} \rightarrow \mathcal{I}_{A_1}) \rightarrow \ldots ) \rightarrow \mathcal{I}_{A_{2n-1}}\:.
\end{equation}

To answer the above interrogations, it can be found in the literature 1) that quantum channels are no signaling from output to input, and are equivalent to quantum 1-network and quantum 1-combs; 2) that the first two state structures defined in \eqref{eq:proj_n_network_chan} and \eqref{eq:proj_n_comb_chan} above are equivalent \cite[Theo. 8]{Chiribella2009}; 3) that the elements $M$ of these state structures can be characterized in the CJ picture by the causality condition \cite[Theo. 5]{Chiribella2009}:
\begin{equation}\label{eq:causality_cond}
\begin{gathered}
\forall \in 1,\ldots n,\\
    \TrX{A_{2i+1}}{M^{(i)}} = \frac{1}{d_{A_{2i}}}\TrX{A_{2i}A_{2i+1}}{M^{(i)}} \otimes \mathds{1}_{A_{2i}} \:,\\
    M^{(n)}\equiv M, \; M^{(j)} \equiv \TrX{A_{2j} \ldots A_{2n-1} }{M}, \; j<n\:;
\end{gathered}%
\end{equation}%
4) that the last two state structures defined in \eqref{eq:proj_n_comb_chan} and \eqref{eq:proj_2n_comb_states} above are equivalent \cite[Prop. 6]{Bisio2018}.

We now recover and explain these results using the projective characterization. These four statements can indeed be proven simply by algebraic manipulations of the projectors. They amount to proving that 1) $\mathcal{I}_{A_0} \rightarrow \mathcal{I}_{A_1} = \overline{\mathcal{I}}_{A_0} \prec \mathcal{I}_{A_1}$; 2) Equations \eqref{eq:proj_n_network_chan} and \eqref{eq:proj_n_comb_chan} are equivalent, i.e. $\Proj{\text{(n-comb)}}{\mathscr{A}_{\text{channel}}} = \Proj{\text{(n-network)}}{\mathscr{A}_{\text{channel}}}$; 3) Equation \eqref{eq:causality_cond} is an equivalent way of defining the validity subspace of quantum networks, Eq. \eqref{eq:proj_n_network_chan}; 4) Equations \eqref{eq:proj_n_comb_chan} and \eqref{eq:proj_2n_comb_states} are equivalent, i.e. $\Proj{\text{(n-comb)}}{\mathscr{A}_{\text{channel}}} = \Proj{\text{(2n-comb)}}{\mathscr{A}_{\text{state}}}$.

\paragraph{Transformations between quantum states.}
The case $n=1$ is treated as a warm-up before the general case. Equations \eqref{eq:proj_n_comb_chan}, \eqref{eq:proj_n_network_chan} and \eqref{eq:proj_2n_comb_states} all reduce to $\mathcal{I}_{A_0} \rightarrow \mathcal{I}_{A_1}$, so items 2) and 4) hold. What remains to be proven are items 1) and 3) in the single-node case. The following lemma does that, and it underlies the proof for the general case. 

\begin{figure}[htb]
    \centering
    \includegraphics[width=\linewidth]{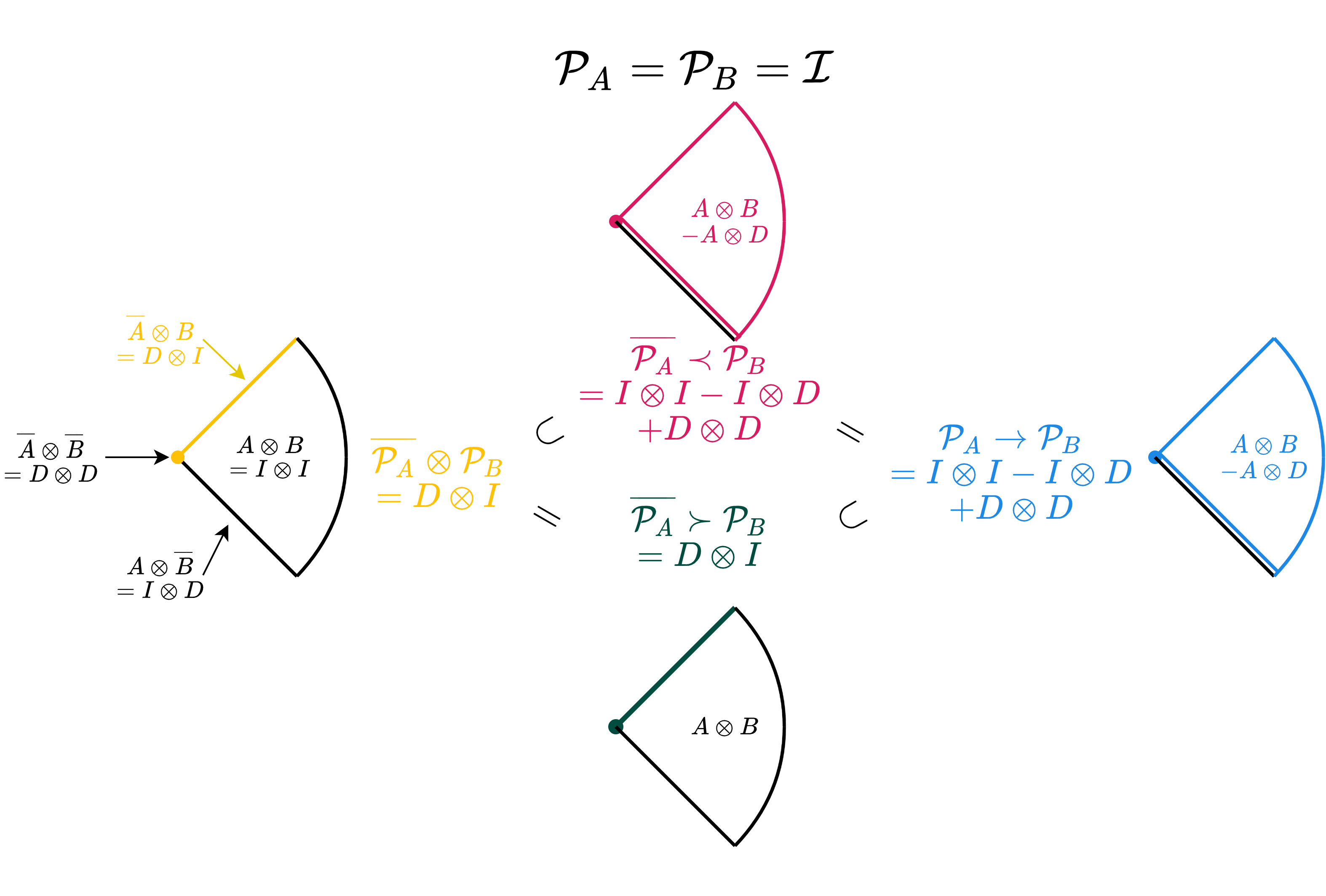}
    \caption{Diagrammatic representation of Lemma \ref{lem:accidental}: isomorphisms between the compositions when the base projectors are the identity, $\mathcal{I}$, i.e. in the case of transformations between quantum states. Compared to the general situation, Fig. \ref{fig:compos}, the A-to-B one-way signaling composition becomes equivalent to the two-way one, whereas the B-to-A one-way signaling composition becomes equivalent to the no signaling one.}
    \label{fig:Accident}
\end{figure}

\begin{lemm}\label{lem:accidental}
    The transformation operation on projectors, $\rightarrow$, simplifies into a prec operation when either of the projectors is the identity. The prec operation on projectors, $\prec$, simplifies into a tensor operation when either of the projectors is depolarizing. In equation,
\begin{subequations}\label{eq:isomorphisms}
\begin{align}
     &\Proj{}{A} \rightarrow \Proj{}{B} = \CompProj{}{A} \prec \Proj{}{B} \notag \\
     &\iff\: \Proj{}{A}=\mathcal{I}_A \:\text{or}\: \Proj{}{B}=\mathcal{I}_B\:; \label{eq:isomorphisms_transfo}\\
    &\CompProj{}{A} \prec \Proj{}{B} = \CompProj{}{A} \otimes \Proj{}{B} \notag \\
    &\iff\: \Proj{}{A}=\mathcal{D}_A \:\text{or}\: \Proj{}{B}=\mathcal{D}_B\:. \label{eq:isomorphisms_tensor}
\end{align}
\end{subequations}
\end{lemm}
(This follows from the definitions of the operations, see App. \ref{sec:projo_prec_iso} for the proof.)

The first equation of this lemma implies that transformations between quantum states, characterized by $\mathcal{I}_{A_0} \rightarrow \mathcal{I}_{A_1}$, are equivalent to a causal succession of a functional and a state, $\mathcal{D}_{A_0} \prec \mathcal{I}_{A_1}$. This proves statement 1), $\mathcal{I}_{A_0} \rightarrow \mathcal{I}_{A_1} = \overline{\mathcal{I}}_{A_0} \prec \mathcal{I}_{A_1}$. Remark these actually are two different ways of defining the quantum channel: as the most general transformation that takes a quantum state to a quantum state, or as a measurement coherently followed by a preparation of a quantum state\footnote{By coherent, it is meant `as any element of the convex hull of measure-and-reprepare maps', so this set does not represent only entanglement breaking maps, i.e. of the form $\mathcal{M}(\bullet)=\rho \TrX{}{\sigma \cdot \bullet}$, but any linear map $\mathcal{M}(\bullet)=\sum_i \rho_i \TrX{}{\sigma_i \cdot \bullet}$). Recall that in this non-CJ representation, the CPTP condition comes as constraints on the $\rho_i$'s and $\sigma_i$'s.} (this is also mentioned in example \ref{sec:examples_single_channel}). 

As for the second equation of the lemma, it is for example the reason why the single-partite process matrix reduces to an effect and a state in tensor product \cite{OCB2012}: the 1-PM is a functional on channels, characterized by the negated projector $\overline{\mathcal{I}_{A_0} \rightarrow \mathcal{I}_{A_1}}$. Using the first part of the lemma, $\overline{\mathcal{I}_{A_0} \rightarrow \mathcal{I}_{A_1}} = \overline{\mathcal{D}_{A_0} \prec \mathcal{I}_{A_1}}$. The negation is distributed over a prec, $\overline{\mathcal{D}_{A_0} \prec \mathcal{I}_{A_1}} = \mathcal{I}_{A_0} \prec \mathcal{D}_{A_1}$, and finally the second part of the lemma is used, $\mathcal{I}_{A_0} \prec \mathcal{D}_{A_1} = \mathcal{I}_{A_0} \otimes \mathcal{D}_{A_1}$. These algebraic manipulations on projectors quickly led to the conclusion that the functionals on quantum channels, characterized by $\overline{\mathcal{I}_{A_0} \rightarrow \mathcal{I}_{A_1}}$, are equivalent to a quantum state on the input system followed by a measurement applied on its output, $\mathcal{I}_{A_0} \prec \mathcal{D}_{A_1}$, and that the state and measurement are causally disconnected, $\mathcal{I}_{A_0} \otimes \mathcal{D}_{A_1}$.

More generally, the first equation \eqref{eq:isomorphisms_transfo} states that any transformation from (or to) an arbitrary state structure
to (or from) the state structure of quantum states will automatically be no signaling from output to input.
Contrastingly, the second equation \eqref{eq:isomorphisms_tensor}  states that any way of combining a state structure with the single element `quantum measurement' state structure, i.e. $\{\mathds{1}\}$, is automatically no signaling from and to it.
These relations are depicted in Fig. \ref{fig:Accident} for the case where both $A$ and $B$ are quantum states. This should be compared with the general case of Fig. \ref{fig:compos}.

Finally, as previously discussed in the remark of Sec. \ref{sec:NS_def}, Eq. \eqref{eq:isomorphisms_transfo} is the case where the definition of no signaling, Eq. \eqref{eq:caus}, overlaps with the one of a transformation, see Eq. \eqref{eq:NS_usual}, so the usual definition of no signaling is recovered. The fact that these two equations coincide in this case is exactly statement 3) in the $n=1$ case: it suffices to notice that $\mathcal{I}_{A_0} \rightarrow \mathcal{I}_{A_1} = \mathcal{D}_{A_0} \leftarrow \mathcal{D}_{A_1}$ and that this projector applied on bipartite operators is exactly enforcing the condition \eqref{eq:causality_cond} in the $n=1$ case, i.e. $\forall M^{(1)} \in \LinOp{\Hilb{A_0}\otimes \Hilb{A_1}}: \left(\mathcal{D}_{A_0} \leftarrow \mathcal{D}_{A_1}\right)\{M^{(1)}\} = M^{(1)} \iff \TrX{A_1}{M^{(1)}} = \mathds{1}_{A_0}$.
\paragraph{Equivalence of quantum combs and networks.} 
We shall now generalize this equivalence for all $n$.
Reinterpreting $\mathcal{I}_{A_0} \rightarrow \mathcal{I}_{A_1} = \mathcal{D}_{A_0} \prec \mathcal{I}_{A_1}$, the quantum channel can be seen as a sort of network in which the nodes alternate between an effect and a state (i.e., the wires alternate between pointing up and down). This was shown to imply \eqref{eq:causality_cond} in the $n=1$ case. Going to $n>1$, a network of quantum channels is then an alternating network of effects and states as associativity can be used, and it indeed implies \eqref{eq:causality_cond} in general.

\begin{theo}\label{theo:combs=networks}
    The projectors to a quantum network \eqref{eq:proj_n_network_chan}, to a quantum comb \eqref{eq:proj_n_comb_chan}, and to a (twicely longer) comb based on quantum states \eqref{eq:proj_2n_comb_states} are all equivalent to the following projector, characterizing an alternating network of quantum effects and states:
    \begin{equation}\label{eq:prec_chain}
    \Proj{(2n)}{} = \overline{\mathcal{I}}_{A_0} \prec \mathcal{I}_{A_1} \prec \ldots \prec \overline{\mathcal{I}}_{A_{2n-2}} \prec \mathcal{I}_{A_{2n-1}}\:.
    \end{equation}
    Meaning that with suitable normalization, the state structure of networks of order $n$ based on quantum channels is equivalent to the one of combs of order $n$ based on quantum channels, and it is also equivalent to the one of combs of order $2n$ based on quantum states. 
    
    In addition, the elements $M$ of the state structure defined by this projector obey equation \eqref{eq:causality_cond}.
\end{theo}
The proof is presented in App. \ref{sec:combs=networks_proof}. This theorem means that an operator $M$ obeying the causality condition \eqref{eq:causality_cond} can be interpreted as an element of four equivalent state structures: (1) from Eq. \eqref{eq:proj_n_comb_chan} as a valid quantum $n$-comb, i.e. (the CJ representation of) a map on $2n$ systems that transforms the recursively defined map on the $n-1$ nodes split between the $2n-2$ systems $A_1, \ldots A_{2n-2}$ into a single-node map, i.e. a quantum channel from $A_0$ to $A_{2n-1}$.
(2) from Eq. \eqref{eq:proj_n_network_chan} as a valid quantum network of channels with $n$ nodes, i.e. (the CJ representation of) a map on $2n$ systems $A_0, \ldots A_{2n-1}$ that represents the $n$ successive operations of a party on $n$ nodes: at node $j$ she applies a channel between systems $A_{2j-2}$ and $A_{2j-1}$ that can depend on (that is, use any size of ancillary memory from) the $j-1$ previous nodes.
(3) from Eq. \eqref{eq:proj_2n_comb_states} as a valid comb of states on $2n$ nodes, i.e. (the CJ representation of) a map on $2n$ systems that transforms the recursively defined map on the $2n-1$ systems $A_0, \ldots A_{2n-2}$ into a quantum state on the $2n-$th system $A_{2n-1}$.
And (4) from Eq. \eqref{eq:prec_chain} as a valid quantum network alternating between quantum measurements and states at each of the $2n$ nodes, i.e. (the CJ representation of) a map on $2n$ systems $A_0, \ldots A_{2n-1}$ that represents the $2n$ successive operations of a party on $2n$ nodes: at odd nodes, she measures a quantum state, at even nodes, she prepares one. Each operation is potentially conditioned by the previous nodes but is independent of the future ones.

Although these four definitions appear quite different, notice that equation \eqref{eq:prec_chain} is actually the normal form of Eqs. \eqref{eq:proj_n_network_chan}, \eqref{eq:proj_n_comb_chan}, and \eqref{eq:proj_2n_comb_states}! Hence the equivalence between seemingly unrelated definitions of higher-order maps is inferred simply by working out their normal form.
The fact that the two comb-like definitions result in structures with a single signaling direction is now directly apparent, as the structure features a single `prec chain' in normal form. 
As mentioned above, this is a very peculiar behavior: From equation \eqref{eq:relations_2} one would have expected the normal forms to be the union of many prec chains with different signaling orderings. That is, the combs should typically have an indefinite causal order between their nodes. 
Lemma \ref{lem:accidental} auspiciously prevents that. The isomorphisms between $\prec$ and $\rightarrow$ in the case of composite state structure whose base structures are sets of quantum states explain this apparent counter-logical equivalence of one-way signaling objects (the networks of states) with various a priori two-way signaling ones. Quantum theory is thus very tame in the sense that successive and/or nested transformations of quantum states or channels happen to automatically have no signaling from output to input.

On the contrary, nesting any other state structures in a comb-like construction will in general result in indefinite causal order. An example of this in the case of Multi-round Process Matrices \cite{MPM} is provided in App. \ref{sec:examples_dynamics_constr}. It also features another example of the use of the normal form.

\section{Discussion}
In this work, a correspondence was developed between several characterizations of the theory of higher-order quantum transformations. 
Starting at the most abstract layer, its formulation in terms of types, we added a new layer, its formulation in terms of superoperator projectors, to bridge to a less abstract layer, its formulation in terms of sets of CJ operators -- here called state structures.
Correspondences between the characterizations can be summed up as follows: Given types $A$ and $B$, types $\overline{A}, A\otimes B$, and $A\rightarrow B$ can be constructed using the semantic rules $\{1,(,),\rightarrow\}$. 
To these types correspond projectors: given projectors $\Proj{}{A}$ and $\Proj{}{B}$, projectors $\CompProj{}{A}, \Proj{}{A}\otimes \Proj{}{B}$ and $\Proj{}{A}\rightarrow \Proj{}{B}$ can be constructed using the algebraic rules $\{\overline{\:\cdot\:},\otimes\}$ as in equations \eqref{eq:det_fctal_proj}, \eqref{eq:det_tensor_proj}, and \eqref{eq:det_map_proj}.
To these types and projectors correspond state structures: given state structures $\Alg{A}{}$ and $\Alg{B}{}$ as in equation \eqref{def:struct}, the state structures $\CompAlg{A}, \CompAlg{A} \otimes \Alg{B}{}$, and $\Alg{A}{}\rightarrow\Alg{B}{}$  can be constructed by requiring that, respectively, conditions \eqref{eq:fctal_def}, \eqref{eq:A_otimes_B_def}, and \eqref{eq:A_to_B} hold.

The main novelty is the prec connector, Prop. \ref{prop:semi-causal_compo}, giving the state structure $\CompAlg{A}\prec \Alg{B}{}$ of maps which are no signaling from the output to the input. 
This allows to speak about the possible signaling structure of the maps characterized by the projective framework.
In addition to that, the properties of the prec hints towards a systematic way to compare different projective expressions as they allow to expand them as unions and intersections of expressions involving only this prec connector.

In the meantime, we identified the features of the algebra of projectors. We showed in particular that it was a Boolean algebra under the rules $\{\overline{\:\cdot\:}, \cap, \cup\}$, which is subsequently extended into a model of linear logic under the rules $\{\overline{\:\cdot\:}, \cap, \cup, \otimes, \overline{\:\cdot\:} \rightarrow \:\cdot\: \}$. 
The introduced prec fits naturally as a multiplicative rule in-between $\otimes$ and $\overline{\:\cdot\:} \rightarrow \:\cdot\:$, which moreover has the property of commuting with the negation.

Finally, we rederived the proof of equivalence between combs based on quantum states, quantum combs, and networks based on quantum channels \cite{Bisio2018} as an example of the use of the projective characterization.
This equivalence explains in particular why these structures feature a definite signaling direction, which comes with a clear physical implementation.
One can wonder how often two differently defined structures with a common no signaling subset are equivalent. And how often do the structures based on the transformation connector (comb-like, which can present ICO in general) reduce in normal form to structures based on a prec (network-like, which cannot present ICO)? Preliminary results indicate these equivalences to be rare, but the question of classifying the isomorphic constructions is left open for future work. 
Another related but more general question that can be asked in terms of projectors is to determine for which base state structures a given construction based on the transformation reduces to a construction based on the one-way signaling composition. In other words, given an abstract way of constructing higher-order transformations, for which base state structures does it reduce to a network?

Several aspects of the theory of higher-order transformations have not been explored here and remain open for future work.
The composition of types in the sense of Ref. \cite{Bisio2018} has only been partially addressed. Knowing the projector characterizing a transformation is sufficient to know what will be the projector characterizing its set of outputs. However, if the input state structure is now restricted to a subset with a different projector, is there a rule to apply to the projector of the transformation in order to get the projector of the possibly now restricted output set?
For example, an $n$-comb takes $(n-1)$-combs as inputs and outputs a 1-comb. One can ask what happens when the inputs are restricted to the no signaling subset: when $(n-1)$ 1-combs are plugged into an $n-$comb instead, is the output set still the full set of 1-combs? 

The notion of causal separability \cite{OCB2012,Giarmatzi2016,Wechs2019,MPM} still is to be defined for general higher-order transformations. The link between the prec and single signaling direction, as well as the normal form, should give a basis for the definition: the normal form can be used to decompose any process into a sum of causally ordered pieces. If this sum is convex, this gives a sufficient condition for causal separability. Nonetheless, there remains to tackle the question of dynamical causal orders, which is a probabilistic notion, and therefore not naturally captured by the formalism of projectors. 

The issue of realizability has not been explored. That is, given a state structure, is it possible to realize each of its elements in a lab experiment? Is there a systematic way to relate these abstract mathematical objects to a circuit realization, as is the case for quantum combs \cite{Chiribella2009} and for some process matrices in terms of time-delocalized subsystems \cite{Oreshkov_2019,Wechs2022}? One may expect that the normal decomposition of general state structures into the union of several one-way signaling structures will prove useful for answering this question.

\textbf{Note added:} After completion of the main results of this work and during the preparation of the manuscript, we became aware of an independent work by Simmons and Kissinger \cite{Simmons2022}, which also investigates the signaling structure of higher-order quantum transformations. Some of the results presented here were also found in this other work using a different formalism.

\begin{acknowledgments}
The authors are grateful for the extensive reviews provided by anonymous reviewers, which contributed greatly to improving this work. They also acknowledge useful technical discussions with Jessica Bavaresco and Aleks Kissinger. 

T. H. is grateful to the organizers and participants of the 2021 Sejny Summer Institute for listening to his rambling interventions about nesting supermaps -- in particular to Pablo Arnault, Titouan Carette, Nat{\'a}lia M{\'o}ller, Eleftherios Tselentis, and Augustin Vanrietvelde for technical discussions. T. H. is especially indebted to Titouan Carette for pointing out the connection with linear logic. In addition, T.H. would like to thank Esteban Castro-Ruiz and Joseph Cunningham for their help and comments, as well as Alexandra Elbakyan for providing access to the scientific literature. 
T. H. benefited from the support of the French Community of Belgium within the framework of the financing of a FRIA grant, then from the Slovak grants ID\# APVV-22-0570 and ID\# VEGA 2/0128/24.

This publication was made possible through the support of the ID\# 61466 grant and ID\# 62312 grant from the John Templeton Foundation, as part of the  \href{https://www.templeton.org/grant/the-quantum-information-structure-of-spacetime-qiss-second-phase}{`The Quantum Information Structure of Spacetime' Project (QISS)}, and by the ID\# 63683 grant from the John Templeton Foundation, as part of the \href{https://www.withoutspacetime.org}{``WithOut SpaceTime'' Project (WOST)}.. The opinions expressed in this publication are those of the authors and do not necessarily reflect the views of the John Templeton Foundation. This work was supported by the Program of Concerted Research Actions (ARC) of the Université libre de Bruxelles and from the F.R.S.-FNRS under project CHEQS within the Excellence of Science (EOS) program. 
O. O. is a Senior Research Associate of the Fonds de la Recherche Scientifique (F.R.S.-FNRS). 

Illustrations were drawn using \href{https://www.drawio.com/}{draw.io}.
\end{acknowledgments}

\bibliography{references.bib}

\providecommand{\noopsort}[1]{}\providecommand{\singleletter}[1]{#1}%
\begin{thebibliography}{10}

\bibitem{Kraus1983}
Karl Kraus.
\newblock ``States, effects, and operations: Fundamental notions of quantum
  theory''.
\newblock \href{https://dx.doi.org/10.1007/3-540-12732-1}{Volume 190 of Lecture
  notes in physics}.
\newblock Springer-Verlag Berlin Heidelberg. ~(1983).
\newblock 1 edition.

\bibitem{Chiribella2008}
G.~Chiribella, G.~M. D’Ariano, and P.~Perinotti.
\newblock ``Transforming quantum operations: Quantum supermaps''.
\newblock \href{https://dx.doi.org/10.1209/0295-5075/83/30004}{EPL (Europhysics
  Letters) {\bf 83}, 30004}~(2008).
\newblock  \href{http://arxiv.org/abs/0804.0180}{arXiv:0804.0180}.

\bibitem{Chiribella2009}
Giulio Chiribella, Giacomo~Mauro D'Ariano, and Paolo Perinotti.
\newblock ``{Theoretical framework for quantum networks}''.
\newblock \href{https://dx.doi.org/10.1103/PhysRevA.80.022339}{Physical Review
  A{\bf 80}}~(2009).
\newblock  \href{http://arxiv.org/abs/0904.4483}{arXiv:0904.4483}.

\bibitem{Perinotti2016}
Paolo Perinotti.
\newblock ``Causal structures and the classification of higher order quantum
  computations''.
\newblock \href{https://dx.doi.org/10.1007/978-3-319-68655-4_7}{Tutorials,
  Schools, and Workshops in the Mathematical SciencesPage 103–127}~(2017).
\newblock  \href{http://arxiv.org/abs/1612.05099}{arXiv:1612.05099}.

\bibitem{Bisio2018}
Alessandro Bisio and Paolo Perinotti.
\newblock ``{Theoretical framework for higher-order quantum theory}''.
\newblock \href{https://dx.doi.org/10.1098/rspa.2018.0706}{Proceedings of the
  Royal Society A: Mathematical, Physical and Engineering Sciences {\bf 475},
  20180706}~(2019).
\newblock  \href{http://arxiv.org/abs/1806.09554}{arXiv:1806.09554}.

\bibitem{Chiribella2013}
Giulio Chiribella, Giacomo~Mauro D'Ariano, Paolo Perinotti, and Benoit Valiron.
\newblock ``Quantum computations without definite causal structure''.
\newblock \href{https://dx.doi.org/10.1103/PhysRevA.88.022318}{Physical Review
  A {\bf 88}, 022318}~(2013).
\newblock  \href{http://arxiv.org/abs/0912.0195}{arXiv:0912.0195}.

\bibitem{Davies1970}
E.~B. Davies and J.~T. Lewis.
\newblock ``{An operational approach to quantum probability}''.
\newblock \href{https://dx.doi.org/10.1007/BF01647093}{Communications in
  Mathematical Physics {\bf 17}, 239--260}~(1970).

\bibitem{OCB2012}
Ognyan Oreshkov, Fabio Costa, and {\v{C}}aslav Brukner.
\newblock ``{Quantum correlations with no causal order}''.
\newblock \href{https://dx.doi.org/10.1038/ncomms2076}{Nature Communications
  {\bf 3}, 1--13}~(2012).
\newblock  \href{http://arxiv.org/abs/1105.4464}{arXiv:1105.4464}.

\bibitem{Kissinger_2019}
Aleks Kissinger and Sander Uijlen.
\newblock ``{A categorical semantics for causal structure}''.
\newblock \href{https://dx.doi.org/10.23638/LMCS-15(3:15)2019}{{Logical Methods
  in Computer Science}{\bf {Volume 15, Issue 3}}}~(2019).
\newblock  \href{http://arxiv.org/abs/1701.04732}{arXiv:1701.04732}.

\bibitem{Araujo2015}
Mateus Ara{\'u}jo, Cyril Branciard, Fabio Costa, Adrien Feix, Christina
  Giarmatzi, and {\v{C}}aslav Brukner.
\newblock ``Witnessing causal nonseparability''.
\newblock \href{https://dx.doi.org/10.1088/1367-2630/17/10/102001}{New Journal
  of Physics {\bf 17}, 102001}~(2015).
\newblock  \href{http://arxiv.org/abs/1506.03776}{arXiv:1506.03776}.

\bibitem{MPM}
Timoth{\'e}e Hoffreumon and Ognyan Oreshkov.
\newblock ``The multi-round process matrix''.
\newblock \href{https://dx.doi.org/10.22331/q-2021-01-20-384}{Quantum {\bf 5},
  384}~(2021).
\newblock  \href{http://arxiv.org/abs/2005.04204}{arXiv:2005.04204}.

\bibitem{Milz_2022}
Simon Milz, Jessica Bavaresco, and Giulio Chiribella.
\newblock ``Resource theory of causal connection''.
\newblock \href{https://dx.doi.org/10.22331/q-2022-08-25-788}{Quantum {\bf 6},
  788}~(2022).
\newblock  \href{http://arxiv.org/abs/2110.03233}{arXiv:2110.03233}.

\bibitem{Milz_2024}
Simon Milz and Marco~Túlio Quintino.
\newblock ``Characterising transformations between quantum objects, of quantum
  properties, and transformations without a fixed causal order''.
\newblock \href{https://dx.doi.org/10.22331/q-2024-07-17-1415}{Quantum {\bf 8},
  1415}~(2024).
\newblock  \href{http://arxiv.org/abs/2305.01247}{arXiv:2305.01247}.

\bibitem{GIRARD1987}
Jean-Yves Girard.
\newblock ``Linear logic''.
\newblock \href{https://dx.doi.org/10.1016/0304-3975(87)90045-4}{Theoretical
  Computer Science {\bf 50}, 1--101}~(1987).

\bibitem{Jamiolkowski1972}
Andrej Jamio{\l}kowski.
\newblock ``{Linear transformations which preserve trace and positive
  semidefiniteness of operators}''.
\newblock \href{https://dx.doi.org/10.1016/0034-4877(72)90011-0}{Reports on
  Mathematical Physics {\bf 3}, 275--278}~(1972).

\bibitem{Choi1975}
Man-Duen Choi.
\newblock ``Positive linear maps on complex matrices''.
\newblock \href{https://dx.doi.org/10.1016/0024-3795(75)90075-0}{Linear Algebra
  and its Applications {\bf 10}, 285--290}~(1975).

\bibitem{Beckman2001}
David Beckman, Daniel Gottesman, Michael~A. Nielsen, and John Preskill.
\newblock ``Causal and localizable quantum operations''.
\newblock \href{https://dx.doi.org/10.1103/physreva.64.052309}{Physical Review
  A{\bf 64}}~(2001).
\newblock
  \href{http://arxiv.org/abs/quant-ph/0102043}{arXiv:quant-ph/0102043}.

\bibitem{Piani2006}
Marco Piani, Micha{\l} Horodecki, Pawe{\l} Horodecki, and Ryszard Horodecki.
\newblock ``Properties of quantum nonsignaling boxes''.
\newblock \href{https://dx.doi.org/10.1103/physreva.74.012305}{Physical Review
  A{\bf 74}}~(2006).
\newblock
  \href{http://arxiv.org/abs/quant-ph/0505110}{arXiv:quant-ph/0505110}.

\bibitem{Nielsen2009}
Michael~A. Nielsen and Isaac~L. Chuang.
\newblock ``Quantum computation and quantum information''.
\newblock \href{https://dx.doi.org/10.1017/cbo9780511976667}{Cambridge
  University Press}. ~(2009).
\newblock 2 edition.

\bibitem{Chiri2008memory}
Giulio Chiribella, Giacomo~M. D'Ariano, and Paolo Perinotti.
\newblock ``Memory effects in quantum channel discrimination''.
\newblock \href{https://dx.doi.org/10.1103/PhysRevLett.101.180501}{Physical
  Review Lett. {\bf 101}, 180501}~(2008).
\newblock  \href{http://arxiv.org/abs/0803.3237}{arXiv:0803.3237}.

\bibitem{Shrapnel2018}
Sally Shrapnel, Fabio Costa, and Gerard Milburn.
\newblock ``{Updating the Born rule}''.
\newblock \href{https://dx.doi.org/10.1088/1367-2630/aabe12}{New Journal of
  Physics{\bf 20}}~(2018).
\newblock  \href{http://arxiv.org/abs/1702.01845}{arXiv:1702.01845}.

\bibitem{ChoiEffros1977}
Man-Duen Choi and Edward~G. Effros.
\newblock ``Injectivity and operator spaces''.
\newblock \href{https://dx.doi.org/10.1016/0022-1236(77)90052-0}{Journal of
  Functional Analysis {\bf 24}, 156 -- 209}~(1977).

\bibitem{Hiai2014}
Fumio Hiai and D{\'e}nes Petz.
\newblock ``Introduction to matrix analysis and applications''.
\newblock
  \href{https://dx.doi.org/https://doi.org/10.1007/978-3-319-04150-6}{Universitext}.
  Springer, Cham. ~(2014).

\bibitem{Dynamics}
Esteban Castro-Ruiz, Flaminia Giacomini, and {\v{C}}aslav Brukner.
\newblock ``{Dynamics of Quantum Causal Structures}''.
\newblock \href{https://dx.doi.org/10.1103/PhysRevX.8.011047}{Physical Review X
  {\bf 8}, 1--18}~(2018).
\newblock  \href{http://arxiv.org/abs/1710.03139}{arXiv:1710.03139}.

\bibitem{ECR_comm}
Esteban Castro-Ruiz~(2019).
\newblock private communication.

\bibitem{Simmons2022}
Will Simmons and Aleks Kissinger.
\newblock ``Higher-order causal theories are models of bv-logic''.
\newblock In Stefan Szeider, Robert Ganian, and Alexandra Silva, editors, 47th
  International Symposium on Mathematical Foundations of Computer Science (MFCS
  2022).
\newblock \href{https://dx.doi.org/10.4230/LIPIcs.MFCS.2022.80}{Volume 241 of
  Leibniz International Proceedings in Informatics (LIPIcs), pages
  80:1--80:14}.
\newblock Dagstuhl, Germany~(2022). Schloss Dagstuhl -- Leibniz-Zentrum f{\"u}r
  Informatik.
\newblock  \href{http://arxiv.org/abs/2205.11219}{arXiv:2205.11219}.

\bibitem{Chiribella_2011}
Giulio Chiribella, Giacomo~Mauro D'Ariano, and Paolo Perinotti.
\newblock ``Informational derivation of quantum theory''.
\newblock \href{https://dx.doi.org/10.1103/physreva.84.012311}{Physical Review
  A{\bf 84}}~(2011).
\newblock  \href{http://arxiv.org/abs/1011.6451}{arXiv:1011.6451}.

\bibitem{Giarmatzi2016}
Ognyan Oreshkov and Christina Giarmatzi.
\newblock ``Causal and causally separable processes''.
\newblock \href{https://dx.doi.org/10.1088/1367-2630/18/9/093020}{New Journal
  of Physics {\bf 18}, 093020}~(2016).
\newblock  \href{http://arxiv.org/abs/1506.05449}{arXiv:1506.05449}.

\bibitem{Wechs2019}
Julian Wechs, Alastair~A Abbott, and Cyril Branciard.
\newblock ``On the definition and characterisation of multipartite causal
  (non)separability''.
\newblock \href{https://dx.doi.org/10.1088/1367-2630/aaf352}{New Journal of
  Physics {\bf 21}, 013027}~(2019).
\newblock  \href{http://arxiv.org/abs/1807.10557}{arXiv:1807.10557}.

\bibitem{Oreshkov_2019}
Ognyan Oreshkov.
\newblock ``Time-delocalized quantum subsystems and operations: on the
  existence of processes with indefinite causal structure in quantum
  mechanics''.
\newblock \href{https://dx.doi.org/10.22331/q-2019-12-02-206}{Quantum {\bf 3},
  206}~(2019).
\newblock  \href{http://arxiv.org/abs/1801.07594}{arXiv:1801.07594}.

\bibitem{Wechs2022}
Julian Wechs, Cyril Branciard, and Ognyan Oreshkov.
\newblock ``Existence of processes violating causal inequalities on
  time-delocalised subsystems''.
\newblock \href{https://dx.doi.org/10.1038/s41467-023-36893-3}{Nature
  Communications {\bf 14}, 1471}~(2023).
\newblock  \href{http://arxiv.org/abs/2201.11832}{arXiv:2201.11832}.

\bibitem{Roman2008}
Steven Roman.
\newblock ``Advanced linear algebra''.
\newblock \href{https://dx.doi.org/10.1007/978-0-387-72831-5}{Springer}.
  ~(2008).
\newblock 3 edition.

\bibitem{Piziak1999a}
Robert Piziak, Patrick~L. Odell, and R.~Hahn.
\newblock ``{Constructing projections on sums and intersections}''.
\newblock \href{https://dx.doi.org/10.1016/S0898-1221(98)00242-9}{Computers
  {\&} Mathematics with Applications {\bf 37}, 67--74}~(1999).

\bibitem{morimae2014process}
Tomoyuki Morimae.
\newblock ``The process matrix framework for a single-party system''~(2014).
\newblock  \href{http://arxiv.org/abs/1408.1464}{arXiv:1408.1464}.

\end{thebibliography}
\appendix
\section{Superoperator projectors\label{sec:projo}}
This appendix explores the properties of the orthogonal superoperator projectors that are used as the principal tool for the characterization in the main text. These are defined under the notion of a projector on an operator system in Section \ref{sec:projos}. Six operations whose purpose is motivated by the main text are then defined on this set. These are: the intersection and union, $\cap,\cup$, the negation, $\overline{\:\cdot\:}$, the tensor and transformation, $\otimes,\rightarrow$, and the prec, $\prec$, successively reviewed in Sec. \ref{sec:projos_prop}, \ref{sec:projo_Boolean}, \ref{sec:projo_LL}, and \ref{sec:projo_prec}. Some important properties of these operations are proven for the purposes of the main text, among which the fact that each operation applied on commuting projector(s) on operator system(s) results in commuting projector(s) on operator system(s), justifying the assumption of commutation while proving Theorem \ref{prop:algebra}.

\subsection{Definition and required properties of superoperator projectors \label{sec:projos}}
\paragraph{Projectors on operator systems.\label{sec:projos_on_os}}
As explained below Def. \ref{def:OS}, a projector on an operator system $\Proj{}{A}: \LinOp{\Hilb{A}}\rightarrow \LinOp{\Hilb{A}}$ is a linear orthogonal superoperator projector to a self-adjoint subspace that contains the identity. 

A linear superoperator $\Proj{}{A}$ is a \textit{projector} on a subspace of $\LinOp{\Hilb{A}}$ if and only if it is idempotent,
\begin{equation}\label{eq:projo_1}
    \Proj{}{A} \circ \Proj{}{A} = \Proj{}{A}\:.
\end{equation}
An operator $V\in \LinOp{\Hilb{A}}$ belongs to the subspace defined by $\Proj{}{A}$ if and only if
\begin{equation}
    \ProjOn{}{A}{V} = V \:.
\end{equation}%
$\Proj{}{A}$ projects to a subspace closed under the adjoint if and only if it obeys 
\begin{equation}\label{eq:projo_2}
    \ProjOn{}{A}{V} = V \Rightarrow \ProjOn{}{A}{V^\dag} = V^\dag\:,
\end{equation}
which can be written concisely as $\Proj{}{A} \circ \dag  = \dag \circ \Proj{}{A}$ where $\dag$ means `taking the adjoint in $\LinOp{\Hilb{A}}$'. 

A projector is \textit{orthogonal} if it does not increase the norm of operators:
\begin{equation}
	\big|\big| \ProjOn{}{A}{V} \big|\big| \leq \big|\big| V \big|\big| \:,
\end{equation}
where $\big|\big| V \big|\big| \equiv \sqrt{\TrX{}{V^\dag \cdot V}}$ is the Hilbert-Schmidt norm. 
Note that this condition is equivalent to requiring the projector to be self-adjoint with respect to the Hilbert-Schmidt inner product $\InProd{V'}{V}\equiv \TrX{}{V^{'\dag}\cdot V}$ (see \textit{e.g.} Ref. \cite[Theorem 10.5]{Roman2008}),
\begin{equation}\label{eq:projo_3}
	\TrX{}{\ProjOn{}{A}{V'}^\dag\cdot V} = \TrX{}{V'^\dag \cdot \ProjOn{}{A}{V}}\:,\quad \forall V,V'\:.
\end{equation} 
As the projectors appearing in the main text are used to talk about the subspace spanned by a collection of elements rather than the elements themselves, what is truly relevant in the projective characterization is the image of the projectors, not the projectors themselves. Consequently, there is no loss of generality in assuming every projector to be orthogonal since there always exists an orthogonal projector with the same image as any given projector. Having this property will greatly simplify some of the proofs. 
For a linear map $\mathcal{M}\in\LinOpB{\LinOp{\Hilb{A}}}{\LinOp{\Hilb{A}}}$, we note its Hilbert-Schmidt adjoint with a $*$ i.e., $\mathcal{M}^*$ is the unique map satisfying $\InProd{V'}{\MOn{V}} = \InProd{\mathcal{M}^*(V')}{V}$ for all $V, V'$. Using this notation, the self-adjointness condition can be written concisely as $\Proj{}{A}=\Proj{*}{A}$.

As the projectors should project onto operator systems, which are subspaces that contain the identity, it should be true that their images contain the span of the identity as a subspace. A necessary and sufficient condition for that to be true is \begin{equation}\label{eq:projo_4}
    \Proj{}{A} \circ \mathcal{D}_A  = \mathcal{D}_A \circ \Proj{}{A} = \mathcal{D}_A\:,
\end{equation}
with $\mathcal{D}_A$ defined as in Eq. \eqref{eq:depolop}. 

\paragraph*{Definition and example.}
Summarizing, a projector on an operator system is a linear superoperator $\Proj{}{A} \in \LinOpB{\LinOp{\Hilb{A}}}{\LinOp{\Hilb{A}}}$ obeying the defining conditions \eqref{eq:projo_1}, \eqref{eq:projo_2}, and \eqref{eq:projo_4}. Moreover, for the purpose of this article, it is assumed to be self-adjoint without loss of generality, condition \eqref{eq:projo_3}. 

Putting everything together, the projector appearing in Eq. \eqref{eq:state_char_proj} is an orthogonal projector on an operator system, meaning it obeys the following set of conditions: 
\begin{subequations}\label{eq:projo}%
    \begin{gather}%
        \Proj{}{A} \circ \Proj{}{A} = \Proj{}{A}\:;\label{eq:projo_cond}\\
        \Proj{}{A} \circ \dag  = \dag \circ \Proj{}{A}\:; \label{eq:projo_HP} \\
        \Proj{*}{A}= \Proj{}{A}\:; \label{eq:projo_SA}\\
        \Proj{}{A} \circ \mathcal{D}_A  = \mathcal{D}_A \circ \Proj{}{A} = \mathcal{D}_A \label{eq:projo_unital}\:.
    \end{gather}%
\end{subequations}
When referring to `a projector' in this article, we implicitly mean `an orthogonal projector on an operator system', i.e., `a projector obeying conditions \eqref{eq:projo}', unless said otherwise.

Examples of projectors on operator systems over a 2-dimensional Hilbert space can be generated using the Pauli basis $\{\mathds{1},\sigma_x,\sigma_y,\sigma_z\}$ (see e.g., Ref. \cite{Nielsen2009} for its definition and properties): it can be checked that any projector `removing' one or several basis elements other than the identity is a projector on an operator system i.e., obeys \eqref{eq:projo}. 
Such is the case of the projector to the subspace of operators that are diagonal in the computational basis, which is used for quantum decoherence. This projector, noted $\Delta: \LinOp{\Hilb{A}} \rightarrow \LinOp{\Hilb{A}}$ and which acts as $\Delta(\cdot) \equiv \frac{\mathds{1}}{2} \TrX{}{\cdot} + \frac{\sigma_z}{2} \TrX{}{\sigma_z\: \cdot }$, is the projector to the subspace spanned by $\{\mathds{1},\sigma_{z}\}$, thus a projector on an operator system. 

\paragraph{Sets of commuting base projectors.\label{sec:projos_base_sets}}

For the purposes of this work, it is assumed that all the projectors under consideration always commute pairwise when acting on the same tensor factor of a Hilbert space. 
Remark that the commutation of two projectors is a sufficient condition for their composition to be a projector as well (see e.g., Ref. \cite[Theorem 2.26 3)]{Roman2008}). 

At first glance, the physical heuristic behind the assumption is that we are only interested in building (higher-order) maps between subsystems of a fixed kind: if Alice's system is a qubit and Bob's is a classical bit, this will not change whether we consider their joint state or a channel between them; the framework developed in this article is only used to characterize their joint state space or the set of channels they can share, not what happens when Alice suddenly demotes her system to a bit. The characterization is done by combining the base projectors of Alice and Bob, which in this case are $\mathcal{I}_A$ and $\Delta_B$ (where $\mathcal{I}$ is the identity projector and $\Delta$ is the projector to the diagonal subspace as defined above). For that reason, it can be assumed that there is only one base projector associated with each factor of a Hilbert space at the start of the construction. 

The rules introduced in the main text and detailed below then lead to several ways to combine the projectors of Alice and Bob into a projector acting on space $\LinOp{\Hilb{A}\otimes \Hilb{B}}$. Hence, the `Alice+Bob' system can be of several different natures: it can be bipartite states, but also channels (in CJ representation), bipartite measurements, etc. All of these correspond to different projectors, and each of these projectors can in turn be seen as the base projectors in the coarse-grained, `Alice+Bob' description. One of the things shown in this appendix is precisely that every introduced rule to build new projectors out of known ones will preserve commutation. Hence, any two projectors characterizing systems of a different type and obtained by relating the same base projectors on the same subsystems using any of the various algebraic rules defined in the following will commute by construction. 

\paragraph*{Non-singleton set of base projectors.} 
Upon closer inspection, it may still be interesting to study what happens when the base subsystems can be picked among a set of possible types. For instance, it is very natural to extrapolate the characterization of a class of processes to a case in which, say, Alice's state preparation becomes a discarding procedure (graphically, it amounts to changing Alice's wire from facing up to facing down). This replacement actually corresponds to the `not' operation in the algebra, so it will preserve commutation as will be proven below. Hence, any base projector can be associated with a second base projector, its negation. In addition, the operation of `intersection' and `union' will also result in two new projectors that commute with both the base projector and its negation. Therefore, even when one considers one base projector per subsystem, it is always extensible to a set of at least four pairwise commuting base projectors. 
There is no point in restricting the set of base projectors then and one can consider any commuting set of base projectors in general. 
Therefore, for the projective characterization developed in this work, the general setting considers $k$ parties $A,B,...$, each of which having access to $n_A, n_B, ...$ base projectors $\{\Proj{(1)}{A},\Proj{(2)}{A},..., \Proj{(n_A)}{A}\}$, $\{\Proj{(1)}{B},\Proj{(2)}{B},..., \Proj{(n_B)}{B}\}$, ... corresponding to the different types of subsystems they may have access to. The commutation assumption, which can be expressed as
\begin{equation}\label{eq:projos_comm_1}
    \Proj{(i)}{X} \circ \Proj{(j)}{X} =  \Proj{(j)}{X} \circ  \Proj{(i)}{X} \:, \quad \forall X,\,\forall i,j\:.
\end{equation}
As mentioned in the related comment made at the end of Sec. \ref{sec:algebra_the_algebra}, these will not only encompass all higher-order objects that could have been built using a fine-graining of the subsystem. For example, it can also include scenarios in which Alice can choose whether she wants to use a classical or quantum bit as her system to model scenarios featuring potential decoherence. In our framework, this corresponds to Alice having a set of base projectors made of two elements, $\{\mathcal{I}_A,\Delta_A\}$ (which will then be extended to a set of eight elements by using the not, intersection and union operations). We want to bring to the reader's attention that since these two projectors commute, all the results of this paper will still hold, albeit it is impossible to build $\Delta$ out of $\mathcal{I}$ using any of the rules presented in this article. This is thus a situation characterized by our result but not considered in this article, which opens a potential direction for future research.

\paragraph*{Commutation with the transposition.} A follow-up consequence of assuming that all projectors commute pairwise is that, when acting on self-adjoint operators, they can be assumed to commute with the transposition induced by the CJ isomorphism. This condition will be required simply because it removes the necessity of keeping track of the transposes, resulting in less clustered formulae. 

Here, transposition is meant in a more general sense than matrix transpose. The CJ correspondence depends on a choice of basis, namely the one in which the matrix transpose that appears in Eq. \eqref{eq:CJ} is performed. This implicitly means that the standard bases of $\Hilb{A}$ and $\Hilb{B}$ were chosen to define the isomorphism. However, this is an arbitrary choice. Eq. \eqref{eq:CJ} is a shortcut for
\begin{equation}
\begin{aligned}
    M &= \mathrm{Tr}_{A'B'}\big[(\phi^+_{AA'} \otimes \phi^+_{B'B}) \\
     &\big(\mathds{1}_A \otimes (\mathcal{I} \otimes \mathcal{M})\{\phi^+\}_{A'B'} \otimes \mathds{1}_B \big)\big] ,
\end{aligned}
\end{equation}
and it is only because $\phi^+ $ was chosen to be $ \phi^+=\sum_{i,j} \ket{i}\bra{j}\otimes \ket{i}\bra{j}$ that it corresponds to matrix transposition in the standard basis. 
Yet, any other choice of unnormalized maximally entangled density matrix $\Phi$ is equally good. Let $\{e_\mu\}_{\mu=0}^{d_A^2-1}$ be an orthonormal basis of $\LinOp{\Hilb{A}}$. Then this basis can always be associated with an unnormalized maximally entangled state of the form $\Phi_{A'A} = \sum_\mu e_\mu \otimes e_\mu$, which is in turn used to define a transposition with respect to $\{e_\mu\}$ by
\begin{equation}
    V_A^T \equiv \TrX{A'}{ \Phi \: (V_{A'}\otimes \mathds{1}_A)} ,\quad \forall V_{A'} \in \LinOp{\Hilb{A'}}\:.
\end{equation}
This abstract notion of transposition happens to correspond to the one of matrix transpose with respect to a basis $\{\ket{\varphi_i}\}_{i=0}^{d_A-1}$ of $\Hilb{A}$ precisely when the basis $\{e_\mu\}$ is $\{e_\mu := \ket{\varphi_i}\bra{\varphi_j}\}_{\mu=(i,j)}$, but it does not have to have this form to define a CJ correspondence. All that matters is that a pair of bases for $\LinOp{\Hilb{A}}$ and $\LinOp{\Hilb{B}}$ is chosen.

This freedom of basis in the definition of the CJ correspondence can be leveraged to make all projectors commute with the transpositions it induces. For instance, this work typically deals with equations of the form
\begin{equation}
    \mathcal{M}(V_A) = \ProjOn{}{B}{\MOn{\ProjOn{}{A}{V_A}}} \:,
\end{equation}
where $\mathcal{M}$ is CP map, $V_A$ a positive operator, and $\Proj{}{A}$ and $\Proj{}{B}$ are projectors on operator systems. 
Our methods will use the reverse direction of the CJ correspondence \eqref{eq:CJ^-1} to rephrase the above into  
\begin{equation}\label{eq:app_transpo1}
    \mathcal{M}(V_A) = \ProjOn{}{B}{\big(\TrX{A}{M\:(\ProjOn{}{A}{V_A} \otimes \mathds{1})}\big)^T} \:,
\end{equation}
and then gather all projectors into a single expression as in
\begin{equation}\label{eq:app_transpo2}
    \mathcal{M}(V_A) = \big(\TrX{A}{((\Proj{}{A} \otimes \Proj{}{B})\{M\})\:(V_A \otimes \mathds{1})}\big)^T \:,
\end{equation}
where the tensor product of projectors is defined in a standard way (see Eq. \eqref{eq:tensor} or in App. \ref{sec:projo_prop_tensor} below). Going from Eq. \eqref{eq:app_transpo1} to \eqref{eq:app_transpo2} assumed that $\big(\ProjOn{}{B}{N_B}\big)^T=\ProjOn{}{B}{N_B^T}$ for all $N_B$. A similar rewriting can also be made when computing $\mathcal{M}^*(N_B)$, in which case one needs to assume that $\big(\ProjOn{}{A}{V_A}\big)^T=\ProjOn{}{A}{V_A^T}$ for all $V_A$. When all operators are self-adjoint, we now show that these assumptions can always be satisfied by picking a specific choice of bases to define the CJ isomorphism.

The general statement we prove is that if the base projectors form a set of $n_X$ pairwise commuting orthogonal projectors on operator systems over subsystem $X$, $\{\Proj{(1)}{X},\Proj{(2)}{X},\Proj{(3)}{X},\ldots,\Proj{(n_X)}{X}\}$, the following can be ensured:
\begin{equation}
    \big(\ProjOn{(i)}{X}{V}\big)^T = \ProjOn{(i)}{X}{V^T}\:, \quad \forall i, \: \forall V: \: V=V^\dag\:, \label{eq:projos_trans_1}
\end{equation}
for the transpositions induced by the CJ isomorphism. This is done by choosing the basis for the isomorphism to be the one diagonalizing the set of projectors. 

If this basis is chosen to be the one that diagonalizes the set of base projectors, i.e. if the base projectors is a set of $n$ pairwise commuting orthogonal projectors $\{\Proj{(1)}{X},\Proj{(2)}{X},\Proj{(3)}{X},\ldots,\Proj{(n_X)}{X}\}$ expressed in this basis as $\ProjOn{(i)}{X}{\cdot} = \sum_{\mu} c^{(i)}_\mu \: e_\mu \TrX{}{e_\mu^\dag \: \cdot}$ with $i=1,\ldots,n_X$ and $c^{(i)}_\mu \in \{0,1\}$, then 
\begin{equation}\label{eq:projo_commutes_1}
    \sum_\nu \ProjOn{(i)}{X'}{e_\nu} \otimes e_\nu = \sum_\nu e_\nu \otimes \ProjOn{(i)}{X}{e_\nu}\:,
\end{equation}
and it follows that the projectors commute with the transpose when applied on self-adjoint operators. Indeed, $\forall V:\: V = V^\dag$ and $\forall i$:
\begin{equation}\label{eq:projo_commutes_2}
    \begin{aligned}
        &\ProjOn{(i)}{X}{V^T}\\
        &= \ProjOn{(i)}{A}{ \TrX{X'}{(V_{X'} \otimes \mathds{1}_X) (\sum_\nu e_\nu \otimes e_\nu)} }\\
        &=  \TrX{X'}{(V_{X'} \otimes \mathds{1}_X) (\sum_\nu e_\nu \otimes \ProjOn{(i)}{X}{e_\nu}) } \\
        &= \TrX{X'}{(V_{X'} \otimes \mathds{1}_X) (\sum_\nu \ProjOn{(i)}{X'}{e_\nu} \otimes e_\nu) }\\
        &=\TrX{A'}{(\ProjOn{(i)}{X'}{V_{X'}} \otimes \mathds{1}_X) (\sum_\nu e_\nu \otimes e_\nu)}\\
        &= \left(\ProjOn{(i)}{X}{V}\right)^T\:,
    \end{aligned}
\end{equation}
where passing from the third to the fourth line was possible only because $V=V^\dag$ and $\Proj{(i)}{X}$ is self-adjoint by definition.

\paragraph*{Definition.}
Summarizing, any subsystem $X$ will be associated with $n_X$ different types given a priori, called their base types. These base types correspond to $n_X$ different state structures to which are associated $n_X$ base projectors noted $\{\Proj{(1)}{X},\Proj{(2)}{X},\Proj{(3)}{X},\ldots,\Proj{(n_X)}{X}\}$ or $\{\Proj{}{X},\Proj{'}{X},\Proj{''}{X},\ldots\}$ . For every self-adjoint operator $V\in \LinOp{\Hilb{X}}$, the following is always assumed to hold for a set of base projectors:
\begin{subequations}\label{eq:projos}
    \begin{gather}
        \Proj{(i)}{X} \circ \Proj{(j)}{X} =  \Proj{(j)}{X} \circ  \Proj{(i)}{X} \:, \quad \forall i,j\:; \label{eq:projos_comm} \\
        \big(\ProjOn{(i)}{X}{V}\big)^T = \ProjOn{(i)}{X}{V^T}\:, \quad \forall i\:. \label{eq:projos_trans} 
    \end{gather}
\end{subequations}

\paragraph*{Remark: non-pairwise-commuting set of base projectors.} 
The assumption that the base projectors all commute is reasonable for the scope of this paper and potential extensions, as argued in the above discussion. Yet it is not necessary: it was only made because there are no heuristical reasons to go beyond it. While it greatly simplifies the mathematics by allowing the intersection (or cap) operation $\cap$ to be defined as the composition of two projectors (Eq. \eqref{eq:cap} in the main text; its mathematical properties are studied in the next subsection below), dropping it may also open a direction for further research: the natural follow-up question to considering a higher-order theory based on a set of pairwise commuting base projectors would be to study what happens when the base projectors are no longer commuting. In this case, most of the results would have to be generalized in a version that does not use this assumption nor the related one that the projectors commute with the transpose. Nonetheless, besides pure mathematical curiosity, we do not have a clear idea of what the purpose of such a framework may be; let alone its physical interpretation.

\subsection{Adding cap and cup operations makes an algebra \label{sec:projos_prop}}
Here, we prove that superoperator projectors on operator systems constitute an algebra under certain rules. Assume $\Proj{}{A}$ and $\Proj{'}{A}$ to be two arbitrary projectors on operator systems, not necessarily the same. By definition, a projector is $\mathbb{C}$-linear:
\begin{equation}
    \ProjOn{}{A}{\sum_i q_i V_i} = \sum_i \;q_i\; \ProjOn{}{A}{V_i}\:,
\end{equation}
$q_i \in \mathbb{C}$.
This linearity allows us to carry on the addition `+' of $\LinOp{\Hilb{A}}$ at the level of the projectors:
\begin{equation}
    \left(\Proj{}{A}+\Proj{'}{A}\right)\{V\} \equiv \ProjOn{}{A}{V} + \ProjOn{'}{A}{V}\:.
\end{equation}
Using linearity again, one can define the negation `-' and scalar multiplication of projectors, thereby defining a vector space over $\mathbb{C}$. Associativity and commutativity are not hard to prove from there.

Projectors also have a natural `conjunction' operation that can be interpreted as a multiplication or as a logic `and', here nicknamed the \textbf{intersection} or the \textbf{cap} and noted via the `$\cap$' symbol. In the case of commuting projectors $\Proj{}{A}$ and $\Proj{'}{A}$, it is
\begin{equation}
    \Proj{}{A} \cap \Proj{'}{A} \equiv \Proj{}{A} \circ \Proj{'}{A} \:.
\end{equation}
This projector $\Proj{}{A} \cap \Proj{'}{A}$ characterizes the intersection of the operator systems $\Alg{A}{}$ and $\Alg{A}{}'$, see the main text.
It is straightforward to prove that it fits the definition of a conjunction: it is a binary operation distributive under the addition defined above and compatible with the scalar multiplication; it is both associative and commutative; for projectors satisfying Eqs. \eqref{eq:projos}, it is resulting in a projector satisfying Eqs. \eqref{eq:projos}; and it is idempotent when acting on projectors since:
\begin{equation}\label{eq:idempot}
    \Proj{}{A} \cap \Proj{}{A} \equiv \Proj{}{A} \circ \Proj{}{A} \equiv \Proj{2}{A} = \Proj{}{A} \:.
\end{equation}
In the above equation, we have also defined a shorthand notation for `squaring' under the multiplication $\cap$: $\Proj{}{A} \cap \Proj{}{A} \equiv \Proj{2}{A}$.

An issue revealed by squaring is that the operation `+' does not necessarily map projectors to a projector:
\begin{equation}
    \left(\Proj{}{A} + \Proj{'}{A}\right)^2 = \Proj{2}{A} + \Proj{'2}{A} + 2 \left( \Proj{}{A} \cap \Proj{'}{A} \right)\:.
\end{equation}
As the last term of the right-hand side is not always zero, $\left(\Proj{}{A} + \Proj{'}{A}\right)^2 \neq \Proj{}{A} + \Proj{'}{A}$ in general, so addition does not preserve idempotency. We wish nonetheless to have operations that keep us in the vector space of commuting projectors, so we redefine the addition to be the union of two projectors (see \textit{e.g.} Ref. \cite{Piziak1999a}). This is inspired by set theory: as $\Proj{}{A}$ maps to a subspace $\mathscr{A}\subset \LinOp{\Hilb{A}}$ and $\Proj{'}{A}$ to $\mathscr{A}' \subset \LinOp{\Hilb{A}}$ we want an addition `$\cup$' so that $\Proj{}{A} \cup \Proj{'}{A}$ maps to $\mathscr{A}\cup \mathscr{A}'$. This requirement is realized by the `disjunction' operation, that can be interpreted as an addition or as logic `or', nicknamed the \textbf{union} or the \textbf{cup} (see also Ref. \cite{MPM}):
\begin{equation}
    \Proj{}{A} \cup \Proj{'}{A} \equiv \Proj{}{A} + \Proj{'}{A} - \Proj{}{A} \cap \Proj{'}{A} \:.
\end{equation}
As the name hints, $\Proj{}{A} \cup \Proj{'}{A}$ characterizes the union of the operator systems $\Alg{A}{}$ and $\Alg{A}{}'$. It is again straightforward to prove that it fits the definition from the properties inherited by the addition and the cap: the cup is a binary operation distributive under the addition defined above and compatible with the scalar multiplication; it is both associative and commutative; for projectors satisfying Eqs. \eqref{eq:projos}, it results in a projector satisfying Eqs. \eqref{eq:projos}; and it inherits the idempotency from the idempotency of the cap as $\Proj{}{A} \cup \Proj{}{A} = \Proj{}{A} + \Proj{}{A} - (\Proj{}{A} \cap \Proj{}{A}) = \Proj{}{A} + \Proj{}{A} - \Proj{}{A}$.

These caps and cup connectors give a quick way to prove subspace inclusions. The subspace spanned by state structure $\mathscr{A}'$ is embedded in the subspace spanned by another state structure $\mathscr{A}$ if and only if either of the following holds (see e.g., Ref. \cite{Roman2008})
\begin{subequations}\label{eq:cap_cup_inclu_app}
    \begin{align}
        \Proj{}{A} \cap \Proj{'}{A} = \Proj{'}{A} \:;\\
        \Proj{}{A} \cup \Proj{'}{A} = \Proj{}{A} \:.
    \end{align}
\end{subequations}
The intersection and union of two projectors project to, respectively, the intersection and union of their images \cite[Theorem 2.26]{Roman2008} i.e.,
\begin{subequations}
    \begin{align}
        \mathrm{Im}\left\{\Proj{}{A} \cap \Proj{'}{A}\right\} = \mathrm{Im}\left\{\Proj{}{A}\right\}\cap \mathrm{Im}\left\{\Proj{'}{A}\right\} \:,\\
        \mathrm{Im}\left\{\Proj{}{A} \cup \Proj{'}{A}\right\} = \mathrm{Im}\left\{\Proj{}{A}\right\} \cup \mathrm{Im}\left\{\Proj{'}{A}\right\}\:,
    \end{align}
\end{subequations}
where the $\cap$ and $\cup$ appearing in the right-hand side of the above are the set-theoretic notions of intersection and union, respectively.  
Condition \eqref{eq:projo_unital} is then recast into
\begin{equation}\label{eq:id_pres}
    \Proj{}{A} \cap \mathcal{D}_A = \mathcal{D}_A \:.
\end{equation}

As defined in Eq. \eqref{eq:subset_def} of the main text, obeying Eqs. \eqref{eq:cap_cup_inclu_app} is concisely noted as $\Proj{'}{A}\subseteq\Proj{}{A}$ and means that such an equation must be understood in terms of the image of the projectors. Accordingly, an equivalence of two projectors is defined as the equivalence of their images:
\begin{equation}
    \Proj{}{A} = \Proj{'}{A} \iff \Im{\Proj{}{A}} = \Im{\Proj{'}{A}}\:.
\end{equation}

\paragraph*{The algebra.} Commuting projectors on operator systems thus form an algebra over $\mathbb{C}$ since they are closed under $(\cap,\cup)$ operations as well as scalar multiplication. 
The cap and cup can indeed be seen as some multiplication and addition of projectors as the former distributes over the latter:
\begin{equation}\label{eq:cap_cup_dist}
    \left( \Proj{}{A} \cup \Proj{'}{A}\right)\cap \Proj{''}{A} = \left(\Proj{}{A} \cap\Proj{''}{A} \right)\cup \left(\Proj{'}{A} \cap \Proj{''}{A} \right)\:.
\end{equation}
Indeed, $\left( \Proj{}{A} \cup \Proj{'}{A}\right)\cap \Proj{''}{A} = \left( \Proj{}{A} + \Proj{'}{A} - (\Proj{}{A} \cap \Proj{'}{A}) \right)\cap \Proj{''}{A} = \big(\Proj{}{A} \cap \Proj{''}{A}\big) + \big(\Proj{'}{A} \cap \Proj{''}{A}\big) - \big(\Proj{}{A} \cap \Proj{'}{A}\big) \cap \Proj{''}{A}$, and $\big(\Proj{}{A} \cap \Proj{'}{A}\big) \cap \Proj{''}{A} = \big(\Proj{}{A} \cap \Proj{'}{A}\big) \cap \big(\Proj{''}{A} \cap \Proj{''}{A}\big) = \big(\Proj{}{A} \cap \Proj{''}{A}\big) \cap \big(\Proj{'}{A} \cap \Proj{''}{A}\big)$.

The projectors $\mathcal{I}$ and $\mathcal{D}$ defined in Eqs. \eqref{eq:Id} and \eqref{eq:depolop} above play a special role in it: one can identify the identity superoperator $\mathcal{I}$ as the multiplicative identity or `\textbf{unit}' of the algebra of projectors since for all projectors $\Proj{}{}$,
\begin{equation}\label{eq:cap_unit}
    \mathcal{I} \cap \Proj{}{} = \Proj{}{} \:;
\end{equation}
it is also the additive absorbing element i.e.
\begin{equation}
    \mathcal{I} \cup \Proj{}{} = \mathcal{I} \:.
\end{equation}
The other way around, the depolarizing superoperator $\mathcal{D}$ is the multiplicative absorbing element or `\textbf{zero}' of the algebra,
\begin{equation}
    \mathcal{D} \cap \Proj{}{} = \mathcal{D}\:,
\end{equation}
and the additive identity,
\begin{equation}\label{eq:cup_unit}
    \mathcal{D} \cup \Proj{}{} = \Proj{}{}\:.
\end{equation}
Because of that property, it should be clear that for any element of the algebra, the following is true:
\begin{equation}
    \mathcal{D} \subseteq \Proj{}{} \subseteq \mathcal{I} \:.
\end{equation}
As a consequence, this is an algebra of idempotent elements equipped with a partial ordering $\subseteq$ that have a common greatest element $\mathcal{I}$ and a common least element $\mathcal{D}$. 
Moreover, the intersection and union of any two elements $\Proj{}{},\Proj{'}{}$ are uniquely defined elements of the algebra. It is direct to check that the following property holds
\begin{equation}
    \Proj{}{} \cap \Proj{'}{} \subseteq \Proj{}{} \subseteq \Proj{}{} \cup \Proj{'}{}\:,
\end{equation}
and so does the analog with $\Proj{'}{}$. In other words, the algebra has a lattice structure (which directly comes from the lattice structure of subspace inclusions; see e.g., Ref. \cite{Roman2008}).

\paragraph*{Remark.}Although referring to the cap and cup as `multiplication' and `addition' provides an intuitive meaning to the connectors, in the following we will call them respectively `additive disjunction and conjunction'. This choice is made in order to avoid ambiguity with the multiplicative conjunction and disjunction that will be introduced in a next section.

\subsection{Adding negation makes a Boolean algebra \label{sec:projo_Boolean}}
In the algebra of superoperator projectors, the new object that naturally appears in Proposition \ref{theo:det_fctal}, $\CompProj{}{A}$, can be seen as an operation on the original projector $\Proj{}{A}$, whence the `bar over $\Proj{}{A}$' notation. This new operation is defined for any projector, and it promotes the algebra to a \textbf{Boolean algebra}. 

It is a Boolean algebra if, in addition to the conjunction (`multiplication') operation $\cap$ and the disjunction (`addition') operation $\cup$, any element is idempotent and has a well-defined complement (or logic `not', thereby referred to as `negation'\footnote{We could also have used the terminology `inverse', `dual', or we could have stuck with `complement', but these words are already used for different concepts in quantum information and functional analysis. This choice is made by consequence to avoid ambiguities.}) characterized by the condition that the addition of any projector with its negation yields the additive identity, i.e. 
\begin{equation}\label{eq:A_cup_nA}
    \Proj{}{A} \cup \CompProj{}{A} = \mathcal{I}_A\:,
\end{equation}
where $\CompProj{}{A}$ is the negation of $\Proj{}{A}$. This is exactly what the quasi-orthogonal complement of Proposition \ref{theo:det_fctal} is doing. First, compute the cap as $\Proj{}{A} \cap \CompProj{}{A} = \Proj{}{A} \circ \left(\mathcal{I}_A - \Proj{}{A} + \mathcal{D}_A\right) = \Proj{}{A}-\Proj{2}{A} +\mathcal{D}_A = \Proj{}{A}-\Proj{}{A}+\mathcal{D}_A$, so
\begin{equation}\label{eq:A_cap_nA}
    \Proj{}{A} \cap \CompProj{}{A} = \mathcal{D}_A\:.
\end{equation}
Then, it is true for the cup because $\Proj{}{A} \cup \CompProj{}{A} = \Proj{}{A} + \mathcal{I}_A - \Proj{}{A} + \mathcal{D}_A  - (\Proj{}{A} \cap \CompProj{}{A}) = \mathcal{I}_A$. 
Thus, we define the \textbf{negation} of any projector $\Proj{}{A}$ as
\begin{equation} \label{eq:neg}
    \CompProj{}{A} \equiv \mathcal{I}_A - \Proj{}{A} + \mathcal{D}_A  \:.
\end{equation}
In addition to property \eqref{eq:A_cup_nA}, one can verify from Eq. \eqref{eq:A_cap_nA} that the projector and its negation intersect at the zero of the algebra, which in this case is indeed $\mathcal{D}_A$ (refer to the diagram in Fig. \ref{fig:fctal}). 
As it should be, negation is an involution: 
\begin{equation}\label{eq:invol}
    \overline{\overline{\mathcal{P}}}_A = \mathcal{P}_A\:,
\end{equation}
and it is closed so that the negation of any projector is a projector as we now prove.
\paragraph*{Negation preserves projector on operator system:}
Note that since $\mathcal{I}_A$ and $\mathcal{D}_A$ are Hermitian-preserving (HP) self-adjoint (SA), and commute with the transpose (CT), Eqs. \eqref{eq:projo_HP},  \eqref{eq:projo_SA}, and \eqref{eq:projos_trans}, the negation of a projector is HP, SA, and CT provided the original projector is.
The idempotency property \eqref{eq:projo_cond} is preserved, $\CompProj{2}{A} = \left(\mathcal{I}_A - \Proj{}{A} + \mathcal{D}_A \right) \cap \left(\mathcal{I}_A - \Proj{}{A} + \mathcal{D}_A \right) = \mathcal{I}_A - \Proj{}{A} + \mathcal{D}_A  - \Proj{}{A} + \Proj{2}{A} - \mathcal{D}_A  + \mathcal{D}_A  -\mathcal{D}_A + \mathcal{D}_A = \mathcal{I}_A - \Proj{}{A}+ \mathcal{D}_A $, hence
\begin{equation}\label{eq:n^2}
    \CompProj{2}{A} = \CompProj{}{A}\:.
\end{equation}
To prove that the negation of a projector on an operator system is a projector on an operator system itself, it remains to show that the identity is still contained in the negation, Eq. \eqref{eq:projo_unital},
\begin{equation}\label{eq:n_depol}
    \mathcal{D}_A \cap \CompProj{}{A} = \mathcal{D}_A\:.
\end{equation}
This follows from $\mathcal{D}_A \cap\CompProj{}{A}=\mathcal{D}_A\cap \left(\mathcal{I}_A - \Proj{}{A} +\mathcal{D}_A \right) = \mathcal{D}_A -\mathcal{D}_A  + \mathcal{D}_A$. 
If we did not have the heuristic of Proposition \ref{theo:det_fctal}, this could be interpreted as a reason for choosing Eq. \eqref{eq:neg} as the definition of negation instead of the orthogonal complement $\mathcal{I}_A - \Proj{}{A}$. In this latter case, identity does not belong to the negated state structure.

\paragraph*{Negation preserves commutativity:}It should also be clear from $\Proj{}{A} \cap \CompProj{}{A} = \Proj{}{A} \cap \left(\mathcal{I}_A - \Proj{}{A} + \mathcal{D}_A\right) = \Proj{}{A}-\Proj{2}{A} +\mathcal{D}_A = \left(\mathcal{I}_A - \Proj{}{A} + \mathcal{D}_A\right) \cap \Proj{}{A}$ that negated projectors commute with the original ones, and, by the same kind of proof, that the negations of two commuting projectors still commute:
\begin{equation}
	\Proj{}{A} \cap \Proj{'}{A} = \Proj{'}{A} \cap \Proj{}{A} \Longrightarrow \CompProj{}{A} \cap \overline{\Proj{'}{A}} = \overline{\Proj{'}{A}} \cap \CompProj{}{A} \:.
\end{equation}

\paragraph*{De Morgan duality:}Moreover, the De Morgan laws are valid for the Boolean algebra of projectors:
\begin{subequations} \label{eq:n_deMorgan}
    \begin{gather}
        \overline{\Proj{}{A}\cup\Proj{'}{A}} = \overline{\Proj{}{A}} \cap \overline{\Proj{'}{A}} \:, \label{eq:deMorgan_cup_to_cap} \\
        \overline{\Proj{}{A}\cap\Proj{'}{A}} = \overline{\Proj{}{A}} \cup \overline{\Proj{'}{A}} \:. \label{eq:deMorgan_cap_to_cup} 
    \end{gather}
\end{subequations}
The proof is more involved: $\overline{\Proj{}{A}\cup\Proj{'}{A}} = \mathcal{I}_A - \Proj{}{A}\cup \Proj{'}{A} + \mathcal{D}_A = \mathcal{I}_A - (\Proj{}{A}+\Proj{'}{A}-(\Proj{}{A}\cap \Proj{'}{A})) + \mathcal{D}_A = \mathcal{I}_A - \Proj{}{A} + \mathcal{D}_A + \mathcal{D}_A - \Proj{'}{A} - \mathcal{D}_A + (\Proj{}{A}\cap \Proj{'}{A}) - \mathcal{D}_A + \mathcal{D}_A = (\mathcal{I}_A-\Proj{}{A}+D)\cap(\mathcal{I}_A-\Proj{'}{A}+\mathcal{D}_A)=\CompProj{}{A} \cap \overline{\Proj{'}{A}}$. The second identity directly ensues. 
The inclusion relation $\subseteq$, defined by conditions \eqref{eq:inclusion_condition}, is by consequence reversed when the projectors it involves are negated. Indeed, if $\Proj{'}{A} \subseteq \Proj{}{A}\Rightarrow \Proj{'}{A} \cap \Proj{}{A} = \Proj{'}{A}$, negating both sides of the cap yields $\overline{\Proj{'}{A}} \cap \CompProj{}{A} = \overline{\Proj{'}{A} \cup \Proj{}{A}}$. But the term under the negation in $\overline{\Proj{'}{A} \cup \Proj{}{A}}$ is equivalent to $\Proj{}{A}$ since the inclusion also implies that $\Proj{'}{A} \subseteq \Proj{}{A}\Rightarrow \Proj{'}{A} \cup \Proj{}{A} = \Proj{}{A}$. Whence, $\overline{\Proj{'}{A}} \cap \CompProj{}{A} = \overline{\Proj{'}{A} \cup \Proj{}{A}} = \CompProj{}{A}$, which by definition is the inclusion $\CompProj{}{A} \subseteq \overline{\Proj{'}{A}}$. This reasoning works in both ways, proving
\begin{equation}\label{eq:inclusion_duality_app}
   \Proj{'}{A} \subseteq \Proj{}{A}\:\iff\: \CompProj{}{A} \subseteq \overline{\Proj{'}{A}}\:,
\end{equation}
which is equation \eqref{eq:inclusion_duality} in the main text.

\subsection{Adding tensor and transformation almost makes Linear Logic\label{sec:projo_LL}}
\subsubsection{The tensor operation \label{sec:projo_prop_tensor}}
For Definition \ref{prop:tensor}, we also defined a tensor product at the level of the projectors. Here we now present its algebraic properties so that we can meaningfully consider $\mathscr{A} \otimes \mathscr{B} \subset \LinOp{\Hilb{A}\otimes\Hilb{B}}$ when we know that $\mathscr{A}$ is characterized by a projector $\Proj{}{A}$ acting on $\LinOp{\Hilb{A}}$ and that $\mathscr{B}$ is by $\Proj{}{B}$. Because of the isomorphism $\LinOp{\Hilb{A}\otimes\Hilb{B}} \cong \LinOp{\Hilb{A}} \otimes \LinOp{\Hilb{B}}$, the definition of \textbf{no signaling composition}, nicknamed \textbf{tensor}, is straightforward:
\begin{multline}
    \left(\Proj{}{A} \otimes \Proj{}{B} \right) \left\{\sum_i \; q_i \; \left(V_i \otimes U_i\right)\right\} \equiv\\ \sum_i\;q_i\;\left(\ProjOn{}{A}{V_i}\otimes \ProjOn{}{B}{U_i}\right) \:,
\end{multline}
and it should hold for all $i$ such that $q_i \in \mathbb{C}$, $V_i \in \LinOp{\Hilb{A}}$, $U_i \in \LinOp{\Hilb{B}}$ as any operator in $\LinOp{\Hilb{A}\otimes \Hilb{B}}$ can be expressed as such. 

By the inherited properties of the tensor product at the level of operators, the tensor at the level of superoperators is associative; it is moreover distributive with respect to the cap,
\begin{equation}\label{eq:cap_tenor_dist}
    \left( \Proj{}{A} \cap \Proj{'}{A} \right) \otimes \Proj{}{B} = \left(\Proj{}{A} \otimes \Proj{}{B}\right) \cap \left(\Proj{'}{A} \otimes \Proj{}{B}\right)\:,
\end{equation}
as well as to the cup,
\begin{equation}\label{eq:cup_tensor_dist}
    \left( \Proj{}{A} \cup \Proj{'}{A} \right) \otimes \Proj{}{B} = \left(\Proj{}{A} \otimes \Proj{}{B}\right) \cup \left(\Proj{'}{A} \otimes \Proj{}{B}\right)\:.
\end{equation}
This directly follows from idempotency:
\begin{equation}
    \left( \Proj{}{A} \cap \Proj{'}{A} \right) \otimes \Proj{}{B} = \left( \Proj{}{A} \cap \Proj{'}{A} \right) \otimes \left( \Proj{}{B} \cap \Proj{}{B} \right),
\end{equation}
and because the intersection of projectors is party-wise, thus it commutes with the tensor product. Another way to see it is that composition of superoperators obeys an interchange law with the tensor product of superoperators:
\begin{multline}\label{eq:interchange_tensor_cap}
    \left(\Proj{}{A} \cap \Proj{'}{A}\right) \otimes \left(\Proj{}{B} \cap \Proj{'}{B}\right)\\ = \left(\Proj{}{A} \otimes \Proj{}{B}\right) \cap \left(\Proj{'}{A} \otimes \Proj{'}{B}\right) \:.
\end{multline}

\paragraph*{Tensor preserves projectors on operator systems:}Expressions built from the tensor product of Hermitian-preserving projectors are automatically HP since the dagger distributes over the tensor, $\left(\rho \otimes \sigma\right)^\dag = \rho^\dag \otimes \sigma^\dag$. If the composed projectors are self-adjoint then so is their tensor product since the inner product is $\mathbb{C}$-linear and splits as $\InProd{\ProjOn{}{A}{\rho_A} \otimes \ProjOn{}{B}{\sigma_B}}{\eta_A \otimes \chi_B}_{AB} = \InProd{\ProjOn{}{A}{\rho_A}}{\eta_A}_A\InProd{\ProjOn{}{B}{\sigma_B}}{\chi_B}_B $ so that $\left(\Proj{}{A} \otimes \Proj{}{B}\right)^*=\Proj{*}{A} \otimes \Proj{*}{B}$. 
They preserve idempotency,
\begin{multline}
    \left(\Proj{}{A}\otimes\Proj{}{B}\right)^2 =\left(\Proj{}{A} \otimes \Proj{}{B}\right) \cap \left(\Proj{}{A} \otimes \Proj{}{B}\right)\\
    \overset{\eqref{eq:interchange_tensor_cap}}{=}\left(\Proj{}{A} \cap \Proj{}{A}\right) \otimes \left(\Proj{}{B} \cap \Proj{}{B}\right)\:,
\end{multline}
because of interchange law, thus
\begin{equation}\label{eq:tensor^2}
    \left(\Proj{}{A}\otimes\Proj{}{B}\right)^2=\left(\Proj{}{A}\otimes\Proj{}{B}\right)\:.
\end{equation}
A further consequence of the interchange law is that the tensor composition of, respectively, the units and the zeroes on $A$ and $B$, $\mathcal{I}_A \otimes \mathcal{I}_B$ and $\mathcal{D}_A \otimes \mathcal{D}_B$, respectively yield the unit and zero of the projector algebra on $\LinOp{\Hilb{A}\otimes\Hilb{B}}$, $\mathcal{I}_{AB}$ and $\mathcal{D}_{AB}$. Indeed, they can be used to define the negation of $\Proj{}{A}\otimes\Proj{}{B}$, $\overline{\Proj{}{A}\otimes\Proj{}{B}}$, and this gives the correct Boolean completions: $\left(\Proj{}{A}\otimes\Proj{}{B}\right) \cap \overline{\left(\Proj{}{A}\otimes\Proj{}{B}\right)} = \left(\Proj{}{A}\otimes\Proj{}{B}\right) \cap \left(\mathcal{I}_A \otimes \mathcal{I}_B - \Proj{}{A}\otimes\Proj{}{B} + \mathcal{D}_A \otimes \mathcal{D}_B\right) = \left(\Proj{}{A}\otimes\Proj{}{B} - \Proj{2}{A}\otimes\Proj{2}{B} + \mathcal{D}_A \otimes \mathcal{D}_B\right) = \mathcal{D}_A \otimes \mathcal{D}_B$ (at the second equality sign, the interchange law was used to compute the cap of $\Proj{}{A}\otimes\Proj{}{B}$ with each of the three components of the negation). $\left(\Proj{}{A}\otimes\Proj{}{B}\right) \cup \overline{\left(\Proj{}{A}\otimes\Proj{}{B}\right)} = \left(\Proj{}{A}\otimes\Proj{}{B}\right) \cup \left(\mathcal{I}_A \otimes \mathcal{I}_B - \Proj{}{A}\otimes\Proj{}{B} + \mathcal{D}_A \otimes \mathcal{D}_B\right) = \Proj{}{A}\otimes\Proj{}{B} + \mathcal{I}_A \otimes \mathcal{I}_B - \Proj{}{A}\otimes\Proj{}{B} + \mathcal{D}_A \otimes \mathcal{D}_B - \mathcal{D}_A \otimes \mathcal{D}_B = \mathcal{I}_A \otimes \mathcal{I}_B$. Hence,
\begin{subequations}\label{eq:composite_D_I}
    \begin{gather}
        \left(\Proj{}{A}\otimes\Proj{}{B}\right) \cap \overline{\left(\Proj{}{A}\otimes\Proj{}{B}\right)} = \mathcal{D}_A \otimes \mathcal{D}_B \equiv \mathcal{D}_{AB} \:;\label{eq:composite_D}\\
        \left(\Proj{}{A}\otimes\Proj{}{B}\right) \cup \overline{\left(\Proj{}{A}\otimes\Proj{}{B}\right)} = \mathcal{I}_A \otimes \mathcal{I}_B \equiv \mathcal{I}_{AB}\label{eq:composite_I} \:. 
    \end{gather}
\end{subequations} 
Tensor composition then ``preserves the identity'' simply because the identity on a joint system is the tensor product of the identities on each of the spaces being combined, $\mathds{1}_{AB} = \mathds{1}_A \otimes \mathds{1}_B$. In again because of the interchange law, $\mathcal{D}_{AB} \cap \left(\Proj{}{A}\otimes\Proj{}{B}\right)=\left(\mathcal{D}_A \otimes \mathcal{D}_B\right) \cap \left(\Proj{}{A}\otimes\Proj{}{B}\right)=\left(\mathcal{D}_A \cap \Proj{}{A}\right)\otimes \left(\mathcal{D}_B \cap\Proj{}{B}\right) = \mathcal{D}_A \otimes \mathcal{D}_B$,
\begin{equation}\label{eq:tensor_depol}
    \mathcal{D}_{AB} \cap \left(\Proj{}{A}\otimes\Proj{}{B}\right) = \mathcal{D}_{AB} \:.
\end{equation}
This shows that the tensor product of two projectors satisfying Eqs. \eqref{eq:projo} satisfies them as well.

\paragraph*{Tensor product preserves commutation:}Since, if $\Proj{}{A}\cap\Proj{'}{A}= \Proj{'}{A}\cap\Proj{}{A}$ and $ \Proj{}{B}\cap\Proj{'}{B} = \Proj{'}{B}\cap\Proj{}{B}$, then $\left(\Proj{}{A}\otimes\Proj{}{B}\right) \cap \left(\Proj{'}{A}\otimes\Proj{'}{B}\right) = \left(\Proj{}{A}\cap\Proj{'}{A}\right)\otimes \left(\Proj{}{B}\cap\Proj{'}{B}\right) = \left(\Proj{'}{A}\cap\Proj{}{A}\right)\otimes \left(\Proj{'}{B}\cap\Proj{}{B}\right)$, and so
\begin{multline}
	\left(\Proj{}{A}\otimes\Proj{}{B}\right) \cap \left(\Proj{'}{A}\otimes\Proj{'}{B}\right) \\
	= \left(\Proj{'}{A}\otimes\Proj{'}{B}\right) \cap \left(\Proj{}{A}\otimes\Proj{}{B}\right)\:.
\end{multline}

\paragraph*{Negation and tensor:}Interestingly, however, the tensor and negation do not commute with each other: $\overline{\Proj{}{A} \otimes \Proj{}{B} } \neq \overline{\Proj{}{A}} \otimes \overline{\Proj{}{B}}$. Actually, the latter expression defines a subspace in the former since
\begin{multline}\label{eq:proof_tensor_in_par}
    \overline{\Proj{}{A} \otimes \Proj{}{B} } \cap \overline{\Proj{}{A}} \otimes \overline{\Proj{}{B}} \\
    =\left(\mathcal{I}_A \otimes \mathcal{I}_B - \Proj{}{A} \otimes \Proj{}{B} + \mathcal{D}_A \otimes \mathcal{D}_B\right) \cap  \overline{\Proj{}{A}} \otimes \overline{\Proj{}{B}} \\
    =\overline{\Proj{}{A}} \otimes \overline{\Proj{}{B}} - \mathcal{D}_A \otimes \mathcal{D}_B + \mathcal{D}_A \otimes \mathcal{D}_B\:,
\end{multline}
so that 
\begin{equation}
    \overline{\Proj{}{A} \otimes \Proj{}{B} } \cap \overline{\Proj{}{A}} \otimes \overline{\Proj{}{B}} = \overline{\Proj{}{A}} \otimes \overline{\Proj{}{B}}\:.
\end{equation}
And from this result
\begin{equation}
    \overline{\Proj{}{A} \otimes \Proj{}{B} } \cup \overline{\Proj{}{A}} \otimes \overline{\Proj{}{B}} = \overline{\Proj{}{A} \otimes \Proj{}{B} } \:.
\end{equation}
Hence,
\begin{equation}\label{eq:notAnotB<notAB}
    \overline{\Proj{}{A}} \otimes \overline{\Proj{}{B}} \subseteq \overline{\Proj{}{A} \otimes \Proj{}{B} }\:.
\end{equation}
(Recall that the subset symbol, $\subseteq$, implicitly means that we are considering the images of these projectors.) This identity which, using \eqref{eq:inclusion_duality_app}, can be recast as
\begin{equation}\label{eq:AB<n(nAnB)}
    \Proj{}{A} \otimes \Proj{}{B} \subseteq \overline{\CompProj{}{A}\otimes \CompProj{}{B}} \:,
\end{equation}
is Eq. \eqref{eq:tensor_in_par} in the main text. That these two projectors are not equal is fundamentally the reason why an indefinite causal order may arise; this will become clearer after the prec is introduced in the algebra so that equations \eqref{eq:causal=0ways} and \eqref{eq:transfo=2way} can be written down. 
For an illustration of this relation, the left-hand side of this identity is depicted in the leftmost diagram in Fig. \ref{fig:diag_compo}, whereas the right-hand side is the rightmost.

Because of the interchange law \eqref{eq:interchange_tensor_cap}, relation \eqref{eq:notAnotB<notAB} is stable when composed with more parties. More generally, the inclusion relation $\subseteq$ induced by conditions \eqref{eq:subset_def} is stable under the tensor; if $\Proj{}{A} \subseteq \Proj{'}{A}$, then the following hold
\begin{subequations}\label{eq:subset_tensor}
    \begin{gather}
        \Proj{}{A}\otimes \Proj{}{B} \subseteq \Proj{'}{A}\otimes \Proj{}{B} \:,\\
        \overline{\CompProj{}{A}\otimes \CompProj{}{B}} \subseteq \overline{\overline{\Proj{'}{A}}\otimes \CompProj{}{B}}\:,\label{eq:subset_tensor_par-par}\\
        \Proj{}{A}\otimes \Proj{}{B} \subseteq \overline{\overline{\Proj{'}{A}}\otimes \CompProj{}{B}}\:.
    \end{gather}
\end{subequations}
The last equation is a consequence of the first using \eqref{eq:notAnotB<notAB}. The first equation holds because 
\begin{multline}
    \left(\Proj{}{A}\otimes \Proj{}{B}\right)\cap \left(\Proj{'}{A}\otimes \Proj{}{B} \right) \\
   \overset{\eqref{eq:cap_tenor_dist}}{=}\left(\Proj{}{A}\cap \Proj{'}{A} \right)\otimes \Proj{}{B}\:,
\end{multline} 
and the term in parenthesis is equal to $\Proj{}{A}$ because $\Proj{}{A} \subseteq \Proj{'}{A}$, therefore the first equation holds. The second equation is proven using De Morgan duality:
\begin{multline}
    \overline{\CompProj{}{A}\otimes \CompProj{}{B}} \cap \overline{\overline{\Proj{'}{A}}\otimes \CompProj{}{B}}\\ \overset{\eqref{eq:deMorgan_cup_to_cap}}{=} \overline{\left(\CompProj{}{A}\otimes \CompProj{}{B}\right)\cup\left(\overline{\Proj{'}{A}}\otimes \CompProj{}{B}\right) } \\
    \overset{\eqref{eq:cup_tensor_dist}}{=} \overline{\left(\overline{\Proj{'}{A}}\cup \CompProj{}{A}\right)\otimes \CompProj{}{B}}\\
    \overset{\eqref{eq:deMorgan_cap_to_cup}}{=} \overline{\overline{\Proj{'}{A}\cap \Proj{}{A}}\otimes \CompProj{}{B}}\\
    =\overline{\overline{\Proj{'}{A}}\otimes \CompProj{}{B}}\:.
\end{multline} 
Relations \eqref{eq:notAnotB<notAB} and \eqref{eq:subset_tensor} are important to define the no signaling subset of an arbitrary projector, as explained in section \ref{sec:Projos_NS_lattice}.

Finally, observe that the only time the identity $\Proj{}{A} \otimes \Proj{}{B} = \overline{\CompProj{}{A}\otimes \CompProj{}{B}}$ holds is when the projector is either characterizing the totality of the space or only the span of identity. In equation,
\begin{equation}\label{eq:Iso_par_tensor}
    \Proj{}{A} \otimes \Proj{}{B} = \overline{\CompProj{}{A}\otimes \CompProj{}{B}} \: \iff \: \Proj{}{A}=\Proj{}{B}=\mathcal{I}\:\text{or}\:\mathcal{D}\:.
\end{equation}
This is proven by rewriting $\overline{\CompProj{}{A}\otimes \CompProj{}{B}}$ into 
\begin{multline}
    \overline{\CompProj{}{A}\otimes \CompProj{}{B}} =  \Proj{}{A}\otimes\Proj{}{B} +
    \Big(\CompProj{}{A}\otimes \Proj{}{B} \\
     - \CompProj{}{A} \otimes \mathcal{D}_B - \mathcal{D}_A \otimes \Proj{}{B} + \mathcal{D}_A\otimes\mathcal{D}_B\Big) \\
    + \Big( \Proj{}{A} \otimes \CompProj{}{B} \\- \Proj{}{A}\otimes \mathcal{D}_B - \mathcal{D}_A \otimes \CompProj{}{B} + \mathcal{D}_A\otimes\mathcal{D}_B\Big) \:,
\end{multline}
using the definition of negation and algebraic properties. It can then be understood from Fig. \ref{fig:diag_compo}: the first term of the above is the right quarter of the wheel, the second is the top quarter with its boundary removed, and the third is the bottom also without boundaries. As the three parts share no intersection, the regular addition `+' is equivalent to a conjunction `$\cup$'.
Next, more algebraic manipulations lead to 
\begin{multline}
    \overline{\CompProj{}{A}\otimes \CompProj{}{B}} =  \Proj{}{A}\otimes\Proj{}{B} \:+ \Big(\CompProj{}{A}\otimes \Proj{}{B} \\+ \Proj{}{A} \otimes \CompProj{}{B} - \mathcal{I}_A\otimes\mathcal{D}_B - \mathcal{D}_A \otimes \mathcal{I}_B \Big)\:,
\end{multline}
and from this expression it is direct to check that the term in parentheses vanishes if and only if either of the conditions in Eq. \eqref{eq:Iso_par_tensor} holds.

\subsubsection{The transformation operation \label{sec:projo_transfo_prop}}
The projector appearing in Proposition \ref{theo:det_map} is built from the projectors characterizing its input and output. The corresponding operation is defined as the \textbf{transformation} in Eq. \eqref{eq:A_to_B}, represented by $\rightarrow$ :
\begin{multline}\label{eq:A_to_B_app}
    \Proj{}{A} \rightarrow \Proj{}{B} \equiv \mathcal{I}_A\otimes \mathcal{I}_B\\
     - \mathcal{P}_A \otimes \mathcal{I}_B + \Proj{}{A} \otimes \Proj{}{B} - \mathcal{P}_A \otimes \mathcal{D}_B + \mathcal{D}_A\otimes \mathcal{D}_B .
\end{multline}
This additional operation in the Boolean algebra of projectors is actually secondary since it can be entirely defined using the negation and the no signaling composition (i.e., the tensor):
\begin{equation}\label{eq:transformation}
    \Proj{}{A} \rightarrow \Proj{}{B} \equiv \overline{\Proj{}{A} \otimes \CompProj{}{B}} \:.
\end{equation}
Therefore it will automatically be a valid projector in the sense that it will obey Eqs. \eqref{eq:projo} if its constituents do, and will accordingly preserve commutation if they do. 
Furthermore, as this expression involves both composition and negation, it will not be associative in general. For example, $\left(\left(\Proj{}{A} \rightarrow \Proj{}{B}\right)\rightarrow \Proj{}{C}\right) \neq \left(\Proj{}{A} \rightarrow \left( \Proj{}{B} \rightarrow \Proj{}{C}\right)\right)$, because the \textit{uncurrying rule} proven in Ref. \cite{Perinotti2016} can be applied to the right-hand side:
\begin{equation}\label{eq:uncurrying}
    \Proj{}{A} \rightarrow \left( \Proj{}{B} \rightarrow \Proj{}{C}\right) = \left(\Proj{}{A} \otimes \Proj{}{B}\right)\rightarrow \Proj{}{C} \:,
\end{equation}
which is obviously different than $\left(\left(\Proj{}{A} \rightarrow \Proj{}{B}\right)\rightarrow \Proj{}{C}\right)$ by Eq. \eqref{eq:AB<n(nAnB)}. In the algebra, the uncurrying rule relies on the associativity of the tensor: 
\begin{multline}
    \Proj{}{A} \rightarrow \left( \Proj{}{B} \rightarrow \Proj{}{C}\right) = \overline{\Proj{}{A}\otimes \overline{\left(\overline{\Proj{}{B} \otimes \CompProj{}{C}}\right)}}\\
    = \overline{\Proj{}{A}\otimes \overline{\overline{\Proj{}{B} \otimes \CompProj{}{C}}}} \overset{\eqref{eq:invol}}{=} \overline{\Proj{}{A}\otimes \Proj{}{B} \otimes \CompProj{}{C}} \\
    = \overline{\left(\Proj{}{A}\otimes \Proj{}{B}\right) \otimes \CompProj{}{C}}
    = \left(\Proj{}{A} \otimes \Proj{}{B}\right)\rightarrow \Proj{}{C} \:.
\end{multline}
Remark that because of Eq. \eqref{eq:subset_tensor_par-par}, the transformation also preserves the inclusion relations:
\begin{equation}
    \Proj{}{A}\subseteq\Proj{'}{A}\:\Rightarrow\: \left(\Proj{}{A} \rightarrow \Proj{}{B}\right)\subseteq\left(\Proj{'}{A} \rightarrow \Proj{}{B}\right)\:.
\end{equation}

At the level of Boolean logic, the transformation operation can be understood as a logical implication. Indeed, the transformation is equal to its inverse implication, meaning that it satisfies
\begin{equation}\label{eq:A->B=nA<-nB}
    \Proj{}{A} \rightarrow \Proj{}{B} = \overline{\Proj{}{A}} \leftarrow \overline{\Proj{}{B}}\:,
\end{equation}
which comes from the definition. Additionally, note that it is equivalent to $\overline{\Proj{}{B}} \rightarrow \overline{\Proj{}{A}}$ since the order of the systems in the tensor product does not matter as $\Hilb{A}\otimes \Hilb{B} \cong \Hilb{B} \otimes \Hilb{A}$, we only keep it fixed to avoid adding confusion. 

The $\rightarrow$ operation on projectors recovers the $\rightarrow$ type constructor of Bisio and Perinotti type theory \cite{Perinotti2016,Bisio2018} (see Sec. \ref{sec:types}): if we define the trivial system as `1', that is to say, the 1-dimensional state structure \{1\} made of the number 1, one can interpret the measurement (thus the negation of a given state structure) as a transformation into the trivial system. This leads to the identity
\begin{equation}\label{eq:proj->1}
    \overline{\Proj{}{A}} = \Proj{}{A} \rightarrow 1 \:,
\end{equation}
which justifies the notation $\mathscr{A}\rightarrow 1 \equiv \overline{\mathscr{A}}$. The proof is straightforward from the definition since $1$ is 1-dimensional: $\Proj{}{A} \rightarrow 1 = \overline{\Proj{}{A} \otimes \overline{1}} = \overline{\Proj{}{A}}$. In the same way, one can prove that
\begin{equation}\label{eq:1->proj}
    \Proj{}{A} = 1 \rightarrow \Proj{}{A} \:.
\end{equation}
Therefore, a state structure can be seen as a transformation from the trivial system to itself and its negation as a transformation from itself to the trivial system. 
In view of the link between composition and transformation, one may also interpret the former in terms of the latter. A bipartite system in tensor-composed state structures, $\mathscr{A}\otimes \mathscr{B}$ for instance, can actually be seen as characterized by the following transformations
\begin{equation}
    \Proj{}{A} \otimes \Proj{}{B} = \overline{\Proj{}{A}\rightarrow\CompProj{}{B}} = \overline{\CompProj{}{A}\leftarrow \Proj{}{B}} \:.
\end{equation}
This means that a no signaling composite bipartite system is equivalent to a functional on a transformation from one state structure to the functionals on the other of the other as explained in Sec. \ref{sec:types}. In the above equation, we see that the direction of the transformation has no influence, which is expected since it is a no signaling composition; this interpretation is explained in Sec. \ref{sec:NS}.

\subsubsection{Linear Logic\label{sec:projo_prop_LL}}
We now show under which definition that the rules $\{\overline{\:\cdot\:}, \cap, \cup, \otimes, \overline{\:\cdot\:} \rightarrow \:\cdot \:\}$ form a model of logic which is almost multiplicative additive linear logic (MALL). The correspondence with MALL is summarized in Table \ref{tab:LL_comparison}, where the rules we have introduced are put in correspondence with their usual notation.
\begin{table*}[ht]
    \centering
    \begin{tabular}{|c|c|c||c|c|}
    \hline
        Name & \multicolumn{2}{c||}{Symbol} & \multicolumn{2}{c|}{Unit} \\
         & proj. & LL & proj. & LL \\
        \hline
        Negation & $\overline{\:\cdot\:}$ & $\cdot^\perp$ & / & /  \\
        \hline
        Additive conjunction & $\cap $ & $ \&  $ & $\mathcal{I}$ & {\sffamily T} \\
        Additive disjunction & $ \cup $ & $ \oplus  $ & $\mathcal{D}$ & 0 \\
        Multiplicative conjunction & $ \otimes $ & $ \otimes $ & 1 & 1 \\
        Multiplicative disjunction& $\overline{\:\cdot\:} \rightarrow \cdot$ & $ \parr $ & 1 & $\perp$ \\
        \hline
        Linear Implication & $\rightarrow$ & $\multimap$ & 1 & / \\
        \hline
    \end{tabular}
    \caption{Correspondence of the algebraic rules of the projective characterization (proj.) with linear logic (LL)}
    \label{tab:LL_comparison}
\end{table*}

Linear logic is a formal system of logic, of which MALL is a fragment, that is a restriction of the logic to fewer rules. It can be defined as a \textit{sequent calculus}, that is the formalization of proof systems in which each \textit{proposition} follows under some \textit{structural and inference rules} from other propositions (it is a \textit{sequent} of these propositions, whence the name) see Ref. \cite{GIRARD1987}. 
If one sees the projectors algebra and the operations on it as the propositions, the starting set of propositions is indeed similar to those of MALL:
\begin{enumerate}
    \item The set from which we start is made of valid projectors;
    \item For every projector $\Proj{}{}$, there exists a projector $\CompProj{}{}$;
    \item For every projectors $\Proj{}{A}$ and $\Proj{'}{A}$, there exist an additive conjunction $\Proj{}{A}\cap\Proj{'}{A}$ as well as an additive disjunction $\Proj{}{A}\cup\Proj{'}{A}$;
    \item For every projectors $\Proj{}{A}$ and $\Proj{}{B}$, there exist a multiplicative conjunction $\Proj{}{A}\otimes\Proj{}{B}$ as well as a multiplicative disjunction $\CompProj{}{A}\rightarrow\Proj{}{B}$;
    \item There are two constants $\left(\mathcal{I},\mathcal{D}\right)$ that go with each additive binary connectors;
    \item There are two constants $\left(1,1\right)$ that go with each multiplicative binary connectors.
\end{enumerate}
And the properties that can be proven for the projectors are similar to those that can be proven in MALL using the sequent calculus:
\begin{enumerate}
    \item All binary connectors are commutative;
    \item Multiplicative connectors distribute over additive ones;
    \item All propositions have a negation obeying the De Morgan rules \eqref{eq:deMorgan};
    \item Additive constants are the negation of each other;
    \item Multiplicative constants are the negation of each other.
\end{enumerate}
Proposition 1 is always true as we take it as an axiom. 
Propositions 2 to 4 are true from the fact that the definition of these various connectors implies closure, which, in the case of projectors, is the conservation of properties \eqref{eq:projo}. These were proven for each connector in the previous sections.
Proposition 5 follows from equations \eqref{eq:cap_unit} and \eqref{eq:cup_unit}.
Proposition 6 happens because of the isomorphism $\LinOp{\Hilb{}}\otimes \mathbb{C}\cong 1$ so that $\Proj{}{} \otimes 1 = \Proj{}{}$ and the same way, $\overline{\CompProj{}{}\otimes \overline{1}} = \overline{\overline{\mathcal{P}}} = \Proj{}{}$, $\overline{1\otimes \CompProj{}{}} = \overline{\overline{\mathcal{P}}} = \Proj{}{}$.

Property 1 is true from the definition. In the case of the mutiplicative connectors, $\Proj{}{A}\otimes \Proj{}{B}$ and $\overline{\CompProj{}{A}\otimes \CompProj{}{B}}$, the isomorphism $\Hilb{A}\otimes\Hilb{B}\cong \Hilb{B}\otimes\Hilb{A}$ should of course be used.

Property 2 follows from Eqs. \eqref{eq:cap_tenor_dist} and \eqref{eq:cup_tensor_dist} in the case of $\otimes$ and application of the De Morgan rules \eqref{eq:n_deMorgan} on these two equations can be used to prove the property in the case of $\overline{\:\cdot\:}\rightarrow\:\cdot\:$. Put differently, the multiplicative conjunction (respectively, disjunction) distributes over the additive conjunction (disjunction)
\begin{subequations}
    \begin{align}
    	&\begin{multlined}
        \Proj{}{A} \otimes \left(\Proj{}{B} \cap \Proj{'}{B} \right) \\
        = \left(\Proj{}{A} \otimes \Proj{}{B}\right) \cap  \left(\Proj{}{A} \otimes \Proj{'}{B}\right)\:;\\
        \end{multlined}\\
        &\begin{multlined}
        \CompProj{}{A} \rightarrow \left(\Proj{}{B} \cup \Proj{'}{B} \right) \\ 
        =  \left(\CompProj{}{A} \rightarrow \Proj{}{B} \right)\cup \left(\CompProj{}{A} \rightarrow \Proj{'}{B} \right)\:.
        \end{multlined}
    \end{align}
\end{subequations}
But as $\otimes$ and $\rightarrow$ are both operations that merge subspaces, the converse obviously does not hold. Indeed an expression like $\Proj{}{X} \cap \left(\Proj{}{A} \otimes \Proj{}{B} \right) = \left(\Proj{}{X} \otimes \Proj{}{A}\right) \cap \left(\Proj{}{X} \otimes \Proj{}{B}\right)$ makes no sense as the right-hand side feature a cap between two superoperator projectors defined on different spaces that are not isomorphic in general. 

Property 3 was proven at Eq. \eqref{eq:invol} for the negation, Eq. \eqref{eq:n_deMorgan} for the additive connectors, and follows from the definition in the case of multiplicative connectors.

Property 4 is the statement $\overline{\mathcal{I}}=\mathcal{I}-\mathcal{I}+ \mathcal{D}=\mathcal{D}$, whose converse holds by idempotency or can be proven by the same kind of computation.

Property 5 is the statement $\overline{1}=1-1+1$, as in the 1-dimensional case $\mathcal{D}=\mathcal{I}=1$.

There are however some discrepancies with MALL, that we now discuss. First, two small issues that indicate a departure from a faithful model of MALL as in Ref. \cite{GIRARD1987}:
\begin{enumerate}
    \item The multiplicative units are equivalent;
    \item The additive falsity (i.e. the unit for $\cup$) is not absorbing.
\end{enumerate}
The first issue is why we use the terminology `degenerate' in the main text to refer to the model of MALL formed by the algebra of projectors. The second issue --that $\Proj{}{A}\otimes\mathcal{D}_B \neq \mathcal{D}_A\otimes \mathcal{D}_B$-- is more severe as it goes against an equality that can be proven in MALL. A way to circumvent this is to redefine the additive falsity as the number 0, but then the trace normalization of all operator systems the projectors characterize should be set to zero which would jeopardize the probabilistic interpretation.

A more substantial problem is the fact that the cap and cup are only well-defined for expressions featuring the same set of base projectors: two projectors can only be composed with a cap if they can be embedded in the same space. This is part of the problem we evoked when discussing Property 2 above. For the model of logic to work properly, the number of subsystems as well as the dimension of each subsystem must be fixed \textit{a priori}, and each expression is associated with a given (set of) subsystem(s), so that the composition is well-defined. Knowing this information allows one to `pad' the expressions with identities so as to embed them in the global Hilbert space. For example, an expression like $\Proj{}{X} \cap \left(\Proj{}{A} \otimes \Proj{}{B} \right)$ can be made definite by knowing for example that there are two systems $A$ and $B$, and that $\Proj{}{X}$ is associated with system $B$ so that the padding reads $\Proj{}{X}\cong\mathcal{I}_A \otimes \Proj{X}{B}$, where $\Proj{X}{B}$ is the projector $\Proj{}{X}$ acting on system $B$. Another possibility is that $\Proj{}{X}$ is associated with both systems, so that the padding is trivial $\Proj{}{X}\cong\Proj{X}{AB}$. Yet another possibility is that the projector is associated with system $A$ and $B$, but the global Hilbert space is tripartite: $\Hilb{}=\Hilb{A}\otimes\Hilb{B}\otimes\Hilb{C}$. In that case, the padding should be done on the two projectors: $\Proj{}{X}\cong\Proj{X}{AB}\otimes \mathcal{I}_C$, $\Proj{}{A} \otimes \Proj{}{B} \cong \Proj{}{A} \otimes \Proj{}{B} \otimes \mathcal{I}_C$.

Despite these particularities, the connection between the algebra of projectors and MALL is too obvious not to be pointed out.

Remark that the issues with interpreting the algebra as a model of LL come from the choice of additive connectives as $\cap$ and $\cup$. This choice is motivated because we want to compare the underlying operator systems associated with different types of higher-order quantum transformations. Another choice proposed in Ref. \cite{Simmons2022} is to use the Cartesian product and direct sum as the additive conjunction and disjunction respectively. In that case, the issues can all be alleviated except for the degeneracy of the multiplicative units... but the additive connectors can no longer be used for comparison. 

\subsection{Adding the prec \label{sec:projo_prec}}
Define the \textbf{A-to-B one-way signaling composition}, nicknamed \textbf{prec}, as
\begin{equation}\label{eq:semi-causal_comp}
    \Proj{}{A} \prec \Proj{}{B} \equiv \mathcal{I}_A \otimes \Proj{}{B} - \CompProj{}{A} \otimes \mathcal{D}_B + \mathcal{D}_A \otimes \mathcal{D}_B \:.
\end{equation}
The reversed sign $\succ$ (also nicknamed the prec) can be defined accordingly: the \textbf{B-to-A one-way signaling composition} is given by
\begin{equation}
    \Proj{}{A} \succ \Proj{}{B} \equiv \Proj{}{A} \otimes \mathcal{I}_B  - \mathcal{D}_A \otimes \CompProj{}{B} + \mathcal{D}_A \otimes \mathcal{D}_B \:.
\end{equation}
As is the case with the transformation $\rightarrow$, note that $\Proj{}{A} \succ \Proj{}{B} \cong \Proj{}{B} \prec \Proj{}{A}$ since $\Hilb{A}\otimes \Hilb{B} \cong \Hilb{B}\otimes \Hilb{A}$.

\paragraph*{Prec preserves projectors on operator systems:}When introduced in the main text, Eq. \eqref{eq:transfo_cap_lemma}, this projector was obtained as the intersection of the following two projectors:
\begin{equation}
    \Proj{}{A} \prec \Proj{}{B} \equiv \left(\CompProj{}{A}\rightarrow\Proj{}{B}\right)\cap\left(\mathcal{I}_A \otimes \Proj{}{B}\right)\:.
\end{equation}
Because $\mathcal{I}_A,\Proj{}{A},\Proj{}{B}$ are projectors on operator systems, meaning they obey Eqs. \eqref{eq:projo}, and we proved that operations $\overline{\:\cdot\:}, \otimes$ and  $\cap$ preserve this property in the last sections, $\Proj{}{A} \prec \Proj{}{B}$ obeys Eqs. \eqref{eq:projo} as well and it is therefore a projector on an operator system in $\LinOp{\Hilb{A}\otimes \Hilb{B}}$.

\paragraph*{Prec preserves commutation:}For the same reason as above, if $\Proj{}{A}\cap\Proj{'}{A} = \Proj{'}{A}\cap\Proj{}{A}$ and $\Proj{}{B}\cap\Proj{'}{B} = \Proj{'}{B}\cap\Proj{}{B}$, then
\begin{multline}
    \left(\Proj{}{A} \prec \Proj{}{B}\right)\cap \left(\Proj{'}{A} \prec \Proj{'}{B}\right)\\ = \left(\Proj{'}{A} \prec \Proj{'}{B}\right) \cap \left(\Proj{}{A} \prec \Proj{}{B}\right)\:.
\end{multline}
It also commutes with the B-to-A one-way signaling composition:
\begin{multline}\label{eq:prec_comm_succ}
    \left(\Proj{}{A} \prec \Proj{}{B}\right)\cap \left(\Proj{'}{A} \succ \Proj{'}{B}\right)\\ = \left(\Proj{'}{A} \succ \Proj{'}{B}\right) \cap \left(\Proj{}{A} \prec \Proj{}{B}\right)\:,
\end{multline}
again because 
\begin{equation}
    \Proj{}{A} \succ \Proj{}{B} \equiv \Proj{}{A} \otimes \mathcal{I}_B - \mathcal{D}_A \otimes  \CompProj{}{B} + \mathcal{D}_A \otimes \mathcal{D}_B\;,
\end{equation}
is an expression derived from connectives that preserve commutation:
\begin{equation}
    \Proj{}{A} \succ \Proj{}{B} \equiv \left(\CompProj{}{A}\rightarrow\Proj{}{B}\right)\cap\left(\Proj{}{A} \otimes \mathcal{I}_B\right)\:.
\end{equation}
Proving \eqref{eq:prec_comm_succ} thus reduces to prove the commutation of the terms in parenthesis in $\left(\left(\CompProj{}{A}\rightarrow\Proj{}{B}\right)\cap\left(\mathcal{I}_A \otimes \Proj{}{B}\right)\right) \cap \left(\left(\overline{\Proj{'}{A}}\rightarrow\Proj{'}{B}\right)\cap\left(\Proj{'}{A} \otimes \mathcal{I}_B\right)\right)$, which follows from the associativity of the $\cap$ and the commutation of the $\rightarrow$ with the $\otimes$. The same way,
\begin{multline}
    \left(\Proj{}{A} \succ \Proj{}{B}\right)\cap \left(\Proj{'}{A} \succ \Proj{'}{B}\right)\\ = \left(\Proj{'}{A} \succ \Proj{'}{B}\right) \cap \left(\Proj{}{A} \succ \Proj{}{B}\right)\:.
\end{multline}

\subsubsection{The prec is associative and commutes with the negation\label{sec:projo_prec_assoc}}
Compared to transformation, one-way signaling composition is better behaved in the sense that it is associative
\begin{equation}\label{eq:semi-causal_assoc}
    \Proj{}{A} \prec \left(\Proj{}{B} \prec \Proj{}{C} \right) = \left(\Proj{}{A} \prec \Proj{}{B}\right) \prec \Proj{}{C} \:.
\end{equation}
So that $\Proj{}{A} \prec \left(\Proj{}{B} \prec \Proj{}{C} \right) = \Proj{}{A} \prec \Proj{}{B} \prec \Proj{}{C}$ can be written unambiguously. 
Even more, it is better behaved than both the no signaling (the tensor) and the two-way signaling (the transformation) compositions as it is distributive with respect to the negation
\begin{equation}\label{eq:semi-causal_negation}
    \overline{\Proj{}{A} \prec \Proj{}{B}} = \CompProj{}{A} \prec \CompProj{}{B} \:.
\end{equation}

We first prove the distributivity of the negation by direct computation,
\begin{multline}\label{eq:proof_commute_neg}
    \overline{\Proj{}{A} \prec \Proj{}{B}} \\
    = \mathcal{I}_A \otimes \mathcal{I}_B - \Proj{}{A} \prec \Proj{}{B} + \mathcal{D}_A\otimes\mathcal{D}_B \\
    = \mathcal{I}_A \otimes \mathcal{I}_B - \mathcal{I}_A \otimes \Proj{}{B} + \CompProj{}{A} \otimes \mathcal{D}_B \\- \mathcal{D}_A\otimes\mathcal{D}_B + \mathcal{D}_A\otimes\mathcal{D}_B\\
    = \big(\mathcal{I}_A \otimes \mathcal{I}_B - \mathcal{I}_A \otimes \Proj{}{B} + \mathcal{I}_A \otimes  \mathcal{D}_B  \big)\\
    - \Proj{}{A} \otimes \mathcal{D}_B + \mathcal{D}_A\otimes\mathcal{D}_B\\
    =\big(\mathcal{I}_A \otimes \CompProj{}{B} \big) - \overline{\overline{\mathcal{P}}}_A \otimes \mathcal{D}_B + \mathcal{D}_A\otimes\mathcal{D}_B\\
     =\CompProj{}{A} \prec \CompProj{}{B}
    \:.
\end{multline}
This can be used to prove associativity,
\begin{multline}\label{eq:proof_assoc}
    \Proj{}{A} \prec \left(\Proj{}{B} \prec \Proj{}{C} \right) \\ 
    =\mathcal{I}_A \otimes \left(\Proj{}{B} \prec \Proj{}{C} \right) \\
    - \CompProj{}{A} \otimes \mathcal{D}_B  \otimes \mathcal{D}_C + \mathcal{D}_A\otimes\mathcal{D}_B\otimes\mathcal{D}_C \\
    = \mathcal{I}_A \otimes \mathcal{I}_B \otimes \Proj{}{C} - \mathcal{I}_A \otimes \CompProj{}{B} \otimes \mathcal{D}_C + \mathcal{I}_A \otimes \mathcal{D}_B \otimes \mathcal{D}_C \\
    -\CompProj{}{A} \otimes \mathcal{D}_B  \otimes \mathcal{D}_C + \mathcal{D}_A\otimes\mathcal{D}_B\otimes\mathcal{D}_C \\
    = \big( \mathcal{I}_A \otimes \mathcal{I}_B \big) \otimes \Proj{}{C} + \mathcal{D}_A\otimes\mathcal{D}_B\otimes\mathcal{D}_C  \\
    - \big( \mathcal{I}_A \otimes \CompProj{}{B} - \Proj{}{A} \otimes \mathcal{D}_B + \mathcal{D}_A\otimes\mathcal{D}_B \big) \otimes \mathcal{D}_C\\ 
    =\big( \mathcal{I}_A \otimes \mathcal{I}_B \big) \otimes \Proj{}{C}  \\
    - \big( \CompProj{}{A} \prec \CompProj{}{B} \big) \otimes \mathcal{D}_C + \mathcal{D}_A\otimes\mathcal{D}_B\otimes\mathcal{D}_C \\
    = \big( \mathcal{I}_A \otimes \mathcal{I}_B \big) \otimes \Proj{}{C}\\
    -  \overline{\big(\Proj{}{A} \prec \Proj{}{B}\big)}  \otimes \mathcal{D}_C  + \left(\mathcal{D}_A\otimes\mathcal{D}_B\right)\otimes\mathcal{D}_C \\
     =\big(\Proj{}{A} \prec \Proj{}{B}\big) \prec \Proj{}{C}\:,
\end{multline}
where we have used the definition to go to the last line, and commutation of the prec with the negation to go to the penultimate one.

\subsubsection{Inclusion relations of the prec\label{sec:projo_prec_inclusion}}
In Sec. \ref{sec:projo_prop_tensor}, it was shown that the tensor product of two projectors defines a subset of the composition of the same two projectors using the transformation connective seen as a composition, proving Eq. \eqref{eq:tensor_in_par} of the main text. That is, $\Proj{}{A}\otimes\Proj{}{B}\subseteq \CompProj{}{A}\rightarrow\Proj{}{B} = \overline{\CompProj{}{A}\otimes\CompProj{}{B}}$, see the discussion around Eq. \eqref{eq:AB<n(nAnB)}. 
As the prec is itself a form of composition as argued in Sec. \ref{sec:NS_def}, similar inclusion relations can be proven. 
Like the tensor product, one-way signaling composition with the trivial system amounts to doing nothing since $\Proj{}{A} \prec 1 = \mathcal{I}_A \otimes 1 - \CompProj{}{A} \otimes 1 + \mathcal{D}_A \otimes 1 = \overline{\CompProj{}{A}} \otimes 1$ and $1 \prec \Proj{}{A} = \Proj{}{A} \otimes 1 - \mathcal{D}_A \otimes 1 + \mathcal{D}_A \otimes 1$ so that
\begin{subequations}
    \begin{gather}
        \Proj{}{A} \prec 1  = \Proj{}{A}\:,\\
        1 \prec \Proj{}{A} = \Proj{}{A} \:.
    \end{gather}
\end{subequations}
One can link the two-way signaling, one-way signaling, and no signaling compositions, respectively, $\overline{\:\cdot\:}\rightarrow\cdot$, $\cdot \prec \cdot$, and $\cdot\otimes\cdot$, by noticing that 
\begin{equation}\label{eq:causal=0ways}
    \left(\Proj{}{A} \prec \Proj{}{B}\right) \cap \left(\Proj{}{A} \succ \Proj{}{B}\right) = \Proj{}{A} \otimes \Proj{}{B} \:,
\end{equation}
and 
\begin{equation} \label{eq:transfo=2way}
    \left(\Proj{}{A} \prec \Proj{}{B}\right) \cup \left(\Proj{}{A} \succ \Proj{}{B}\right) = \CompProj{}{A} \rightarrow \Proj{}{B} \:.
\end{equation}
Negating $\Proj{}{A}$ yields relations \eqref{eq:relations} in the main text. 
Notice that the input of the transformation has to be negated to be interpreted as a composition because of how it was defined in Ref. \cite{Perinotti2016}: it \textit{transforms} an \textit{input} in $\LinOp{\Hilb{A}}$ to an output in $\LinOp{\Hilb{B}}$, hence it must be a functional on $\LinOp{\Hilb{A}}$, meaning it should locally belong to $\CompAlg{A}$ instead of $\Alg{A}{}$ (see the discussion in Sec. \ref{sec:Projos_NS_lattice}). 
On the contrary, the tensor and the prec \textit{compose} an \textit{output} in $\LinOp{\Hilb{A}}$ with an output in $\LinOp{\Hilb{B}}$. 

The first equation is proven by developing it,
    $\left(\Proj{}{A} \prec \Proj{}{B}\right) \cap \left(\Proj{}{A} \succ \Proj{}{B}\right) = 
    \big( \mathcal{I}_A \otimes \Proj{}{B} - \CompProj{}{A} \otimes \mathcal{D}_B + \mathcal{D}_A \otimes \mathcal{D}_B  \big) \cap
    \big( \Proj{}{A} \otimes \mathcal{I}_B - \mathcal{D}_A \otimes \CompProj{}{B} + \mathcal{D}_A \otimes \mathcal{D}_B  \big)$, 
and noting that the intersection of any two elements but $\big(\mathcal{I}_A \otimes \Proj{}{B}\big) \cap \big(\Proj{}{A} \otimes \mathcal{I}_B \big) = \Proj{}{A} \otimes \Proj{}{B}$ is giving $\mathcal{D}_A \otimes \mathcal{D}_B$. Thus, the expression reduces to $ \Proj{}{A} \otimes \Proj{}{B}$ followed by eight occurrences of $\mathcal{D}_A \otimes \mathcal{D}_B$ alternating between a plus and minus sign, therefore canceling each other.
The second equation has a quick proof using the De Morgan rule and the distributivity of the negation:
\begin{multline}
    \left(\Proj{}{A} \prec \Proj{}{B}\right) \cup \left(\Proj{}{A} \succ \Proj{}{B}\right) \\
    \overset{\eqref{eq:invol}}{=} \overline{\overline{\left(\Proj{}{A} \prec \Proj{}{B}\right) \cup \left(\Proj{}{A} \succ \Proj{}{B}\right)}} \\
    \overset{\eqref{eq:deMorgan_cup_to_cap}}{=} \overline{\overline{\left(\Proj{}{A} \prec \Proj{}{B}\right)} \cap \overline{\left(\Proj{}{A} \succ \Proj{}{B}\right)}}\\
    \overset{\eqref{eq:semi-causal_negation}}{=} \overline{\left(\CompProj{}{A} \prec \CompProj{}{B}\right) \cap \left(\CompProj{}{A} \succ \CompProj{}{B}\right)}\\
    \overset{\eqref{eq:causal=0ways}}{=} \overline{\CompProj{}{A} \otimes \CompProj{}{B}}= \CompProj{}{A} \rightarrow \Proj{}{B}\:.
\end{multline}

From Eq. \eqref{eq:causal=0ways}, it can also be inferred that
\begin{subequations}
    \begin{gather}
        \Proj{}{A} \otimes \Proj{}{B} \subseteq \Proj{}{A} \prec \Proj{}{B}\:,\\
        \Proj{}{A} \otimes \Proj{}{B} \subseteq \Proj{}{A} \succ \Proj{}{B}\:,
    \end{gather}
\end{subequations}
as, obviously, $ \left(\Proj{}{A} \otimes \Proj{}{B}\right) \cap \left(\Proj{}{A} \prec \Proj{}{B}\right) = \left(\Proj{}{A} \prec \Proj{}{B}\right) \cap \left(\Proj{}{A} \succ \Proj{}{B}\right) \cap \left(\Proj{}{A} \prec \Proj{}{B}\right) = \left(\Proj{}{A} \prec \Proj{}{B}\right) \cap \left(\Proj{}{A} \succ \Proj{}{B}\right) =  \Proj{}{A} \otimes \Proj{}{B}$, and the second equation is proven analogously. The same way, from Eq. \eqref{eq:transfo=2way}, 
\begin{subequations}
    \begin{gather}
        \Proj{}{A} \prec \Proj{}{B} \subseteq \CompProj{}{A} \rightarrow \Proj{}{B}\:,\\
        \Proj{}{A} \succ \Proj{}{B} \subseteq \CompProj{}{A} \rightarrow \Proj{}{B}\:.
    \end{gather}
\end{subequations}
Putting these relations together,
\begin{subequations}\label{eq:inclusions_compo}
    \begin{gather}
        \Proj{}{A} \otimes \Proj{}{B} \subseteq \Proj{}{A} \prec \Proj{}{B} \subseteq \CompProj{}{A} \rightarrow \Proj{}{B}\:,\\
        \Proj{}{A} \otimes \Proj{}{B} \subseteq \Proj{}{A} \succ \Proj{}{B} \subseteq \CompProj{}{A} \rightarrow \Proj{}{B}\:.
    \end{gather}
\end{subequations}
\begin{figure}[htb]
    \centering
    \includegraphics[width=\linewidth]{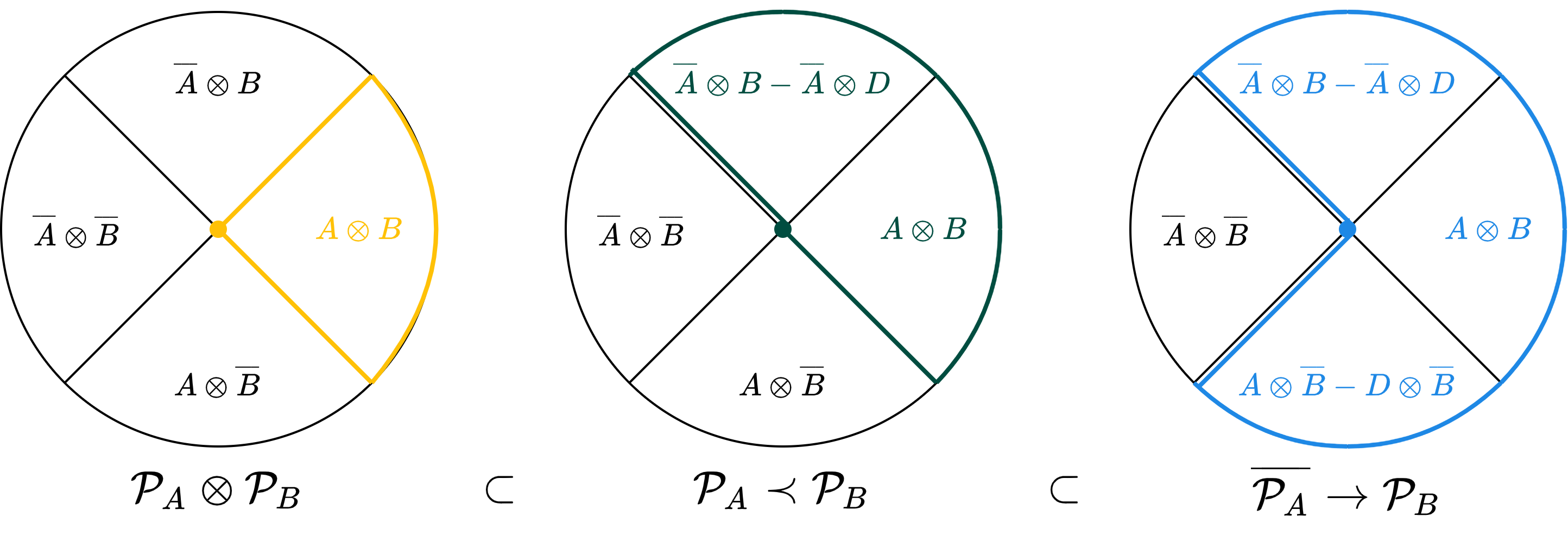}
    \caption{Diagram of the subspaces defined by the three ways of defining a composite projector.}
    \label{fig:diag_compo}
\end{figure}
The first line is diagrammatically depicted in Fig. \ref{fig:diag_compo}: the subspace associated with $\CompProj{}{A} \rightarrow \Proj{}{B}$ is represented in blue on the right. The intermediate case, associated with $\Proj{}{A} \prec \Proj{}{B}$ is depicted in green on the centre; from Eq. \eqref{eq:transfo_cap_lemma}, the green part, $\Proj{}{A} \prec \Proj{}{B} \equiv \left(\CompProj{}{A}\rightarrow\Proj{}{B}\right)\cap\left(\mathcal{I}_A \otimes \Proj{}{B}\right)$, is obtained as the intersection of the blue part, $\CompProj{}{A}\rightarrow\Proj{}{B}$, with $\mathcal{I}_A \otimes \Proj{}{B} = \left(\Proj{}{A}\otimes\Proj{}{B}\right)\cup\left(\CompProj{}{A} \otimes \Proj{}{B}\right)$ which correspond to the right and top quadrants. It indeed recovers the other definition \eqref{eq:semi-causal_comp}, $\Proj{}{A} \prec \Proj{}{B} \equiv \mathcal{I}_A \otimes \Proj{}{B} - \CompProj{}{A} \otimes \mathcal{D}_B + \mathcal{D}_A \otimes \mathcal{D}_B$: the green part made of the right and top quadrants ($\mathcal{I}_A \otimes \Proj{}{B}$) minus the line separating the left and top quadrants ($\CompProj{}{A} \otimes \mathcal{D}_B$) except at the central dot ($\mathcal{D}_A \otimes \mathcal{D}_B$).
The smallest case, associated with $\Proj{}{A} \prec \Proj{}{B}$ is depicted in yellow on the left; as discussed in Corollary \ref{theo:tensor} it is effectively defined by the intersection of no signaling subspaces of Lemma \ref{theo:causal_map}, $\left(\mathcal{I}_A \otimes \Proj{}{B}\right) \cap \left(\Proj{}{A} \otimes \mathcal{I}_B \right)$, that is the intersection of the right and top quadrants ($\mathcal{I}_A \otimes \Proj{}{B}$) with the right and bottom ones ($\Proj{}{A} \otimes \mathcal{I}_B$); the intersection with $\CompProj{}{A}\rightarrow\Proj{}{B}$ brings no new constraints.

Note that similar diagrams appear in Figs. \ref{fig:compos} and \ref{fig:accidental_isomorphism}, but they once again correspond to the situation where $\Proj{}{A}$ is negated. Negating $\Proj{}{A}$ in Eqs. \eqref{eq:inclusions_compo} indeed yields the relations \eqref{eq:inclusions} of the main text depicted in these figures.  

\subsubsection{Distribution properties of the prec \label{sec:projo_prec_dist}}
The cap and prec satisfy an interchange law,
\begin{multline}\label{eq:prec_cap_interchange}
    \big(\Proj{}{A}\cap \Proj{'}{A}\big) \prec \big(\Proj{}{B} \cap \Proj{'}{B}\big) \\
    = \big(\Proj{}{A}\prec \Proj{}{B}\big) \cap \big(\Proj{'}{A}\prec \Proj{'}{B}\big)\:.
\end{multline}
This can be shown as follows:
\begin{multline}
    \left(\Proj{}{A}\cap \Proj{'}{A}\right) \prec \left(\Proj{}{B} \cap \Proj{'}{B}\right) \\
    = \mathcal{I}_A \otimes \left(\Proj{}{B} \cap \Proj{'}{B}\right) - \overline{\left(\Proj{}{A}\cap \Proj{'}{A}\right)} \otimes \mathcal{D}_B + \mathcal{D}_A \otimes \mathcal{D}_B \\
    \overset{\eqref{eq:deMorgan_cap_to_cup}}{=} \mathcal{I}_A \otimes \left(\Proj{}{B} \cap \Proj{'}{B}\right) - \left(\CompProj{}{A}\cup \overline{\Proj{'}{A}}\right) \otimes \mathcal{D}_B \\+ \mathcal{D}_A \otimes \mathcal{D}_B \\
    = \mathcal{I}_A \otimes \big(\Proj{}{B} \cap \Proj{'}{B}\big) - \CompProj{}{A} \otimes \mathcal{D}_B  - \overline{\Proj{'}{A}} \otimes \mathcal{D}_B \\+ \big(\CompProj{}{A}\cap \overline{\Proj{'}{A}}\big) \otimes \mathcal{D}_B + \mathcal{D}_A \otimes \mathcal{D}_B \\
    = \big(\mathcal{I}_A \otimes \Proj{}{B} - \CompProj{}{A} \otimes \mathcal{D}_B + \mathcal{D}_A \otimes \mathcal{D}_B \big) \\
    \cap \big(\mathcal{I}_A \otimes \Proj{'}{B} - \overline{\Proj{'}{A}} \otimes \mathcal{D}_B + \mathcal{D}_A \otimes \mathcal{D}_B  \big)\\
    =\big(\Proj{}{A} \prec \Proj{}{B}\big) \cap \big(\Proj{'}{A}\prec \Proj{'}{B}\big)\:.
\end{multline}
Remark that the grouping at the penultimate line is arbitrary. In different terms, because the cap is commutative we have $\big(\Proj{}{A}\cap \Proj{'}{A}\big) \prec \big(\Proj{}{B} \cap \Proj{'}{B}\big) = \big(\Proj{}{A}\cap \Proj{'}{A}\big) \prec \big(\Proj{'}{B} \cap \Proj{}{B}\big) = \big(\Proj{}{A}\prec \Proj{'}{B}\big) \cap \big(\Proj{'}{A}\prec \Proj{}{B}\big)$. Hence, the exact grouping does not matter as long as the $A$'s and $B$'s are on the correct side of the prec connector.

The cup and the prec satisfy an interchange law as well,
\begin{multline}\label{eq:prec_cup_interchange}
    \big(\Proj{}{A}\cup \Proj{'}{A}\big) \prec \big(\Proj{}{B} \cup \Proj{'}{B}\big)\\
    = \big(\Proj{}{A}\prec \Proj{}{B}\big) \cup \big(\Proj{'}{A}\prec \Proj{'}{B}\big)\:.
\end{multline}
This can be shown the same way,
\begin{multline}
    \big(\Proj{}{A}\prec \Proj{}{B}\big) \cup \big(\Proj{'}{A}\prec \Proj{'}{B}\big)
    \\
    =\big(\mathcal{I}_A \otimes \Proj{}{B} - \CompProj{}{A} \otimes \mathcal{D}_B + \mathcal{D}_A \otimes \mathcal{D}_B \big) \\
    \cup \big(\mathcal{I}_A \otimes \Proj{'}{B} - \overline{\Proj{'}{A}} \otimes \mathcal{D}_B + \mathcal{D}_A \otimes \mathcal{D}_B  \big)\\
    =\big(\mathcal{I}_A \otimes \Proj{}{B} - \CompProj{}{A} \otimes \mathcal{D}_B + \mathcal{D}_A \otimes \mathcal{D}_B \big) \\
    + \big(\mathcal{I}_A \otimes \Proj{'}{B} - \overline{\Proj{'}{A}} \otimes \mathcal{D}_B + \mathcal{D}_A \otimes \mathcal{D}_B  \big) \\
    - \mathcal{I}_A \otimes \big(\Proj{}{B} \cap \Proj{'}{B}\big) + \CompProj{}{A} \otimes \mathcal{D}_B  + \overline{\Proj{'}{A}} \otimes \mathcal{D}_B \\- \big(\CompProj{}{A}\cap \overline{\Proj{'}{A}}\big) \otimes \mathcal{D}_B - \mathcal{D}_A \otimes \mathcal{D}_B\\
    = \mathcal{I}_A \otimes \Proj{}{B} + \mathcal{I}_A \otimes \Proj{'}{B} -\mathcal{I}_A \otimes \big(\Proj{}{B} \cap \Proj{'}{B}\big)\\
    - \big(\CompProj{}{A}\cap \overline{\Proj{'}{A}}\big) \otimes \mathcal{D}_B + \mathcal{D}_A \otimes \mathcal{D}_B\\
    = \mathcal{I}_A \otimes \big(\Proj{}{B} \cup \Proj{'}{B}\big) - \overline{\big(\Proj{}{A}\cup \Proj{'}{A}\big)} \otimes \mathcal{D}_B + \mathcal{D}_A \otimes \mathcal{D}_B \\
    = \big(\Proj{}{A}\cup \Proj{'}{A}\big) \prec \big(\Proj{}{B} \cup \Proj{'}{B}\big) \:.
\end{multline}

These two interchange laws imply distributivity as a special case because the elements in the algebra are all idempotent, so that $\Proj{}{} = \Proj{}{} \cap \Proj{}{} = \Proj{}{} \cup \Proj{}{}$ and we can do the following:
$\big(\Proj{}{A}\cup \Proj{'}{A}\big) \prec \Proj{}{B} = \big(\Proj{}{A}\cup \Proj{'}{A}\big) \prec \big(\Proj{}{B} \cup \Proj{}{B}\big) = \big(\Proj{}{A}\prec \Proj{}{B}\big) \cup \big(\Proj{'}{A}\prec \Proj{}{B}\big)$. 
Therefore, the following relations hold:
\begin{subequations}
    \begin{gather}
        (\Proj{}{A}\cap \Proj{'}{A}) \prec \Proj{}{B} = (\Proj{}{A} \prec \Proj{}{B}) \cap (\Proj{'}{A} \prec \Proj{}{B}) \:,\\
        (\Proj{}{A}\cup \Proj{'}{A}) \prec \Proj{}{B} = (\Proj{}{A} \prec \Proj{}{B}) \cup (\Proj{'}{A} \prec \Proj{}{B}) \:,\\
        \Proj{}{A} \prec (\Proj{}{B}\cup \Proj{'}{B}) = (\Proj{}{A} \prec \Proj{}{B}) \cup (\Proj{}{A}\prec \Proj{'}{B}) \:,\\
        \Proj{}{A} \prec (\Proj{}{B}\cap \Proj{'}{B}) = (\Proj{}{A} \prec \Proj{}{B}) \cap (\Proj{}{A}\prec \Proj{'}{B}) \:.
    \end{gather}
\end{subequations}

\section{Accidental isomorphisms in the case of quantum theory\label{sec:projo_prec_iso}}
The set inclusions \eqref{eq:inclusions_compo}, appearing in the main text as Eqs. \eqref{eq:inclusions}:
\begin{subequations}
    \begin{gather}
        \CompProj{}{A} \otimes \Proj{}{B} \subseteq \CompProj{}{A} \prec \Proj{}{B} \subseteq \Proj{}{A}\rightarrow \Proj{}{B}\:,\\
        \CompProj{}{A} \otimes \Proj{}{B} \subseteq \CompProj{}{A} \succ \Proj{}{B} \subseteq \Proj{}{A}\rightarrow \Proj{}{B} \:,
    \end{gather}
\end{subequations}
can become equivalences in the special cases where the projectors are either identity or depolarizing. These imply set isomorphisms that are not without consequences: subsets with different signaling constraints get accidentally equivalent. 

Putting an identity on the right side of the $\rightarrow$ gives
\begin{equation}
    \Proj{}{A} \rightarrow \mathcal{I}_B= \overline{\Proj{}{A}\otimes \mathcal{D}_B}\:,
\end{equation}%
which happens to be equivalent to
\begin{equation}
    \CompProj{}{A} \prec \mathcal{I}_B= \mathcal{I}_A \otimes \mathcal{I}_B - \Proj{}{A} \otimes\mathcal{D}_B + \mathcal{D}_A \otimes \mathcal{D}_B\:.
\end{equation}
As it can be shown directly from the definitions. In that case, the two-way and one-way signaling compositions coincide. The same way, putting one on the left side gives
\begin{equation}
\begin{aligned}
   \mathcal{I}_A \rightarrow \Proj{}{B} =& \overline{\mathcal{I}_A \otimes \CompProj{}{B}}\\
   =& \mathcal{I}_A \otimes \mathcal{I}_B - \mathcal{I}_A \otimes\CompProj{}{B} + \mathcal{D}_A \otimes \mathcal{D}_B\\
   =&\mathcal{I}_A \otimes \Proj{}{B} - \mathcal{I}_A \otimes \mathcal{D}_B +  \mathcal{D}_A\otimes \mathcal{D}_B \:,
\end{aligned}
\end{equation}
which is equivalent to
\begin{equation}
   \mathcal{D}_A \prec \Proj{}{B} = \mathcal{I}_A \otimes \Proj{}{B} - \mathcal{I}_A \otimes \mathcal{D}_B +  \mathcal{D}_A\otimes \mathcal{D}_B\:.
\end{equation}
One has then the following identities:
\begin{subequations}
\begin{align}
    \CompProj{}{A} \prec \mathcal{I}_B = \Proj{}{A} \rightarrow \mathcal{I}_B \label{eq:Isom_prec_I_r}\:;\\
    \mathcal{D}_A \prec \Proj{}{B} = \mathcal{I}_A \rightarrow \Proj{}{B} \label{eq:Isom_prec_D_l} \:. 
\end{align}
\end{subequations}
These are the only two cases for which the transformation is equivalent to the prec as the following rewriting shows:
\begin{equation}
    \Proj{}{A} \rightarrow \Proj{}{B} = \CompProj{}{A} \prec \Proj{}{B} + \left( \mathcal{I}_A - \Proj{}{A}\right) \otimes \left( \mathcal{I}_B - \Proj{}{B}\right)\:. 
\end{equation}
These relations can be concisely recast as
\begin{multline}\label{eq:Iso_semi-causal_transfo}
    \CompProj{}{A} \prec \Proj{}{B} = \Proj{}{A} \rightarrow \Proj{}{B} \\
    \iff\: \Proj{}{A}=\mathcal{I}_A \:\text{or}\: \Proj{}{B}=\mathcal{I}_B\:.
\end{multline}

In addition to that, one-way signaling composition is equivalent to the no signaling composition in the following cases:
\begin{subequations}
\begin{align}
    \Proj{}{A} \prec \mathcal{D}_B = \Proj{}{A} \otimes \mathcal{D}_B \label{eq:Isom_prec_D_r}\:;\\
    \mathcal{I}_A \prec \Proj{}{B} = \mathcal{I}_A \otimes \Proj{}{B}\label{eq:Isom_prec_I_l} \:.
\end{align}
\end{subequations}
Again, this directly follows from the definition by re-expressing it as
\begin{equation}
        \CompProj{}{A} \prec \Proj{}{B} = \CompProj{}{A} \otimes \Proj{}{B} + \left( \Proj{}{A} - \mathcal{D}_A \right) \otimes \left( \Proj{}{B} - \mathcal{D}_B \right) \:,
\end{equation}
and this can be concisely recast as
\begin{multline}\label{eq:Iso_semi-causal_causal}
    \CompProj{}{A} \prec \Proj{}{B} = \CompProj{}{A} \otimes \Proj{}{B} \\
    \iff\: \Proj{}{A}=\mathcal{D}_A \:\text{or}\: \Proj{}{B}=\mathcal{D}_B\:.
\end{multline} 
It should be noted that condition \eqref{eq:Iso_semi-causal_causal} is stronger than \eqref{eq:Iso_semi-causal_transfo}. Actually, when both conditions are satisfied at once, one has either of the two identities:
\begin{subequations}
\begin{align}
    &\mathcal{I}_A \rightarrow \mathcal{D}_B = \mathcal{D}_A \prec \mathcal{D}_B =  \mathcal{D}_A \otimes \mathcal{D}_B \:;\\
    &\mathcal{D}_A \rightarrow \mathcal{I}_A = \mathcal{I}_A \prec  \mathcal{I}_B = \mathcal{I}_A \otimes  \mathcal{I}_B \:. 
\end{align}
\end{subequations}
The reason this is the case comes from isomorphism \eqref{eq:Iso_par_tensor}, which reduces the transformation into a no signaling composition. Indeed,
\begin{multline}
    \mathcal{D}_A \prec \mathcal{D}_B \overset{\eqref{eq:Iso_semi-causal_transfo}}{=} \mathcal{I}_A \rightarrow \mathcal{D}_B \overset{\eqref{eq:transformation}}{=} \overline{\mathcal{I}_A \otimes \overline{\mathcal{D}_B}}\\
    =\overline{\mathcal{I}_A \otimes \mathcal{I}_B} \overset{\eqref{eq:Iso_par_tensor}}{=} \mathcal{D}_A \otimes \mathcal{D}_B \:.
\end{multline}
And the same way,
\begin{multline}
    \mathcal{I}_A \prec \mathcal{I}_B \overset{\eqref{eq:Iso_semi-causal_transfo}}{=} \mathcal{D}_A \rightarrow \mathcal{I}_B \overset{\eqref{eq:transformation}}{=} \overline{\mathcal{D}_A \otimes \overline{\mathcal{I}_B}}\\
    =\overline{\mathcal{D}_A \otimes \mathcal{D}_B} \overset{\eqref{eq:Iso_par_tensor}}{=} \mathcal{I}_A \otimes \mathcal{I}_B \:.
\end{multline}

Therefore, the isomorphisms \eqref{eq:Iso_par_tensor}, \eqref{eq:Iso_semi-causal_transfo}, and \eqref{eq:Iso_semi-causal_causal} give the conditions for set equivalences in the composition rules.

\begin{figure}[htb]
    \centering
    \subfloat[General projectors\label{fig:accidental_isomorphism_a_b}]{
        \includegraphics[width=.97\linewidth]{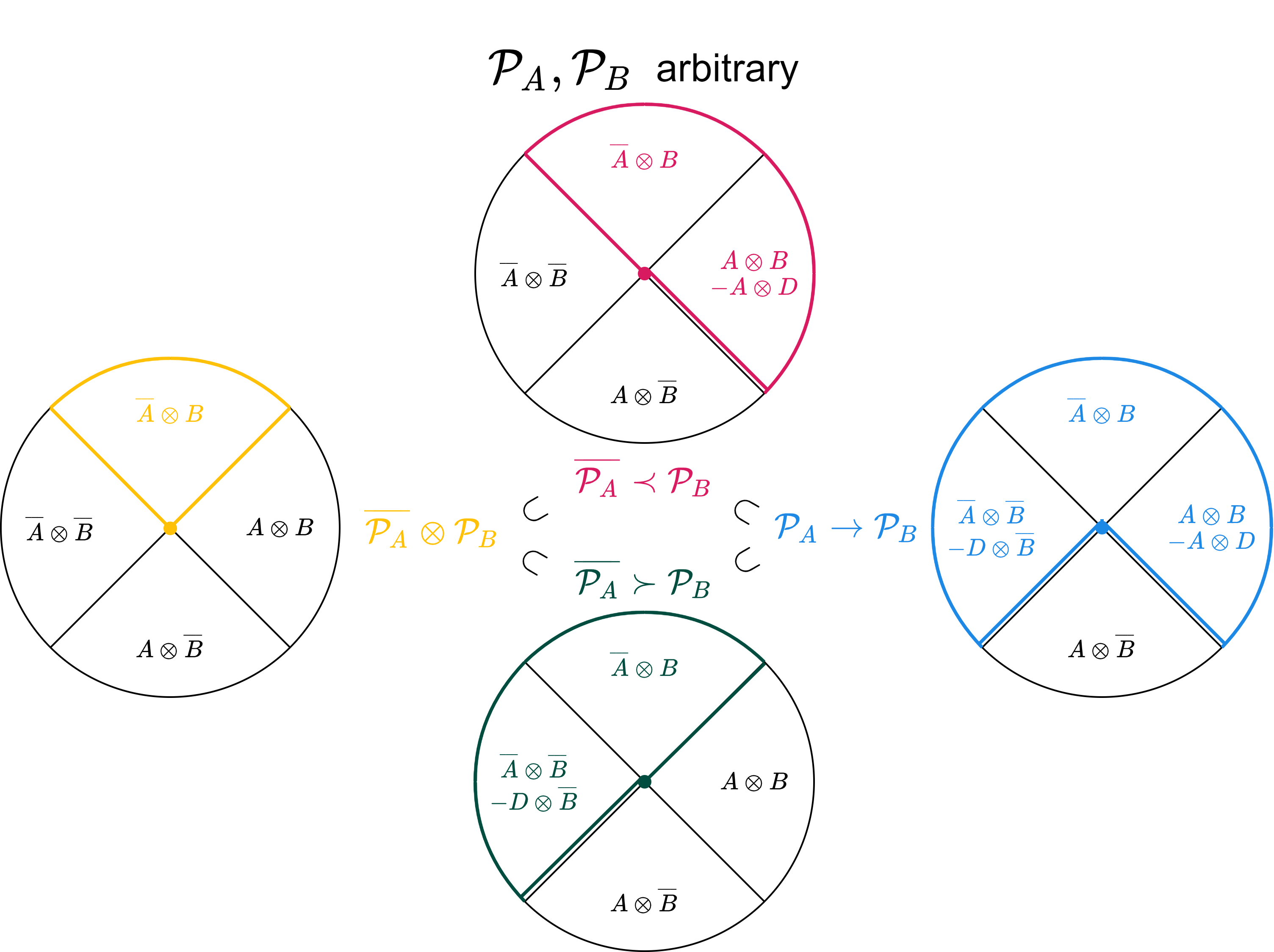}
    }\\
    \subfloat[Quantum theory (identity projectors)\label{fig:accidental_isomorphism_id_id}]{
        \includegraphics[width=.97\linewidth]{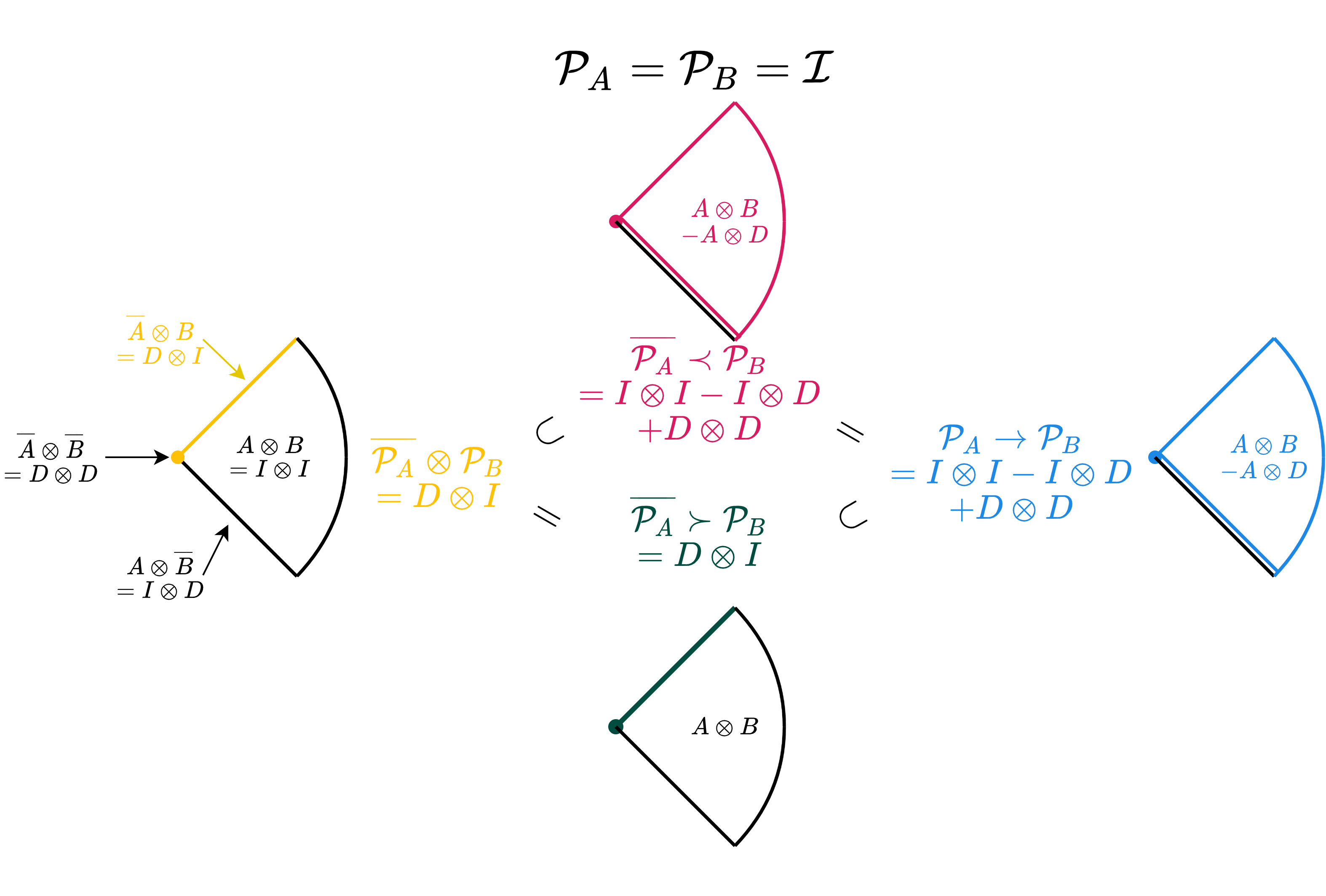}
    }
    \caption{Diagrams depicting the subspaces spanned by the four different ways of defining a transformation (top).
    When $\CompProj{}{A} = \mathcal{D}_A$ and $\Proj{}{B} = \mathcal{I}_B$ (bottom), the yellow and green zones are shrunk to the segment $D\otimes B$: Equations akin to \eqref{eq:Isom_prec_D_r} and \eqref{eq:Isom_prec_I_l} are simultaneously satisfied so no signaling (yellow) is equivalent to one-way to $A$ (green). At the same time, the pink and blue zones are reduced to $A\otimes B - A\otimes D$: Eqs. \eqref{eq:Isom_prec_I_r}and \eqref{eq:Isom_prec_D_l} are simultaneously satisfied so two-way signaling (blue) is equivalent to one-way signaling to $B$ (pink).
    \label{fig:accidental_isomorphism}}
\end{figure}
One can understand these relations using a diagram of the kind depicted in Fig. \ref{fig:accidental_isomorphism}.
For a transformation $\Proj{}{A} \rightarrow \Proj{}{B}$ made of arbitrary projectors (Fig. \ref{fig:accidental_isomorphism_a_b}, in blue), its different substructures are effectively distinct (no signaling in yellow, $A$-to-$B$ one-way in pink, and $B$-to-$A$ one-way in green). 
Recall that one can infer the inclusion relations from their overlap, e.g. $\CompProj{}{A} \otimes \Proj{}{B} \subseteq \CompProj{}{A} \prec \Proj{}{B} \subseteq \Proj{}{A} \rightarrow \Proj{}{B}$ can be inferred from the fact that the yellow part is contained within the pink one which is contained within the blue one. Another example, $\left(\CompProj{}{A} \prec \Proj{}{B}\right)\cap \left(\CompProj{}{A} \succ \Proj{}{B}\right) = \Proj{}{A} \rightarrow \Proj{}{B}$ can be inferred from the fact that zone covered by the pink and green areas is equivalent to the blue one.

An accidental isomorphism is present when either the input or the output is a density operator (featuring an identity projector). In the first case, $\Proj{}{A} = \mathcal{I}_A$ and $\CompProj{}{A} = \mathcal{D}_A$, so the area $\overline{A}$ shrinks into a line: the top and left quadrants of the circle are shrunk into the diagonal that goes from bottom left to top right. The yellow zone is thus shrunk into $D\otimes B$, the upper border of $A\otimes B$, and so is the green zone, this equivalence of yellow and green is similar to Eq. \eqref{eq:Isom_prec_D_r}, but with $A$ and $B$ swapped: $\CompProj{}{A} =\mathcal{D}_A \:\Rightarrow\:\CompProj{}{A} \succ \Proj{}{B} = \mathcal{D}_A \otimes \Proj{}{B}$. 
At the same time, the pink zone is shrunk to $A\otimes B -  A \otimes D$, i.e. $A\otimes B$ without its bottom border, and so is the blue zone, this equivalence of pink and blue is Eq. \eqref{eq:Isom_prec_D_l}. 
In the second case,  $\Proj{}{B} = \mathcal{I}_B$ and $\CompProj{}{B} = \mathcal{D}_B$, so the area $\overline{B}$ also shrinks to a line: the bottom and left quadrants are shrunk into the diagonal that goes from top left to bottom right. The yellow zone is untouched but the green one is shrunk into the same zone as $\overline{A}\otimes \overline{B}$ becomes the border $\overline{A}\otimes D$, i.e. the top-left segment. This equivalence when $\Proj{}{B}=\mathcal{I}_B$ is similar to Eq. \eqref{eq:Isom_prec_I_l}, but with $A$ and $B$ swapped: $\Proj{}{B}=\mathcal{I}_B\:\Rightarrow\: \CompProj{}{A}\succ \Proj{}{B} = \CompProj{}{A}\otimes \mathcal{I}_B$. As for the green zone, it is untouched as well, but the blue zone also sees its left quadrant shrunk into the top-left segment, so the green and blue zones become equivalent. This is Eq. \eqref{eq:Isom_prec_I_r}.

When $\CompProj{}{A} = \mathcal{D}_A$ and $\Proj{}{B} = \mathcal{I}_B$ (Fig. \ref{fig:accidental_isomorphism_id_id}), the yellow and green zones are shrunk to the segment $D\otimes B$. The versions of Eqs. \eqref{eq:Isom_prec_D_r} and \eqref{eq:Isom_prec_I_l} where A and B have swapped roles are simultaneously satisfied, so the diagram depicting the no signaling transformation (yellow) is equivalent to the one depicting the transformation which is one-way signaling to $A$ (green). At the same time, the pink and blue zones are reduced to $A\otimes B - A\otimes D$. Eqs. \eqref{eq:Isom_prec_I_r} and \eqref{eq:Isom_prec_D_l} are simultaneously satisfied so the diagram depicting the two-way signaling transformation (blue) is equivalent to the one depicting the transformation which is one-way signaling to $B$ (pink).

\section{No signaling is local quasi-orthogonality\label{app:NS=QO}}
Here we give more details about how the generalization of a quasi-orthogonality, Eq. \eqref{eq:QO}, into party-wise quasi-orthogonality, Eq. \eqref{eq:partTrAB=TrATrB}, yields the condition \eqref{eq:a}. In order to justify its denomination of ``no signaling'', we also show how this condition relates to the model-independent notion of no signaling. That is, to a constraint on the joint outcome distribution that two parties observe when they act probabilistically on a state structure, i.e., when applying resolution of functionals locally. 

Translating this statement into mathematical conditions, if Alice chooses her measurement settings depending on a classical random variable $x$ and gets classical outcome $a$ and if Bob does the same with settings $y$ and outcome $b$ we want the following constraints on their correlations:
\begin{subequations}\label{eq:no_sign_corr}
    \begin{gather}
        \forall y, y', \quad \sum_b p(a,b|x,y) = \sum_b p(a,b|x,y')\:;\label{eq:no_sign_corr_BtoA}\\
        \forall x, x', \quad \sum_a p(a,b|x,y) = \sum_a p(a,b|x',y)\:.\label{eq:no_sign_corr_AtoB}
    \end{gather}
\end{subequations}
This is the standard statement that no local measurement scheme can be used to deterministically gain knowledge of the other party's actions. The first condition, Eq. \eqref{eq:no_sign_corr_BtoA}, states that Alice's distribution of outcome $a$ cannot be used to determine which setting $y$ Bob has used, thus that the measurement result of Alice cannot be used to guess Bob's own choice of measurement. These correlations are \textit{no signaling} from Bob to Alice. Respectively, the second condition, Eq. \eqref{eq:no_sign_corr_AtoB}, defines no signaling from Alice to Bob correlations.

For formulating this statement in the formalism developed in this article, assume two parties Alice and Bob trying to measure a shared bipartite state structure. Each locally sees its own measurement as part of a functional state structure, respectively $\CompAlg{A}$ and $\CompAlg{B}$. Alice's and Bob's measurements are locally disconnected and thus assumed in tensor product (otherwise, they can be trivially signaling to one another just by using the non-commutative nature of their operations). Thus, their joint operations belong to a state structure $\overline{\mathscr{X}}$ having the inner structure $\overline{\mathscr{X}} \subseteq \overline{\mathscr{A}} \otimes \overline{\mathscr{B}}$. The question reduces to finding what kind of shared state $W$ Alice and Bob can locally measure so that their outcome distributions are no signaling.

Let $W$ be a state shared by Alice and Bob
, $W \in \mathscr{X} \supseteq \overline{\overline{\mathscr{A}} \otimes \overline{\mathscr{B}}}$. A measurement of the part of the state that is under the control of Alice is represented by a collection of operators $\left\{N_{a|x}\right\}$ resolving some operator $N_{|x} \in \overline{\mathscr{A}}$. Each $N_{a|x}$ is thus associated with Alice seeing outcome `$a$' given setting `$x$'. The `$_{|x}$' subscript notation is here to emphasize that $x$ can condition the choice of the resolved operator as well as the choice of how it splits into effects, depending on Alice's choice of strategy. The same way, Bob's measurement is given by $\left\{N_{b|y}\right\}$. The probability distribution reads
\begin{equation}
    p(a,b|x,y) = \TrX{}{\left(N_{a|x} \otimes N_{b|y}\right)\cdot W} \:,
\end{equation}
which translates the no signaling conditions \eqref{eq:no_sign_corr} into
\begin{subequations}
    \begin{align}
        &\forall y, y',\notag\\
        &\TrX{}{\left(N_{a|x} \otimes N_{|y}\right)\cdot W} = \TrX{}{\left(N_{a|x} \otimes N_{|y'}\right)\cdot W}\!; \label{eq:no_sign_comput_step}\\
        &\forall x, x',\notag\\
        & \TrX{}{\left(N_{|x} \otimes N_{b|y}\right)\cdot W} = \TrX{}{\left(N_{|x'} \otimes N_{b|y}\right)\cdot W}\!.
    \end{align}
\end{subequations}
These equations can be reduced into a more concise condition that resembles Eq. \eqref{eq:QO}: take Eq. \eqref{eq:no_sign_comput_step}, since it should hold for any $y,y'$ the choice of a particular setting is no longer needed, and one can simply consider different $N \in \overline{\mathscr{B}}$. Rewriting it as
\begin{multline}
    \forall N, N',\\
    \quad \TrX{A}{N_{a|x} \cdot \TrX{B}{\left(\mathds{1} \otimes N\right)\cdot W}} = \\
    \TrX{A}{N_{a|x} \cdot  \TrX{B}{\left(\mathds{1} \otimes N'\right)\cdot W}} \:,
\end{multline}
one can simplify further by noticing that the possible $N_{a|x}$ actually range over the whole of $\LinOp{\Hilb{A}}$ (up to trace normalization of course), so that only the trace over $B$ part is relevant:
\begin{equation}
    \forall N, N',\: \TrX{B}{\left(\mathds{1} \otimes N\right)\cdot W} =  \TrX{B}{\left(\mathds{1} \otimes N'\right)\cdot W}\!.
\end{equation}
Finally, as $\mathds{1}/c_B$ is a valid element of $\overline{\mathscr{B}}$, we can replace $N'$ by it to obtain the shortened
\begin{equation}
    \forall N, \: \TrX{B}{\left(\mathds{1} \otimes N\right)\cdot W} =  \frac{1}{c_B} \: \TrX{B}{\left(\mathds{1} \otimes \mathds{1} \right)\cdot W}.
\end{equation}
Using Proposition \ref{theo:det_fctal}, observe that $1/c_B = \TrX{}{N}/d_B$ because of Eq. \eqref{eq:det_fctal_norm}, we can thus make $\TrX{}{N}$ appear in the right-hand side and reach the desired formulation. Doing the same reasoning for condition \eqref{eq:no_sign_corr_AtoB}, we obtain the following rephrasing of Eqs. \eqref{eq:no_sign_corr}:
\begin{subequations}\label{eq:no_sign}
    \begin{gather}
        \forall N_B, \: \TrX{B}{\left(\mathds{1} \otimes N_B\right)\cdot W} =  \frac{\TrX{}{N_B}\TrX{B}{ W}}{d_B}\:;\label{eq:no_sign_BtoA}\\
        \forall N_A, \: \TrX{A}{\left(N_A \otimes \mathds{1}\right)\cdot W} =  \frac{\TrX{}{N_A}\TrX{A}{ W}}{d_A}\:,\label{eq:no_sign_AtoB}
    \end{gather}
\end{subequations}
Hence, no signaling \eqref{eq:no_sign_corr} can be recast into the conditions \eqref{eq:no_sign}, which amounts to requiring quasi-orthogonality for only one of the tensor factors of an operator. This leads to Eq. \eqref{eq:causal_notA_to_B} and the ensuing characterization in Lemma \ref{theo:causal_map}, and in particular to Eqs. \eqref{eq:A_otimes_B_def} and Corollary \ref{theo:tensor} in the main text. 

\section{Proofs}
\subsection{Proof of Proposition \ref{theo:det_fctal}\label{app:fctal_proof}}
As $c_A/d_A \mathds{1}$ is a valid element of $\mathscr{A}\subseteq\LinOp{\Hilb{A}}$, any element $N$ of $\CompAlg{A}$ must satisfy $1 = \InProd{N}{c_A/d_A \mathds{1}} = \TrX{}{N \cdot c_A/d_A \mathds{1}} = c_A/d_A \TrX{}{N}$. Since $N$ is arbitrary, this fixes the normalization for all elements of $\overline{\mathscr{A}}$: $c_{\overline{A}} \equiv d_A/c_A = \TrX{}{N}$; this is condition \eqref{eq:det_fctal_norm}.

Positivity follows from the requirement of element-wise positivity with each resolution. For any $N$, let $\{\ket{e_j}\}$ is an orthonormal basis of $\Hilb{A}$ so that $N$ is diagonal in this basis. Then its eigendecomposition is given by $N = \sum_j a_j \dyad{e_j}$, where $a_j\in \mathbb{C}$ are the eigenvalues of $N$. 
Since $c_A/d_A \mathds{1} =  c_A/d_A \sum_j\dyad{e_j}$, a resolution of $c_A/d_A \mathds{1}$ can be defined as $\{E_j = c_A/d_A \dyad{e_j}\}_{j=0}^{d_A-1}$. That each element in the resolution gives a number between 0 and 1 reads $\forall j, \InProd{N}{E_j} \in [0,1]$. But $\InProd{N}{E_j} = \TrX{}{\left(\sum_i a_i \dyad{e_i}\right)^\dag\cdot \left(c_A/d_A\dyad{e_j}\right)} = c_A/d_A\: \sum_i \: a_i^*\:\delta_{i,j}$, where $*$ indicates complex conjugation and $\delta_{i,j}$ is the Kronecker delta. Hence, $\forall j, \InProd{N}{E_j} = c_A/d_A\:a_j^* \in [0,1]$, and since $c_A/d_A$ is a positive real constant, this means that each $a_j$ is real and greater or equal to zero, therefore that $N$ is positive semi-definite, condition \eqref{eq:det_fctal_pos}. 

Finally, the projective condition remains. First, $1/c_A \mathds{1}$ must belong to $\overline{\mathscr{A}}$ as it is positive, properly normalized, and respects $\TrX{}{1/c_A \mathds{1} \cdot V} = 1$ for all $V\in \mathscr{A}$. Assuming that $\overline{\mathscr{A}}$ is a deterministic state structure with projector $\CompProj{}{A}\equiv\mathcal{I}-\Proj{}{A}+\mathcal{D}$, the if part follows: Any positive and properly normalized operator $N$ in $\LinOp{\Hilb{A}}$ on which the projector $\CompProj{}{A}$ is applied obeys $\InProd{\CompProjOn{}{A}{N}}{V} = 1$ for all $V \in \mathscr{A}$ because $\InProd{\CompProjOn{}{A}{N}}{V} \equiv \InProd{\left(\mathcal{I}-\Proj{}{A}\right)\{N\}}{V} + \InProd{\mathcal{D}\{N\}}{V}$. The first member on the right part of the equality vanishes because it belongs to the orthogonal complement of $\mathscr{A}$; the second member is normalized so that $\InProd{\mathcal{D}\{N\}}{V} = \InProd{\TrX{}{N}/d_A\mathds{1}}{V} = \InProd{\frac{d_A}{c_A d_A}\mathds{1}}{V}= 1/c_A\TrX{}{V} = 1$.

To prove the only if part, assume that there is a positive $X$ with trace norm $c_{\overline{A}} = d_A/c_A$ that gives $\InProd{X}{V}=1$ for all $V\in \Alg{A}{}$. Then, $\InProd{X}{\ProjOn{}{A}{V}}=1$ by definition of $V$. Since $\Proj{}{A}$ is self-adjoint, this further implies $\InProd{\ProjOn{}{A}{X}}{V} = 1$. 
Remark that $\Proj{}{A} = \mathcal{D} + \Proj{}{A} - \mathcal{D}$ and that $\Proj{}{A} - \mathcal{D} = \mathcal{I} - \CompProj{}{A}$ so $\Proj{}{A} = \mathcal{D} + (\mathcal{I} - \CompProj{}{A})$. Inserting this in the inner product and using its linearity, it must be that $\InProd{\mathcal{D}\{X \}}{V} + \InProd{(\mathcal{I} - \CompProj{}{A})\{X\}}{V} =1$. The first term of the l.h.s. simplifies into $\InProd{\mathcal{D}\{X \}}{V} = d_A/c_A\InProd{\frac{\mathds{1}}{d_A}}{V}$ because the trace of $X$ was assumed to be $d_A/c_A$ and so this term must be equal to 1 as $V$ has a trace norm of $c_A$ by definition. Therefore, it must be the case that $\InProd{(\mathcal{I} - \CompProj{}{A})\{X\}}{V} = 0$. Since both $X$ and $V$ are generally non-zero (as they are arbitrary), and we want this inner product to be zero for all $V$, there are two possibilities for $X$: it can either belong to the orthogonal complement of $V$, i.e. to the image of $\mathcal{I}- \Proj{}{A}$, or else it can be that $(\mathcal{I} - \CompProj{}{A})\{X\} = 0$, i.e. that $\CompProj{}{A}\{X\} = X$. Because $X$ has a non-zero trace while the orthogonal complement of $\Proj{}{A}$ is spanned by traceless matrices, we conclude that only the latter condition can be true, i.e. $\CompProj{}{A}\{X\} = X$ if $\InProd{X}{V}=1$ for all $V\in \Alg{A}{}$. This proves the necessity of the projective condition.

\subsection{Proof of Proposition \ref{theo:det_map}\label{sec:det_map_proof}}
    Following Ref. \cite{Dynamics}, the second requirement of Definition \ref{def:struc_pres} states that:
    \begin{subequations}\label{eq:struc_pres}
    \begin{gather}
    \MOn{V} \geq 0\:, \label{eq:struc_pres_pos}\\
    \TrX{}{\MOn{V}} = c_B\:, \label{eq:struc_pres_resc}\\
    \Proj{}{B} \circ \mathcal{M} \circ \Proj{}{A} = \mathcal{M} \circ \Proj{}{A}\:. \label{eq:struc_pres_proj}
    \end{gather}
    \end{subequations}
    The third requirement imposes condition \eqref{eq:det_map_pos}. This condition implies in turn the positivity of $\mathcal{M}$, Eq. \eqref{eq:struc_pres_pos}, because $\mathcal{M}$ having a positive Choi matrix is equivalent to $\mathcal{M}$ being completely positive by Choi theorem \cite{Choi1975}.
    
    Using the reverse definition of the CJ isomorphism, Eq. \eqref{eq:CJ^-1}, condition \eqref{eq:struc_pres_resc} becomes (subscripts have been put on $M$ and $V$ for clarity)
    \begin{equation}
        \TrX{}{M_{AB} \cdot (V_A \otimes \mathds{1}_B)} = c_B \;,
    \end{equation}
    from which the normalization condition \eqref{eq:det_map_norm} follows, as $V$ can be $c_A\mathds{1}/d_A$. 
    Finally, as $\MOn{V} = (\TrX{A}{M \cdot (V\otimes \mathds{1}_B)})^T$ and $\MOn{V}\in \mathscr{B}$, it should be true that, $\forall N \in \CompAlg{B}$,
    \begin{equation}
        \begin{aligned}
        1 &= \TrX{B}{N\cdot \MOn{V}}\\
        &= \TrX{B}{N\cdot (\TrX{A}{M \cdot (V\otimes \mathds{1}_B)})^T} \\
        &= \left(\TrX{B}{\TrX{A}{M \cdot (V\otimes \mathds{1}_B)} \cdot N^T }\right)^T \\
        &= \TrX{}{M \cdot (V \otimes N^T)} \:.
        \end{aligned}
    \end{equation}
    Since $M$ is positive and normalized, and since $V\in \mathscr{A}$ and $N\in\overline{\mathscr{B}}$ are arbitrary, it follows from this last equation that the set of all $M$ obeying the above is a functional state structure on a spanning set for $\mathscr{A} \otimes \overline{\mathscr{B}}^T  \cong \mathscr{A} \otimes \overline{\mathscr{B}}$ (recall that a state structure is closed under transposition).
    From linearity, it means that $M$ is also normalized on affine combinations of this set. Indeed, let $\{V_i\}$ and $\{N_i\}$ two arbitrary sets of states in $\Alg{A}{}$ and $\CompAlg{B}$, respectively, and let $ 1 = \sum_i q_i$ where $\forall i,\,q_i\in \mathbb{R}$. Then,
    \begin{equation}
        \begin{aligned}
            1 &= \sum_i q_i \\
            &= \sum_i q_i \TrX{}{M \cdot (V_i \otimes N_i^T)}\\
            &= \TrX{}{M \cdot \big(\sum_i q_i(V_i \otimes N_i^T)\big)}\:.
        \end{aligned}
    \end{equation}
    
    Hence, $M$ is also normalized on the affine hull of the pure tensor products of $\Alg{A}{}$ and $\CompAlg{B}$, which is $\Alg{A}{} \otimes \CompAlg{B}$ by Definition \ref{prop:tensor}.
    Therefore, using Proposition \ref{theo:det_fctal}, the set of positive and normalized operators on this set is characterized by a projector $\overline{\Proj{}{A}\otimes \CompProj{}{B}}$ yielding the projective condition \eqref{eq:det_map_proj} (note this was first proven in \cite[Proposition 1]{Bisio2018}).


\subsection{Proof of Proposition \ref{theo:Bisio_equivalence}\label{sec:Bisio_equivalence_proof}}
     The characterization of the set of deterministic events is one of the core results of Ref. \cite{Bisio2018}. The part of their derivation that is of interest for this proof is the characterization of what is called the \textit{deterministic events} of an arbitrary type $X$ constructed from elementary types $A, B, ..., K$ and using the $\{1,(,),\rightarrow\}$ connectors. 
    The authors start by 1) defining the deterministic events of elementary type (Definition 6, p.6); then they 2) generalize the definition to deterministic events of arbitrary type $X \rightarrow Y$, where $X$ and $Y$ can be any type (Definition 7, p. 7); and finally they 3) characterize the deterministic events of type $X\rightarrow Y$ (Proposition 1, p. 11). 
    The end goal of this proof is then to show that their Proposition 1 is our Proposition \ref{theo:det_map} in the special case where the base state structures are sets of density matrices.

    Our proof strategy will consist of showing that each of these steps has an equivalent in the projective characterization under the assumption that the base state structures are the quantum states. 
    Precisely, we will show 1) what is called `deterministic events of elementary types $A, \ldots, K$', and noted $\mathrm{T}_1(A), \ldots, \mathrm{T}_1(K)$ in Ref. \cite{Bisio2018}, correspond to the base state structures $\mathscr{A}, \ldots, \mathscr{K}$, all associated with the identity projector, $\Proj{}{A} = \mathcal{I}_A, \ldots , \Proj{}{K} = \mathcal{I}_K$, and with unit normalization, $c_A = 1, \ldots, c_K = 1$; 
    2) that what is called `the sets of deterministic events of type $X$ obtained from elementary types $A,\ldots, K$', and noted $\mathrm{T}_1(X)$ in Ref. \cite{Bisio2018} (remark that for non-elementary types, the authors of this work prefer to use small case letters like $x$, but this convention is not followed in this article) corresponds to a state structure $\mathscr{X}$ obtained from base structures $\mathscr{A}, \ldots, \mathscr{K}$; 
    3) that their characterization of the deterministic events of type $X \rightarrow Y$ corresponds to our characterization of the state structure $\mathscr{X} \rightarrow \mathscr{Y}$.

    The first item is trivial to show. By Definition 6 in Ref. \cite{Bisio2018}, an operator $M$ is a deterministic event of elementary type $A$ if $M \in \LinOp{\Hilb{A}}:$ $M\geq 0$ and $\TrX{}{M}=1$. This is exactly the state structure of density operators as characterized by Eqs. \eqref{eq:state_char}.

    The second item is shown by first using Theorem 3 of Ref. \cite{Bisio2018}, which characterizes the set of deterministic events of arbitrary type by the following constraints:
    \begin{subequations}
        \begin{gather}
            M \in \mathrm{T}_1(X \rightarrow Y) \iff \notag \\
            M \geq 0 \:;\\
            \forall V \in \mathrm{T}_1(X): \: \mathcal{M}(V) \in \mathrm{T}_1(Y) \:;
        \end{gather}
    \end{subequations}
    where $\mathcal{M}$ is the linear map obtained when applying the reverse direction of the CJ correspondence on $M$. 
    In such a form, it is direct to see that it corresponds to our definition of a structure-preserving map, Def. \ref{def:struc_pres}.

    For the third item, Proposition 1 of Ref. \cite{Bisio2018} provides the following characterization result:
    \begin{subequations}\label{eq:Bisio_char}
        \begin{gather}
            M \in \mathrm{T}_1(X \rightarrow Y) \iff \notag \\
            M \geq 0 \:, \label{eq:Bisio_char_positivity}\\
            M = \lambda_{X \rightarrow Y} \mathds{1} + H_{X\rightarrow Y} \label{eq:Bisio_char_rest} \:,
        \end{gather}
    \end{subequations}
    where $H_{X\rightarrow Y}$ belongs to a subspace $\Delta_{X\rightarrow Y}$. This subspace and the constant $\lambda_{X \rightarrow Y}$ are defined recursively by:
    \begin{subequations}
        \begin{gather}
            \text{For $X=1$ and $Y$ elementary, } \notag \\
            \lambda_Y = \frac{1}{d_Y} \:, \label{eq:lambda_elem}\\
            \Delta_Y = \{M \in \LinOp{\Hilb{Y}}| \TrX{}{M} = 0\} \:. \label{eq:Delta_elem}\\
            \text{Else, } \notag \\
            \lambda_{X\rightarrow Y} = \frac{\lambda_Y}{d_X\lambda_X}\:, \label{eq:lambda_gen}\\
            \Delta_{X\rightarrow Y} = \left[ \LinOp{\Hilb{X}} \otimes \Delta_Y \right] \oplus \left[ \overline{\Delta}_X \otimes \Delta_Y^\perp \right] \:. \label{eq:Delta_gen}
        \end{gather}
    \end{subequations}
    In the last formula, the $\overline{\:\cdot\:}$ indicates a quasi-orthogonal complement whilst the ${}^\perp$ indicates an orthogonal complement.

    What we want to show is that the conditions \eqref{eq:Bisio_char} are exactly the conditions \eqref{eq:det_map}, i.e.,
    \begin{subequations}\label{eq:det_map_app}
    \begin{gather}
        M\in \mathscr{X}\rightarrow \mathscr{Y} \: \iff\notag \\
        M \geq 0 \:\label{eq:det_map_app_pos},\\
        \TrX{}{M} = \frac{c_Y}{c_X}d_X\label{eq:det_map_app_norm} \:,\\
        \left(\Proj{}{X}\rightarrow \Proj{}{Y}\right)\{M\} = M \:,\label{eq:det_map_app_projo}
    \end{gather}
    \end{subequations}
    despite our vastly different approaches. 
    Firstly, notice that conditions \eqref{eq:Bisio_char_positivity} and \eqref{eq:det_map_app_pos} are the same, so the proof will consist in showing that condition \eqref{eq:Bisio_char_rest} is equivalent to conditions \eqref{eq:det_map_app_norm} and \eqref{eq:det_map_app_projo}. 

    Since any $\Delta_{X}$ space is traceless by construction (see Sec. 5. b. of Ref. \cite{Bisio2018}), we can extract the trace condition by taking the trace of Eq.\eqref{eq:Bisio_char_rest}, yielding $\TrX{}{M} = 1$ for Eq. \eqref{eq:lambda_elem}, and $\TrX{}{M} = \frac{\lambda_Y}{\lambda_X}d_Y$ for Eq. \eqref{eq:lambda_gen}. By direct inspection, one can identify that $c_X = \lambda_X d_X$, so that $\TrX{}{M} = \frac{c_Y}{c_X}d_X$, recovering Eq. \eqref{eq:det_map_app_norm}, and that $c_Y = 1$ when $\mathscr{Y}$ is a base state structure, i.e. when $\Proj{}{Y} = \mathcal{I}_Y$.
    
    We can characterize the $\Delta$ subspaces using superoperator projectors: Eq. \eqref{eq:Delta_elem} is equivalent to the condition $\left(\mathcal{I}-\mathcal{D}\right)\{M\} = M$ since the trace of an operator is contained in the subspace invariant under the depolarizing superoperator $\mathcal{D}$. Injecting this in Eq. \eqref{eq:Bisio_char_rest}, we recover the (trivial) condition $\mathcal{I}\{M\} = M$ characterizing the set of density matrices in the case where $X=1$ and $Y$ is elementary. 
    
    Now, for the general case, since $\Delta_{X}$ (with $X$ being any type) must be a traceless subspace, we can posit it to correspond to some projector $\Proj{}{X} - \mathcal{D}_X$ such that $V \in \Delta_{X} \iff (\Proj{}{X} - \mathcal{D}_X)\{V\} = V$. Similarly, to $\Delta_Y$ will correspond some projector $\Proj{}{Y} - \mathcal{D}_Y$ and to $\Delta_{X\rightarrow Y}$ some $\Proj{}{X\rightarrow Y} - \mathcal{D}_X \otimes \mathcal{D}_Y$. Injecting the three ans\"{a}tze into Eq. \eqref{eq:Delta_gen} yields
    \begin{multline}
        \Proj{}{X\rightarrow Y} - \mathcal{D}_X \otimes \mathcal{D}_Y =\left[\mathcal{I}_X \otimes (\Proj{}{Y}- \mathcal{D}_Y)\right]  \\
        + \left[\overline{(\Proj{}{X}- \mathcal{D}_X)} \otimes (\Proj{}{Y}- \mathcal{D}_Y)^\perp\right]\:.
    \end{multline}
    (Note that the correspondence between the direct sum of spaces $\oplus$ and the sum of projectors is only possible because the intersection of each part of the sum is the zero projector. In general, the union of spaces through $\oplus$ must correspond to the union of projectors $\cup$ introduced in Sec. \ref{sec:Projos_alg}.)

    The orthogonal complement $\mathcal{P}^\perp$ of a projector $\Proj{}{}$ is given by $\mathcal{I} - \Proj{}{}$ whereas his quasi-orthogonal complement $\CompProj{}{}$ is given by $\mathcal{I} - \Proj{}{} + \mathcal{D}$ (see Eq. \eqref{eq:det_fctal_proj}). Injecting these into the right-hand side, then subsequently simplifying, gives
    \begin{multline}
        \left[\mathcal{I}_X \otimes (\Proj{}{Y}- \mathcal{D}_Y)\right] \\
        + \left[\overline{(\Proj{}{X}- \mathcal{D}_X)} \otimes (\Proj{}{Y}- \mathcal{D}_Y)^\perp\right] \\
        =\left[\mathcal{I}_X \otimes (\Proj{}{Y}- \mathcal{D}_Y)\right] + \left[(\mathcal{I}_X - \Proj{}{X}) \otimes \CompProj{}{Y}\right] \\
        =\mathcal{I}_X \otimes \Proj{}{Y} - \mathcal{I}_X \otimes \mathcal{D}_Y + \mathcal{I}_X \otimes \CompProj{}{Y} - \mathcal{D}_X \otimes \CompProj{}{Y} \\
        =\mathcal{I}_X \otimes\mathcal{I}_Y - \Proj{}{X} \otimes \CompProj{}{Y} \:.
    \end{multline}
    (All the identities used and projector manipulations are detailed in Sec. \ref{sec:Projos_alg} and App. \ref{eq:projo}.)
    Hence, Eq. \eqref{eq:Delta_gen} is recast as the following projector identity:
    \begin{equation}
        \Proj{}{X\rightarrow Y} - \mathcal{D}_X \otimes \mathcal{D}_Y = \mathcal{I}_X \otimes\mathcal{I}_Y - \Proj{}{X} \otimes \CompProj{}{Y} \:.
    \end{equation}
    Passing the $\mathcal{D}_X \otimes \mathcal{D}_Y$ term on the other side, we obtain
    \begin{equation}
    \begin{aligned}
        \Proj{}{X\rightarrow Y} &= \mathcal{I}_X \otimes\mathcal{I}_Y - \Proj{}{X} \otimes \CompProj{}{Y} + \mathcal{D}_X \otimes \mathcal{D}_Y\\
        &= \overline{\Proj{}{X} \otimes \CompProj{}{Y}} \\
        &=\Proj{}{X} \rightarrow \Proj{}{Y} \:,
    \end{aligned}
    \end{equation}
    recovering Eq. \eqref{eq:det_map_app_projo}.

    Therefore, we showed that conditions \eqref{eq:Bisio_char} are equivalent to the conditions \eqref{eq:det_map_app}. We conclude that Proposition 1 of Ref. \cite{Bisio2018} is equivalent to Proposition \ref{theo:det_map} of this paper. This concludes the proof.
    
\subsection{Proof of Lemma \ref{lem:NS_lattice}\label{sec:NS_lattice_proof}}

    Let $\Proj{(n)}{}$ be the projector obtained after $n$ `steps' that can be categorized as 1) do a global negation of the projector, $\Proj{(n)}{} = \CompProj{(n-1)}{}$; 2) add a base projector on the right using tensor product, $\Proj{(n)}{} = \Proj{(n-1)}{}\otimes \Proj{}{X}$; 3) add a similarly obtained projector after $k$ steps, $\Proj{(k)}{}$, on the right using tensor product, $\Proj{(n)}{} = \Proj{(n-1)}{} \otimes \Proj{(k)}{}$. This covers all cases as the $\rightarrow$ can be split into a negation and a tensor, $\cdot\rightarrow \cdot\equiv \overline{\cdot\otimes\overline{\;\cdot\;}}$, and one can always redefine the tensor factors labeling so that the added system is on the right since $\Hilb{A}\otimes \Hilb{B} \cong \Hilb{B}\otimes \Hilb{A}$.

    The only non-trivial first step is to choose a base projector, say $\Proj{}{A}$, in which case the claim trivially holds, $\Proj{}{A}\subseteq \Proj{}{A} \subseteq \overline{\overline{\mathcal{P}}}_{A}$. Suppose it holds after $n-1$ steps during which a certain amount of base projectors were added so that the joint projectors act on a Hilbert space whose factors range from $\Hilb{A}$ up to $\Hilb{J}$ (this is without loss of generality as the index of the last factor, $J$, was arbitrarily chosen). Then, the hypothesis for the induction reads 
    \begin{equation}
        \TProj{}{A} \otimes \ldots \otimes \TProj{}{J} \subseteq \Proj{(n-1)}{} \subseteq \overline{\overline{\TProj{}{A}} \otimes \ldots \otimes  \overline{\TProj{}{J}}}\:.
    \end{equation}
    Let $\TProj{}{A} \otimes \ldots \otimes \TProj{}{J} \equiv \Proj{(n-1)}{NS}$ and $\overline{\overline{\TProj{}{A}} \otimes \ldots \otimes  \overline{\TProj{}{J}}} \equiv \Proj{(n-1)}{FS}$. Then, $ \Proj{(n-1)}{NS} \subseteq  \Proj{(n-1)}{} \subseteq \Proj{(n-1)}{FS}$. The following holds because of Eqs. \eqref{eq:inclusion_duality} 
    \begin{equation}
        \CompProj{(n-1)}{NS} \supseteq  \CompProj{(n-1)}{} \supseteq \CompProj{(n-1)}{FS}\:,
    \end{equation}
    This corresponds to doing a step of category 1), $\Proj{(n)}{} = \CompProj{(n-1)}{}$, in which case $\CompProj{(n-1)}{NS} = \overline{\TProj{}{A} \otimes \ldots \otimes \TProj{}{J}}$. A negation can be put over every single projector by redefining the tilde, $\TProj{}{A} \mapsto \overline{\TProj{}{A}}$, which implies that  $\CompProj{(n-1)}{NS}$ is redefined as $\CompProj{(n-1)}{NS} = \Proj{(n)}{FS}=\overline{\overline{\TProj{}{A}} \otimes \ldots \otimes \overline{\TProj{}{J}}}$. The same can be done on $\CompProj{(n-1)}{FS}$, which proves the induction for 1).

    Let $\Proj{'}{}$ be an arbitrary projector on operator system, the following holds by \eqref{eq:subset_tensor} and \eqref{eq:tensor_in_par} proven in App. \ref{sec:projo_prop_tensor}:
    \begin{subequations}
        \begin{gather}
            \Proj{(n-1)}{NS} \otimes \Proj{'}{} \subseteq  \Proj{(n-1)}{} \otimes \Proj{'}{} \subseteq \Proj{(n-1)}{FS} \otimes \Proj{'}{}\:,\\
            \Proj{(n-1)}{FS} \otimes \Proj{'}{} \subseteq \overline{\CompProj{(n-1)}{FS} \otimes \CompProj{'}{}}\:,\\
            \Proj{(n-1)}{NS} \otimes \Proj{'}{} \subseteq  \Proj{(n-1)}{} \otimes \Proj{'}{} \subseteq \overline{\CompProj{(n-1)}{FS} \otimes \CompProj{'}{}}\:.
        \end{gather}
    \end{subequations}
    Since $\Proj{'}{}$ is arbitrary, the first equation corresponds to either doing step 2) or 3). For case 2), $\Proj{'}{}\equiv \Proj{}{L}$ is the (j+1)-th subsystem added, corresponding to some party $L$, so that $\Proj{(n-1)}{} \otimes \Proj{}{L} = \Proj{(n)}{}$. The third equation then reads $\Proj{(n-1)}{NS} \otimes \Proj{}{L} \subseteq  \Proj{(n)}{} \subseteq \overline{\CompProj{(n-1)}{FS} \otimes \CompProj{'}{}}$. Set  $\Proj{}{L} \equiv \TProj{}{L}$ and the induction is proven. 
    The reasoning is analog in case 3). Note that the added system, $\Proj{'}{}\equiv\Proj{(k)}{}$, is also included in some $\Proj{(k)}{NS}\subseteq \Proj{(k)}{}\subseteq\Proj{(k)}{FS}$ by assumption. Using Eq. \eqref{eq:subset_tensor} again, one has $\Proj{(n-1)}{NS} \otimes \Proj{(k)}{NS} \subseteq  \Proj{(n-1)}{} \otimes \Proj{(k)}{} \subseteq \overline{\CompProj{(n-1)}{FS} \otimes \CompProj{(k)}{FS}}\:$. As $\Proj{(n)}{} = \Proj{(n-1)}{} \otimes \Proj{(k)}{}$, the induction is proven by using the fact that the tensor as well as $\overline{\:\cdot\:}\rightarrow\cdot$ are associative operations to extend the expressions on both side of $\Proj{(n-1)}{} \otimes \Proj{(k)}{}$ and then to define the $\TProj{}{}$'s according to the sought expression.

\subsection{Proof of Lemma \ref{theo:causal_map}\label{sec:proof_causal_map}}
    Let $\left\{\GGB{X}{i}\right\}$ be a basis of $\LinOp{\Hilb{X}}$ so that
    \begin{subequations}\label{eq:traceless_basis}
    \begin{gather}
        \left(\GGB{X }{i}\right)^\dag =\GGB{X}{i} \:; \label{eq:traceless_basis_hermit}\\
        \GGB{X}{0} \equiv \mathds{1}/\sqrt{d_A} \:; \label{eq:traceless_basis_unit}\\
        \TrX{}{\GGB{X}{i\neq 0}} = 0 \:;\label{eq:traceless_traceless}\\
        \InProd{\GGB{X}{i}}{\GGB{X}{j}} \equiv \TrX{}{\GGB{X}{i}\cdot\GGB{X}{j}} = \delta_{i,j}\:.\label{eq:traceless_basis:ortho}
    \end{gather}
\end{subequations}
    And choose this basis on $\LinOp{\Hilb{A}}$ such that the first $n< d_A^2$ elements after $\GGB{A}{0}$, i.e. the set $\{\GGB{A}{1}, \GGB{A}{2}, \ldots ,\GGB{A}{n}\}$, form a spanning set of $\mathscr{A}\setminus\{\mathds{1} = \sqrt{d_A}\sigma_0\}$. This implies that if $\Proj{}{A}$ is the projector on $\mathscr{A}$, then $\sum_{i=0}^{d_A^2-1}\ProjOn{}{A}{\GGB{A}{i}} = \sum_{i=0}^{n}\ProjOn{}{A}{\GGB{A}{i}}$.
    We expand $M$ and $V$ using this basis, let $M=\sum_{i,j}^{d_A^2-1,d_B^2-1} m_{ij} \; \GGB{A}{i} \otimes \GGB{B}{j}$; $V= c_A/\sqrt{d_A} \; \GGB{A}{0} + \sum_{k=1}^{n} v_k\; \GGB{A}{k}$. By direct computation:
    \begin{equation}\label{eq:proof_1}
        \TrX{A}{M} = d_A \; \sum_{j=0}^{d_B^2-1} m_{0j} \; \GGB{B}{j} \:;
    \end{equation}%
    \begin{equation}\label{eq:proof_2}
        \TrX{}{V} = c_A\;;
    \end{equation}
    \begin{equation}\label{eq:proof_3}
        \TrX{A}{M \cdot \left(V \otimes N\right)} = \left(\sum_{i=0}^{n}\sum_{j=0}^{d_B^2-1} d_A \; m_{ij} \; v_i \GGB{B}{j}\right)\cdot N\!.
    \end{equation}
    Combining Eqs. \eqref{eq:proof_1} and \eqref{eq:proof_2} in the right-hand side of Eq. \eqref{eq:partTrAB=TrATrB} and \eqref{eq:proof_3} in its left-hand side yields
    \begin{multline}
        \TrX{A}{M \cdot \left(V \otimes N\right)} = \frac{\TrX{A}{M}\TrX{}{V}}{d_A} \cdot N \Rightarrow \\
            \left(\sum_{i=0}^{n}\sum_{j=0}^{d_B^2-1} d_A \; m_{ij} \; v_i \GGB{B}{j}\right)\cdot N = \\
            \frac{\left(d_A \; \sum_{j=0}^{d_B^2-1} m_{0j} \; \GGB{B}{j}\right) \left( c_A\right)}{d_A} \cdot N
            \:,
    \end{multline}
    the equality is true for
    \begin{equation}
        \left(\sum_{i > 0}^{n}\sum_{j=0}^{d_B^2-1} \; m_{ij} \; v_i\; \GGB{B}{j}\right)\cdot N = 0\:.
    \end{equation}
    This reduces to $\sum_{i > 0}^n\sum_{j=0}^{d_B^2-1} \; m_{ij} \; v_i\; \GGB{B}{j} = 0$ because $N$ is arbitrary and can therefore be taken as $\mathds{1}_B$. 
    On the other hand, this equation has to hold for all $V$ in $\Alg{A}{}$, so one can consider the family $V_i\equiv \sigma_i^A$. Each $V_i$ then induces a condition $\sum_j \; m_{ij} \; \GGB{B}{j} = 0$ for $i\in 1,\ldots,n$. Because the $\GGB{B}{j}$'s are orthonormal vectors, and $M$ is arbitrary, this condition is true for a given $V_i$ if and only if $m_{ij}=0 \:\forall j$ whenever $i$ corresponds to a basis element of $\mathscr{A}$, i.e., whenever $0<i\leq n$. The condition for the equality to hold is therefore that $M$ has the form $\sum_{j=0}^{d_B^2-1} \left(m_{0j}\: \GGB{A}{0} \otimes \GGB{B}{j} + \sum_{i>n}^{d_A^2-1} m_{ij} \:\GGB{A}{i} \otimes \GGB{B}{j} \right)$. This is equivalent to requiring Eq. \eqref{eq:nAotimesiB} on $M$, concluding the proof.

\subsection{Proof of Proposition \ref{prop:lattice}\label{sec:lattice_proof}}
General case: That $\Gamma$ is a projector on operator system acting on $\LinOp{\Hilb{A}\otimes\ldots\otimes\Hilb{K}}$ is direct from the properties of the operations in the algebra proven in App. \ref{sec:projo}. Note that these operations also guarantee commutativity, so that the composition (cap) of any two such $\Gamma$ is always well-defined and results in a projector on an operator system. As a consequence, the cup is also well-defined. The least and greatest elements directly follow as any expression must be contained between the projector on the span of the identity $\mathcal{D}_{A\ldots K} $ and the identity projector $\mathcal{I}_{A\ldots K}$, so that $\mathcal{D}_{A\ldots K} \overset{\eqref{eq:composite_D}}{=} \mathcal{D}_{A} \otimes \ldots \otimes  \mathcal{D}_{K} \subseteq \Gamma \subseteq \mathcal{I}_{A} \otimes \ldots \otimes  \mathcal{I}_{K} \overset{\eqref{eq:composite_I}}{=} \mathcal{I}_{A\ldots K}$. Finally, the least and greatest elements are reached since (for example) one can always consider the elementary projector $\Proj{}{A}\otimes \ldots\otimes \Proj{}{K}$ and its party-wise-negated counterpart $\CompProj{}{A}\otimes \ldots\otimes \CompProj{}{K}$. These two projectors are valid $\Gamma$'s by construction, and their intersection and union yield $\mathcal{D}_{A} \otimes \ldots \otimes  \mathcal{D}_{K}$ and $\mathcal{I}_{A} \otimes \ldots \otimes  \mathcal{I}_{K}$, respectively.\\
Restricted case: This is essentially proven the same way as Lemma \ref{lem:NS_lattice}, see App. \ref{sec:NS_lattice_proof}, but with two new categories of `steps': 4) $\Proj{(n)}{}= \Proj{(n-1)}{}\prec \Proj{}{L}$ and 5) $\Proj{(n)}{}= \Proj{(n-1)}{} \prec \Proj{(k)}{}$. Reducing these steps into a chain of inclusion like $\Proj{(n-1)}{NS} \otimes  \Proj{'}{} \subseteq \Proj{(n-1)}{}\prec\Proj{'}{} \subseteq \overline{\CompProj{(n-1)}{NS} \otimes  \CompProj{'}{}}$ is again obtained by first noticing that $\Proj{}{A}\subseteq\Proj{'}{A}\subseteq\Proj{''}{A} \Rightarrow \Proj{}{A}\prec \Proj{}{B}\subseteq\Proj{'}{A} \prec \Proj{}{B} \subseteq \Proj{''}{A}\prec \Proj{}{B}$ and then by using relations \eqref{eq:inclusions} so that $\Proj{}{A}\otimes \Proj{}{B} \subseteq \Proj{}{A}\prec \Proj{}{B}$ and $\Proj{''}{A}\prec \Proj{}{B} \subseteq \overline{\CompProj{''}{A}\otimes \CompProj{}{B} }$.

\subsection{Proof of Theorem \ref{theo:normal_form}\label{sec:normal_form_proof}}
Let $\Gamma = \bigcup_{i=1}^{x}\left(\bigcap_{j=1}^{y_i} \TProj{}{{\sigma_{ij}(A)}} \prec \ldots \prec \TProj{}{{\sigma_{ij}(K)}} \right)$ and $\Xi = \bigcup_{m=1}^{z}\left(\bigcap_{n=1}^{t_m} \TProj{}{{\sigma_{mn}(A)}} \prec \ldots \prec \TProj{}{{\sigma_{mn}(K)}} \right)$ be two normal forms involving $k$ base projectors. We define the shorthand notation: $\Gamma_i \equiv \left(\bigcap_{j=1}^{y_i} \TProj{}{{\sigma_{ij}(A)}} \prec \ldots \prec \TProj{}{{\sigma_{ij}(K)}} \right)$ so that
    \begin{equation}
        \Gamma = \bigcup_{i=1}^{x} \Gamma_i = \Gamma_1 \cup \Gamma_2 \cup ... \cup \Gamma_x\:,
    \end{equation} 
    and $\Gamma_{ij} \equiv \TProj{}{{\sigma_{ij}(A)}} \prec \ldots \prec \TProj{}{\sigma_{ij}(K)}$ so that $\Gamma_i = \left(\bigcap_{j=1}^{y_i} \Gamma_{ij}\right)$,
    \begin{equation}
        \Gamma = \bigcup_{i=1}^{x}\left( \bigcap_{j=1}^{y_i} \Gamma_{ij} \right) \:.
    \end{equation}
    Hence,
    \begin{multline}
        \Gamma = \bigcup_{i=1}^{x}\left( \Gamma_{i1} \cap \Gamma_{i2} \cap \ldots \cap \Gamma_{iy_i} \right)\\
        = \left( \bigcap_{j=1}^{y_1} \Gamma_{1j} \right) \cup \left( \bigcap_{j=1}^{y_2} \Gamma_{2j} \right) \cup \ldots \cup \left( \bigcap_{j=1}^{y_x} \Gamma_{xj} \right)\:;
    \end{multline}
    and $\Xi_m$ and $\Xi_{mn}$ are defined accordingly.

    First, remark that the negation, intersection, and union of normal forms can be put into normal forms: $\Gamma \cup \Xi$ is a normal form obtained simply by merging the unions:
    \begin{multline}
        \Gamma \cup \Xi = \left(\bigcup_{i=1}^{x} \Gamma_i\right)\cup \left(\bigcup_{m=1}^{z} \Xi_m\right)\\
        = \Gamma_1 \cup \Gamma_2 \cup ... \cup \Gamma_x \cup \Xi_1 \cup \Xi_2 \cup ... \cup \Xi_z\:,
    \end{multline}
    which is a normal form once the redundant terms in the series --the $\Gamma_i$'s that happen to be equivalent to some $\Xi_m$'s-- have been removed by commutativity and associativity. To prove that $\Gamma \cap \Xi$ can be put in normal form requires to use of the distribution law \eqref{eq:cap_cup_dist}
    \begin{multline}
        \Gamma \cap \Xi = \bigcup_{i=1}^{x}\left( \bigcap_{j=1}^{y_i} \Gamma_{ij} \right)\cap \bigcup_{m=1}^{z}\left( \bigcap_{n=1}^{t_m} \Xi_{mn} \right)\\
        = \left(\Gamma_1 \cup \Gamma_2 \cup ... \cup \Gamma_x\right)\cap \left(\Xi_1 \cup \Xi_2 \cup ... \cup \Xi_z\right)\\
        = \left(\Gamma_1 \cap \Xi_1\right) \cup \left(\Gamma_2 \cap \Xi_1\right) \cup \ldots \cup \left(\Gamma_z \cap \Xi_1\right) \\
        \cup \left(\Gamma_1 \cap \Xi_2\right) \cup \ldots \cup \left(\Gamma_z \cap \Xi_t\right)\:.
    \end{multline}
    Each $\left(\Gamma_i \cap \Xi_m\right)$ term is an intersection of intersections, therefore they can be merged as an overall intersection by associativity and the redundancy can be removed the same way as for the union case above. Then again, the unions of these terms is a normal form once the redundancies like $\left(\Gamma_i \cap \Xi_m\right) = \left(\Gamma_{i'} \cap \Xi_{m'}\right)$ for some $(i,m)\neq (i',m')$ have been removed.
    To prove that $\overline{\Gamma}$ can be put in normal form requires to use successively the De Morgan laws, 
    \begin{equation}
    \begin{aligned}
        \overline{\Gamma} = &\overline{\bigcup_{i=1}^{x}\left( \bigcap_{j=1}^{y_i} \Gamma_{ij} \right)}\\
        \overset{\eqref{eq:deMorgan_cup_to_cap}}{=} &\bigcap_{i=1}^{x}\overline{\left( \bigcap_{j=1}^{y_i} \Gamma_{ij} \right)}\\
        \overset{\eqref{eq:deMorgan_cap_to_cup}}{=} &\bigcap_{i=1}^{x}\left( \bigcup_{j=1}^{y_i} \overline{\Gamma_{ij}} \right)\:,
    \end{aligned}
    \end{equation}
    then the commutation of the negation with the prec, Eq. \eqref{eq:prec_neg} (proven at Eq. \eqref{eq:proof_commute_neg}), is used on each term,
    \begin{multline}
        \overline{\Gamma_{ij}} = \overline{\TProj{}{{\sigma_{ij}(A)}} \prec \ldots \prec \TProj{}{\sigma_{ij}(K)}}\\
        = \overline{\TProj{}{{\sigma_{ij}(A)}}} \prec \ldots \prec \overline{\TProj{}{\sigma_{ij}(K)}}\\
        =\TProj{'}{{\sigma_{ij}(A)}} \prec \ldots \prec \TProj{'}{\sigma_{ij}(K)}
        =\Gamma_{ij}'
    \end{multline}
    where the base projectors have been redefined so as to incorporate the negation. Thus, the $\Gamma_{ij}'$ are normal forms, and so their intersections of unions can be put into normal form by using the two properties proven just before this one.

    Next, let $\Upsilon$ be a projector on an operator system over $l$ subsystems in $\LinOp{\Hilb{L}\otimes \ldots \otimes \Hilb{R}}$ that is in normal form, $\Upsilon=\bigcup_{m=1}^{z}\left(\bigcap_{n=1}^{t_m} \TProj{}{{\chi_{mn}(L)}} \prec \ldots \prec \TProj{}{{\chi_{mn}(R)}} \right)$ where $\chi_{mn}$ is an element of the permutation group over $l$ symbols. 
    Then the one-way signaling composition $\Gamma \prec \Upsilon$ is a projector on operator system over $\LinOp{\Hilb{A}\otimes \ldots \otimes \Hilb{K}\otimes \Hilb{L}\otimes \ldots \otimes \Hilb{R}}$ that can be put into a normal form as well. This is proven using the interchange laws \eqref{eq:prec_cap_cup}. Define $\Upsilon_m$ and $\Upsilon_{mn}$ like above, let $u$ ranging from 1 to $x\times z$ such that $u=1$ is identified with $(i,m)=(1,1)$, $u=2$ with $(i,m)=(1,2)$, etc., let $v_u$ ranging from 1 to $y_i\times t_m$ such that $v_u = 1$ is identified with $(j,n)=(1,1)$, etc., and let $\zeta_{uv}=(\sigma_{ij},\chi_{mn})$ be an element of the permutation group over $k+l$ elements indexed by $uv$. That way, an element like $\Gamma_{ij}\prec \Upsilon_{mn}$ can be rewritten as $\Theta_{uv}$ in the following manner:
    \begin{multline}
        \Gamma_{ij}\prec \Upsilon_{mn} = \left( \TProj{}{{\sigma_{ij}(A)}} \prec \ldots \prec \TProj{}{\sigma_{ij}(K)} \right)\\
        \prec \left(\TProj{}{{\chi_{mn}(L)}} \prec \ldots \prec \TProj{}{{\chi_{mn}(R)}} \right)\\
        =\TProj{}{{\sigma_{ij}(A)}} \prec \ldots \prec \TProj{}{\sigma_{ij}(K)} \\
        \prec \TProj{}{{\chi_{mn}(L)}} \prec \ldots \prec \TProj{}{{\chi_{mn}(R)}}\\
        =\TProj{}{{\zeta_{uv}(A)}} \prec \ldots \prec \TProj{}{\zeta_{uv}(K)} \\
        \prec \TProj{}{{\zeta_{uv}(L)}} \prec \ldots \prec \TProj{}{{\zeta_{uv}(R)}}\\
        \equiv \Theta_{uv}\:;
    \end{multline}
    This is but a relabelling using associativity of the prec \eqref{eq:prec_assoc} (proven at Eq. \eqref{eq:proof_assoc}). But the same way, an element like $\Gamma_{ij}\succ \Upsilon_{mn}$ can also be identified with some $\Theta_{u'v'}$ by finding the permutation $\zeta_{u'v'}$ that exactly corresponds to $\TProj{}{{\zeta_{u'v'}(A)}} \prec \ldots \prec \TProj{}{\zeta_{u'v'}(K)} \prec \TProj{}{{\zeta_{u'v'}(L)}} \prec \ldots \prec \TProj{}{{\zeta_{u'v'}(R)}} = \TProj{}{{\chi_{mn}(L)}} \prec \ldots \prec \TProj{}{{\chi_{mn}(R)}} \prec \TProj{}{{\sigma_{ij}(A)}} \prec \ldots \prec \TProj{}{\sigma_{ij}(K)}$. The important thing to notice is by definition, $\Theta_{uv}$ is a normal form. The one-way signaling composition of $\Gamma$ with $\Upsilon$ then reads:
    \begin{multline}
        \Gamma \prec \Upsilon = \bigcup_{i=1}^{x}\left( \bigcap_{j=1}^{y_i} \Gamma_{ij} \right)\prec \bigcup_{m=1}^{z}\left( \bigcap_{n=1}^{t_m} \Upsilon_{mn} \right)\\
        \overset{\eqref{eq:prec_cup_interchange}}{=} \bigcup_{i=1}^{x}\left(\left( \bigcap_{j=1}^{y_i} \Gamma_{ij} \right)\prec \left(\bigcup_{m=1}^{z}\left( \bigcap_{n=1}^{t_m} \Upsilon_{mn} \right)\right)\right)\\
        \overset{\eqref{eq:prec_cup_interchange}}{=} \bigcup_{i=1}^{x}\left(\bigcup_{m=1}^{z}\left(\left( \bigcap_{j=1}^{y_i} \Gamma_{ij} \right)\prec \left( \bigcap_{n=1}^{t_m} \Upsilon_{mn} \right)\right)\right)\\
        \overset{\eqref{eq:prec_cap_interchange}}{=} \bigcup_{i=1}^{x}\left(\bigcup_{m=1}^{z}\left( \bigcap_{j=1}^{y_i} \left(\Gamma_{ij} \prec \left( \bigcap_{n=1}^{t_m} \Upsilon_{mn} \right)\right)\right)\right)\\
        \overset{\eqref{eq:prec_cap_interchange}}{=} \bigcup_{i=1}^{x}\left(\bigcup_{m=1}^{z}\left( \bigcap_{j=1}^{y_i} \left( \bigcap_{n=1}^{t_m}\Gamma_{ij} \prec  \Upsilon_{mn} \right)\right)\right)\\
        = \bigcup_{(i=1,m=1)}^{(x,z)}\left(\bigcap_{(j=1,n=1)}^{(y_i,t_m)} \Gamma_{ij} \prec  \Upsilon_{mn} \right)\\
        = \bigcup_{u=1}^{x\times z}\left(\bigcap_{v=1}^{y_i \times t_m} \Theta_{uv} \right)\:.
    \end{multline}
    In the above rewriting, associativity of intersections and unions was used to go to the penultimate line, and then the definition of the $\Theta_{uv}$'s was injected to go the last line. Note that compared to the intersection and union cases, there is no risk of redundancies since the permutations of $\Gamma_{ij}$ and $\Upsilon_{mn}$ run over different sets ($\{A,\ldots K\}$ and $\{L,\ldots R\}$ respectively). As each $\Theta_{ab}$ is a ``prec chain'', it is direct to see that the last line is a normal form of $\Gamma \prec \Upsilon$.
    
    Using relations \eqref{eq:relations} (proven in App. \ref{sec:projo_prec_inclusion}), the no signaling composition of two normal forms, $\Gamma \otimes \Upsilon$, as well as their transformation of a normal form into a normal form, $\Gamma \rightarrow \Upsilon$ can be expressed in terms of intersections, unions, negations or precs: $\Gamma \otimes \Upsilon = (\Gamma \prec \Upsilon) \cap (\Gamma \succ \Upsilon)$ and $\Gamma \rightarrow \Upsilon= (\overline{\Gamma} \prec \Upsilon) \cap (\overline{\Gamma} \succ \Upsilon)$.
    Thus, they can be put in normal form as well because of the above discussion.

    Therefore, we showed that the negation of a normal form, the intersection and union of two normal forms, as well as the composition of two normal forms using $\otimes,\prec$, and $\rightarrow$ can all be rewritten into a normal form. The proof is completed by noticing that a single base projector like $\Proj{}{A}$ is by definition in a normal form.
    
\subsection{Proof of Theorem \ref{theo:combs=networks}\label{sec:combs=networks_proof}} 
    Using Lemma \ref{lem:accidental}, the equivalence between \eqref{eq:prec_chain} and \eqref{eq:proj_n_network_chan} is almost immediate: inject the result on each node, $\Proj{\text{(n-network)}}{\mathscr{A}_{\text{channel}}} \equiv \left(\mathcal{I}_{A_0} \rightarrow \mathcal{I}_{A_1}\right) \prec \ldots \prec \left( \mathcal{I}_{A_{2n-2}} \rightarrow \mathcal{I}_{A_{2n-1}}\right) \overset{\eqref{eq:isomorphisms_transfo}}{=} \left(\overline{\mathcal{I}}_{A_0} \prec \mathcal{I}_{A_1}\right) \prec \ldots \prec \left( \overline{\mathcal{I}}_{A_{2n-2}} \prec \mathcal{I}_{A_{2n-1}}\right)$, then use the associativity of the prec.
    
    This in turn can be used to recursively prove the equivalence between \eqref{eq:prec_chain} and \eqref{eq:proj_2n_comb_states}. Indeed, observe that the 1-comb is characterized by $\mathcal{I}_{A_0} \rightarrow \mathcal{I}_{A_1}$ which is equivalent to
    \begin{equation}
    \Proj{\text{(2-comb)}}{\mathscr{A}_{\text{state}}} \equiv \mathcal{I}_{A_0} \rightarrow \mathcal{I}_{A_1} \overset{\eqref{eq:isomorphisms_transfo}}{=} \overline{\mathcal{I}}_{A_0} \prec \mathcal{I}_{A_1} \:.
    \end{equation}
    Then, each iteration of combs satisfies the latter condition of Eq. \eqref{eq:isomorphisms_transfo} on the right side of the $\rightarrow$. E.g, for $\Proj{\text{(4-comb)}}{\mathscr{A}_{\text{state}}}$:
    \begin{multline}
    \Proj{\text{(4-comb)}}{\mathscr{A}_{\text{state}}} \equiv \left( \left( \mathcal{I}_{A_0} \rightarrow \mathcal{I}_{A_1}\right)\rightarrow \mathcal{I}_{A_2} \right) \rightarrow \mathcal{I}_{A_3} \\
    =\left( \left( \overline{\mathcal{I}}_{A_0} \prec \mathcal{I}_{A_1} \right)\rightarrow \mathcal{I}_{A_2} \right) \rightarrow \mathcal{I}_{A_3} \\
    = \left( \overline {\overline{\mathcal{I}}_{A_0} \prec \mathcal{I}_{A_1} }\prec \mathcal{I}_{A_2} \right) \rightarrow \mathcal{I}_{A_3} \\
    =\left( \mathcal{I}_{A_0} \prec \overline{\mathcal{I}}_{A_1} \prec \mathcal{I}_{A_2} \right) \rightarrow \mathcal{I}_{A_3} \\
    =\overline{\mathcal{I}_{A_0} \prec \overline{\mathcal{I}}_{A_1} \prec \mathcal{I}_{A_2}}\prec \mathcal{I}_{A_3} \\
    =\overline{\mathcal{I}}_{A_0} \prec \mathcal{I}_{A_1} \prec \overline{\mathcal{I}}_{A_2} \prec \mathcal{I}_{A_3}\:.
    \end{multline}
    Where the associativity of one-way signaling composition, Eq. \eqref{eq:semi-causal_assoc}, and distributed negation, Eq. \eqref{eq:semi-causal_negation}, has been used to simplify in between each step.
    The proof for the $n$-comb directly follows by induction on the above computation: suppose it holds for $n$, $\Proj{\text{(n-network)}}{\mathscr{A}_{\text{channel}}} = \Proj{\text{(2n-comb)}}{\mathscr{A}_{\text{state}}}$ then for $n+1$ the projector is
    \begin{multline}
        \Proj{\text{(2(n+1)-comb)}}{\mathscr{A}_{\text{state}}} \\
        = \left(\Proj{\text{(2n-comb)}}{\mathscr{A}_{\text{state}}} \rightarrow \mathcal{I}_{A_{2n}} \right)\rightarrow\mathcal{I}_{A_{2n+1}}\\
        = \left(\CompProj{\text{(2n-comb)}}{\mathscr{A}_{\text{state}}} \prec \mathcal{I}_{A_{2n}} \right)\rightarrow\mathcal{I}_{A_{2n+1}}\\
        = \overline{\CompProj{\text{(2n-comb)}}{\mathscr{A}_{\text{state}}} \prec \mathcal{I}_{A_{2n}} }\prec \mathcal{I}_{A_{2n+1}}\\
        = \Proj{\text{(2n-comb)}}{\mathscr{A}_{\text{state}}} \prec \overline{\mathcal{I}}_{A_{2n}} \prec \mathcal{I}_{A_{2n+1}}\\
        = \Proj{\text{(n-network)}}{\mathscr{A}_{\text{channel}}} \prec \left(\mathcal{I}_{A_{2n}} \rightarrow \mathcal{I}_{A_{2n+1}}\right)\\
        \equiv \Proj{\text{((n+1)-network)}}{\mathscr{A}_{\text{channel}}} \:,
    \end{multline}
    where the hypothesis was injected in between the antepenultimate and penultimate lines as well as identity $\overline{\mathcal{I}}_{A_{2n}} \prec \mathcal{I}_{A_{2n+1}} = \left(\mathcal{I}_{A_{2n}} \rightarrow \mathcal{I}_{A_{2n+1}}\right) $.
    
    Next, the equivalence between \eqref{eq:prec_chain} and \eqref{eq:proj_n_comb_chan} is also proven by induction. It holds by definition for the $n=1$ case, suppose it holds for $n$, $\Proj{\text{(n-comb)}}{\mathscr{A}_{\text{channel}}} =  \overline{\mathcal{I}}_{A_0} \prec \mathcal{I}_{A_1} \prec \ldots \prec \overline{\mathcal{I}}_{A_{2n-2}} \prec \mathcal{I}_{A_{2n-1}}$, and define the relabelling $\Proj{\text{(n-comb)}'}{\mathscr{A}_{\text{channel}}} \equiv  \overline{\mathcal{I}}_{A_1} \prec \mathcal{I}_{A_2} \prec \ldots \prec \overline{\mathcal{I}}_{A_{2n-1}} \prec \mathcal{I}_{A_{2n}}$ where all indices have been incremented by 1. Then,
    \begin{multline}
        \Proj{\text{(2(n+1)-comb)}}{\mathscr{A}_{\text{state}}} \equiv \\
        (\ldots(\mathcal{I}_{A_{n}} \rightarrow \mathcal{I}_{A_{n+1}}) \rightarrow \ldots ) \rightarrow \left(\mathcal{I}_{A_0} \rightarrow \mathcal{I}_{A_{2n+1}}\right)\\
        = \Proj{\text{(n-comb)}'}{\mathscr{A}_{\text{channel}}} \rightarrow \left(\mathcal{I}_{A_0} \rightarrow \mathcal{I}_{A_{2n+1}}\right)\\
        \overset{\eqref{eq:uncurrying}}{=} (\mathcal{I}_{A_0} \otimes \Proj{\text{(n-comb)}'}{\mathscr{A}_{\text{channel}}}) \rightarrow  \mathcal{I}_{A_{2n+1}}\\
        \overset{\eqref{eq:isomorphisms_tensor}}{=}(\mathcal{I}_{A_0} \prec \Proj{\text{(n-comb)}'}{\mathscr{A}_{\text{channel}}}) \rightarrow  \mathcal{I}_{A_{2n+1}}\\
        \overset{\eqref{eq:isomorphisms_transfo}}{=} \overline{\mathcal{I}_{A_0} \prec \Proj{\text{(n-comb)}'}{\mathscr{A}_{\text{channel}}}} \prec  \mathcal{I}_{A_{2n+1}}\\
        =\overline{\mathcal{I}}_{A_0} \prec \CompProj{\text{(n-comb)}'}{\mathscr{A}_{\text{channel}}} \prec  \mathcal{I}_{A_{2n+1}}\\
        =\overline{\mathcal{I}}_{A_0} \prec \overline{\overline{\mathcal{I}}_{A_1} \prec \ldots \prec \mathcal{I}_{A_{2n}}} \prec  \mathcal{I}_{A_{2n+1}}\\
        =\overline{\mathcal{I}}_{A_0} \prec \mathcal{I}_{A_1} \prec \ldots \prec \overline{\mathcal{I}}_{A_{2n}} \prec  \mathcal{I}_{A_{2n+1}} \:.
    \end{multline}
    Here, the \textit{uncurrying rule} \eqref{eq:uncurrying} (first proven for types \cite{Perinotti2016}; it is derived for projectors in App. \ref{sec:projo_transfo_prop}) was used between the second and third lines as a computational shortcut.

    Finally, the equivalence between \eqref{eq:prec_chain} and \eqref{eq:causality_cond} follows by induction as well. In the case $n=1$, it is proven by writing explicitly the content of the projector \eqref{eq:prec_chain},
    \begin{multline}\label{ea:proof_a}
        \left(\overline{\mathcal{I}}_{A_0} \prec \mathcal{I}_{A_1}\right)\{M\} = M\\
        \left(\mathcal{I}_{A_0} \otimes \mathcal{I}_{A_1} - \mathcal{I}_{A_0} \otimes \mathcal{D}_{A_1} + \mathcal{D}_{A_0} \otimes \mathcal{D}_{A_1}\right)\{M\}\\
         = M\\
        M - \MapX{A_1}{M} \\
        + \MapX{A_0}{\MapX{A_1}{M}} = M\\
        \MapX{A_0A_1}{M} = \MapX{A_1}{M}\\
        \TrX{A_1}{M}=\frac{1}{d_{A_0}}\TrX{A_0A_1}{M}\:\mathds{1}_{A_0}.
    \end{multline}  
    Suppose it holds for $n$ nodes, i.e., $\Proj{(2n)}{}\{M^{(n)}\}=M^{(n)} \iff \TrX{A_1}{M^{(1)}}=\frac{1}{d_{A_0}}\TrX{A_0A_1}{M^{(1)}}\:\mathds{1}_{A_0} \:\wedge\:\ldots\:\wedge\: \TrX{A_{2n-1}}{M}$ $= \frac{1}{d_{A_{2n-2}}}\TrX{A_{2n-2}A_{2n-1}}{M} \otimes \mathds{1}_{A_{2n-1}}$. Then, for $n+1$ nodes, let $M^{(n+1)} \equiv M$, and
    \begin{multline}
        M = \Proj{(2n+2)}{}\{M\}\\
         = \big[\Proj{(2n)}{}\prec\left(\overline{\mathcal{I}}_{A_{2n}} \prec \mathcal{I}_{A_{2n+1}}\right)\big]\{M\}\\
        = \big[ \left(\mathcal{I}_{A_0} \otimes \ldots \otimes \mathcal{I}_{A_{2n-1}} \right) \otimes \left(\overline{\mathcal{I}}_{A_{2n}} \prec \mathcal{I}_{A_{2n+1}}\right) \\
        - \CompProj{(2n)}{}\otimes \mathcal{D}_{A_{2n}} \otimes \mathcal{D}_{A_{2n+1}} \\
        +  \mathcal{D}_{A_{0}} \otimes \mathcal{D}_{A_{1}}\otimes \ldots\otimes \mathcal{D}_{A_{2n}} \otimes \mathcal{D}_{A_{2n+1}}\big]\{M\} \\
        =\big[ \left(\mathcal{I}_{A_0} \otimes \ldots \otimes \mathcal{I}_{A_{2n-1}} \right) \otimes \left(\overline{\mathcal{I}}_{A_{2n}} \prec \mathcal{I}_{A_{2n+1}}\right) \\
        - \mathcal{I}_{A_{0}} \otimes \ldots \otimes \mathcal{I}_{A_{2n-1}}\otimes \mathcal{D}_{A_{2n}} \otimes \mathcal{D}_{A_{2n+1}}+\\
        \Proj{(2n)}{}\otimes \mathcal{D}_{A_{2n}} \otimes \mathcal{D}_{A_{2n+1}} \big]\{M\} \:.
    \end{multline}
    Using this last equality, one can regroup terms as
    \begin{multline}
        0= \big[\left(\mathcal{I}_{A_0} \otimes \ldots \otimes \mathcal{I}_{A_{2n-1}} \right) \otimes \\
        \left(\left(\overline{\mathcal{I}}_{A_{2n}} \prec \mathcal{I}_{A_{2n+1}}\right) -  \mathcal{I}_{A_{2n}}\otimes \mathcal{I}_{A_{2n+1}}\right)\big]\{M\} \\
        - \Big[\left(\left(\mathcal{I}_{A_{0}} \otimes \ldots \otimes \mathcal{I}_{A_{2n-1}}\right) - \Proj{(2n)}{}\right)\otimes\\ \mathcal{D}_{A_{2n}} \otimes \mathcal{D}_{A_{2n+1}} \Big]\{M\}\:.
    \end{multline}
    This defines two projectors with zero intersection, therefore each piece in square brackets must be zero independently of the other. The first piece, $0 = \big[\left(\mathcal{I}_{A_0} \otimes \ldots \otimes \mathcal{I}_{A_{2n-1}} \right) \otimes \left(\left(\overline{\mathcal{I}}_{A_{2n}} \prec \mathcal{I}_{A_{2n+1}}\right) -  \mathcal{I}_{A_{2n}}\otimes \mathcal{I}_{A_{2n+1}}\right)\big]\{M\}$ is exactly equation \eqref{ea:proof_a} applied on systems $A_{2n+1}A_{2n}$. Whereas the second piece, $0=\Big[\left(\left(\mathcal{I}_{A_{0}} \otimes \ldots \otimes \mathcal{I}_{A_{2n-1}}\right) - \Proj{(2n)}{}\right)\otimes \mathcal{D}_{A_{2n}} \otimes \mathcal{D}_{A_{2n+1}} \Big]\{M\}$ can be recast into $0=\big[(\left(\mathcal{I}_{A_{0}} \otimes \ldots \otimes \mathcal{I}_{A_{2n-1}}\right) - \Proj{(2n)}{}\big]\{\TrX{A_{2n}A_{2n+1}}{M}\}$. Using $M^{(n)} \equiv \TrX{A_{2n}A_{2n+1}}{M}$, this last equation must contain by hypothesis the $n$ other causality conditions. 
    Therefore, the $n+1$ causality conditions have been recovered from the projector, completing the proof.

\section{Detailed examples for Section \ref{sec:Proj_char} -- (Some) state structures built with quantum theory as base state structure\label{sec:examples_QT}}
Here we present some examples of state structures that can be built using the three operations $\{\overline{\:\cdot\:},\otimes,\rightarrow\}$, respectively called \textit{negation, tensor, and transformation} and presented in Sec. \ref{sec:Proj_char}. The base state structures $\mathscr{A},\mathscr{B},\ldots$ associated with each party in the following examples are assumed to be the set of quantum states, to which correspond projectors $\mathcal{I}_A, \mathcal{I}_B, \ldots$. 

With these base state structures, we consider first single-partite POVM measurements, which are the resolutions of the unit effect (or functional) $\overline{\mathscr{A}}$ characterized by $\overline{\mathcal{I}}_A$ as an example of negation. 
Then the bipartite quantum states $\mathscr{A} \otimes \mathscr{B}$, characterized by $\mathcal{I}_A \otimes \mathcal{I}_B$, are considered as an example of tensor, followed by bipartite functional $\overline{\mathscr{A} \otimes \mathscr{B}}$, as an example of tensor and negation.
Next, we consider quantum channels $\mathscr{A} \rightarrow \mathscr{B}$, characterized by $\mathcal{I}_A \rightarrow \mathcal{I}_B$ as an example of transformation, followed by single-partite process matrix $\overline{\mathscr{A} \rightarrow \mathscr{B}}$ as an example of transformation and negation. 

We then move on to ``higher-order theory'', interpreting the transformations as higher-order states. Our aim is to present the difference between the set of bipartite quantum channels, $\left(\Alg{A}{0}\otimes \Alg{B}{0} \right) \rightarrow \left(\Alg{A}{1}\otimes \Alg{B}{1} \right)$, and its subset of no signaling channels, $\left(\Alg{A}{0}\rightarrow \Alg{A}{1} \right) \otimes \left(\Alg{B}{0}\rightarrow \Alg{B}{1} \right)$. The former indeed has the type of a first-order transformation: it transforms the bipartite states $\Alg{A}{0}\otimes \Alg{B}{0}$ into bipartite states $\Alg{A}{1}\otimes \Alg{B}{1}$. The latter is also of the first order, but built differently: it is a tensor product of transformations rather than a transformation between tensor products. 
This difference is used to define a higher-order: the state structure $\left(\Alg{A}{0}\rightarrow \Alg{A}{1} \right) \otimes \left(\Alg{B}{0}\rightarrow \Alg{B}{1} \right)$ is a first-order transformation with respect to the state structures of quantum states (characterized by $\mathcal{I}_X$), but it can be interpreted with respect to the state structure of quantum channels (characterized by $\mathcal{I}_{X_0} \rightarrow \mathcal{I}_{X_1}$), effectively moving its interpretation to a higher-order transformation theory. 
In such a case, the state structure $\left(\Alg{A}{0}\rightarrow \Alg{A}{1} \right) \otimes \left(\Alg{B}{0}\rightarrow \Alg{B}{1} \right)$ is the parallel composition of two systems A and B of the base type (i.e. represented by quantum channels); the no signaling channels are seen as bipartite higher-order states. The last example then concerns the functionals normalized on these, which are the bipartite process matrices.

\subsection{Single party quantum theory\label{sec:examples_single_QT}}
This is the simplest example of a theory that can be built assuming Proposition \ref{theo:det_fctal}. 
In this case, we have a state structure $\mathscr{A}$ whose operator system spans the whole of $\LinOp{\Hilb{A}}$, which will constitute the states. To it, we associate its negation $\overline{\mathscr{A}}$, which constitutes the unit effects (or deterministic functionals) interpreted as `the need of a state of $\mathscr{A}$ to obtain a probability of 1'. This results in the state being the regular notion of the quantum state in density matrix form, as defined in Eq. \eqref{eq:state_char}, 
\begin{subequations}\label{eq:state_qu}
    \begin{gather}
        V\geq 0 \:, \\
        \TrX{}{V} = 1 \:, \\
        \mathcal{I}\{V\} = V \:.
    \end{gather}
\end{subequations}
To it corresponds the regular notion of effect as $\overline{\mathscr{A}} = \{\mathds{1}\}$ because Proposition \ref{theo:det_fctal} implies that the valid $N \in \overline{\mathscr{A}}$ have to obey
\begin{subequations}\label{eq:effect_qu}
    \begin{gather}
        N \geq 0 \:, \\
        \TrX{}{N} = d_A \:, \\
        \mathcal{D}\{N\} = N \:,
    \end{gather}
\end{subequations}
where we used $\overline{\mathcal{I}} = \mathcal{I}-\mathcal{I}+\mathcal{D}=\mathcal{D}$. These requirements single out $N = \mathds{1}$. That is to say, the set of destructive measurements of any quantum state whose outcome can be predicted with certainty (whence the terminology \textit{deterministic} functional) is the singleton $\{\mathds{1}\}$, i.e. the discard. The destructive measurement whose outcome can be predicted only up to a probability, the effects (or probabilistic functionals), are (represented by) the operators resolving $\mathds{1}$; we have derived the POVM formalism. By construction, the probability rule in that case indeed recovers the Born rule \eqref{eq:Born}:
\begin{equation}
    p(i|V,N) = \InProd{N_i}{V} \:,
\end{equation}
where each $N_i$ belongs to a family $\{N_i\}$ resolving $\mathds{1}$. With respect to the introductory example of Sec. \ref{sec:Proj_char}, only the notation has been changed: $(\rho,\mathds{1},E_i)\mapsto(V,N,N_i)$.

\subsection{Bipartite quantum theory\label{sec:examples_bipartite_QT}}
Using Definition \ref{prop:tensor}, the bipartite quantum theory is an example of a theory featuring a tensor of projectors. From the previous example of single-party quantum theory, we can define a bipartite theory on space $\LinOp{\Hilb{X}}$ by requiring a bipartition $\LinOp{\Hilb{X}} \cong \LinOp{\Hilb{A}\otimes \Hilb{B}}$ so that there exist states upon which bipartite measurements are well-defined (see the discussion around Corollary \ref{theo:tensor} as well as App. \ref{app:NS=QO}). The one-party example has shown that local measurements of, say, Alice resolve the state structure $\overline{\mathscr{A}}$ as in Eqs. \eqref{eq:state_qu}. Assume Bob has a similarly defined $\overline{\mathscr{B}}$ so that the parallel composition of the measurements of Alice and Bob are resolving $\overline{\mathscr{A}} \otimes \overline{\mathscr{B}}$. By Definition \ref{prop:tensor}, this set is characterized by
\begin{subequations}\label{eq:effect_qu_otimes_qu}
    \begin{gather}
        M \geq 0 \:, \\
        \TrX{AB}{M} = 1 \:, \\
        \left(\mathcal{D}_A \otimes \mathcal{D}_B\right)\{M\} = M \:.
    \end{gather}
\end{subequations}
It is not hard to see from Eq. \eqref{eq:composite_D} that also in the bipartite case, the effects of quantum theory resolve a single element, $\{M=\mathds{1}= \mathds{1}_A \otimes \mathds{1}_B\}$. Note, however, that the effects can now in general be entangled in the sense that there are resolutions $\{M_i\}$ for which there is no possibility to find a decomposition $\{M_i = \sum_{i} q_i E_{i}^A \otimes F_{i}^B\}$ where, for all $q_i$ and $i$, $q_i\geq 0\:, \sum_{i} q_i = 1$ and with $E_{i}$ (respectively, $F_{i}$) being a valid effect resolving an element of $\overline{\mathscr{A}}$ (resp., $\overline{\mathscr{B}}$). A Bell measurement is an instance of such a collection of entangled effects resolving $\mathds{1}$ in the $d_A=d_B=2$ case.

We now use Proposition \ref{theo:det_fctal} to characterize the valid states. An operator $W \in \LinOp{\Hilb{A}\otimes \Hilb{B}}$ is a valid state if it belongs to the set $\overline{\overline{\mathscr{A}} \otimes \overline{\mathscr{B}}}$, i.e.
\begin{subequations}\label{eq:state_qu_otimes_qu}
    \begin{gather}
        W \geq 0 \:, \\
        \TrX{AB}{W} = 1 \:, \\
        \overline{\left(\mathcal{D}_A \otimes \mathcal{D}_B\right)}\{W\} = W\:.
    \end{gather}
\end{subequations} 
Remark that property \eqref{eq:Iso_par_tensor} actually applies to the projective constraint, so that it can be simplified into 
\begin{equation}
    \left(\mathcal{I}_A \otimes \mathcal{I}_B\right)\{W\} = W \:.
\end{equation}
Meaning that $\overline{\overline{\mathscr{A}} \otimes \overline{\mathscr{B}}} \cong \mathscr{A}\otimes\mathscr{B}$ actually holds for quantum theory\footnote{In categorical language, this is the fact that the category of quantum theory is compact closed, whereas the more general category of higher-order quantum transformations is $*$-autonomous, see Ref. \cite{Kissinger_2019}.}, and surprisingly it does so only in this case (this is a consequence of Lemma \ref{lem:accidental}, proven in App. \ref{sec:projo_prec_iso}). 
In other words, in quantum theory the set of valid states normalized on local measurements is exactly the set of no signaling composite states: we recover the intuition that it is impossible for parties measuring a part of a shared quantum state to signal to the other one. What is more surprising is that it is actually the only theory having this property: in general a state normalized on a pair of local measurements can be used for signaling; this is the physical content of the inclusion $\overline{\overline{\mathscr{A}} \otimes \overline{\mathscr{B}}} \supseteq \mathscr{A}\otimes\mathscr{B}$ which follows from Eq. \eqref{eq:tensor_in_par}. This ability to signal is made apparent in the bipartite process matrix and in the biased quantum theory examples below (respectively presented in App. \ref{sec:examples_bipartite_PM} and \ref{sec:examples_bipartite_biaisedQT}).

\subsection{Single party quantum channel and process matrix formalisms\label{sec:examples_single_channel}}
The first example of a state structure built using the transformation connector `$\rightarrow$' described in Proposition \ref{theo:det_map} is the one of quantum channels. With the associated quantum instrument formalism and the single-partite process matrices, they form a higher-order state and effect pair. This is the same pair as the one considered in the introductory example of Sec. \ref{sec:Proj_char}, rephrased using the language developed in that section (it will be rephrased again in App. \ref{sec:examples_dynamics_constr_1} using all of the language developed in this paper). The construction is schematically represented in Fig. \ref{fig:1stTo2ndOrder}. 
\begin{figure*}[htb]
    \centering
    \includegraphics[width=.8\linewidth]{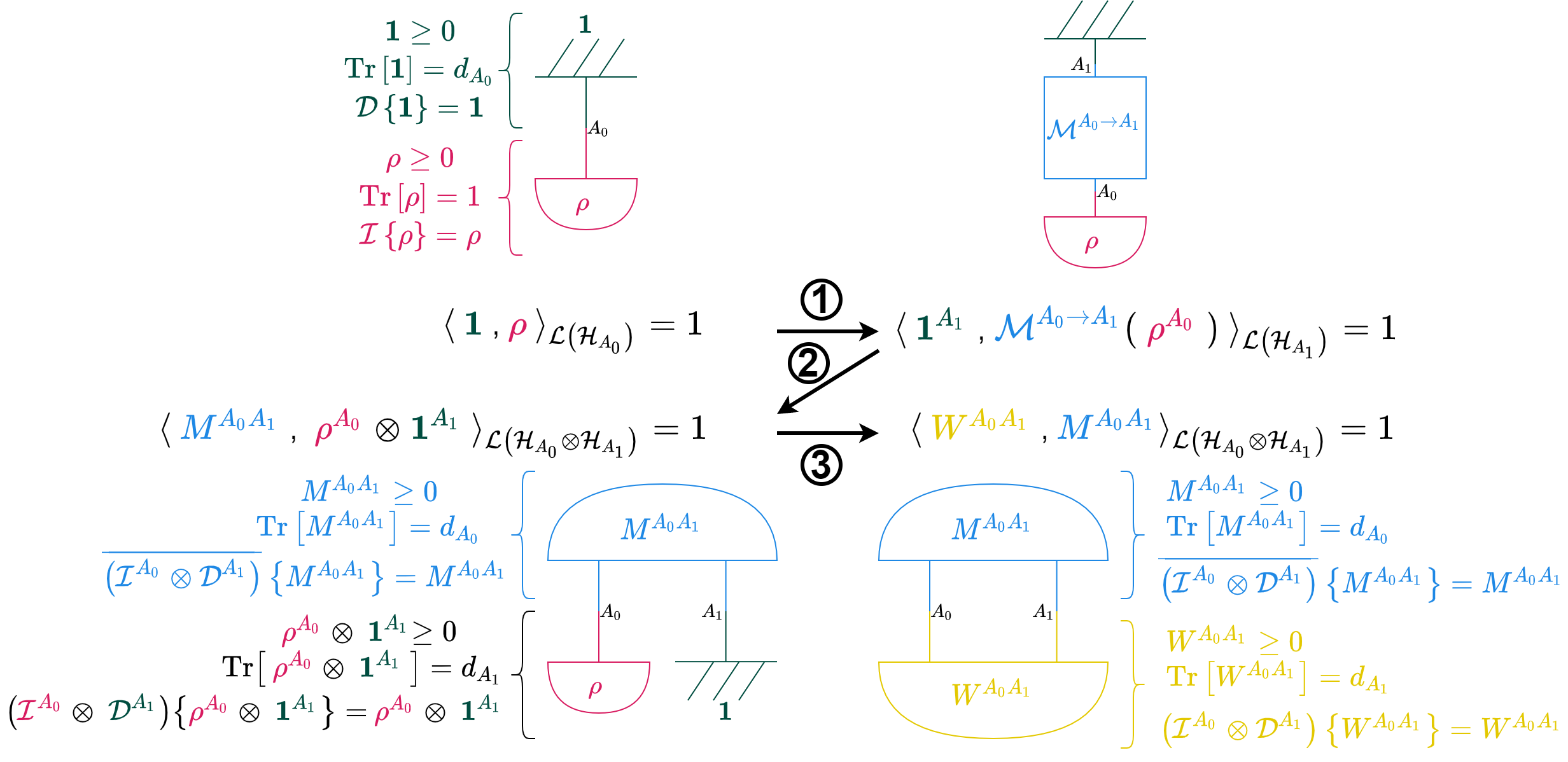}
    \caption{Dynamics-like construction from quantum states to single-partite PM: 1) adding CPTP dynamics to states (in pink) yields channels (in blue), 2) going to the CJ picture and characterizing it using Prop. \ref{theo:det_map}, 3) taking the `negation' of the channels to `complete' the states and effects pair (pink and green) into a single-partite PM (in yellow)}
    \label{fig:1stTo2ndOrder}
\end{figure*}

Let $\mathscr{A}_0$ and $\mathscr{A}_1$ be state structures characterized by similar constraints as Eqs. \eqref{eq:state_qu}. In accordance with Definition \ref{def:struc_pres}, we introduce dynamics from $\mathscr{A}_0$ to $\mathscr{A}_1$ as the set of all CPTP maps $\{\mathcal{M}\}$ from $\mathscr{A}_0$ to $\mathscr{A}_1$; this is represented as step 1 in Fig. \ref{fig:1stTo2ndOrder}. 
Call $\mathscr{A}_0 \rightarrow \mathscr{A}_1$ the set of all operators $M \in \LinOp{\Hilb{A_0} \otimes \Hilb{A_1}}$ which are the CJ representation of an admissible transformation; these are the $M$ such that for all quantum state $\rho_{A_0}\in \mathscr{A}_0 \subset \LinOp{\Hilb{A_0}}$, there exists a state $\sigma_{A_1} \in \mathscr{A}_1 \subset \LinOp{\Hilb{A_1}}$ so that
\begin{equation}
    \left[\TrX{A_0}{M\cdot\left(\rho_{A_0} \otimes \mathds{1}\right)}\right]^T = \sigma_{A_1} \:.
\end{equation}
By definition, the operators $M$'s are (representing) quantum channels, so by Proposition \ref{theo:det_map}, they must satisfy
\begin{subequations} \label{eq:state_qu_to_state_qu}
\begin{gather}
    M \geq 0 \:,\\
    \TrX{}{M} = d_{A_0} \;,\\
    \left( \mathcal{I}_{A_0} \rightarrow \mathcal{I}_{A_1}\right)\{M\} = M \:,
\end{gather}
\end{subequations}
to be valid; this is depicted as step 2 in Fig. \ref{fig:1stTo2ndOrder}. 
To see that these conditions indeed single out the quantum channels in CJ form, remark that the last two conditions can be combined in the more common quantum `1-comb' condition, $\TrX{A_1}{M} = \mathds{1}_{A_0}$ (see Ref. \cite{Chiribella2009,MPM}; see also the discussion in Sec. \ref{sec:applications_iso}). It is obtained by explicitly writing the projector
\begin{equation}
    \begin{gathered}
         \left( \mathcal{I}_{A_0} \rightarrow \mathcal{I}_{A_1}\right)\{M\} = \overline{ \mathcal{I}_{A_0} \otimes \mathcal{D}_{A_1}}\{M\}\\ = \left(\mathcal{I}_{A_0} \otimes \mathcal{I}_{A_1} - \mathcal{I}_{A_0} \otimes \mathcal{D}_{A_1} + \mathcal{D}_{A_0} \otimes \mathcal{D}_{A_1} \right)\{M\}    \:,
    \end{gathered}
\end{equation}
and noticing that
\begin{equation}
    \left(\mathcal{D}_{A_0} \otimes \mathcal{D}_{A_1}\right)\{M\} = \frac{\TrX{A_0A_1}{M}}{d_{A_0}d_{A_1}}\:\mathds{1}_{A_0A_1} \:.
\end{equation}

Seeing these as the effects of a higher-order theory, the corresponding set of states is obtained through Prop. \ref{theo:det_fctal} (notice we are actually working the proof of Prop. \ref{theo:det_map} backward; see App. \ref{sec:det_map_proof}):
\begin{subequations}
\begin{gather}
    W \geq 0 \:,\\
    \TrX{}{W} = d_{A_1} \;,\\
    \overline{\mathcal{I}_{A_0} \rightarrow \mathcal{I}_{A_1}}\{W\} = W \:.
\end{gather}
\end{subequations}
This is depicted as step 3 in Fig. \ref{fig:1stTo2ndOrder}; this, by definition, is the set of single-partite process matrices, as it is the set of functionals normalized on quantum instruments. Indeed, any resolution $\{M_i\}$ of a channel leads to a probability rule of the form of Eq. \eqref{eq:Born_gen}:
\begin{equation}
    p(i|W,M) = \InProd{M_i}{W}\:.
\end{equation}
The set of valid $W$'s is characterized by projector $\overline{\mathcal{I}_{A_0} \rightarrow \mathcal{I}_{A_1}} = \overline{\overline{\mathcal{I}_{A_0} \otimes \overline{\mathcal{I}}_{A_1}}} = \mathcal{I}_{A_0} \otimes \overline{\mathcal{I}}_{A_1}$. Thus it corresponds to state structure $\Alg{A}{0} \otimes \CompAlg{A}_1$, characterized by $\mathcal{I}_{A_0} \otimes \mathcal{D}_{A_1}$. By Definition \ref{theo:tensor}, it is supposedly in the affine span of the objects of the form $\rho_{A_0} \otimes \mathds{1}_{A_1}$. 
However, the remarkable thing here is that, since there is only one element in $\CompAlg{A}_{1}= \{\mathds{1}_{A_1}\}$, there is no need to take the affine span as $\Alg{A}{0}$ is convex closed: in Fig. \ref{fig:1stTo2ndOrder} it means that the pure tensor products of the pink and green objects are exactly the yellow one. In other words, this recovers the known result that a single-partite process matrix reduces to inputting a quantum state and tracing out the output (see the supplementary material of Ref. \cite{OCB2012}, alternatively see Ref. \cite{morimae2014process}).

\paragraph{Remark: the no signaling subset of a quantum channel.\label{sec:examples_single_channel_remNS}}From the above considerations, it can be guessed that the set of elements of the form $M = \mathds{1}_{A_0} \otimes \sigma_{A_1}^T$, which is $\CompAlg{A}_0 \otimes \Alg{A}{1}$, is a subset of the valid quantum channels. And indeed they are the `trace-and-replace' channels such that any input state $\rho_{A_0}$ is traced out and replaced by $\sigma_{A_1}$:
\begin{multline}
    \left(\TrX{A_0}{\left(\mathds{1}_{A_0} \otimes \sigma_{A_1}^T\right)\cdot\left(\rho_{A_0} \otimes \mathds{1}_{A_1}\right)}\right)^T = \\
    \TrX{A_0}{\rho_{A_0}}\;\left(\sigma_{A_1}^T\right)^T = \sigma_{A_1} \:.
\end{multline}
In terms of state structures this is the inclusion $\CompAlg{A}_0 \otimes \Alg{A}{1} \subseteq \mathscr{A}_0 \rightarrow \mathscr{A}_1 = \overline{\mathscr{A}_0 \otimes \overline{\mathscr{A}_1}}$ so it is again Eq. \eqref{eq:tensor_in_par} at play: the transformation can be seen as another composition of $\CompAlg{A}_0$ with $\Alg{A}{1}$ that is larger than the one provided by their tensor product, $\CompAlg{A}_0 \otimes \Alg{A}{1}$. 
But here, compared to the bipartite quantum theory, the inclusion is strict, $\CompAlg{A}_0 \otimes \Alg{A}{1} \subsetneq \mathscr{A}_0 \rightarrow \mathscr{A}_1 = \overline{\mathscr{A}_0 \otimes \overline{\mathscr{A}_1}}$; the `trace-and-replace' are indeed but a specific kind of channels.

Actually, a special property of the subset $\CompAlg{A}_0 \otimes \Alg{A}{1}$ is that neither side of the channel can have a deterministic influence over the other. To see it, treat $\CompAlg{A}_0$ and $\mathscr{A}_1$ as if they were under the control of different parties so the `trace' part is controlled by $A_0$ and the `replace' part by $A_1$. The resolutions $\{M_i\}$ of $M$ now depend on 2 classical variables: $i\mapsto (j,k)$ such that $j$ is known only by $A_0$ and $k$ only by $A_1$. Any such `bipartite' resolution $\{M_{j,k}\}$ has well-defined `single-partite marginals':
\begin{equation}\label{eq:app_M_jk}
    \begin{aligned}
        \sum_j M_{j,k} &= \mathds{1} \otimes \sigma_k \:, \\
        \sum_k M_{j,k} &= E_j \otimes \sigma^T \:, \\
        \sum_{j,k} M_{j,k} &= \mathds{1} \otimes \sigma^T \:,
    \end{aligned}
\end{equation}
in the sense that $\{E_j\}$ is a resolution of $\mathds{1}$ and $\{\sigma_k\}$ is one of $\sigma^T$. Indeed, if such partial sums were not possible to obtain, the interpretation as local parties acting probabilistically would fail; one could not interpret $\sum_j M_{j,k}$ (respectively, $\sum_k M_{j,k}$) as a resolution where the action of $A_0$ (resp., $A_1$) on her side of the shared channel is deterministic. Such `bipartite' resolution $M_{j,k}$ of a `trace-and-replace' map has interpretation as a `measure-and-reprepare' scenario; the variable $j$ is associated with the measurement outcome seen by $A_0$ whereas $k$ conditions the choice of repreparation made by $A_1$. 

The probability rule becomes $p(j,k|W,M) = \InProd{M_{j,k}}{W}$ and can be simplified using the fact that $W$ factorises into $\rho \otimes \mathds{1}$, 
\begin{equation}
    p(j,k|W,M) = \TrX{}{M_{j,k}\:\left(\rho \otimes \mathds{1}\right)} \:.
\end{equation}
Talking about one side's `deterministic influence' over the other requires summing over her classical variable. The probabilistic interpretation of such summation is precisely a conditional distribution. For example, after $A_0$ has acted on her side the `reduced' distribution as seen on $A_1$ side is $\sum_j p(j,k|W,M) = p(k|j;W,M)$; it is conditioned by `quantum variables' $W,M$ (as well as how $M$ has been resolved into $\{M_{j,k}\}$) and by the actual value of the classical variable $j$ that $A_0$ has observed. 
However, for an $M \in \CompAlg{A}_0 \otimes \Alg{A}{1} \subset \mathscr{A}_0 \rightarrow \mathscr{A}_1$, the conditional distributions seen by each side drastically simplify because of Eqs. \eqref{eq:app_M_jk},
\begin{equation}
    \begin{aligned}
        \sum_j p(j,k|W,M) &= \TrX{}{\sigma_k} \:;\\
        \sum_k p(j,k|W,M) &= \TrX{}{E_j \: \rho} \:.
    \end{aligned}
\end{equation}
Consequently, the classical variables are independent of each other, $p(k|j;W,M) = p(k|W,M)$ and vice-versa. In other words, the knowledge of $j$ cannot help determine $k$, and vice-versa. In such a channel no information can be deterministically passed from the input side at $A_0$ to the output side at $A_1$; the subset $\CompAlg{A}_0 \otimes \Alg{A}{1}$ is dubbed the `no signaling' subset of $\mathscr{A}_0 \rightarrow \mathscr{A}_1$ for this reason. This notion of `no signaling' as the independence of conditional probabilities is generalized to any state structure in Sec. \ref{sec:Projos_NS_lattice} under Definition \ref{def:NS_FS_sets}, and it is generally proven to correspond to the tensor composition of state structures in Proposition \ref{prop:causal_trans}.

\paragraph{Remark: channels are no signaling from output to input.\label{sec:examples_single_channels_remprec}}Another occurrence of no signaling in state structures is the property of quantum channels to be automatically no signaling from the output to the input. 
So far we considered $M$ to represent the action of the parties and $W$ to be the environment, in the sense that $M$ was being resolved and $W$ was deterministic and fixed. Reversing the picture, the transformation $M$ in $\overline{\mathscr{A}_0 \otimes \overline{\mathscr{A}_1}}$ can be seen as a channel shared by two parties $A_0$ and $A_1$ trying to communicate through it. The description of their joint intervention is thus an element of $\mathscr{A}_0 \otimes \overline{\mathscr{A}_1}$ (or a resolution of one such element if they are acting probabilistically). 
As in the above remark, the probabilistic dependency of $W$ can be split between two variables $j$ and $k$, each known only by one of the parties. $A_0$ prepares a state according to $j$, inputs it through the channel, and then $A_1$ measures its output. In that case, however, party $A_1$'s action is to perform a POVM, but it means she can apply only one deterministic (or unit) effect at the output of a quantum channel (i.e. $\mathds{1}_{A_1}$). It is impossible for her to deterministically influence the input of the channel by her action on the output as there are no two different deterministic effects $N$ and $N'$ such that $\TrX{A_1}{M\cdot\left(\mathds{1} \otimes N\right)} \neq \TrX{A_1}{M\cdot\left(\mathds{1} \otimes N'\right)}$; no matter what $A_1$ does deterministically, $A_0$ always sees the same thing on her end of the channel. This absence of deterministic influence is then to be interpreted as the no signaling from the output to the input condition \eqref{eq:nA_succ_B} defined in App. \ref{app:NS=QO}. 
This property is not true for transformations in general (again, it follows as a consequence of Lemma \ref{lem:accidental}): if $M$ was not a quantum channel but, say, a transformation between process matrices, the party at the output could have deterministic influence over what the input sees. This observation is a motivation for the introduction of the one-way signaling composition of two state structures in Sec. \ref{sec:NS} as well as the discussion about why the quantum comb formalism is ``accidentally'' one-way signaling in Sec. \ref{sec:applications_iso}. 

\subsection{Bipartite quantum channel theory and the no signaling channels\label{sec:examples_bipartite_channel}}
This example is one of the simplest that mixes the tensor and the transformation, i.e. Definition \ref{prop:tensor} and Proposition \ref{theo:det_map}. A bipartite quantum channel is a CPTP map $\mathcal{M}$ between bipartite quantum states. For $\mathcal{M}: \LinOp{\Hilb{A_0}\otimes \Hilb{B_0}} \rightarrow \LinOp{\Hilb{A_1}\otimes \Hilb{B_1}}$, the set of all such channels corresponds to the state structure $\left(\Alg{A}{0}\otimes\Alg{B}{0}\right) \rightarrow \left(\Alg{A}{1}\otimes\Alg{B}{1}\right)$. 
We will see in this example that the ordering of parentheses and connectors $\otimes,\rightarrow$ is important, and that this semantic difference underlies the difference between the bipartite channels and their no signaling subset. Like the previous subsection, this example also hints at how the notion of no signaling is generalized in Sec. \ref{sec:NS} and especially how it relates to the notion of no signaling at the level of the outcomes probability distribution that is developed in App. \ref{app:NS=QO}. 

\paragraph{Bipartite quantum channels\label{sec:examples_bipartite_channel_char}}
A bipartite channel is example \ref{sec:examples_single_channel} in the bipartite case, refer to the top right of Fig. \ref{fig:bipartite_PM_construction} for an illustration. 
The deterministic probability rule reads:
\begin{equation}
    \InProd{\mathds{1}_{A_1}\otimes \mathds{1}_{B_1}}{\MOn{\rho_{A_0B_0}}} = 1\:,
\end{equation}
which holds for $\mathds{1}_{A_1}\otimes \mathds{1}_{B_1} \in \overline{\Alg{A}{1} \otimes \Alg{B}{1}} $ and for all $\rho_{A_0B_0}\in \Alg{A}{0} \otimes \Alg{B}{0}$. In the CJ picture, this is 
\begin{equation}
    \InProd{M}{\rho_{A_0B_0} \otimes \mathds{1}_{A_1}\otimes \mathds{1}_{B_1}} = 1\:,
\end{equation}
with $M \in \LinOp{\Hilb{A_0}\otimes\Hilb{B_0} \otimes \Hilb{A_1}\otimes \Hilb{B_1}}$, leading to the characterization of valid $M$'s as:
\begin{subequations}
\begin{gather}
    M \geq 0 \:,\\
    \TrX{}{M} = d_{A_0}d_{B_0} \;,\\
    \left( \left( \mathcal{I}_{A_0}\otimes\mathcal{I}_{B_0}\right) \rightarrow \left( \mathcal{I}_{A_1} \otimes \mathcal{I}_{B_1}\right)\right)\{M\} = M \:.
\end{gather}
\end{subequations}
The projector is $\left( \mathcal{I}_{A_0}\otimes\mathcal{I}_{B_0}\right) \rightarrow \left( \mathcal{I}_{A_1} \otimes \mathcal{I}_{B_1}\right) = \mathcal{I}_{A_0}\otimes\mathcal{I}_{B_0}\otimes \mathcal{I}_{A_1} \otimes \mathcal{I}_{B_1} -  \mathcal{I}_{A_0}\otimes\mathcal{I}_{B_0} \otimes \overline{\mathcal{I}_{A_1} \otimes \mathcal{I}_{B_1}} +  \mathcal{D}_{A_0}\otimes\mathcal{D}_{B_0}\otimes \mathcal{D}_{A_1} \otimes \mathcal{D}_{B_1} $; 
It can be further simplified by noticing that $\overline{\mathcal{I}_{A_1} \otimes \mathcal{I}_{B_1}} = \mathcal{D}_{A_1} \otimes \mathcal{D}_{B_1}$. This simplification then gives condition $ M - \TrX{A_1B_1}{M} \otimes \mathds{1}_{A_1}\otimes \mathds{1}_{B_1} +  \mathds{1}_{A_0}\otimes \mathds{1}_{B_0} \otimes \mathds{1}_{A_1}\otimes \mathds{1}_{B_1}= M$ which, once the trace is fixed, is equivalent to the usual condition $\TrX{A_1B_1}{M} = \mathds{1}_{A_0}\otimes \mathds{1}_{B_0}$. One indeed recovers the same condition as for the single-partite case but applied on two subsystems.

\paragraph{Remark: Bipartite channels from the composition of two channels.\label{sec:examples_bipartite_channel_remcommut}} We defined the state structure of bipartite channels as the set of transformations between bipartite states, one can nonetheless wonder if it can be defined as the composition of two state structures of single-partite channels instead. In other words, can we define the state structure $\left(\Alg{A}{0}\otimes\Alg{B}{0}\right) \rightarrow \left(\Alg{A}{1}\otimes\Alg{B}{1}\right)$ by combining state structures $\Alg{A}{0}\rightarrow \Alg{A}{1}$ and $\Alg{B}{0} \rightarrow \Alg{B}{1}$? 
The tensor product appears as a good candidate to do so. However, it will be shown next that it only yields the subset of no signaling bipartite quantum channels. As it is discussed in Sec. \ref{sec:Projos_NS_lattice} under Eq. \eqref{eq:tensor_in_par}, there exists a larger composition than the tensor, based on the non-commutation of the negation with the tensor. It amounts to taking the negation of the tensor product of the negations or, equivalently, to taking the transformation with a negated input. In symbols, this last sentence amounts to the following inclusion:
\begin{equation}
    \Alg{A}{} \otimes \Alg{B}{} \subseteq \overline{\CompAlg{A} \otimes \CompAlg{B}} = \CompAlg{A} \rightarrow \Alg{B}{} = \Alg{A}{} \leftarrow \CompAlg{B} \:.
\end{equation}
So is the right-hand side the proper composition? It turns out that it is for $\Alg{A}{} = \Alg{A}{0}\rightarrow \Alg{A}{1}$ and $\Alg{B}{} = \Alg{B}{0} \rightarrow \Alg{B}{1}$ being the state structures of quantum channels. This can be shown by some projector algebra and, as for the case of bipartite quantum states above, this crucially relies on the fact that the identity projector on two systems is the same as the tensor of two identities as in Eq. \eqref{eq:composite_I}. The projector associated to state structure $\overline{\overline{\Alg{A}{0}\rightarrow \Alg{A}{1}} \otimes \overline{\Alg{B}{0} \rightarrow \Alg{B}{1}}}$ can undergo the following rewriting:
\begin{multline}
    \overline{\overline{\mathcal{I}_{A_0} \rightarrow \mathcal{I}_{A_1}} \otimes  \overline{\mathcal{I}_{B_0} \rightarrow \mathcal{I}_{B_1}}} \\
    = \overline{(\mathcal{I}_{A_0} \otimes \mathcal{D}_{A_1}) \otimes (\mathcal{I}_{B_0} \otimes \mathcal{D}_{B_1})} \\
    = \overline{\mathcal{I}_{A_0} \otimes \mathcal{I}_{B_0} \otimes \mathcal{D}_{A_1} \otimes \mathcal{D}_{B_1}}\\
    \overset{\eqref{eq:composite_D}}{=} \overline{\mathcal{I}_{A_0} \otimes \mathcal{I}_{B_0} \otimes \overline{\mathcal{I}_{A_1} \otimes \mathcal{I}_{B_1}}} \\
    = (\mathcal{I}_{A_0} \otimes \mathcal{I}_{B_0}) \rightarrow (\mathcal{I}_{A_1} \otimes \mathcal{I}_{B_1})\:,
\end{multline}
and this last line is indeed the projector associated with bipartite quantum channels.

The fact that the bipartite quantum channels are equivalent to the composition of single-partite channels is not trivial. This only happens in a handful of cases in which the base state structures are quantum, i.e. when the projectors associated with subsystems are the identity or depolarizing. (The cases are inferred from Eq. \eqref{eq:Iso_par_tensor} which is again a consequence of Lemma \ref{lem:accidental} presented in Sec. \eqref{sec:applications_iso}.) In general, the following inclusion is tight, 
\begin{multline}\label{eq:app_equiv_proj_bipartite_chan}
    (\Alg{A}{0}\rightarrow \Alg{A}{1}) \otimes (\Alg{B}{0} \rightarrow \Alg{B}{1}) \\
    \subseteq \left(\Alg{A}{0}\otimes\Alg{B}{0}\right) \rightarrow \left(\Alg{A}{1}\otimes\Alg{B}{1}\right)  \\
     \subseteq \overline{\overline{(\Alg{A}{0}\rightarrow \Alg{A}{1})} \otimes \overline{(\Alg{B}{0} \rightarrow \Alg{B}{1})}} \:,
\end{multline}
so a bipartite transformation is a genuinely different state structure than a composition of transformations.

\paragraph{No signaling bipartite channels.\label{sec:examples_no_signaling_channel}}
A concrete example of the above inclusions is obtained by comparing (the state structures of) the tensor product of quantum channels to a bipartite quantum channel, i.e. $\left(\Alg{A}{0}\otimes\Alg{B}{0}\right) \rightarrow \left(\Alg{A}{1}\otimes\Alg{B}{1}\right)$ compared to $\left(\Alg{A}{0}\rightarrow\Alg{A}{1}\right) \otimes \left(\Alg{B}{0}\rightarrow\Alg{B}{1}\right)$ in the case of base state structures of quantum states. 
The first case, as we have shown, corresponds to the set of bipartite channels, whereas the second is the set of no signaling \cite{Piani2006} (also known as causal \cite{Beckman2001}) channels as we will now show. 
The tightness of the inclusion can be proven with a dash of the algebra of projectors: using Eq. \eqref{eq:app_equiv_proj_bipartite_chan}, the bipartite channels are associated with projector $\overline{\left(\mathcal{I}_{A_0} \otimes \mathcal{D}_{A_1}\right) \otimes \left(\mathcal{I}_{B_0} \otimes \mathcal{D}_{B_1}\right)}$ whereas $\left(\Alg{A}{0}\rightarrow\Alg{A}{1}\right) \otimes \left(\Alg{B}{0}\rightarrow\Alg{B}{1}\right)$ is associated to $\overline{\mathcal{I}_{A_0} \otimes \mathcal{D}_{A_1}} \otimes \overline{\mathcal{I}_{B_0} \otimes \mathcal{D}_{B_1}}$ so Eqs. \eqref{eq:notAnotB<notAB} and \eqref{eq:Iso_par_tensor} can be used: 
\begin{multline}
    \overline{\left(\mathcal{I}_{A_0} \otimes \mathcal{D}_{A_1}\right) \otimes \left(\mathcal{I}_{B_0} \otimes \mathcal{D}_{B_1}\right)} \\
    \supset \overline{\mathcal{I}_{A_0} \otimes \mathcal{D}_{A_1}} \otimes \overline{\mathcal{I}_{B_0} \otimes \mathcal{D}_{B_1}} \:.
\end{multline}

Next, we prove that it is indeed the set of no signaling channels, that is, the subset of channels forbidding deterministic signaling from Alice's side to Bob's side and vice-versa. We already mentioned in Remark \ref{sec:examples_single_channels_remprec} that channels are no signaling from the output to input and that local operation on a channel is a single-partite PM, corresponding to a state preparation followed by a discard. Putting these two pieces of information together, what is to be shown here is that the choice of input state on Alice's side cannot induce a deterministic influence on Bob's measurement and vice-versa. In equation, Alice and Bob share a bipartite channel $M$, and each acts locally by state preparation and measurement. Call $x$ and $y$ the settings of respectively Alice and Bob, and $a$ and $b$ their outcomes. That is, depending on some input $x$, Alice prepares some state $\rho_{|x}$ that she inputs in her side of the channel, $A_0$. Then, she measures at the output $A_1$ some POVM $\{E_{a|x}\}$ which choice can also depend on $x$, and sees outcome $a$. Bob does the same at $B_0$ with state $\sigma_{|y}$ at $B_0$ and POVM  $\{F_{b|y}\}$ at $B_1$. The probability rule then reads
\begin{equation}
    p(a,b|x,y)=\TrX{}{M \cdot \left(\rho_{|x} \otimes \sigma_{|y}  \otimes E_{a|x}^T \otimes F_{b|y}^T \right)}\:.
\end{equation}
(We omit reference to the `quantum variables' $M$, $\rho$, $\sigma$, etc. in the probability distribution to lessen clutter, as we are only interested in the correlations between the classical variables $a,b,x,$ and $y$.) Alice is no signaling to Bob if her choice of setting $x$ cannot deterministically influence his measurement outcome, that is (see e.g., \cite{Piani2006})
\begin{equation}\label{eq:NS_chan_1}
    \forall x,x', \: \sum_a p(a,b|x,y)=\sum_a p(a,b|x',y)\:.
\end{equation}
The same way Bob is no signaling if
\begin{equation}\label{eq:NS_chan_2}
    \forall y,y', \: \sum_b p(a,b|x,y)=\sum_b p(a,b|x,y')\:.
\end{equation}
In terms of the channel, these two conditions can be shown to be equivalent to (see \cite{Chiribella2009}):
\begin{subequations}
    \begin{gather}
        \TrX{A_1}{M} = \mathds{1}_{A_0} \otimes \TrX{A_0A_1}{M} \:;\label{eq:2-comb_cd_1}\\
        \TrX{B_1}{M} = \mathds{1}_{B_0} \otimes \TrX{B_0B_1}{M}\:.\label{eq:2-comb_cd_2}
    \end{gather}
\end{subequations}
which form a stronger constraint than the channel condition $\TrX{A_1B_1}{M} = \mathds{1}_{A_0} \otimes \mathds{1}_{B_0}$. 

These two conditions are exactly those encoded in the projector of the state structure $\left(\Alg{A}{0}\rightarrow\Alg{A}{1}\right) \otimes \left(\Alg{B}{0}\rightarrow\Alg{B}{1}\right)$, which reads $\left(\mathcal{I}_{A_0} \rightarrow \mathcal{I}_{A_1}\right) \otimes \left(\mathcal{I}_{B_0}\rightarrow \mathcal{I}_{B_1}\right) = \overline{\mathcal{I}_{A_0} \otimes \mathcal{D}_{A_1}} \otimes \overline{\mathcal{I}_{B_0} \otimes \mathcal{D}_{B_1}}$. Indeed, notice that this is the composition of two projectors acting on different subspaces (and thus commuting): \begin{multline}
    \left(\mathcal{I}_{A_0} \rightarrow \mathcal{I}_{A_1}\right) \otimes \left(\mathcal{I}_{B_0}\rightarrow \mathcal{I}_{B_1}\right) = \\
    \left(\left(\mathcal{I}_{A_0} \rightarrow \mathcal{I}_{A_1}\right) \otimes \mathcal{I}_{B_0}\otimes \mathcal{I}_{B_1}\right) \cap \\
    \left(\mathcal{I}_{A_0} \otimes  \mathcal{I}_{A_1} \otimes \left(\mathcal{I}_{B_0}\rightarrow \mathcal{I}_{B_1}\right)\right)\:,
\end{multline}
Where the $\cap$ symbol represents a composition of commuting projectors as defined in Sec. \ref{sec:Projos_alg} (see App. \ref{sec:projos_prop} for a review of its properties). Hence, each of them enforces a condition independent of the other. 
The first one, $\left(\overline{\mathcal{I}_{A_0} \otimes \mathcal{D}_{A_1}} \otimes \mathcal{I}_{B_0} \otimes \mathcal{I}_{B_1} \right)\{M\} = M$, is explicitly 
\begin{equation}
    M - \mathds{1}_{A_1} \otimes \TrX{A_1}{M} + \mathds{1}_{A_0} \otimes \mathds{1}_{A_1} \otimes \TrX{A_0A_1}{M} = M \:,
\end{equation}
which is equivalent to condition \eqref{eq:2-comb_cd_1}, and thus to \eqref{eq:NS_chan_1}. The same way, $\left(\mathcal{I}_{A_0} \otimes  \mathcal{I}_{A_1} \otimes \overline{\mathcal{I}_{B_0}\rightarrow \mathcal{D}_{B_1}}\right)$ can be shown equivalent to condition \eqref{eq:2-comb_cd_2} and thus to \eqref{eq:NS_chan_2}.

This example of no signaling bipartite channel illustrates how Definition \ref{prop:tensor} is exactly a composition of two state structures done in a manner that forbids a deterministic influence of each state structure over the other. In the CJ picture, this definition tells us that the set of these channels obeys
\begin{subequations}
    \begin{gather}
        M \geq 0\:;\\
        \TrX{}{M} = d_{A_0}d_{B_0}\:;\\
        \left(\left(\mathcal{I}_{A_0} \rightarrow \mathcal{I}_{A_1}\right) \otimes \left(\mathcal{I}_{B_0}\rightarrow \mathcal{I}_{B_1}\right)\right)\{M\} = M\:;
    \end{gather}
\end{subequations}
or, equivalently, is the set of trace-normalized positive operators in the affine hull of tensor products of single-partite channels, i.e.,
\begin{equation}\label{eq:aff_sum_channels}
    M = \sum_i q_i \: M^A_i \otimes M^B_i \:,
\end{equation}
where $M\geq 0$, $q_i \in \mathbb{R}$, $\sum_i q_i = 1$, each $M^A_i \in \LinOp{\Hilb{A_0}\otimes\Hilb{A_1}}$ obeys single quantum channel conditions \eqref{eq:state_qu_to_state_qu}, and so does each $M^B_i \in \LinOp{\Hilb{B_0}\otimes\Hilb{B_1}}$.
Generalizing this notion of no signaling to general state structures is the purpose of Sec. \ref{sec:NS} in which the tensor product of state structure is effectively proven to be no signaling composition, see Corollary \ref{theo:tensor}.

\paragraph{Remark: separability and localizability.} Notice that the projective methods are linear constraints, so while they can give information about the deterministic signaling structure, nothing can be said about separability or localizability. 
In the context of a no signaling quantum channel, separability will be the analog of separable quantum states: the subset of channels whose decomposition like Eq. \eqref{eq:aff_sum_channels} is a convex sum, for which each $q_i \in [0,1]$. 
Localizability on the other hand is the possibility of obtaining the channel as local operations applied on a joined ancillary entangled state \cite{Beckman2001}. It is known that separable and localizable channels are subsets of the set of no signaling channels. Yet, the characterization of these sets is much more involved than working out the subspaces they span, as their constraints are non-linear in the operators. 
Generalizing separability and localizability to higher-order quantum transformations is left open as a future research direction.

\subsection{Bipartite no signaling channel and process matrix formalism. \label{sec:examples_bipartite_PM}}
\begin{figure}[htb]
    \centering
    \includegraphics[width=\linewidth]{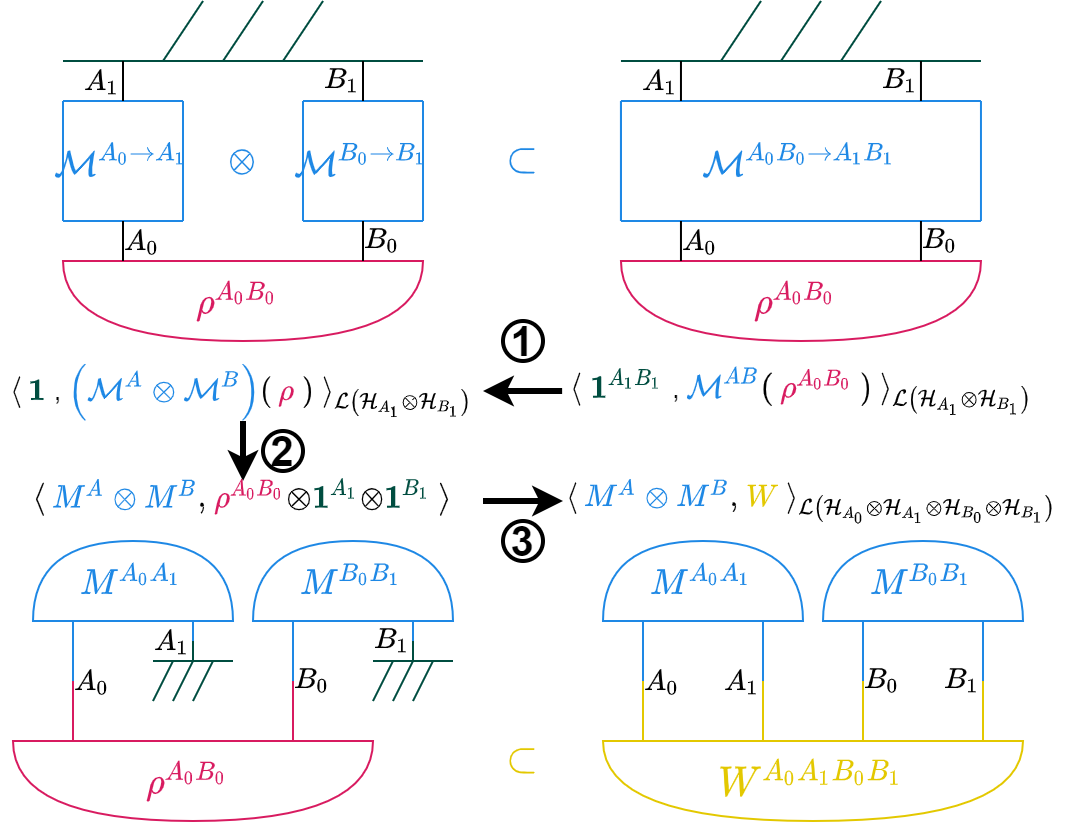}
    \caption{Gleason-like construction of a bipartite PM: starting from a bipartite channel normalized on a state (in blue and pink, top right), 1) Restrict the set of channels to local operations of Alice and Bob subset (in blue, top left); 2) Go to the CJ picture so that the channels are interpreted as two effects composed in tensor (in blue, bottom left); 3) Extend the allowed higher-order `state' to the most general operator normalized on a tensor product of channels, the bipartite PM (in yellow, bottom right)}
    \label{fig:bipartite_PM_construction}
\end{figure}

This last example mixes all three operations on state structures (i.e., $\otimes,\rightarrow,$ and $\overline{\:\cdot\:}$) in order to recover the bipartite process matrix framework from the four base state structures of quantum states. With respect to the characterization of no signaling channels in the previous section, the bipartite process matrices are nothing short of the functionals normalized on this set \cite{OCB2012,MPM}. We will do this construction step by step, relating all the other structures characterized in this appendix.

The construction is similar to that of Sec. \ref{sec:examples_single_channel} but assumes the induced dynamics to be bipartite to start with. Refer to Fig. \ref{fig:bipartite_PM_construction}: Starting with the regular bipartite quantum theory as in Sec. \ref{sec:examples_bipartite_QT} (top right corner, with states in pink and effects in green), one first considers no signaling at the level of effects to argue for bipartite states, $\rho \in \mathscr{A}_0\otimes\mathscr{B}_0$ (as we have seen, this is automatically the case in the case of quantum theory).   
Then, the theory is extended by allowing quantum channels (top right corner, the set of channels $\mathcal{M}$ is in blue) as in Sec. \ref{sec:examples_bipartite_channel} so that these maps correspond to the state structure $\left(\Alg{A}{0}\otimes\Alg{B}{0}\right)\rightarrow \left(\Alg{A}{1}\otimes\Alg{B}{1}\right)$. 

Step 1 (going to the top left situation of Fig. \ref{fig:bipartite_PM_construction}) consists of restricting the set of valid bipartite channels to the set of no signaling channels, so as to guarantee that parties $A$ and $B$ are causally disconnected. In other words, one considers another requirement of no signaling, at the level of the transformation this time. The heuristic of which is straightforward: if there are local effects, the dynamics may be constrained to local operations as well because e.g. Alice and Bob's labs are space-like separated. Another way to put it: Alice and Bob are each allowed to do any physical transformation on their qudit before measuring, but no global transformation is allowed. Both points of view conclude that the set of channels is restricted to its local subset spanned by the tensor product of local transformations. 
Step 2 (going to the bottom left situation of Fig. \ref{fig:bipartite_PM_construction}) is to pass to the CJ picture so to obtain their characterization. The operations of Alice and Bob have the form $M^{A_0A_1} \otimes M^{B_0B_1} \in \left(\Alg{A}{0}\rightarrow\Alg{A}{1}\right)\otimes \left(\Alg{B}{0}\rightarrow\Alg{B}{1}\right) \subset \left(\Alg{A}{0}\otimes\Alg{B}{0}\right)\rightarrow \left(\Alg{A}{1}\otimes\Alg{B}{1}\right)$ (or, more generally, the affine span of these pure tensors, see Sec. \ref{sec:examples_no_signaling_channel}). 
Finally, step 3 (going to the bottom right situation of Fig. \ref{fig:bipartite_PM_construction}) is to extend the set of states and effects to all functionals as in Proposition \ref{theo:det_fctal}. By doing so, we have inferred from a few heuristics what is the set of bipartite process matrices $\{W\}$, and that its state structure is $\overline{\left(\Alg{A}{0}\rightarrow\Alg{A}{1}\right)\otimes \left(\Alg{B}{0}\rightarrow\Alg{B}{1}\right)}$, which characterization directly ensues (this is done explicitly in Ref. \cite{MPM}).

Abstractly, this way of building a higher-order theory consists in inducing a no signaling composition of state structures (a tensor product), then normalizing a functional on it (a negation). We will nickname this kind of construction a `Gleason-like' construction as opposed to a `dynamics-like' construction that involves defining dynamics of dynamics (nested transformations; i.e., supermaps). The dynamics-like way of defining higher-order transformations is considered in Sec. \ref{sec:examples_dynamics_constr}. This choice of nomenclature is made for distinguishing between two lines of works that resulted in the same kind of higher-order formalisms featuring ICO, respectively Refs. \cite{OCB2012,Araujo2015,Shrapnel2018} and Refs. \cite{Chiribella2008,Chiribella2009,Chiribella2013,Perinotti2016,Bisio2018}. In a sense, the Gleason-like way is the special case of the dynamics-like way obtained when the output is the trivial system.

\section{An example with a different base projector -- Biased Quantum Theory\label{sec:examples_biased_QT}} 
One novelty of our approach compared to the type theory is that the base types --generalized into base state structures-- do not have to be the set of quantum states. For example, in a state structure like $\Alg{A}{}\rightarrow\Alg{B}{}$ the base state structures $\Alg{A}{}$ and $\Alg{B}{}$ can be any state structure. To compare the resulting process theories, one can consider various sets of possible base state structures for a fixed number of subsystems associated with a fixed number of parties. Say Alice's system is known to be bipartite and quantum, then her possible base state structures can be the elementary one, $\Alg{A}{}= \Alg{A}{0} \otimes \Alg{A}{1}$, or a transformation $\Alg{A}{}= \Alg{A}{0} \rightarrow \Alg{A}{1}$. The same reasoning applies to Bob.

Nonetheless, the set of base state structures in consideration does not actually have to be restricted to only those that can be built from the quantum states (corresponding to the identity projector) and the operations in the algebra. The only constraint implied by the projective characterization is that the set of base projectors associated with a single party has to contain the depolarizing and the identity projectors and that every projector in the set should commute with the other ones. So what about projectors that cannot be built from the algebra but that commute with the other elements?

As mentioned below Proposition \ref{prop:lattice} in the main text, a physically motivated example is the set of diagonal operators with respect to a fixed basis. These correspond to the dephased states obtained after a measurement in this basis, hence to classical systems to an extent. The projector $\Delta$ to this subspace can be checked to commute with $\mathcal{I}$ and $\mathcal{D}$, so that $\mathcal{I} \supset \Delta \supset \mathcal{D}$, and to obey the other properties of a projector on an operator system \eqref{eq:projo}. The set of base state structures for a single-partite system in this case is characterized by $\{\Delta, \overline{\Delta}\}$ instead of $\{\mathcal{I}, \mathcal{D}\}$. Studying higher-order transformations mixing base state structures of classical and quantum theory is left open as a future research direction.

The purpose of this appendix is instead to present another, less physically relevant, projector that cannot be built from $\mathcal{I}$ and the operations of the algebra, so as to study the properties of the multipartite state structures built out of it. This toy model has different signaling properties than transformations built out of the quantum states as we will show, providing some motivation and insights for Sec. \ref{sec:NS}.

For simplicity let us take a 2-dimensional underlying Hilbert space $\Hilb{A}$. Let $\Proj{}{}$ be the projector that restricts the basis elements to $\{\mathds{1}, \sigma_x, \sigma_y\}$ in the Pauli basis so that we can define $\mathscr{A} = \Span{\mathds{1}, \sigma_x, \sigma_y}$ associated with projector $\Proj{}{A}$ and, by Proposition \ref{theo:det_fctal}, $\overline{\mathscr{A}} = \Span{\mathds{1}, \sigma_z}$ is associated with $\CompProj{}{A}$. Compare it to quantum theory in 2 dimensions: $\mathscr{A}_{quant.}=\Span{\mathds{1}, \sigma_x, \sigma_y, \sigma_z}$ associated with $\mathcal{I}_{A}$ and $\CompAlg{A}_{quant.} = \{\mathds{1}\}$ with $\mathcal{D}_{A}$, it is direct to check that $\Proj{}{A}$ is a projector on operator system and that $\mathcal{I}_{A} \subset \Proj{}{A} \subset \mathcal{D}_{A}$. 

The theory characterized by $\Proj{}{A}$ is a \textit{biased} quantum theory in the sense that the measurement can be `biased' towards an arbitrary element of $\overline{\mathscr{A}}$. For an element $V\in \Alg{A}{}$ and one $N\in \CompAlg{A}$, the probability rule reads 
\begin{equation}
    p(i|V,N) = \InProd{N_i}{V} \:,
\end{equation}
and it can be interpreted in the picture of regular quantum mechanics as a regular measurement followed by a projection onto the state $N$.

Thus, some deterministic postselection of the measurement in a given basis --here in the computational basis-- is allowed in this kind of theory. Nevertheless, the theory by itself does not allow for the usual counter-logical behaviors encountered in theories with postselection because it is inherently constraining the allowed states into a basis that is quasi-orthogonal to the one of the postselection. Another way to picture it is that the theory allows for a postselection but in a basis that is by construction mutually unbiased with respect to the basis of the state: in the example, Alice can in general postselect in any state of the form $\mathds{1}+p\;\sigma_z$, where $p$ real and $p^2 \leq 1$ because of positivity. That is, she can choose to project the state into a mixture of projectors $\dyad{0}$ and $\dyad{1}$, but the state is itself built by superposing $\ket{\pm}\equiv \frac{1}{\sqrt{2}}(\ket{0} \pm \ket{1})$ and $\ket{\pm i}\equiv \frac{1}{\sqrt{2}}(\ket{0} \pm i \ket{1})$, the two mutually unbiased bases w.r.t. the computational one. Therefore, her postselection cannot be used to distinguish which state has been prepared from the maximally mixed state $\mathds{1}/2$. That is, she cannot deterministically distinguish states ensembles by choosing a suitable $p \neq 0$, which explains how the theory is operationally equivalent (sometimes called `tomographically indistinguishable') to a qubit quantum theory with one `forbidden' axis of the Bloch sphere.
In other words, it is also guaranteed to be a perfectly well-behaved theory, in the sense that no matter the choice of $V$, $N$ and its resolution $\{N_i\}$, no incoherence like overnormalized or negative probabilities can be found in the outcomes.

This one-party example is of course trivial, but as soon as more than one party is allowed, the possibility of deterministically sending a signal from one party to another will coincide with this kind of allowed postselection. This fact can be compared to the motivating example, where the choice of \textit{which} quantum channel the elements of an instrument sum up to also amounts to deterministically inducing a bias in the outcome probability. 

\subsection{Bipartite biased quantum states\label{sec:examples_bipartite_biaisedQT}}
In the biased case, where $\overline{\mathscr{A}} = \overline{\mathscr{B}} = \Span{\mathds{1},\sigma_z}$ the difference between a bipartite no signaling state and a general bipartite one becomes striking. 
The set of all valid states normalized on the local effects resolving $\overline{\mathscr{A}} \otimes \overline{\mathscr{B}}$ is $\overline{\overline{\mathscr{A}} \otimes \overline{\mathscr{B}}}$ which, according to Proposition \ref{theo:det_fctal} is made of the following 13 basis elements:
\begin{multline}
    \overline{\overline{\mathscr{A}} \otimes \overline{\mathscr{B}}} \subset \mathrm{Span}\Big\{ \mathds{1}^A \otimes \mathds{1}^B, \mathds{1}^A \otimes \GGB{B}{x}, \mathds{1}^A \otimes \GGB{B}{y}, \\
    \GGB{A}{x} \otimes \mathds{1}^B, \GGB{A}{x} \otimes \GGB{B}{x}, \GGB{A}{x} \otimes \GGB{B}{y}, \GGB{A}{x} \otimes \GGB{B}{z}, \\
    \GGB{A}{y} \otimes \mathds{1}^B, \GGB{A}{y} \otimes \GGB{B}{x}, \GGB{A}{y} \otimes \GGB{B}{y}, \GGB{A}{y} \otimes \GGB{B}{z}, \\
    \GGB{A}{z} \otimes \GGB{B}{x}, \GGB{A}{z} \otimes \GGB{B}{y} \Big\} \:.
\end{multline}
On the other hand, the no signaling composition of $\Alg{A}{}$ and $\Alg{B}{}$, $\mathscr{A}\otimes \mathscr{B}$ as by Definition \ref{prop:tensor}, is only made of 9 elements (notice that the spanning set of $\mathscr{A}\otimes \mathscr{B}$ is indeed contained in the one of $\overline{\overline{\mathscr{A}} \otimes \overline{\mathscr{B}}}$, in accordance with Definition \ref{def:NS_FS_sets}). The four missing elements are
\begin{multline}
    \overline{\overline{\mathscr{A}} \otimes \overline{\mathscr{B}}} \backslash \mathscr{A}\otimes \mathscr{B} \subset \mathrm{Span}\Big\{ \GGB{A}{x} \otimes \GGB{B}{z}, \GGB{A}{y} \otimes \GGB{B}{z},\\
    \GGB{A}{z} \otimes \GGB{B}{x}, \GGB{A}{z} \otimes \GGB{B}{y} \Big\} \:,
\end{multline}
which are exactly the elements containing a $\sigma_z$ term that we will now show to allow signaling between $A$ and $B$. That the elements of $\overline{\overline{\mathscr{A}} \otimes \overline{\mathscr{B}}} \backslash \mathscr{A}\otimes \mathscr{B}$ allow for signaling is the content of Sec. \ref{sec:NS}; this result is proven in the general case in Lemma \ref{theo:causal_map}. Indeed, these basis elements are those which are quasi-orthogonal globally, meaning that an operator $W$ that contains some of them will satisfy $\forall N^A\in \CompAlg{A}$, $\forall N^B\in \CompAlg{B}$
\begin{multline}
    \TrX{}{\left(N^A \otimes N^B\right) \cdot W} \\
    = \frac{1}{d_Ad_B}\TrX{}{\left(N^A \otimes N^B\right)}\TrX{}{W} \:,
\end{multline}
because it satisfies Proposition \eqref{theo:det_fctal} and therefore Eq. \eqref{eq:QO}. But it does not obey quasi-orthogonality with respect to a local measurement, meaning that it will fail to satisfy the analog of Eq. \eqref{eq:partTrAB=TrATrB} in at least one of Eqs. \eqref{eq:A_otimes_B_def} like e.g., $\TrX{A}{\left(N^A \otimes \mathds{1}\right)\cdot W } \neq 1/d_A\TrX{A}{N^A}\TrX{A}{W}$.

The consequence of this observation is that if Alice and Bob share a bipartite no signaling state, they may observe nonlocal entanglement effects on their outcome distributions, but these correlations will obey no signaling constraints in both directions, Eqs. \eqref{eq:no_sign_corr}. 
If, however, they share a general bipartite state, they may use it to achieve deterministic signaling: for certain states, they will be able to signal perfectly in one direction, i.e. Alice can perfectly send a message to Bob, and vice-versa.

For instance, consider the task where Alice receives a classical bit $x$ and she wants to communicate it to Bob, so that his outcome $b$ has the same value, $b=x$. Without any resources, Bob can only guess and thus succeed with p(b=x)=1/2. Now if we allow them to measure a shared state in $\overline{\overline{\mathscr{A}} \otimes \overline{\mathscr{B}}}$, they can pick the following state:
\begin{equation}\label{eq:biased_ancilla}
    W_{A\prec B} = \frac{1}{4}\left( \mathds{1}^A \otimes \mathds{1}^B + \GGB{A}{z} \otimes \GGB{B}{x} \right)\:,
\end{equation}
(the notation in the subscript means `Alice can signal to Bob') and choose to do the following: Alice `steers' her measurement towards $\ket{0}$ or $\ket{1}$ depending on $x$,
\begin{equation}
    N^A_{|x} = \mathds{1}^A + (-1)^x \: \GGB{A}{z} \:,
\end{equation}
while Bob is measuring an unbiased $N^B=\mathds{1}$ resolved into a measurement in the $\ket{\pm}$ basis,
\begin{equation}
    N^B_{b} = \frac{1}{2} \left(\mathds{1}^B + (-1)^b \: \GGB{B}{x}\right)\:,
\end{equation}
where $b={0,1}$ so that his probabilistic effects satisfy $N^B _0 + N^B_1= N^B = \mathds{1}^B$. One can check that they are effectively properly normalized positive operators belonging to the proper state structures, $N^A_{|x} \in \Alg{A}{}$, $N^B \in \Alg{B}{}$, despite $N^B _0,N^B_1\notin \Alg{B}{}$. The measurement yields the following probability distribution:
\begin{equation}
    p(b|x) = \TrX{}{\left(N^A_{|x} \otimes N^B_b\right) \cdot W_{A\prec B} } \:.
\end{equation}
Injecting the above expressions into it yields
\begin{equation}
    p(b|x) = \frac{1}{2}\left(1 + (-1)^{x+b}\right) \:,
\end{equation}
which gives 0 when $x\neq b$ and 1 when $x=b$; Alice's setting is perfectly correlated to Bob's outcome. We conclude that a bit was perfectly sent from $A$ to $B$, or that using the $W_{A\prec B}$ state as their resource, they obtained $p(x=b)=1$.

The bottom line of this example is that there exist states in $\overline{\overline{\mathscr{A}} \otimes \overline{\mathscr{B}}}$ which allow beating one of the no signaling constraints, Eqs. \eqref{eq:no_sign_corr}, with a probability of 1. In this regard, and with respect to every theory that can be characterized by the projective methods of this article, quantum theory plays a special role as it is the only one whose bipartite states normalized on no signaling effects are automatically no signaling. This fact is the deeper meaning of the inclusion $\overline{\CompProj{}{A} \otimes \CompProj{}{B}} \supseteq \Proj{}{A}\otimes \Proj{}{B}$ of Eq. \eqref{eq:tensor_in_par} becoming an equality in the case $\Proj{}{A} = \Proj{}{B} = \mathcal{I}$ (Eq. \eqref{eq:Iso_par_tensor}) which follows from Lemma \ref{lem:accidental}). 

Comparing the biased to the regular quantum theory, the bipartite quantum states do saturate the inclusion (see example \ref{sec:examples_bipartite_QT}), which is why they are no signaling compared to bipartite biased states. Contrastingly, bipartite quantum channels (see example \ref{sec:examples_bipartite_channel}), when seen as the composition of two channels, do not saturate it (\ref{sec:examples_bipartite_channel_remcommut}). This is why information can be passed from one side of bipartite channels to the other, and why the no signaling channels are a proper subset of them (\ref{sec:examples_no_signaling_channel}).

\subsection{Biased quantum channels\label{sec:examples_biasedQT_channel}}
The dynamics of the biased quantum theory can be postulated as the maps from the biased theory to itself. Proposition \ref{theo:det_map} characterizes these as the set $\mathscr{A}\rightarrow\mathscr{B}$ where the state structure of the input and output are the same, $\mathscr{A}=\mathscr{B}=\Span{\mathds{1},\sigma_x,\sigma_y}$:
\begin{multline}
    \mathscr{A} \rightarrow \mathscr{B} \subset \mathrm{Span}\Big\{ \mathds{1}^A \otimes \mathds{1}^B, \mathds{1}^A \otimes \GGB{B}{x}, \mathds{1}^A \otimes \GGB{B}{y}, \\
    \GGB{A}{x} \otimes \GGB{B}{x}, \GGB{A}{x} \otimes \GGB{B}{y}, \GGB{A}{y} \otimes \GGB{B}{x},\GGB{A}{y} \otimes \GGB{B}{y}, \\
    \GGB{A}{z} \otimes \mathds{1}^B, \GGB{A}{z} \otimes \GGB{B}{x}, \GGB{A}{z} \otimes \GGB{B}{y}, \GGB{A}{z} \otimes \GGB{B}{z} \Big\} \:.
\end{multline}

Here, the fact proven in Sec. \ref{sec:NS_def} that a transformation between state structures is a composition that does not forbid signaling in any direction can be proven by an example. Suppose Alice and Bob are sharing the channel $M_{A \prec B}=\frac{1}{2}\left(\mathds{1}^A\otimes\mathds{1}^B+ \sigma_x^A\otimes \sigma_x^B\right)$. Alice can perfectly signal to Bob by encoding her setting $x$ in the $\sigma_x$ basis, $V_{|x}=1/2(\mathds{1}+(-1)^x\sigma_x)$, and if Bob measures in the same basis, they effectively have a perfect single bit channel, $p(b=x)=1$. 

On the other hand, suppose they share the channel $M_{A\succ B}=\frac{1}{2}\left(\mathds{1}^A\otimes\mathds{1}^B+ \sigma_z^A \otimes \sigma_z^B\right)$. 
Now it is Bob who can perfectly signal to Alice: Alice has to use an ancilla so that she can prepare the same joint state as the bipartite example, Eq. \eqref{eq:biased_ancilla},
\begin{equation}
    W_{A\prec A'}= \frac{1}{4}\left(\mathds{1}^{A}\otimes \mathds{1}^{A'} + \GGB{A}{z}\otimes \GGB{A'}{x}\right)\:.
\end{equation}
She sends the $A$ part through the channel and keeps the $A'$ part in her ancilla. Bob can then apply the measurement $N^B_{|y}=\mathds{1}+(-1)^y\GGB{}{z}$ at the outcome of the channel, depending on the variable $y$ he wishes to send. 
Alice can then finally measure her ancilla in the $\GGB{}{x}$ basis, $N^{A'}_{a}=1/2(\mathds{1}+(-1)^a\sigma_x)$, leading to the distribution
\begin{multline}
    p(a|y)\\
    = \TrX{}{\left(N^B_{|y} \otimes N_a^{A'}\right)\cdot\TrX{A}{M_{A\succ B}\cdot W_{A\prec A'}}^T}\:.
\end{multline}
and she will get perfect correlation with Bob setting, $p(a=y)=1$, exactly like in the bipartite state example. 
Therefore, in the state structure $\mathscr{A}\rightarrow\mathscr{B}$, there are elements allowing perfect signaling from Alice to Bob as well as from Bob to Alice. 
Contrasingly, the quantum channels (seen as a composition of an input and an output using the transformation) only allow one direction of signaling, from the input (Alice) to the output (Bob; see Remark \ref{sec:examples_single_channels_remprec}).

\section{Introductory example extended: dynamics-like construction up to fourth order\label{sec:examples_dynamics_constr}}
For this example, we will focus on the case of higher-order transformations built upon quantum mechanics. In a nutshell, we will go back to the construction that lifts the POVM formalism into the quantum instrument formalism that we presented at the beginning of Sec. \ref{sec:Proj_char}, we will abstract it, and we will repeat it until an indefinite causal order (ICO) arises. The reason we do that is two-fold: on the one hand, we want to present how the characterization techniques work in a concrete case (as we will see, all the objects that will be defined have already been studied in the literature). On the other hand, and in accordance with the discussion of Sec. \ref{sec:applications_iso}, we want to stress how tame quantum theory is with respect to other theories that can be characterized using projective means: albeit we will be using the transformation operation $\rightarrow$ repetitively, i.e. we are nesting operations in a way that allows for bidirectional signaling, it will not result in a transformation with ICO before the fourth order. This is in stark contrast with the example of the biased theory of App. \ref{sec:examples_biased_QT} in which one could observe signaling in both directions already at the level of bipartite states as well as of channels.

\subsection{Reformulating the introductory example: the quantum channel is a 1-comb\label{sec:examples_dynamics_constr_1}}
Reformulating the introductory example of Sec. \ref{sec:Proj_char}, we start with a state structure $\mathscr{A}_0\subset \LinOp{\Hilb{A_0}}$ with projector $\mathcal{I}_{A_0}$ and a trace of 1, as in Eqs. \eqref{eq:state_char} or \eqref{eq:state_qu}. Defining the functionals on this state structure, we consider state and effect pairs in complementary (i.e. quasi-orthogonal) state structures, $(\rho,\mathds{1}) \in \mathscr{A}_0 \times \overline{\mathscr{A}_0}$, linked by the normalization of the probability rule
\begin{equation}
    1 = \InProd{\mathds{1}}{\rho}_{\LinOp{\Hilb{A_0}}} \equiv \TrX{}{\mathds{1}\cdot\rho} \:.
\end{equation}
These pairs represent all deterministic preparation and measurement procedures, i.e. those yielding a probability of 1 irrespective of the choice of state. These pairs characterize the `first-order theory' (because these are the states and effects on which can be defined first-order transformations, which in turn will be the effects of the second-order theory). 
Probabilistic assignments are obtained by resolving the effect state structure by a collection of positive operators $\{E_i\}$, 
\begin{equation}
    p(i|\rho,\mathds{1}) = \TrX{}{E_i \cdot \rho} \;,
\end{equation}
yielding the POVM formalism and the usual definition of effects.

To go to the second-order theory, one postulates some dynamics so that the state structure $\mathscr{A}_0$ is mapped to a similarly defined state structure $\mathscr{A}_1$ by first-order transformations, i.e. by some CPTP map $\mathcal{M} \in \LinOpB{\LinOp{\Hilb{A_0}}}{\LinOp{\Hilb{A_1}}}$ so that
\begin{equation}
    \InProd{\mathds{1}}{\rho}_{\LinOp{\Hilb{A_0}}} \mapsto \InProd{\mathds{1}}{\MapOn{A_0\rightarrow A_1}{}{\rho}}_{\LinOp{\Hilb{A_1}}} \:.
\end{equation}
As it was shown in example \ref{sec:examples_single_channel}, in the CJ picture, this results in a new pair $\left(\rho^{A_0} \otimes \mathds{1}^{A_1}, M^{A_0A_1} \right) \in \left(\mathscr{A}_0 \otimes \overline{\mathscr{A}_1}\right) \times \left(\mathscr{A}_0 \rightarrow \mathscr{A}_1\right)$, so that
\begin{equation}
    1 = \TrX{}{M \cdot \left(\rho \otimes \mathds{1}^T\right)}^T = \TrX{}{M \cdot \left(\rho \otimes \mathds{1}\right)} \:.
\end{equation}
Here, $M$ is a quantum channel in CJ representation, which by definition is a 1-comb. 
Resolving state structure $\left(\mathscr{A}_0 \rightarrow \mathscr{A}_1\right)$ yields the quantum instrument formalism,
\begin{equation}
    p(j|M, \rho) = \InProd{M_j}{\rho \otimes \mathds{1}} \:.
\end{equation}

In this example, the dynamics-like construction of a second-order theory is completed by requiring that the set of states be extended to all functionals normalized on the effects. That is, to all $W$ in $\overline{\mathscr{A}_0 \rightarrow \mathscr{A}_1}$ so that the normalization becomes
\begin{equation}
    \InProd{M}{W}_{\LinOp{\Hilb{A_0}\otimes \Hilb{A_1}}} = 1 \:,
\end{equation}
where $W$ is a single-partite PM and $M$ is a 1-comb, so that both are second-order objects. 

Nonetheless, we noted at the end of example \ref{sec:examples_single_channel} that single-partite PM trivially decomposes into states and measurements. Meaning that the states $W$ of the second-order naturally descend to the first-order: $\overline{\mathscr{A}_0 \rightarrow \mathscr{A}_1} = \mathscr{A}_0 \otimes \CompAlg{A}_1$. The explanation in terms of the projector algebra is simply a matter of definition $\Proj{}{A_0}\rightarrow\Proj{}{A_1} = \overline{\Proj{}{A_0}\otimes \CompProj{}{A_1}}$. In terms of signaling structure, because the transformation is two-way signaling, $\Proj{}{A_0}\rightarrow\Proj{}{A_1} = \left(\CompProj{}{A_0}\prec\Proj{}{A_1}\right) \cup \left(\Proj{}{A_0}\succ\Proj{}{A_1}\right)$, the functional on transformations have to be no signaling by De Morgan duality \eqref{eq:deMorgan_mult}: $\overline{\Proj{}{A_0}\rightarrow\Proj{}{A_1}} = \overline{\CompProj{}{A_0}\prec\Proj{}{A_1}} \cap \overline{\Proj{}{A_0}\succ\Proj{}{A_1}} = \left(\Proj{}{A_0}\prec\CompProj{}{A_1}\right) \cap \left(\Proj{}{A_0}\succ\CompProj{}{A_1}\right) \overset{\eqref{eq:relations_1}}{=} \Proj{}{A_0}\otimes \CompProj{}{A_1} $.

Another difference we noticed in Remark \ref{sec:examples_single_channels_remprec} is that the 1-comb was one-way signaling despite being built by transformation. Now we can see this fact quickly from isomorphism \eqref{eq:Iso_semi-causal_transfo} (Eq. \eqref{eq:isomorphisms_transfo} in the main text): observe that the 1-comb is characterized by
$\mathcal{I}_{A_0} \rightarrow \mathcal{I}_{A_1}$ which by Lemma \ref{lem:accidental} is equivalent to
\begin{equation}
    \mathcal{I}_{A_0} \rightarrow \mathcal{I}_{A_1} = \overline{\mathcal{I}}_{A_0} \prec \mathcal{I}_{A_1} \:.
\end{equation}
Hence, we are guaranteed that the 1-combs have a fixed signaling direction, so they cannot show indefinite causal order. If the base state structure was anything else than quantum states, we would not have witnessed such a simplification; the 1-comb projector would have decomposed into two different orderings: $\mathcal{P}_{A_0} \rightarrow \mathcal{P}_{A_1} = \left(\overline{\mathcal{P}}_{A_0} \prec \mathcal{P}_{A_1}\right) \cup \left(\overline{\mathcal{P}}_{A_0} \succ \mathcal{P}_{A_1}\right)$.

\subsection{Dynamics-like construction}
As preparation and measurement are special transformations with respectively, input and output trivial systems, we can also promote the destructive measurement at $A_1$ into a measurement performed by a second party, so that the probability rule becomes
\begin{equation}\label{eq:1-comb_ij}
    p(i,j|M, \rho) = \InProd{M_j}{\rho \otimes E_i^T}\:,
\end{equation}
this is still one-way signaling as it amounts to having a second party making a POVM measurement $\{E_i\}$ (which is a special case of a quantum instrument) after the quantum instrument $\{M_j\}$.

\paragraph*{Remark.} To better stick to our conventions, we could have passed all probabilistic assignments to the left side of the inner product:
\begin{equation}
    p(i,j|M, \rho) = \InProd{E_i \ast M_j}{\rho}_{\LinOp{\Hilb{A_0}}}\:,
\end{equation}
where $E_i \ast M_j \equiv \TrX{A_1}{\left(\mathds{1}^{A_0}\otimes E_i^T\right)\cdot M_j}$ is the link product of the two operators \cite{Chiribella2009}, which corresponds to the CJ representation of the sequential composition of these two instruments over subsystem $A_1$.

This construction of the quantum instrument formalism has been abstracted under the name `dynamics-like construction' at the end of App. \ref{sec:examples_bipartite_PM}. For the sake of the argument, however, we will take a slightly different point of view in the following: the introduced transformations will be considered to be the states of the higher-order theory instead of the effects, and the effects will not be extended to be the functionals on the new states. In other words, the introduced dynamics will be considered to be the deterministic objects over which the parties have no influence (the state), and everything else will be locally controlled by some parties acting probabilistically (the effects). For the case of Eq. \eqref{eq:1-comb_ij}, this amounts to taking the marginal over $j$,  $p(i|M, \rho) = \InProd{M}{\rho \otimes E_i^T}$. The parties $A_0$ and $A_1$ then share a channel in which $A_0$ inputs a state, and of which $A_1$ measures the output. We do so because this will maximize the number of local parties, which makes ICO more obvious: it would be difficult to talk about the ICO \textit{within} the whole operator representing the shared dynamics (i.e. the global functional on the state), whereas it can be more easily discussed by looking at the correlations achievable by a group of parties sharing this global object. In this regard, this amounts to taking the dynamics as the state and everything else as local effects in tensor product. 

Therefore, given states in $\Alg{A}{} \in \LinOp{\Hilb{A_0}}$ and effects in $\CompAlg{A}$, the dynamics-like construction now amounts to defining bipartite states in $\Alg{A}{}\rightarrow\Alg{B}{} \in \LinOp{\Hilb{A_0}\otimes\Hilb{A_1}}$ with $\Hilb{A_1}\cong\Hilb{A_0}$. Associated with these shared states are two local effects $\Alg{A}{0}$ and $\CompAlg{A}_1$, each under the control of an independent party. I.e., the dynamics-like construction of the next order consists of the following redefinition of the states and effects pair:
\begin{equation}
    \Alg{A}{0} \times \CompAlg{A}_0 \mapsto \left(\Alg{A}{0}\rightarrow\Alg{A}{1} \right) \times \left(\overline{\Alg{A}{0}\rightarrow\Alg{A}{1}}\right)\:.
\end{equation}
\subsection{The quantum supermap is a 2-comb\label{sec:examples_dynamics_constr_2}}
The third-order theory is obtained by assuming the existence of dynamics over the current dynamics, i.e. second-order transformations. In the same way that structure-preserving maps can be nicknamed `superoperator', we are introducing a `supermap' $\mathcal{N}$ as a linear map between two maps with the same state structure \cite{Chiribella2008}. If the map $\mathcal{M}$ of the previous section is defined between subsystems $A_1$ and $A_2$, we introduce $\mathcal{N}$ as a CPTP-preserving supermap that send $\mathcal{M}$ to a similar map $\Tilde{\mathcal{M}}$ between subsystems $A_0$ and $A_3$, $\mathcal{N}\left(\mathcal{M}^{A_1\rightarrow A_2}\right) = \Tilde{\mathcal{M}}^{A_0 \rightarrow A_3}$, so that
\begin{equation}
    \InProd{\mathds{1}}{\MOn{\rho}}_{\LinOp{\Hilb{A_1}}} \mapsto \InProd{\mathds{1}}{\left[\mathcal{N}\left(\mathcal{M}\right)\right](\rho)}_{\LinOp{\Hilb{A_3}}} \:.
\end{equation}
Going to the CJ picture, the following probability rule is obtained: 
\begin{equation}
    p(i,j|N, M, \rho) = \InProd{N}{\rho \otimes M_j^{T_{A_2}} \otimes \overline{E}_i^T} \:.
\end{equation}
The theory is then characterized by quasi-orthogonal pairs of the form $\left(N, \rho \otimes M \otimes \mathds{1}\right) \in \left(\left( \mathscr{A}_1 \rightarrow \mathscr{A}_2 \right) \rightarrow \left( \mathscr{A}_0 \rightarrow \mathscr{A}_3\right)\right) \times \left( \mathscr{A}_0 \otimes \left(\mathscr{A}_1\rightarrow\mathscr{A}_2\right) \otimes \overline{\mathscr{A}_3}\right)$.

This way, $N$ is a 2-comb, an object that transforms 1-combs into 1-combs. We have already seen in the last section that the 1-combs themselves can be interpreted as one-way signaling objects, so that the signaling of the 2-comb can be made explicit by looking at its corresponding projector:
\begin{multline}
    \left( \mathcal{I}_{A_1} \rightarrow \mathcal{I}_{A_2}\right) \rightarrow \left( \mathcal{I}_{A_0} \rightarrow \mathcal{I}_{A_3} \right)\\
    \overset{\eqref{eq:Iso_semi-causal_transfo}}{=} \left( \overline{\mathcal{I}}_{A_1} \prec \mathcal{I}_{A_2}\right) \rightarrow \left( \overline{\mathcal{I}}_{A_0} \prec \mathcal{I}_{A_3} \right)\\
    \overset{\eqref{eq:transfo=2way}}{=} \left(\overline{\overline{\mathcal{I}}_{A_1} \prec \mathcal{I}_{A_2}}\right) \prec \left( \overline{\mathcal{I}}_{A_0} \prec \mathcal{I}_{A_3} \right) \\
    \cup \left(\overline{\overline{\mathcal{I}}_{A_1} \prec \mathcal{I}_{A_2}}\right) \succ \left( \overline{\mathcal{I}}_{A_0} \prec \mathcal{I}_{A_3} \right) \\
    = \mathcal{I}_{A_1} \prec \overline{\mathcal{I}}_{A_2} \prec \overline{\mathcal{I}}_{A_0} \prec \mathcal{I}_{A_3} \\
    \cup \overline{\mathcal{I}}_{A_0} \prec \mathcal{I}_{A_3} \prec \mathcal{I}_{A_1} \prec \overline{\mathcal{I}}_{A_2} \:.
\end{multline}
So it may be concluded that the most general 2-comb is a superposition of two possible quantum networks: one in which the channel $M$ is measured, $\mathcal{I}_{A_1} \prec \overline{\mathcal{I}}_{A_2}$, then reprepared as $\Tilde{M}$, $\overline{\mathcal{I}}_{A_0} \prec \mathcal{I}_{A_3}$  in its causal future and one where it is first $\Tilde{M}$ which is prepared.

Yet, there is again an isomorphism at play. This is the isomorphism presented in Sec. \ref{sec:applications_iso}: the $n-$combs are equivalent to quantum networks. This will make the interpretation of the supermap $N$, a second-order transformation, equivalent to a succession of first-order transformations with a single signaling direction.
\begin{multline} \label{eq:supermap=2-comb}
    \left( \mathcal{I}_{A_1} \rightarrow \mathcal{I}_{A_2}\right) \rightarrow \left( \mathcal{I}_{A_0} \rightarrow \mathcal{I}_{A_3} \right)\\
    \overset{\eqref{eq:uncurrying}}{=}\left( \mathcal{I}_{A_0} \otimes \left( \mathcal{I}_{A_1} \rightarrow \mathcal{I}_{A_2}\right)  \right) \rightarrow \mathcal{I}_{A_3} \\
    \overset{\eqref{eq:Iso_semi-causal_transfo}}{=}\overline{\left( \mathcal{I}_{A_0} \otimes \left( \overline{\mathcal{I}}_{A_1} \prec \mathcal{I}_{A_2}\right)  \right)} \prec \mathcal{I}_{A_3} \\
    \overset{\eqref{eq:Iso_semi-causal_causal}}{=}\overline{\left( \mathcal{I}_{A_0} \prec \left( \overline{\mathcal{I}}_{A_1} \prec \mathcal{I}_{A_2}\right)  \right)} \prec \mathcal{I}_{A_3} \\
    =\overline{\mathcal{I}}_{A_0} \prec \mathcal{I}_{A_1}\prec \overline{\mathcal{I}}_{A_2}  \prec \mathcal{I}_{A_3}\:.
\end{multline}
This means that $N$ can equivalently be understood as a succession of two quantum channels \cite{Chiribella2009,Perinotti2016,Bisio2018}; this is a peculiar feature of quantum supermaps, and of quantum combs in general.

This surprisingly tame behavior of having a single signaling direction when reduced to a succession of first-order objects will not be carried out in the fourth order, as we will see next. 
In addition, note that this explains why a single-partite PM does not feature ICO (or anything outside of the usual quantum theory). If the systems $A_0$ and $A_3$ were trivial, so that $N$ is defined as (the CJ of) a supermap that takes a quantum channel to a probability, we see that the associated projector $\overline{\mathcal{I}}_{A_0} \prec \mathcal{I}_{A_1}\prec \overline{\mathcal{I}}_{A_2}  \prec \mathcal{I}_{A_3} \mapsto 1 \prec \mathcal{I}_{A_1}\prec \overline{\mathcal{I}}_{A_2}  \prec 1$, and this can be simplified into $\mathcal{I}_{A_1}\prec \overline{\mathcal{I}}_{A_2} \overset{\eqref{eq:Iso_semi-causal_causal}}{=} \mathcal{I}_{A_1} \otimes \mathcal{D}_{A_2}$.

\subsection{The quantum super-supermap is an MPM\label{sec:examples_dynamics_constr_3}}
To go to the fourth order, the same procedure is applied again: we introduce a `super-supermap' $\mathcal{W}$ so that it maps supermaps like $\mathcal{N}$ to themselves, $\mathcal{W}(\mathcal{N}) = \Tilde{\mathcal{N}}$, where $\mathcal{N}\in\LinOp{\Hilb{A_1}\otimes \Hilb{A_2} \otimes \Hilb{A_5}\otimes \Hilb{A_6}}$ and $\Tilde{\mathcal{N}}\in\LinOp{\Hilb{A_0}\otimes \Hilb{A_3} \otimes \Hilb{A_4}\otimes \Hilb{A_7}}$ :
\begin{multline}
    \InProd{\mathds{1}}{\left[\mathcal{N}\left(\mathcal{M}\right)\right](\rho)}_{\LinOp{\Hilb{A_3}}} \\
    \mapsto \InProd{\mathds{1}}{\left[\left[\mathcal{W}\left(\mathcal{N}\right)\right]\left(\mathcal{M}\right)\right](\rho)}_{\LinOp{\Hilb{A_7}}}\:.
\end{multline}
In the CJ representation, the probability rule is
\begin{multline}
    p(i,j,k|W, M, N, \rho) =\\
    \InProd{W}{\rho \otimes N_j^{T_{A_2}T_{A_6}} \otimes M_j^{T_{A_4}} \otimes \overline{E}_i^{T_{A_7}}} \:.
\end{multline}
Here, $M$ is a transformation of $\rho$ defined in $\mathscr{A}_3 \rightarrow \mathscr{A}_4$; $N$ is a transformation of $M$, $N \in \left( \mathscr{A}_2 \rightarrow \mathscr{A}_5 \right) \rightarrow \left( \mathscr{A}_1 \rightarrow \mathscr{A}_6\right)$, and the $W$ operator is a transformation of $N$,
\begin{multline}
    W \in  \left[\left( \mathscr{A}_2 \rightarrow \mathscr{A}_5 \right) \rightarrow \left( \mathscr{A}_1 \rightarrow
    \mathscr{A}_6\right) \right]\\
    \rightarrow \left[\left( \mathscr{A}_3 \rightarrow \mathscr{A}_4 \right) \rightarrow \left( \mathscr{A}_0 \rightarrow \mathscr{A}_7\right) \right] \:.
\end{multline}
We can now treat $N$ as a party itself, which physically corresponds to someone acting a first time between nodes $A_1$ and $A_2$, and a second time between $A_5$ and $A_6$, using a quantum network.

Again, one can focus on the projector characterizing the state structure of $N$ to extract its signaling structure by putting it in normal form. First, successive applications of the uncurrying rule \eqref{eq:uncurrying} yield
\begin{multline}
    \left[\left( \mathcal{I}_{A_2} \rightarrow \mathcal{I}_{A_5} \right) \rightarrow \left( \mathcal{I}_{A_1} \rightarrow
    \mathcal{I}_{A_6}\right) \right]\\
    \rightarrow \left[\left( \mathcal{I}_{A_3} \rightarrow \mathcal{I}_{A_4} \right) \rightarrow \left( \mathcal{I}_{A_0} \rightarrow \mathcal{I}_{A_7}\right) \right] \\
    = \big[ \big( \mathcal{I}_{A_3} \rightarrow \mathcal{I}_{A_4} \big) \otimes  \big(\big( \mathcal{I}_{A_2} \rightarrow \mathcal{I}_{A_5} \big) \\
    \rightarrow \big( \mathcal{I}_{A_1} 
     \rightarrow \mathcal{I}_{A_6}\big)\big) \big] \rightarrow \left( \mathcal{I}_{A_0} \rightarrow \mathcal{I}_{A_7} \right)
    \\
    = \big[  \mathcal{I}_{A_0} \otimes \big(\left( \mathcal{I}_{A_3} \rightarrow \mathcal{I}_{A_4} \right) \\
    \otimes  \left(\left( \mathcal{I}_{A_2} \rightarrow \mathcal{I}_{A_5} \right) \rightarrow \left( \mathcal{I}_{A_1} \rightarrow
    \mathcal{I}_{A_6}\right)\right)\big) \big] \rightarrow \mathcal{I}_{A_7} \:.
\end{multline}
Next, Eqs. \eqref{eq:Iso_semi-causal_transfo} and \eqref{eq:supermap=2-comb}, \eqref{eq:Iso_semi-causal_causal}, \eqref{eq:Iso_semi-causal_transfo}, and \eqref{eq:semi-causal_negation} are used successively:
\begin{multline}
     \big[  \mathcal{I}_{A_0} \otimes \big(\left( \mathcal{I}_{A_3} \rightarrow \mathcal{I}_{A_4} \right) \\
    \otimes  \left(\left( \mathcal{I}_{A_2} \rightarrow \mathcal{I}_{A_5} \right) \rightarrow \left( \mathcal{I}_{A_1} \rightarrow
    \mathcal{I}_{A_6}\right)\right)\big) \big] \rightarrow \mathcal{I}_{A_7}\\
    = \Big[  \mathcal{I}_{A_0} \otimes \Big(\Big( \overline{\mathcal{I}}_{A_3} \prec \mathcal{I}_{A_4} \Big) \\
    \otimes  \Big(\overline{\mathcal{I}}_{A_1} \prec \mathcal{I}_{A_2} \prec \overline{\mathcal{I}}_{A_5} \prec \mathcal{I}_{A_6}\Big)\Big) \Big]\rightarrow \mathcal{I}_{A_7}\\
    \overset{\eqref{eq:Iso_semi-causal_causal}}{=} \Big[  \mathcal{I}_{A_0} \prec \Big(\Big( \overline{\mathcal{I}}_{A_3} \prec \mathcal{I}_{A_4} \Big)\\
     \otimes  \Big(\overline{\mathcal{I}}_{A_1} \prec \mathcal{I}_{A_2} \prec \overline{\mathcal{I}}_{A_5} \prec \mathcal{I}_{A_6}\Big)\Big) \Big] \rightarrow \mathcal{I}_{A_7} \overset{\eqref{eq:Iso_semi-causal_transfo}}{=} \\
     \overline{\mathcal{I}_{A_0} \prec \left(\left( \overline{\mathcal{I}}_{A_3} \prec \mathcal{I}_{A_4} \right) \otimes  \left(\overline{\mathcal{I}}_{A_1} \prec \mathcal{I}_{A_2} \prec \overline{\mathcal{I}}_{A_5} \prec \mathcal{I}_{A_6}\right)\right)} \\\prec \mathcal{I}_{A_7}\overset{\eqref{eq:semi-causal_negation}}{=} \\
    \overline{\mathcal{I}}_{A_0} \prec \overline{\left(\left( \overline{\mathcal{I}}_{A_3} \prec \mathcal{I}_{A_4} \right) \otimes  \left(\overline{\mathcal{I}}_{A_1} \prec \mathcal{I}_{A_2} \prec \overline{\mathcal{I}}_{A_5} \prec \mathcal{I}_{A_6}\right)\right)} \\\prec \mathcal{I}_{A_7}\:.
\end{multline}

If, for simplicity, we assume the subsystems $A_0$ and $A_7$ to be of dimension 1, we can already notice in the first line of the previous equation that the state structure of the operator $W$ features a projector whose expression involves the negation of a tensor product of two quantum combs:
\begin{equation}
    W \in  \overline{\left( \mathscr{A}_2\rightarrow \mathscr{A}_5 \right) \rightarrow  \left(\mathscr{A}_1 \rightarrow \mathscr{A}_6\right) \otimes \left( \mathscr{A}_3 \rightarrow \mathscr{A}_4 \right)} \:,
\end{equation}
i.e. it is the functional normalized on the tensor product between a 2-comb and 1-comb. This is by definition a multi-round process matrix (MPM). We proved in a previous work that it can manifest indefinite causal order and even more that it can beat causal inequalities \cite{MPM}. The fourth order, therefore, presents multiple directions of signaling in a non-trivial manner.

The fundamental reason why, as we showed in this previous article \cite{MPM}, is that this last expression defines a set which is the union of 3 sets of quantum combs with different signaling orderings,
\begin{multline}
    \overline{\left( \mathscr{A}_2\rightarrow \mathscr{A}_5 \right) \rightarrow  \left(\mathscr{A}_1 \rightarrow \mathscr{A}_6\right) \otimes \left( \mathscr{A}_3 \rightarrow \mathscr{A}_4 \right)}\\
    = \\
    \overline{\left( \left( \left( \left(\mathscr{A}_1 \rightarrow \mathscr{A}_2 \right) \rightarrow \mathscr{A}_5 \right) \rightarrow \mathscr{A}_6\right) \rightarrow  \mathscr{A}_3 \right) \rightarrow \mathscr{A}_4} \: \cup \\
    \overline{\left( \left( \left( \left(\mathscr{A}_1 \rightarrow \mathscr{A}_2 \right) \rightarrow \mathscr{A}_3 \right) \rightarrow \mathscr{A}_4\right) \rightarrow  \mathscr{A}_5 \right) \rightarrow \mathscr{A}_6} \: \cup \\
    \overline{\left( \left( \left( \left(\mathscr{A}_3 \rightarrow \mathscr{A}_4 \right) \rightarrow \mathscr{A}_1 \right) \rightarrow \mathscr{A}_2\right) \rightarrow  \mathscr{A}_5 \right) \rightarrow \mathscr{A}_6} \:.
\end{multline}
In other words, $W$ can be taken as a superposition of all combs that respect the signaling ordering in between the nodes of the 2-comb and the 1-comb that will be plugged into it, but it itself does not assume a global ordering. The above expression thus has 3 possible global orderings, depending on whether the operation of the party with the 1-comb (acting between nodes 3 and 4) is before, in between, or after the two operations of the party with the 2-comb.

Back to the general case of a super-supermap, we can use this insight to present a case of non-unique normal form. One can further simplify the projector to a normal form using Eq. \eqref{eq:relations_1}:
\begin{multline}
\overline{\mathcal{I}}_{A_0} \prec\\
 \overline{\left(\left( \overline{\mathcal{I}}_{A_3} \prec \mathcal{I}_{A_4} \right) \otimes  \left(\overline{\mathcal{I}}_{A_1} \prec \mathcal{I}_{A_2} \prec \overline{\mathcal{I}}_{A_5} \prec \mathcal{I}_{A_6}\right)\right)} \\\prec \mathcal{I}_{A_7}\\
= \overline{\mathcal{I}}_{A_0} \prec \\
 \overline{\big(\left(\overline{\mathcal{I}}_{A_1} \prec \mathcal{I}_{A_2} \prec \overline{\mathcal{I}}_{A_5} \prec \mathcal{I}_{A_6} \prec \overline{\mathcal{I}}_{A_3} \prec \mathcal{I}_{A_4}\right) \cap}\\ \overline{\left(\overline{\mathcal{I}}_{A_3} \prec \mathcal{I}_{A_4} \prec \overline{\mathcal{I}}_{A_1} \prec \mathcal{I}_{A_2} \prec \overline{\mathcal{I}}_{A_5} \prec \mathcal{I}_{A_6}\right)\big)}\\
\prec \mathcal{I}_{A_7}\\
\overset{\eqref{eq:deMorgan_cap_to_cup}}{=} \overline{\mathcal{I}}_{A_0} \prec\\
 \Big(\overline{\left(\overline{\mathcal{I}}_{A_1} \prec \mathcal{I}_{A_2} \prec \overline{\mathcal{I}}_{A_5} \prec \mathcal{I}_{A_6} \prec \overline{\mathcal{I}}_{A_3} \prec \mathcal{I}_{A_4}\right)} \cup\\ \overline{\left(\overline{\mathcal{I}}_{A_3} \prec \mathcal{I}_{A_4} \prec \overline{\mathcal{I}}_{A_1} \prec \mathcal{I}_{A_2} \prec \overline{\mathcal{I}}_{A_5} \prec \mathcal{I}_{A_6}\right)}\Big)\prec \mathcal{I}_{A_7}\\ 
\overset{\eqref{eq:prec_neg}}{=} \overline{\mathcal{I}}_{A_0} \prec\\ \big(\big(\mathcal{I}_{A_3} \prec \overline{\mathcal{I}}_{A_4} \prec \mathcal{I}_{A_1} \prec \overline{\mathcal{I}}_{A_2} \prec \mathcal{I}_{A_5} \prec \overline{\mathcal{I}}_{A_6}\big)  \cup\\ 
    \big(\mathcal{I}_{A_1} \prec \overline{\mathcal{I}}_{A_2} \prec \mathcal{I}_{A_5} \prec \overline{\mathcal{I}}_{A_6} \prec \mathcal{I}_{A_3} \prec \overline{\mathcal{I}}_{A_4} \big)\big) \prec \mathcal{I}_{A_7} \\
\overset{\eqref{eq:prec_cap_cup}}{=} \overline{\mathcal{I}}_{A_0} \prec 
\mathcal{I}_{A_3} \prec \overline{\mathcal{I}}_{A_4} \prec \mathcal{I}_{A_1} \prec \overline{\mathcal{I}}_{A_2}\\
 \prec \mathcal{I}_{A_5} \prec \overline{\mathcal{I}}_{A_6}  \prec \mathcal{I}_{A_7} \cup\\ 
    \overline{\mathcal{I}}_{A_0} \prec \mathcal{I}_{A_1} \prec \overline{\mathcal{I}}_{A_2} \\ \prec \mathcal{I}_{A_5} \prec \overline{\mathcal{I}}_{A_6} \prec \mathcal{I}_{A_3} \prec \overline{\mathcal{I}}_{A_4} \prec \mathcal{I}_{A_7} \:.
\end{multline}
Thus,
\begin{multline}\label{eq:NF_1}
    \left[\left( \mathcal{I}_{A_2} \rightarrow \mathcal{I}_{A_5} \right) \rightarrow \left( \mathcal{I}_{A_1} \rightarrow
    \mathcal{I}_{A_6}\right) \right]\\
    \rightarrow \left[\left( \mathcal{I}_{A_3} \rightarrow \mathcal{I}_{A_4} \right) \rightarrow \left( \mathcal{I}_{A_0} \rightarrow \mathcal{I}_{A_7}\right) \right] \\
    =\\
    \overline{\mathcal{I}}_{A_0} \prec \mathcal{I}_{A_3} \prec \overline{\mathcal{I}}_{A_4} \prec \mathcal{I}_{A_1} \\
    \prec \overline{\mathcal{I}}_{A_2} \prec \mathcal{I}_{A_5} \prec \overline{\mathcal{I}}_{A_6}  \prec \mathcal{I}_{A_7} \\ 
    \cup \overline{\mathcal{I}}_{A_0} \prec \mathcal{I}_{A_1} \prec \overline{\mathcal{I}}_{A_2} \prec \mathcal{I}_{A_5} \\
    \prec \overline{\mathcal{I}}_{A_6} \prec \mathcal{I}_{A_3} \prec \overline{\mathcal{I}}_{A_4} \prec \mathcal{I}_{A_7} \:.
\end{multline}

But one can also inject the previously discussed simplification for the case of a functional into the central part. Distributing over the union, one obtains a normal form as in Definition \ref{def:normal_form}:
\begin{multline}\label{eq:NF_2}
    \left[\left( \mathcal{I}_{A_2} \rightarrow \mathcal{I}_{A_5} \right) \rightarrow \left( \mathcal{I}_{A_1} \rightarrow
    \mathcal{I}_{A_6}\right) \right]\\
    \rightarrow \left[\left( \mathcal{I}_{A_3} \rightarrow \mathcal{I}_{A_4} \right) \rightarrow \left( \mathcal{I}_{A_0} \rightarrow \mathcal{I}_{A_7}\right) \right] \\
    =\\
    \overline{\mathcal{I}}_{A_0} \prec \mathcal{I}_{A_1} \prec \overline{\mathcal{I}}_{A_2} \prec \mathcal{I}_{A_3} \prec \overline{\mathcal{I}}_{A_4} \prec \mathcal{I}_{A_5} \\
    \prec \overline{\mathcal{I}}_{A_6}  \prec \mathcal{I}_{A_7} \\
    \cup \overline{\mathcal{I}}_{A_0} \prec \mathcal{I}_{A_3} \prec \overline{\mathcal{I}}_{A_4} \prec \mathcal{I}_{A_1} \prec \overline{\mathcal{I}}_{A_2} \prec \mathcal{I}_{A_5} \\
    \prec \overline{\mathcal{I}}_{A_6}  \prec \mathcal{I}_{A_7} \\ 
    \cup \overline{\mathcal{I}}_{A_0} \prec \mathcal{I}_{A_1} \prec \overline{\mathcal{I}}_{A_2} \prec \mathcal{I}_{A_5} \prec \overline{\mathcal{I}}_{A_6} \prec \mathcal{I}_{A_3} \\ \prec \overline{\mathcal{I}}_{A_4} \prec \mathcal{I}_{A_7} \:.
\end{multline}

The reason for these two normal forms to be valid comes from the transitivity of the prec, i.e. that it is associative, Eq. \eqref{eq:prec_assoc}. Call $\Proj{}{A} = \mathcal{I}_{A_3} \prec \overline{\mathcal{I}}_{A_4}$ $\Proj{}{B_0} = \mathcal{I}_{A_1} \prec \overline{\mathcal{I}}_{A_2}$, and $\Proj{}{B_1} = \mathcal{I}_{A_5} \prec \overline{\mathcal{I}}_{A_6}$. Then, in the expression
$\left(\Proj{}{A} \prec \Proj{}{B_0} \prec \Proj{}{B_1}\right) \cup \left(\Proj{}{B_0} \prec \Proj{}{B_1}\prec \Proj{}{A}\right)$ one sees that the state structure characterized by this projector allows for signaling from $A$ to $B_0$ and $B_1$ and at the same time from $B_0$ and $B_1$ to $A$. Obviously, if $A$ can be both before or after the two $B$'s, she can also be in the middle, in equation:
\begin{multline}
    \left(\Proj{}{A} \prec \Proj{}{B_0} \prec \Proj{}{B_1}\right) \cup \left(\Proj{}{B_0} \prec \Proj{}{B_1}\prec \Proj{}{A}\right) =\\
    \left(\Proj{}{A} \prec \Proj{}{B_0} \prec \Proj{}{B_1}\right) \cup \left(\Proj{}{B_0} \prec \Proj{}{B_1}\prec \Proj{}{A}\right) \\ \cup \left(\Proj{}{B_0}\prec \Proj{}{A} \prec \Proj{}{B_1}\right)\:.
\end{multline}
In other words, saying that $B_0$ can signal to $A$ which can signal to $B_1$ does not add new information when one knows that $A$ can signal to $B_0$ and $B_1$ and at the same time that $B_0$ and $B_1$ can signal to $A$. 
This implies by De Morgan duality that 
\begin{multline}
    \left(\Proj{}{A} \prec \Proj{}{B_0} \prec \Proj{}{B_1}\right) \cap \left(\Proj{}{B_0} \prec \Proj{}{B_1}\prec \Proj{}{A}\right) =\\
    \left(\Proj{}{A} \prec \Proj{}{B_0} \prec \Proj{}{B_1}\right) \cap \left(\Proj{}{B_0} \prec \Proj{}{B_1}\prec \Proj{}{A}\right) \\ \cap \left(\Proj{}{B_0}\prec \Proj{}{A} \prec \Proj{}{B_1}\right)\:,
\end{multline}
and thus that Eqs. \eqref{eq:NF_1} and \eqref{eq:NF_2} are equivalent.
Ergo, the normal form is multiply defined in that case because of redundant information in the description of the signaling structure.

Redundancy or not, the presence of a union in the normal form indeed confirms that more than one fixed signaling direction is allowed, hence that there are operators in that state structure that allow for indefinite causal ordering. What we have just done is to track the ICO origin as an ambiguity in the decomposition of higher-order into lower-order: using the algebraic rules of the projectors, we could express the projector associated with the fourth-order in terms of unions of first-order projectors (by that we mean that we expressed the state structure of super-supermap as unions of `prec chains' between quantum states and measurements). This decomposition does not forbid the existence of operators that belong to the affine hull of the union. That is, operators that are obtained from a superposition of several chains with different signaling directions. These are the ones that can show causal non-separability, like the operator studied in Ref. \cite{MPM} for example.

The bottom line of this example is that we had to go to the fourth order to observe signaling in more than one direction and an indefinite causal order. In a general process theory built using the dynamics-like construction, there should have been two directions already at the second order, possibly featuring ICO (see the biased quantum theory example), and $2^{n-1}$ directions at order $n$. This tameness is a concrete consequence of Lemma \ref{lem:accidental} and the ensuing Theorem \ref{theo:combs=networks}.
\end{document}